\documentclass[
reprint,
nobibnotes,
 aps,
 pra,
]{revtex4-2}

\usepackage{amsmath,amssymb,amsthm,amsfonts,amsbsy}
\usepackage{braket}
\usepackage{bbold}

\usepackage[colorlinks]{hyperref}
\hypersetup{
pdfstartview={FitH},
pdfnewwindow=true,
colorlinks=true,
linkcolor=blue,
citecolor=blue,
filecolor=blue,
urlcolor=blue}
\usepackage{cleveref}
\usepackage{acro}

\usepackage{graphicx}
\usepackage{tikz}
\usetikzlibrary{arrows.meta}
\usepackage{circuitikz}
\usetikzlibrary{patterns}
\usepackage{complexity}

\newtheorem{theorem}{Theorem}
\newtheorem{definition}[theorem]{Definition}

\newtheorem{corollary}[theorem]{Corollary}
\newtheorem{proposition}[theorem]{Proposition}
\newtheorem{lemma}[theorem]{Lemma}

\newtheorem{problem}[theorem]{Problem}

\DeclareAcronym{QC}{short=QC, long=Quantum Computer,}
\DeclareAcronym{QAE}{short=QAE, long=Quantum-Amplitude-Estimation,}
\DeclareAcronym{BE}{short=BE, long=Block Encoding,}
\DeclareAcronym{LCU}{short=LCU, long=Linear-Combination-of-Unitaries,}
\DeclareAcronym{HS}{short=HS, long=Hamiltonian Simulation,}
\DeclareAcronym{LCHS}{short=LCHS, long=Linear Combination of Hamiltonian Simulation,}
\DeclareAcronym{FFT}{short=FFT, long=Fast-Fourier-Transform,}
\DeclareAcronym{ODE}{short=ODE, long=Ordinary Differential Equation,}
\DeclareAcronym{DOF}{short=DOF, long=Degree of Freedom, long-plural-form=Degrees of Freedom}

\newcommand{\ketbra}[2]{|#1\rangle\!\langle#2|}
\newcommand{\diag}{\operatorname{diag}}
\newcommand{\one}{\mathbb{1}}
\newcommand{\lb}{\left(}
\newcommand{\rb}{\right)}
\newcommand\norm[1]{\left\lVert#1\right\rVert}
\newcommand{\bC}{\mathbf{C}}

\newcommand{\bG}{\mathbf{G}}
\newcommand{\bH}{\mathbf{H}}
\newcommand{\bJ}{\mathbf{J}}
\newcommand{\bK}{\mathbf{K}}
\newcommand{\boell}{\boldsymbol{\ell}}
\newcommand{\br}{\mathbf{r}}
\newcommand{\bS}{\mathbf{S}}
\newcommand{\bT}{\mathbf{T}}
\newcommand{\bv}{\mathbf{v}}
\newcommand{\bx}{\mathbf{x}}
\newcommand{\tbx}{\tilde{\mathbf{x}}}

\newcommand{\bzero}{\mathbf{0}}
\newcommand{\cO}{\mathcal{O}}

\newcommand{\nn}{\nonumber \\}
\newcommand{\cl}{\mathrm{cl}}
\newcommand{\balpha}{\boldsymbol{\alpha}}
\newcommand{\bPi}{\boldsymbol{\Pi}}
\newcommand{\rmid}{\mathrm{mid}}
\renewcommand{\poly}{\mathrm{poly}}

\begin{document}


\title{Quantum Simulation of Dissipative Non-Markovian Coupled Classical Oscillators}

\author{Malte Schade$^1$, Sophia Simon$^2$, Nathan Wiebe$^{2,3}$, Scott Keating$^1$, Cyrill Bösch$^{4}$, Andreas Fichtner$^1$}
\affiliation{$^{1}$ Institute of Geophysics, ETH Zürich, Zürich, Switzerland}
\affiliation{$^{2}$ Department of Computer Science, University of Toronto, Toronto, Canada}
\affiliation{$^{3}$ Pacific Northwest National Laboratory, Richland WA, USA}
\affiliation{$^{4}$ Department of Computer Science, Princeton University, Princeton NJ, USA}


\begin{abstract}
We present a quantum algorithm for simulating classical oscillator networks characterized by non-Markovian dissipation and time-varying material properties, extending recent speedups for undamped harmonic systems to viscoacoustic and viscoelastic media.
We embed the history-dependent dynamics into a Markovian state space governed by a non-Hermitian operator by approximating memory kernels through a Prony series.
We then use linear combination of Hamiltonian simulation to estimate the instantaneous kinetic and potential energy for a subset of oscillators at time $t$ within error $\epsilon$ using a number of queries to the oscillator system that scales as $\widetilde{\mathcal{O}}(\alpha_{\rm tot} t/\epsilon)$, where $\alpha_{\rm tot}$ is polynomial in the strength of the dissipation, spring constants, inverse masses, and sparsity of the connections in the network.
We further show that this energy estimation task is in the worst-case classically hard (i.e., a corresponding decision problem is $\mathsf{BQP}$-complete), even in the presence of strong dissipation.
For time-dependent materials, we show that changes in material properties appear as effective dissipation or growth in the energy representation.
Additionally, we prove the infeasibility of exponential quantum advantages in locally coupled topologies through a novel form of Lieb–Robinson-like bounds that apply to differential equations.
This allows our quantum algorithms to provide quartic speedups for locally coupled damped oscillator systems in three dimensions, raising the possibility of practical quantum speedups for realistically damped oscillator networks and approximated wave equations.
\end{abstract}


\maketitle
\newpage

\section{Introduction} \label{s:introduction}

The simulation of non-Markovian systems, i.e.~systems with memory, is fundamental to computational sciences: From dissipative wave propagation, such as viscoacoustics and viscoelasticity, in fields ranging from seismology and acoustics to structural engineering and materials science \cite{coleman_foundations_1961, robertsson_viscoelastic_1994, lakes_viscoelastic_2009} to quantum dynamics such as the Nakajima–Zwanzig equation~\cite{breuer2016colloquium}.
The simulation of such effects is usually computationally harder than Markovian problems, i.e.~memory-less systems, while still being necessary for accurately representing the real world dynamics.
Such physical systems of interest are often high-dimensional, heterogeneous, and possess broad spectra of relaxation times.
Classical integrators face significant bottlenecks in this regime and are bound to simulation costs that scale with the number of system \acp{DOF} and the fastest time scales that the system supports \cite{alves_numerical_2021}.
Specifically, adding non-Markovian effects to such systems typically requires extending the system by auxiliary or history states to formulate a Markovian embedding, or directly solve costly convolution integrals.
As such they act as multiplicative cost factors on an already extensive runtime and memory cost, rendering the simulation of complex relaxation spectra in large, high-dimensional systems expensive.

Quantum algorithms offer a promising alternative for simulating such large-scale evolution problems.
It has been shown that \ac{HS} algorithms \cite{berry2015hamiltonian, low2017optimal, low2019hamiltonian} allow for exponentially faster simulation of lossless Hamiltonian dynamics, i.e.~unitary dynamics.
Recently, Ref.~\cite{babbush2023exponential} proved that systems of undamped classical coupled oscillators can be efficiently mapped to Hamiltonian dynamics and solved with up-to exponential runtime advantages.
This generalized previous work on approximating continuum dynamics, such as wave equations, on quantum computers \cite{costa2019quantum, suau_practical_2021} and showed that a large class of classical hyperbolic many-body systems admit efficient quantum simulation.
However, physical systems are usually not lossless; they exhibit dissipation and history-dependent responses that cannot be naturally described by Hermitian Hamiltonians.
Recent frameworks such as \ac{LCHS} allow for the simulation of non-Hermitian operators \cite{an2023quantumalgorithmlinearnonunitary, low_optimal_2025} with non-positive log norm.
Furthermore, previous literature has targeted the simulation of viscously damped dynamics \cite{krovi2024quantum, sato2024quantum, villanyi2025oscillators_dissipation} or non-Markovian quantum systems \cite{li2204nonmarkovian, walters2024nonmarkovian, christensen2025nonmarkovian, ameri2026quantum}.
However, a unified quantum framework for the simulation of coupled classical oscillators with general dissipation, memory, and time-dependence has remained elusive.

In this work, we target this gap and provide a comprehensive quantum algorithm for simulating coupled oscillators with non-Markovian dissipation.
We formulate the problem using a graph-theoretic approach, where mass-mass and mass-ground couplings are mediated by generalized Maxwell connectors, i.e.~parallel damped springs.
This construction allows us to approximate arbitrary linear rheologies (e.g., Kelvin-Voigt, Standard Linear Solid), which model deformation responses of materials, and memory kernels via a Prony series expansion \cite{tschoegl_phenomenological_2012}, which is a series expansion in terms of exponentials.
By embedding these non-Markovian dynamics into an augmented Markovian state space, we derive a first-order evolution equation generated by a system matrix $\mathbf{C} = i\mathbf{H} + \mathbf{L}$, composed of a Hermitian Hamiltonian $\mathbf{H}$ and a positive semi-definite dissipator $\mathbf{L}$.

We construct explicit quantum circuits to block-encode the necessary operators.
By analyzing the spectral properties of these operators, we demonstrate that our encoding achieves improved normalization compared to prior approaches for specific graph topologies.
Combining these encodings with the optimal \ac{LCHS} algorithm, we achieve a gate complexity that is polylogarithmic in the number of degrees of freedom and linear in the simulation time.
We also extend the \BQP-completeness results of Ref.~\cite{babbush2023exponential} for estimating the energy of a subset of coupled classical oscillators to dissipative systems. 
In contrast to Ref.~\cite{krovi2024quantum}, our \BQP-completeness result holds even in the case of localized strong dissipation.
Furthermore, using information locality arguments, we show that quantum computers can provide only a polynomial advantage over classical computers for simulating non-Markovian oscillators on a lattice, even if the system has sources.
This extends existing literature for source-free local dynamics \cite{nachtergaele2007lieb, sakamoto2025quantum}.
Additionally, we derive that time-dependent materials lead to viscous dissipation or growth, which implies that for any time-dependent undamped oscillator system the generator in the energy metric acquires damping/growth terms that need to be modeled.
We pair this with corresponding block-encoding approaches that allow for the simulation of these systems.
Finally, we provide extensions to active systems, probabilistic systems, and show that ground couplings in the graph topology have an inherently lower spectral cost.

\section{Oscillator Systems} \label{s:osc_systems}

Recent work~\cite{babbush2023exponential} has shown that an exponentially large, energy-conserving network of classical oscillators
can be efficiently simulated on \acp{QC} using \ac{HS}.
Here we extend their algorithm to include general non-Markovian dissipation.
Our formulation follows the notation introduced in \cite{babbush2023exponential, koukoutsis2023} and applied in \cite{bosch_quantum_2025}.

\subsection{Oscillator Dynamics}
In this subsection we review the equations of motion for a single harmonic oscillator. 
First we consider the simple undamped case.
Then we include viscous dissipation and show that the system stays Markovian (memory-free).
Finally, we generalize to non-Markovian dissipation, which introduces a memory integral.
This motivates a subsequent Markovian embedding.

\subsubsection{Simple Harmonic Oscillator}

Let us first consider a simple harmonic oscillator in one dimension. 
It can be described by a point mass $m$ attached to a spring with spring constant $k_1$.
Its motion is governed by two fundamental relations.
Newton's second law postulates that the rate of change of momentum of the point mass equals the net force $-f(t)$ acting on it, where we choose the engineering sign convention:
\begin{equation} \label{eq:law_newton_general}
    m\frac{\mathrm{d}^2x(t)}{\mathrm{d} t^2} = -f(t),
\end{equation}
where $\tfrac{\mathrm{d}}{\mathrm{d}t}(m(\mathrm{d} x(t) / \mathrm{d}t)) = m\tfrac{\mathrm{d}^2x(t)}{\mathrm{d}t^2}$ as $m$ is time-independent. 
Here $x(t)$ denotes the displacement of the mass away from its equilibrium position.
Furthermore, Hooke's law 
\begin{equation} \label{eq:hookes_law}
    f(t) = k_1 x(t)
\end{equation}
states that the tensile force $f(t)$ inside the spring scales linearly with the displacement $x(t)$, multiplied by the spring constant $k_1$. 

By substitution, we obtain the equation of motion that captures the time evolution of the harmonic oscillator:
\begin{equation} \label{eq:undamped_motion}
    m\frac{\mathrm{d}^2x(t)}{\mathrm{d} t^2} = -k_1x(t).
\end{equation}
It describes the rate of change of momentum of the mass as a function of its displacement.

\subsubsection{Damped Harmonic Oscillator}

Let us now add viscous dissipation to the system.
This means that energy is dissipated due to external resistance acting on the mass, such as viscous drag from the surrounding medium.  
We define a dashpot to be a circuit element for this coupled oscillator that applies viscous forces on the oscillator.  Specifically,
the viscous drag is schematically modeled as a dashpot with viscosity $\eta_0$.
It acts in parallel to the original spring with spring constant $k_1$.
These two interaction mechanisms form the Kelvin-Voigt model of viscous dissipation shown in \Cref{fig:kelvin_voigt}.

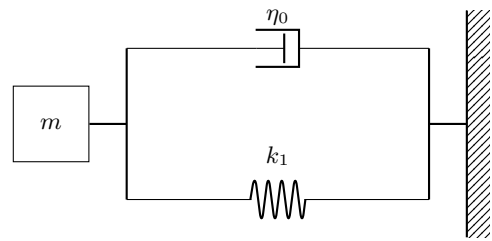
\begin{figure}[t]
\centering
\begin{tikzpicture}[scale=1.0, transform shape]

\node[minimum size=1cm, draw] (mass_m) at (-1, 0) {$m$};

\draw[thick] (mass_m.east) -- (0, 0); 
\draw[thick] (0, 1) -- (0, -1);

\draw (0, 1) to[damper=$\eta_0$] (4, 1);

\draw (0, -1) to[spring=$k_1$] (4, -1);

\draw[thick] (4, 1) -- (4, -1); 
\draw[thick] (4, 0) -- (4.5, 0); 
\draw[thick] (4.5, 1.5) -- (4.5, -1.5); 
\fill[pattern=north east lines] (4.5, -1.5) rectangle (4.8, 1.5); 

\end{tikzpicture}
\caption{Kelvin-Voigt model of a damped harmonic oscillator. 
The point mass ($m$) is coupled to a fixed wall through a purely viscous dashpot ($\eta_0$) in parallel with a spring ($k_1$). \label{fig:kelvin_voigt}}
\end{figure}

By adding viscous drag, Newton's second law, \Cref{eq:law_newton_general}, becomes
\begin{equation} \label{eq:law_newton_viscous}
    m\frac{\mathrm{d}^2x(t)}{\mathrm{d} t^2} = -\Big(\eta_0 \frac{\mathrm{d}x(t)}{\mathrm{d} t} + f_1(t)\Big) ,
\end{equation}
where $\eta_0 \frac{\mathrm{d}x(t)}{\mathrm{d} t}$ captures the instantaneous viscous resistance acting on the mass and $f_1(t)$ is the tensile force of the spring.
After substituting Hooke's law (\Cref{eq:hookes_law}), the equation of motion describing the dynamics of the viscously damped oscillator reads
\begin{equation} \label{eq:eq_motion_visc}
    m\frac{\mathrm{d}^2x(t)}{\mathrm{d} t^2} = -f(t),
\end{equation}
where the total resisting force $f(t)$ is given by
\begin{equation}
    f(t) =  \eta_0 \frac{\mathrm{d} x(t)}{\mathrm{d}t} + k_1x(t).
\end{equation}
The system is Markovian since the resisting force $f(t)$ acting on the point mass only depends on the displacement $x(t)$ and the velocity $\frac{\mathrm{d} x(t)}{\mathrm{d}t}$ at the current time $t$. 
In other words, the system is \emph{memoryless}.

\subsubsection{Non-Markovian Dissipation}

While \Cref{eq:eq_motion_visc} captures purely viscous (instantaneous) dissipation, it cannot resolve all real-world relaxation mechanisms.
Consider again the damped oscillator shown in \Cref{fig:kelvin_voigt}.
It supports a second dissipation mechanism we do not yet cover: the spring itself can dissipate energy through internal friction as it continuously deforms and relaxes.
This is schematically modeled as a dashpot of viscosity $\eta_1$ in series with the existing spring $k_1$; a combination referred to as a \emph{Maxwell body}.
This Maxwell body is characterized by its relaxation rate $\lambda_1=k_1/\eta_1$.
It extends \Cref{fig:kelvin_voigt} to the Standard-Linear-Fluid model of viscoelasticity shown in \Cref{fig:standard_linear_fluid}.
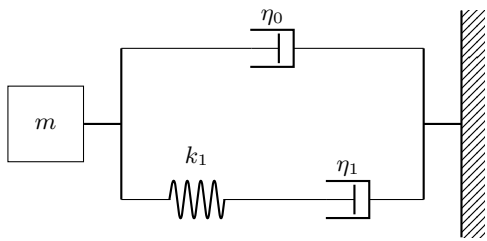
\begin{figure}[t]
\centering
\begin{tikzpicture}[scale=1.0, transform shape]

\node[minimum size=1cm, draw] (mass_m) at (-1, 0) {$m$};

\draw[thick] (mass_m.east) -- (0, 0); 
\draw[thick] (0, 1) -- (0, -1);

\draw (0, 1) to[damper=$\eta_0$] (4, 1);

\draw (0, -1) to[spring=$k_1$] (2, -1) to[damper=$\eta_1$] (4, -1);

\draw[thick] (4, 1) -- (4, -1); 
\draw[thick] (4, 0) -- (4.5, 0); 
\draw[thick] (4.5, 1.5) -- (4.5, -1.5); 
\fill[pattern=north east lines] (4.5, -1.5) rectangle (4.8, 1.5); 

\end{tikzpicture}
\caption{Standard-Linear-Fluid (Jeffreys) model. 
The point mass ($m$) is coupled to a fixed wall through a purely viscous dashpot ($\eta_0$) in parallel with a Maxwell element, which consists of a spring ($k_1$) and a second dashpot ($\eta_1$) in series. \label{fig:standard_linear_fluid}}
\end{figure}

To find its equation of motion, we must evaluate the total force acting on the mass.
In the Maxwell body, the spring and dashpot are in series. 
They therefore share the identical restoring force $f_1(t)$, but their velocities add up:
\begin{equation}
    \frac{\mathrm{d}x(t)}{\mathrm{d}t} = \frac{1}{k_1} \frac{\mathrm{d}f_1(t)}{\mathrm{d}t} + \frac{1}{\eta_1} f_1(t).
\end{equation}
Multiplying by $k_1$ reveals a standard first-order linear differential equation for the force $f_1(t)$. 
Using the relaxation rate $\lambda_1 = k_1 / \eta_1$ and solving it via an integrating factor, this equation yields the explicit Maxwell restoring force
\begin{equation}
    f_1(t) = \int_{-\infty}^{t} k_1 e^{-\lambda_1 (t-\tau)} \frac{\mathrm{d}x(\tau)}{\mathrm{d}\tau} \, \mathrm{d}\tau.
\end{equation}
Substituting this force into Newton's second law alongside the parallel viscous drag gives the equation of motion for the Standard-Linear-Fluid
\begin{equation} \label{eq:eq_motion_non_markovian}
    m\frac{\mathrm{d}^2x(t)}{\mathrm{d} t^2} = -f(t),
\end{equation}
where the total resisting force $f(t)$ acting on the point mass is given by
\begin{equation} \label{eq:memory_simple}
    f(t) = \eta_0 \frac{\mathrm{d}x(t)}{\mathrm{d}t} + \int_{-\infty}^{t} k_1 e^{-\lambda_1 (t-\tau)} \frac{\mathrm{d}x(\tau)}{\mathrm{d}\tau} \, \mathrm{d}\tau.
\end{equation}

Unlike \Cref{eq:eq_motion_visc}, the right-hand side is no longer solely dependent on the current state at time $t$. 
Instead, the force currently acting on the mass depends continuously on its entire past trajectory $\frac{\mathrm{d}x(\tau)}{\mathrm{d}\tau}$. 
Consequently, by introducing internal relaxation to the spring, the dynamics become strictly \emph{non-Markovian}.

Real-world processes often have non-trivial spectra of multiple relaxation rates $\lambda$ that characterize their dissipation behavior.
Thus, the Standard-Linear-Fluid model (\Cref{fig:standard_linear_fluid}) is insufficient to model such processes accurately.
We generalize \Cref{eq:memory_simple} by defining the resisting force $f(t)$ acting on the point mass as the convolution of a memory kernel $G(t)$ with the velocity $v(t)=\frac{\mathrm{d}x(t)}{\mathrm{d}t}$ of that mass. 
The resulting integral is referred to as the \emph{hereditary integral}
\begin{equation} 
    f(t) = \int_{-\infty}^{t} G(t-\tau)v(\tau)\,\mathrm{d}\tau.
\label{eq:law_hereditary_general}
\end{equation}
The support of $G(t)$ dictates the time horizon over which memory effects take place. 

We recover simple viscous dissipation if $G(t-\tau) = \eta_0 \delta(t-\tau)$, where $\delta$ is the Dirac delta distribution.
Similarly, Hooke's law is recovered when choosing $G(t-\tau) = k$ and assuming initial rest.
Interestingly, this shows that a simple harmonic oscillator can be represented as a non-Markovian process that has infinite memory of its past when parametrized solely in terms of its velocity $v(t)$.
This memory only becomes hidden if positions $x(t)$ are used as state variables, which integrate the velocity history and act as a basic \emph{Markovian embedding}.
In the following section, we show how to map general non-Markovian dissipation to Markovian dissipation via a series approximation of the memory kernel $G(t)$.

\subsection{Markovian Embedding}\label{sec:embed}

In this subsection we show how general non-Markovian memory kernels can be approximated by a finite sum of relaxation processes. 
To provide a Markovian embedding for \Cref{eq:law_hereditary_general}, we approximate $G(t)$ using a Prony series \cite{carcione_seismic_1993, robertsson_viscoelastic_1994, blanch_modeling_1995}.
This is referred to as the \emph{generalized Maxwell model}.
We decompose the memory kernel into an instantaneous viscous part (Markovian dissipation) and a memory-dependent relaxation part. 
Introducing the viscosity $\eta_0$ and a discrete spectrum of $C$ Maxwell bodies with stiffnesses $k_c$ and viscosities $\eta_c$, we consider the following approximation:
\begin{equation} \label{eq:prony_series}
    G(t) \approx \underbrace{\eta_0 \delta(t)}_{\text{Instantaneous}} + \sum_{c=1}^C \underbrace{k_c e^{-\lambda_ct}}_{\text{Relaxation}},
\end{equation}
where $\lambda_c=\frac{k_c}{\eta_c}$ is the relaxation rate of the $c^{\text{th}}$ Maxwell body.
Physically, this groups all relaxation processes with relaxation rates faster than we resolve into an instantaneous term (dashpot), while we resolve the rest of the relaxation spectrum through a series of relaxation rates (Maxwell bodies).
Because the Maxwell bodies are mass-free, they possess no inertia. 
This is captured by the relaxation rates $\lambda_c$ being strictly real, ensuring purely monotonic decay without oscillation.

Inserting \Cref{eq:prony_series} into the hereditary integral in \Cref{eq:law_hereditary_general} decomposes $f(t)$ into a  viscous force term, $\eta_0 v(t)$, and a sum of partial memory forces $f_c(t)$.
We allow for body-specific velocity source terms $b_c(t)$ directly within the convolution such that
\begin{equation} 
\label{eq:general_maxwell}
\begin{aligned}
    f(t) &= \eta_0 v(t) \\
         &\quad + \sum_{c=1}^C \underbrace{ \int_{-\infty}^{t} k_c e^{-\lambda_c(t-\tau)} (v(\tau) + b_c(\tau)) \mathrm{d}\tau }_{=: f_c(t)}.
\end{aligned}
\end{equation}
For the memory terms, we differentiate $f_c(t)$ with respect to time which yields the following time evolution for the partial forces:
\begin{equation} \label{eq:maxwell_modes}
    \frac{\mathrm{d} f_c(t)}{\mathrm{d} t} = k_c (v(t) + b_c(t)) - \lambda_c f_c(t).
\end{equation}
We substitute into Newton's second law
\begin{equation}
    m\frac{\mathrm{d}v(t)}{\mathrm{d} t} = b_0(t) - f(t),
\end{equation}
where we have introduced the body force source $b_0$.
This gives the equation of motion
\begin{equation} \label{eq:single_osc_motion}
    m\frac{\mathrm{d}v(t)}{\mathrm{d} t} = b_0(t) -\eta_0v(t) - \sum_{c=1}^Cf_c(t).
\end{equation}
Such a transformation constitutes a \emph{Markovian embedding}. 
The history dependence is now explicitly carried by the instantaneous values of the partial forces $f_c$.
In essence, we have transformed the problem of solving the non-Markovian \ac{ODE}, to the problem of solving a system of $1 + C$ coupled Markovian \acp{ODE}.

\subsection{Coupled System}

\begin{table}[]
    \centering
    \begin{tabular}{|c|c|}
        \hline
         Variable Name & Interpretation  \\
         \hline
         $N$ & Number of degrees of freedom in Sim.\\
         & $N= M + C(N_e + N_g)$\\
         $N_e$ & Number of inter-mass connections.\\
         $N_g$ & Number of connections to ground.\\
         $M$ & Number of masses.
         \\
         $C$ & Number of Maxwell bodies per connection. \\
         \hline
    \end{tabular}
    \caption{Variables used to specify interaction graph for coupled oscillators.}
    \label{tab:symbols}
\end{table}

In this subsection, we generalize the embedding above to an oscillator network of $M$ masses.
These are coupled by $N_e$ inter-mass connectors, which characterize the coupling of masses to masses. 
In the simplest case, each of the $N_e$ inter-mass connectors would just be a single spring (\Cref{eq:hookes_law}). 
More generally though, each of the $N_e$ inter-mass connectors can consist of multiple Maxwell bodies, each with different stiffness and relaxation rate as we explain in more detail below, as well as a direct connection that represents viscous dissipation.
In addition to the inter-mass connectors, an oscillator network can also have $N_g$ ground connectors, which characterize the coupling of masses to ground and possess the same interaction mechanisms as the inter-mass connectors.
Such a network is shown in \Cref{fig:network_topology}.

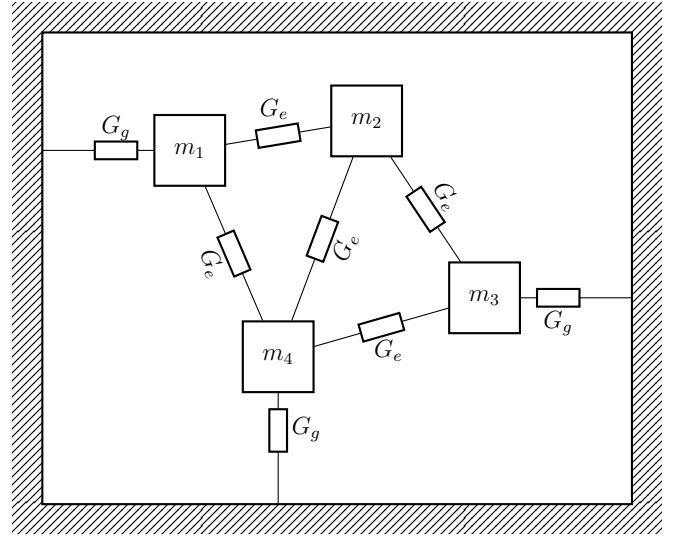
\begin{figure}[tbp]
\centering
\begin{tikzpicture}[scale=0.78, transform shape]

\ctikzset{
    bipoles/length=1.2cm,
    bipoles/generic/width=0.60,
    bipoles/generic/height=0.25
}

\draw[pattern=north east lines, draw=none] (-0.5, -0.5) rectangle (10.5, 8.5);
\fill[white] (0, 0) rectangle (10, 8);
\draw[thick] (0, 0) rectangle (10, 8);

\node[minimum size=1.2cm, draw, thick, fill=white] (m1) at (2.5, 6.0) {\large $m_1$};
\node[minimum size=1.2cm, draw, thick, fill=white] (m2) at (5.5, 6.5) {\large $m_2$};
\node[minimum size=1.2cm, draw, thick, fill=white] (m3) at (7.5, 3.5) {\large $m_3$};
\node[minimum size=1.2cm, draw, thick, fill=white] (m4) at (4.0, 2.5) {\large $m_4$};

\draw (m1) to[generic, l_={\large $G_{g}$}] (0, 6.0);
\draw (m3) to[generic, l_={\large $G_{g}$}] (10.0, 3.5);
\draw (m4) to[generic, l={\large $G_{g}$}] (4.0, 0);

\draw (m1) to[generic, l={\large $G_{e}$}] (m2);
\draw (m2) to[generic, l={\large $G_{e}$}] (m3);
\draw (m3) to[generic, l={\large $G_{e}$}] (m4);
\draw (m4) to[generic, l={\large $G_{e}$}] (m1);

\draw (m4) to[generic, l_={\large $G_{e}$}, pos=0.4] (m2);

\end{tikzpicture}
\caption{Coupled oscillator network within a solid ground frame which is assumed to have infinite mass. 
The interactions between neighboring masses and the ground are abstracted into generalized coupling blocks ($G_{e}$ and $G_{g}$, respectively), representing the composite Maxwell bodies and dashpots. \label{fig:network_topology}}
\end{figure}

An oscillator network can be described by a directed graph with
incidence matrix $\mathbf{D}_0 \in \{-1,0,1\}^{N_e\times M}$, which contains a single $+1$ and $-1$ per row corresponding to the starting and end masses of inter-mass connectors,
and the ground selector $\mathbf{S}_0 \in \{0,1\}^{N_g\times M}$, which contains a single $+1$ per row indicating which masses are connected to the ground.
We summarize the overall topology of the oscillator network in the single-body topology matrix:
\begin{equation}
    \mathbf{T}_0 := 
    \begin{bmatrix} 
        \mathbf{D}_0 \\ \mathbf{S}_0 
    \end{bmatrix} \in \{-1,0,1\}^{(N_e+N_g) \times M}.
\end{equation}
Note that $\mathbf{T}_0^\dagger \mathbf{T}_0$ is the graph Laplacian of the oscillator network. The effect of the ground selector is to add $+1$ to the diagonal entries of the graph Laplacian for vertices connected to the ground.
As mentioned before, inter-mass connection or ground connections can be more complex than just consisting of a single spring. 
In particular, we consider the case where each connector consists of a  viscous dashpot parallel to $C$ Maxwell bodies as shown in \Cref{fig:general_visco_2mass}. 
The  dashpots induce Markovian dissipation, i.e.~they represent the instantaneous viscous part of the memory kernel (\Cref{eq:prony_series}). 
The Maxwell bodies on the other hand model non-Markovian dissipation with varying relaxation rates as discussed in the previous section.
We then define the multi-body topology matrix
\begin{equation}
    \mathbf{T} = [\mathbf{T}_0;\dots;\mathbf{T}_0] \in \{-1,0,1\}^{C(N_e+N_g)\times M},
\end{equation}
which captures the identical topology of all Maxwell bodies in the oscillator network.

\begin{figure}[t]
\centering
\begin{tikzpicture}[scale=1.0, transform shape]

\draw[pattern=north east lines, thick] (2, 4) -- (3, 4) -- (3, 4.3) -- (2, 4.3) -- (2, 4);
\draw[-,thick] (2.5,3.5) -- (2.5,4) node[below right] {$f_{g}$};
\draw[thick] (0, 3.5) -- (5, 3.5);
\draw[thick] (0, 1.5) -- (5, 1.5);
\draw (0, 1.5) to[damper=$\eta_{g,0}$] (0, 3.5);
\foreach \x/\k/\e in {1.25/k_{g,1}/\eta_{g,1}, 2.5/k_{g,2}/\eta_{g,2}, 5/k_{g,C}/\eta_{g,C}}{
\draw (\x, 1.5) to[spring=$\k$] (\x, 2.5) to[damper=$\e$] (\x, 3.5);
}
\node at (3.75, 2.5) {$\dots$};
\draw (2.5, 0.75) node[minimum size=1cm, draw] (m) {m};
\draw[-, thick] (m.north) -- (2.5, 1.5);
\draw[-, thick] (m.south) -- (2.5, 0);
\draw[thick] (0, 0) -- (5, 0);
\draw[thick] (0, -2) -- (5, -2);
\draw (0, 0) to[damper=$\eta_{e,0}$, l_=$\eta_{e,0}$] (0, -2);
\foreach \x/\k/\e in {1.25/k_{e,1}/\eta_{e,1}, 2.5/k_{e,2}/\eta_{e,2}, 5/k_{e,C}/\eta_{e,C}}{
\draw (\x, 0) to[spring=$\k$, l_=$\k$] (\x, -1) to[damper=$\e$, l_=$\e$] (\x, -2);
}
\node at (3.75, -1) {$\dots$};
\draw[->,thick] (2.5,-2) -- (2.5,-2.5) node[above right] {$f_{e}$};
\end{tikzpicture}
\caption{A single coupling block of \Cref{fig:network_topology}.
Mass $m$ coupled to ground (top, $G_{g}$) and neighbor (bottom, $G_{e}$) via  dashpots $\eta_0$ in parallel with $C$ Maxwell bodies ($k_c, \eta_c$). \label{fig:general_visco_2mass}}
\end{figure}
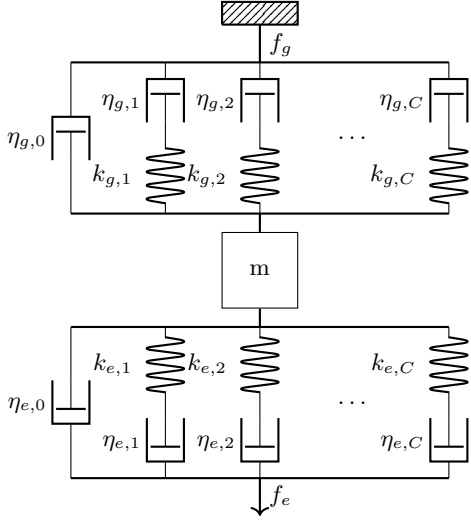

Note that the multi-body topology matrix $\mathbf{T}$ only gives information about the presence or absence of connectors in the oscillator network. 
It does not tell us anything about the material parameters, i.e.~the values of the masses, spring constants or viscosities. 
These parameters are stored in an additional set of matrices, which are diagonal as we have decoupled the individual local \acp{DOF}.
In particular, the (nodal) masses $m_1, m_2, \dots, m_M$ are collected into the diagonal matrix $\mathbf{M}\succ0 \in \mathbb{R}^{M \times M}$. 
The viscosities of the  dashpots are collected into the diagonal matrix $\boldsymbol{\eta}_0 = \diag(\boldsymbol{\eta}_{e,0}, \boldsymbol{\eta}_{g,0}) \in \mathbb{R}^{(N_e + N_g) \times (N_e + N_g)}$.
For the Maxwell bodies $c \in \{1,\dots,C\}$, we define the diagonal stiffness and viscosity matrices
\begin{align}
    \mathbf{K}_c &:= \diag(\mathbf{K}_{e,c} , \mathbf{K}_{g,c}) \in \mathbb{R}^{(N_e + N_g) \times (N_e + N_g)}, \\
    \boldsymbol{\eta}_c &:= \diag(\boldsymbol{\eta}_{e,c}, \boldsymbol{\eta}_{g,c}) \in \mathbb{R}^{(N_e + N_g) \times (N_e + N_g)}.
\end{align}
The stiffnesses and viscosities of all parallel Maxwell bodies of the mass-mass and mass-ground connections are independent parameters.

To describe the overall state of the coupled network through time we need to define a joint representation over all oscillators.
For this we collect the memory forces for the $c^{\text{th}}$ Maxwell body into stacked vectors
\begin{equation}
    \mathbf{f}_c(t) = \begin{bmatrix} \mathbf{f}_{e,c}(t) \\ \mathbf{f}_{g,c}(t) \end{bmatrix},
\end{equation}
again separating mass-mass and mass-ground couplings.
This allows us to define the total system state $\boldsymbol{\phi}(t)$ as a stack of the nodal velocities $\mathbf{v} \in \mathbb{R}^{M}$ with the memory forces $\mathbf{f}_c(t)$ for each body $c$
\begin{equation}
    \boldsymbol{\phi}(t) = [\mathbf{v}(t); \mathbf{f}_1(t); \dots; \mathbf{f}_C(t)] \in \mathbb{R}^{N},
\end{equation}
where the overall system size is $N=M+C(N_e+N_g)$.

Similarly, we stack the velocity source terms for each Maxwell body as
\begin{equation}
    \mathbf{b}_c(t) = \begin{bmatrix} \mathbf{b}_{e,c}(t) \\ \mathbf{b}_{g,c}(t) \end{bmatrix}.
\end{equation}
Together with the body force sources $\mathbf{b}_{0}(t)$ we define the overall source vector
\begin{equation}
    \mathbf{b}(t) = [\mathbf{b}_0(t); \mathbf{b}_1(t); \dots; \mathbf{b}_C(t)] \in \mathbb{R}^{N}.
\end{equation}
It defines the system in-flux or out-flux for every \ac{DOF}, i.e. each spring and each mass can be driven independently.

We rewrite the equation of motion of the single oscillator with non-Markovian dissipation (\Cref{eq:single_osc_motion}) using the network topology $\mathbf{T}_0$ as
\begin{equation}
    \frac{\mathrm{d} \mathbf{v}(t)}{\mathrm{d} t} = \mathbf{M}^{-1}\Big(\mathbf{b}_0(t) - \mathbf{T}_0^{\dagger}\big(\boldsymbol{\eta}_{0}\mathbf{T}_0\mathbf{v}(t) + \sum_{c=1}^{C}{\mathbf{f}}_{c}(t)\big)\Big),
\label{velocity_de}
\end{equation}
where $\dagger$ denotes the conjugate transpose.
The generalized Maxwell response (\Cref{eq:maxwell_modes}) for each body $c$ is expanded as
\begin{equation} \label{eq:maxwell_element_short}
    \frac{\mathrm{d} {\mathbf{f}}_{c}(t)}{\mathrm{d} t} = {\mathbf{K}}_{c}  (\mathbf{T}_0 \mathbf{v}(t) + {\mathbf{b}}_{c}(t)) - \boldsymbol{\Lambda}_c{\mathbf{f}}_{c}(t),
\end{equation}
where $\boldsymbol{\Lambda}_c={\mathbf{K}}_{c}{\boldsymbol{\eta}}_{c}^{-1}\succeq0$ is the matrix of relaxation rates.
This yields the linear system
\begin{equation} \label{eq:general_ode_short}
    {\frac{\mathrm{d} \boldsymbol{\phi}(t)}{\mathrm{d} t} = \mathbf{B}(\mathbf{A} - \boldsymbol{\eta}) \boldsymbol{\phi}(t) + \mathbf{B}\mathbf{b}(t).}
\end{equation}
The material matrix is diagonal
\begin{equation}
    \mathbf B = \diag(\mathbf M^{-1}, \mathbf{K}) \succ 0,
\end{equation}
with $\mathbf{K}=\diag(\mathbf K_1, \dots, \mathbf K_C)$.
The skew-Hermitian coupling matrix $\mathbf{A}$ is given by
\begin{equation}
    \mathbf{A} = 
    \begin{bmatrix}
        \mathbf{0} & -\mathbf{T}^\dagger \\
        \mathbf{T} & \mathbf{0}
    \end{bmatrix} = -\mathbf{A}^\dagger \in \mathbb{R}^{N \times N}  .
\end{equation}

Finally, the positive semi-definite dissipation matrix is block-diagonal
\begin{equation}
    \boldsymbol{\eta} = 
    \diag(\boldsymbol{\eta}_{\mathbf{M}}, \boldsymbol{\eta}_{\mathbf{K}}^{-1}) \succeq 0,
\end{equation}
where 
\begin{equation}
    \boldsymbol{\eta}_{\mathbf{M}}= \mathbf{T}_0^\dagger \boldsymbol{\eta}_0\mathbf{T}_0\label{eq:etaM}
\end{equation} is a grounded and weighted graph Laplacian and $\boldsymbol{\eta}_{\mathbf{K}}=\diag(\boldsymbol{\eta}_1, \dots, \boldsymbol{\eta}_C)$ is diagonal.  Similarly we define 
\begin{equation}
    \boldsymbol{\Lambda} = \mathbf{K}^{1/2}\boldsymbol{\eta}^{-1}_{\mathbf{K}}\mathbf{K}^{1/2}.
\end{equation}

\subsection{Hamiltonian Form} \label{s:port_ham}

While \eqref{eq:general_ode_short} correctly captures the dynamical evolution of the oscillator system, it does not preserve the state norm.
Hence, even for lossless system evolution we have $\|\boldsymbol{\phi}(0)\| \neq \|\boldsymbol{\phi}(t)\|$ in general. 
To address this we utilize that for lossless oscillator systems we have perfect conservation of energy.
As such, by mapping the inner product of the Hilbert space to the energy, we correctly capture the unitary evolution of lossless dynamics.  

We simplify the equation of motion through the similarity transform
\begin{equation} \label{eq:sim_trans}
    \boldsymbol\psi(t)=\mathbf B^{-1/2} \boldsymbol\phi(t), \quad \boldsymbol{\chi}(t) = \mathbf B^{1/2}  \mathbf{b}(t),
\end{equation}
which maps into a basis such that the total energy of the oscillator system obeys
\begin{equation}
    E(t) = \frac{1}{2}\|\boldsymbol{\psi}(t)\|^2.
\end{equation}
Substituting the transformed states into \Cref{eq:general_ode_short} yields the Schrödinger-form evolution equation
\begin{equation}\label{eq:damped_schrodinger_form}
    {\frac{\mathrm{d} \boldsymbol\psi(t)}{\mathrm{d} t}= -\mathbf{C} \boldsymbol\psi(t) + \boldsymbol{\chi}(t),}
\end{equation}
where $\mathbf{C}=i\mathbf H + \mathbf{L} \in \mathbb{R}^{N \times N}$ is the system matrix.
The dynamics are governed by the Hermitian Hamiltonian $\mathbf H$ and the positive semi-definite Dissipator $\mathbf L$
\begin{equation} \label{eq:hamiltonian_dissipator}
    {\mathbf H=i \mathbf B^{1/2} \mathbf A \mathbf B^{1/2} = \mathbf H^\dagger, \quad
    \mathbf L=\mathbf B^{1/2} \boldsymbol{\eta} \mathbf B^{1/2} \succeq 0,}
\end{equation}
or equivalently in expanded form
\begin{align}
    \mathbf{H} &= 
    \begin{bmatrix}
        \mathbf{0} & -i\mathbf{M}^{-1/2}\mathbf{T}^\dagger\mathbf{K}^{1/2} \\
        i\mathbf{K}^{1/2}\mathbf{T}\mathbf{M}^{-1/2} & \mathbf{0}
    \end{bmatrix}, 
    \label{H_expanded}
    \\
    \mathbf{L} &= \begin{bmatrix}
        \mathbf{M}^{-1/2}\mathbf{T}_0^\dagger \boldsymbol{\eta}_0\mathbf{T}_0\mathbf{M}^{-1/2} & \mathbf{0} \\
        \mathbf{0} & \boldsymbol{\Lambda} 
    \end{bmatrix}.
\end{align}

Our quantum algorithm will then simulate the resultant non-unitary dynamics using \ac{LCHS} methods.

\section{Quantum Simulation} \label{s:applications}
In this section we present quantum algorithms for solving \Cref{eq:damped_schrodinger_form} introduced in the last section.
This includes \ac{BE} algorithms for the Hamiltonian $\mathbf{H}$, the Dissipator $\mathbf{L}$, and the evolution operator $\mathbf{G}(T) = e^{-\mathbf{C}T}$, where $T$ is the final simulation time.
First we establish the oracle access model.
Finally we provide algorithms for preparing the normalized evolution state $\ket{\boldsymbol{\psi}(T)}$ and for measuring subspace energies.

We decompose the directed incidence matrix $\mathbf{D}_0 \in \{-1,0,1\}^{N_e\times M}$ as 
\begin{equation}
    \mathbf{D}_0  = \mathbf{S}^+-\mathbf{S}^-,
\end{equation}
where $\mathbf{S}^+\in \{0,1\}^{N_e\times M}$ and $\mathbf{S}^-\in \{0,1\}^{N_e\times M}$ describe the spring heads and tails selection matrices, respectively.
As such, they contain exactly a single $1$ per row and otherwise only $0$.
Note that the orientation of the heads and tails is arbitrary due to gauge freedom, as the underlying graph Laplacian is undirected and loses all notion of orientation.
However, the block-encoding complexities are sensitive to this chosen gauge, allowing us to orient the edges to optimize simulation costs.

We combine these inter-mass connections with the ground selector matrix $\mathbf{S}_0 \in \{0,1\}^{N_g\times M}$ to define the augmented topological selection matrices
\begin{equation}
    \mathbf{S}_A = \begin{bmatrix} \mathbf{S}^+ \\ \mathbf{S}_0 \end{bmatrix}, \quad \mathbf{S}_B = \begin{bmatrix} \mathbf{S}^- \\ \mathbf{S}_0 \end{bmatrix}.
\end{equation}
Because $\mathbf{S}^+$, $\mathbf{S}^-$, and $\mathbf{S}_0$ contain exactly a single $1$ per row, $\mathbf{S}_A$ and $\mathbf{S}_B$ inherit this property.

We define $d_{\operatorname{in}}$ as the maximum column sparsity of $\mathbf{S}_A$, and $d_{\operatorname{out}}$ as the maximum column sparsity of $\mathbf{S}_B$. 
Ground connections constructively contribute to both $d_{\operatorname{in}}$ and $d_{\operatorname{out}}$, which physically exploits the lack of anti-cyclic oscillation supported by grounding edges (see \Cref{s:ground_edges}).
We further assume that $M\leq N_e + N_g$ in our main text notation, which excludes trees and forests.
In \Cref{apx:block_encodings}, we discuss the simple generalization of our results to these settings.

Throughout this work, we assume that the source term $\boldsymbol{\chi}(t)$ is sufficiently smooth with bounded time derivatives. 
This guarantees that the continuous time-integral in the source-driven solution can be approximated with a numerical quadrature discretization error that is absorbed by the overall algorithmic error budget.
For details about source implementations consider \cite{low_hamiltonian_2019}.

\subsection{Access Model}

Here we define the quantum access model we will subsequently use.
We assume that all of these oracles are efficiently controllable:
\begin{definition}[Index Oracles] \label{def:time_independent_access}
Let $\mathcal{M} \in \{\mathbf{S}_A, \mathbf{S}_B\}$ denote the augmented topological selection matrices of the system. 
Let $O_r^{(\mathcal{M})}$ and $O_c^{(\mathcal{M})}$ be index oracles acting as in-place permutation unitaries on two registers of $n_{\mathbf{S}} = \lceil \log_2 (N_e + N_g) \rceil$ qubits, such that
\begin{equation} \label{eq:sparse_oracles_main}
\begin{aligned}
    O_r^{(\mathcal{M})} &: \ket{i} \ket{k} \rightarrow \ket{i} \ket{r^{(\mathcal{M})}_{ik}}, \\
    O_c^{(\mathcal{M})} &: \ket{l} \ket{j} \rightarrow \ket{c^{(\mathcal{M})}_{lj}} \ket{j},
\end{aligned}
\end{equation}
where $r^{(\mathcal{M})}_{ik}$ is the column index of the $k$-th nonzero entry of row $i$, and $c^{(\mathcal{M})}_{lj}$ is the row index of the $l$-th nonzero entry of column $j$ in matrix $\mathcal{M}$. For $\mathcal{M}= \mathbf{S}_A$ then the maximum column sparsity is defined to be at most $d_{\rm in}$ and similarly $d_{\rm out}$ upper bounds the column sparsity of $\mathbf{S}_B$.
\end{definition}

Because the row sparsity of both matrices is $1$, the index $k$ is strictly $0$ and the row oracle simplifies to $O_r^{(\mathcal{M})} \ket{i}\ket{0} \rightarrow \ket{i}\ket{r^{(\mathcal{M})}_{i0}}$.
We refer to the literature for the implementation of such oracles (e.g. \cite{takahashi_quantum_2009, gidney_halving_2018, gaur_logarithmic_2023, gaur_novel_2024}), for example, in the context of lattice topologies. 

\begin{definition} [Phase Oracles]
Let $\mathbf{R} \in \{\mathbf{M}^{-1/2}, \mathbf{K}^{1/2}, \boldsymbol{\eta}_0, \boldsymbol{\Lambda}\}$ be a diagonal matrix of dimension $N_{\mathbf{R}} \times N_{\mathbf{R}}$. We define the phase vector $\boldsymbol{\theta}_{\mathbf{R}} \in \mathbb{R}^{N_{\mathbf{R}}}$, where $(\theta_{\mathbf{R}})_k = \arccos\left(\mathbf{R}_{kk} / \alpha_{\mathbf{R}}\right)$, and $\alpha_{\mathbf{R}} \geq \|\mathbf{R}\|$. 
A $\varepsilon'$-accurate \emph{phase oracle} for $\mathbf{R}$ is a unitary $O_{\mathbf{R}}$ acting on $n_{\mathbf{R}}=\lceil\log_2 N_{\mathbf{R}}\rceil$ qubits such that
\begin{equation}
\bigg\| O_{\mathbf{R}} - \sum_{k=0}^{N_{\mathbf{R}}-1} e^{-i (\theta_{\mathbf{R}})_k} \ketbra{k}{k} \bigg\| \le \varepsilon',
\end{equation}
where $\|\cdot\|$ denotes the operator norm.
\end{definition}

There are efficient implementations of phase oracles using function synthesis over orthogonal bases (e.g. \cite{rosenkranz_quantum_2024, zylberman_efficient_2025, yano2026quantum}).  
The number of gates required to reduce the error in these synthesis approaches scales as $\mathcal{O}(\log(1/\varepsilon'))$, which means that errors can be made small at low cost.

\begin{definition} [State Preparation Oracles]
Let $\boldsymbol{v} \in \{\boldsymbol{\psi}_0, \boldsymbol{\chi}_{\tau}\}$ be a vector in $\mathbb{C}^{N_v}$ with dimension $N_v$ and norm $\alpha_{\boldsymbol{v}} = \|\boldsymbol{v}\|_2$, where $\boldsymbol{\chi}_{\tau} \in \mathbb{R}^{N_t}$ is defined by $(\boldsymbol{\chi}_{\tau})_j = \sqrt{\|\boldsymbol{\chi}(\tau_j)\| \tfrac{T}{N_t}}$.
An $(\alpha_{\boldsymbol{v}}, a, \varepsilon')$ \emph{state preparation oracle} for $\boldsymbol{v}$ is a unitary $O_{\boldsymbol{v}}$ acting on $(a)$ ancilla qubits and an $n_v=\lceil\log_2 N_v\rceil$ qubit system register, such that
\begin{equation}
    \bigl\| O_{\boldsymbol{v}}(|0\rangle_a \otimes |0\rangle_{n_v}) - (|0\rangle_a \otimes \ket{\boldsymbol{v}}) \bigr\|_2 \le \varepsilon',
\end{equation}
where $\ket{\boldsymbol{v}} = \boldsymbol{v}/\alpha_{\boldsymbol{v}}$ is the normalized quantum state. 
Here $N_v = N$ for $\boldsymbol{v} = \boldsymbol{\psi}_0$, and $N_v = N_t$ for $\boldsymbol{v} = \boldsymbol{\chi}_{\tau}$.

For the spatial source states $\boldsymbol{\chi}(\tau_j)$, we define a $\varepsilon'$-accurate \emph{controlled state preparation oracle} $O_{\boldsymbol{\chi}_x}$ acting on an $n_t=\lceil\log_2 N_t\rceil$ qubit clock register, $a$ ancilla qubits, and an $n=\lceil\log_2 N\rceil$ qubit system register. For all $j \in \{0, \dots, N_t-1\}$, it satisfies
\begin{equation}
    \bigl\| O_{\boldsymbol{\chi}_x}(\ket{j} \otimes |0\rangle_a \otimes |0\rangle_{n}) - (\ket{j} \otimes |0\rangle_a \otimes \ket{\bar{\boldsymbol{\chi}}(\tau_j)}) \bigr\|_2 \le \varepsilon',
\end{equation}
where $\ket{\bar{\boldsymbol{\chi}}(\tau_j)} = \boldsymbol{\chi}(\tau_j)/\|\boldsymbol{\chi}(\tau_j)\|$ and $\ket{j}$ encodes the discrete time step.
\end{definition}

For efficient state preparation algorithms consider e.g. \cite{gleinig2021efficient, de2022double,  zhang2022quantum, sun_asymptotically_2023, zylberman_efficient_2024, rosenkranz_quantum_2024}.

\subsection{Algorithms}

We encode the Hamiltonian as specified in \Cref{eq:hamiltonian_dissipator}:
\begin{lemma}[Hamiltonian Encoding, Proof in \Cref{s:be_hamiltonian}] \label{lem:hamiltonian_encoding}
    For any target error $\varepsilon > 0$, the Hamiltonian $\mathbf{H}$ admits an $(\alpha_{\mathbf{H}}, a_{\mathbf{H}}, \varepsilon)$-\ac{BE} $\mathcal{U}_{\mathbf{H}}$ with normalization factor
    \begin{equation}
        {\alpha_{\mathbf{H}} =(\sqrt{d_{\operatorname{in}}}+\sqrt{d_{\operatorname{out}}})\sqrt{\frac{k_{\max}}{m_{\min}}C}.}
    \end{equation}
    The encoding requires $a_{\mathbf{H}} = \mathcal{O}(\log N)$ ancilla qubits, $\mathcal{O}(\log N)$ additional two-qubit gates, and $\mathcal{O}(1)$ queries to the (controlled) index oracles $O_r^{(\mathcal{M})},O_c^{(\mathcal{M})}$ for $\mathcal{M} \in \{\mathbf{S}_A, \mathbf{S}_B\}$ and (controlled) $\mathcal{O}(\varepsilon/\alpha_{\mathbf{H}})$-accurate phase oracles $O_{\mathbf{M}^{-1/2}}, O_{\mathbf{K}^{1/2}}$.
\end{lemma}
This normalization is slightly tighter than the encoding provided in \cite{babbush2023exponential} by a factor of up to $\sqrt{2}$ for certain topologies.
Also notice that multiple parallel springs accumulate stiffness. 
For example, if we assume $k_c=k$ for all $C$ bodies, we have
\begin{equation}
    k_{\max} = \frac{K_{\operatorname{total}}}{C} \implies \alpha_{\mathbf{H}} \propto \sqrt{\frac{K_{\operatorname{total}}}{C} \cdot C} = \sqrt{K_{\operatorname{total}}}.
\end{equation}
This shows that the actual overhead with respect to the number of memory bodies is not necessarily determined by the normalization constant but instead by the gate cost, i.e., $\mathcal{O}(\log C)$. 
One might expect that this gives rise to significant additional advantage compared to classical methods that scale linearly with $C$.
However, in \Cref{s:kernel_approx} we show that such hopes are likely unfounded due to the exceptional approximation accuracy of the utilized Markovian embedding and fast classical alternatives for solving the hereditary integral in \Cref{eq:law_hereditary_general}.

We encode the Dissipator as specified in \Cref{eq:hamiltonian_dissipator}:
\begin{lemma}[Dissipator Encoding, Proof in \Cref{s:be_dissipator}] \label{lem:dissipator_encoding}
    For any target error $\varepsilon > 0$, the Dissipator $\mathbf{L}$ admits an $(\alpha_{\mathbf{L}}, a_{\mathbf{L}}, \varepsilon)$-\ac{BE} $\mathcal{U}_{\mathbf{L}}$ with normalization factor
    \begin{equation}
        {\alpha_{\mathbf{L}} = \max\Big((\sqrt{d_{\operatorname{in}}}+\sqrt{d_{\operatorname{out}}})^2\frac{\|{\boldsymbol{\eta}}_0\|_{\max}}{m_{\min}}, \|\boldsymbol{\Lambda}\|_{\infty}\Big).}
    \end{equation}
    The encoding requires $a_{\mathbf{L}} = \mathcal{O}(\log N)$ ancilla qubits, $\mathcal{O}(\log N)$ additional two-qubit gates, and $\mathcal{O}(1)$ queries to the (controlled) index oracles $O_r^{(\mathcal{M})},O_c^{(\mathcal{M})}$ for $\mathcal{M} \in \{\mathbf{S}_A, \mathbf{S}_B\}$ and (controlled) $\mathcal{O}(\varepsilon/\alpha_{\mathbf{L}})$-accurate phase oracles $O_{\mathbf{M}^{-1/2}}, O_{\boldsymbol{\eta}_0}, O_{\boldsymbol{\Lambda}}$.
\end{lemma}

Finally, we block-encode the system evolution using \ac{LCHS} \cite{low_optimal_2025}.
\begin{lemma}[LCHS Lemma] \label{th:_lchs_op_applied}
    Consider the \ac{ODE} system
    \begin{equation}
        \frac{\mathrm{d} \mathbf{G}(t)}{\mathrm{d} t} = -\mathbf{C}\mathbf{G}(t), \quad \mathbf{G}(0)=\mathbf{I},
    \end{equation}
    where $\mathbf{C} = i\mathbf{H} + \mathbf{L}$, with $\mathbf{H}$ (\Cref{lem:hamiltonian_encoding}), $\mathbf{L}$ (\Cref{lem:dissipator_encoding}) being Hermitian and $\mathbf{L} \succeq 0$.
    
    For any target error $\varepsilon \in (0,1)$, the evolution operator $\mathbf{G}(T) = e^{-\mathbf{C}T}$ admits an $(\alpha_{\mathbf{G}}, a_{\mathbf{G}}, \varepsilon)$-\ac{BE} $\mathcal{U}_{\mathbf{G}}$ with normalization factor $\alpha_{\mathbf{G}}=\mathcal{O}(1)$ that can be constructed using
    \begin{equation}
    {Q_{\mathbf{G}}=\mathcal{O}\Big(T(\alpha_\mathbf{L}\log(1/\varepsilon) + \alpha_\mathbf{H}) + \log(1/\varepsilon) \Big)}
    \end{equation}
    queries to the (controlled) index oracles $O_r^{(\mathcal{M})},O_c^{(\mathcal{M})}$ for $\mathcal{M} \in \{\mathbf{S}_A, \mathbf{S}_B\}$ and (controlled) $\mathcal{O}(\varepsilon/Q_{\mathbf{G}})$-accurate phase oracles $O_{\mathbf{M}^{-1/2}}, O_{\mathbf{K}^{1/2}}, O_{\boldsymbol{\eta}_0}, O_{\boldsymbol{\Lambda}}$.
    The encoding requires $a_{\mathbf{G}} = \mathcal{O}\big(\log(N\alpha_{\mathbf{L}}T/\varepsilon)\big)$ ancilla qubits and uses $\mathcal{O}\Big(\big(Q_\mathbf{G} + \log^{5/2}(1/\varepsilon)\big)\log(N\alpha_{\mathbf{L}}T/\varepsilon)\Big)$ additional two-qubit gates.
\end{lemma}
\begin{proof}
    Proof follows by substitution of variables into the claims of~\cite{low_optimal_2025}, allocating half of the error budget to the \ac{LCHS} algorithmic and quadrature error and half to the accumulated phase oracle errors, and merging all logarithmic factors up to constants in the nontrivial regime $\alpha_{\mathbf{L}}T \geq 1$.
\end{proof}

The  block encoding for \ac{LCHS} can be used for state preparation of the normalized evolution state $\ket{\boldsymbol{\psi}(T)}$.
The result is derived by substitution of parameters into the general form of \Cref{th:_lchs_op_applied}:
\begin{theorem}[State Preparation \cite{low_optimal_2025}] \label{th:_lchs_state_prep_applied}
    Assume that $T \alpha_{\mathbf{L}} \in \Omega(1)$.
    The normalized state in $\mathbb{C}^N$
    \begin{equation}
    \begin{aligned}
        \ket{\boldsymbol{\psi}(T)} &\propto \|\boldsymbol{\psi}(0)\|\mathbf{G}(T) \ket{\boldsymbol{\psi}(0)} \\&+ \int_{0}^{T} \|\boldsymbol{\chi}(\tau)\| \mathbf{G}(T-\tau) \ket{\bar{\boldsymbol{\chi}}(\tau)} \mathrm{d}\tau,
    \end{aligned}
    \end{equation}
    can be prepared to any additive target error $\varepsilon \in (0,1)$ with constant success probability using
    \begin{equation}
        {Q_{\boldsymbol{\psi}} = \mathcal{O}\left(\frac{\|\boldsymbol{\psi}(0)\| + \|\boldsymbol{\chi}\|_{L^1}}{\|\boldsymbol{\psi}(T)\|}\right)}
    \end{equation}
    queries to the $\mathcal{O}(\varepsilon/Q_{\boldsymbol{\psi}})$-accurate state preparation oracles $O_{\boldsymbol{\psi}_0}$, $O_{\boldsymbol{\chi}_\tau}$, and $O_{\boldsymbol{\chi}_x}$, where $\|\boldsymbol{\chi}\|_{L^1} = \int_{0}^{T}\|\boldsymbol{\chi}(\tau)\|\mathrm{d}\tau$, and
    \begin{align}
        Q_{\mathbf{G}'} &= \mathcal{O}\Big(Q_{\boldsymbol{\psi}} T\big(\alpha_\mathbf{L}\log(Q_{\boldsymbol{\psi}}/\varepsilon) + \alpha_\mathbf{H}\big) \Big)
        \\&= \widetilde{\mathcal{O}}\Big(\frac{\|\boldsymbol{\psi}(0)\| + \|\boldsymbol{\chi}\|_{L^1}}{\|\boldsymbol{\psi}(T)\|}T \lb \alpha_\mathbf{L} + \alpha_\mathbf{H} \rb\Big)
    \end{align}
    queries to the (controlled) index oracles $O_r^{(\mathcal{M})},O_c^{(\mathcal{M})}$ for $\mathcal{M} \in \{\mathbf{S}_A, \mathbf{S}_B\}$ and (controlled) $\mathcal{O}(\varepsilon/Q_{\mathbf{G}'})$-accurate phase oracles $O_{\mathbf{M}^{-1/2}}, O_{\mathbf{K}^{1/2}}, O_{\boldsymbol{\eta}_0}, O_{\boldsymbol{\Lambda}}$.
    The algorithm requires $\mathcal{O}\big(\log(N \alpha_{\mathbf{L}} T Q_{\boldsymbol{\psi}}/\varepsilon)\big)$ ancilla qubits and uses $\mathcal{O}\Big(\big(Q_{\mathbf{G}'} + Q_{\boldsymbol{\psi}}\log^{5/2}(Q_{\boldsymbol{\psi}}/\varepsilon)\big)\log(N \alpha_{\mathbf{L}} T Q_{\boldsymbol{\psi}}/\varepsilon)\Big)$ additional two-qubit gates.
\end{theorem}

The similarity transform \Cref{eq:sim_trans} maps the physical energy to the squared Euclidean norm, $E(t) = \frac{1}{2}\|\boldsymbol{\psi}(t)\|^2$.
Because the material matrix $\mathbf{B}$ is diagonal, local degrees of freedom do not mix.
Hence, the energy stored in any subsystem can be determined as the support of $\boldsymbol{\psi}(T)$ on the corresponding subspace:

\begin{problem}[Energy Estimation Problem \cite{babbush2023exponential}, Proof in \Cref{apx:measurements}] \label{prob:subspace_energy}
    Given the same oracles as for the simulation problem (\Cref{th:_lchs_state_prep_applied}), and a (controlled) marking oracle $O_{\mathcal{S}}$ that flips the sign of the target physical subspace $\mathcal{S} \subseteq \{1, \dots, N\}$ such that $O_{\mathcal{S}} = \mathbf{I} - 2\mathbf{P}_{\mathcal{S}}$, where $\mathbf{P}_{\mathcal{S}} = \sum_{j \in \mathcal{S}} \ketbra{j}{j}$ is the projector onto $\mathcal{S}$, output an estimate $\Xi_{\mathcal{S}}(T)$ such that
    \begin{equation}
        {\left|\Xi_{\mathcal{S}}(T) - \frac{E_{\mathcal{S}}(T)}{E_{\operatorname{tot}}}\right| \leq \varepsilon,}
    \end{equation}
    where $E_{\mathcal{S}}(T)$ is the energy stored in the subsystem $\mathcal{S}$ at time $T$
    \begin{equation}
        E_{\mathcal{S}}(T) := \frac{1}{2}\sum_{j \in \mathcal{S}} |\psi_j(T)|^2,
    \end{equation}
    and $E_{\operatorname{tot}}= \tfrac{1}{2}(\|\boldsymbol{\psi}(0)\| + \|\boldsymbol{\chi}\|_{L^1})^2$ is the maximum theoretical energy.
\end{problem}

Building on \cite{babbush2023exponential}, we prove the following result (\Cref{th:energy_estimation}) in \Cref{apx:measurements}:
\begin{theorem}[Subspace Energy Estimation, see \cite{babbush2023exponential}] \label{th:energy_estimation_applied}
    Assume that $T \alpha_{\mathbf{L}} \in \Omega(1)$.
    For any target error $\varepsilon \in (0,1)$ and failure probability $\delta \in (0,1)$, Problem \ref{prob:subspace_energy} can be solved with probability at least $1-\delta$ using 
    \begin{equation}
        {Q_{\Xi} = \mathcal{O}\bigg(\frac{\log(1/\delta)}{\varepsilon}\bigg)}
    \end{equation}
    queries to the (controlled) marking oracle $O_{\mathcal{S}}$ and to the $\mathcal{O}(\varepsilon/Q_{\Xi})$-accurate state preparation oracles $O_{\boldsymbol{\psi}_0}$, $O_{\boldsymbol{\chi}_\tau}$, and $O_{\boldsymbol{\chi}_x}$, and
    \begin{align}
         Q_{\mathbf{G}''} &= \mathcal{O} \Big(Q_{\Xi} T\big(\alpha_\mathbf{L}\log(Q_{\Xi}/\varepsilon) + \alpha_\mathbf{H}\big) \Big) \\
          &= \widetilde{\mathcal{O}} \left( \frac{\log(1/\delta)}{\varepsilon} T \lb \alpha_\mathbf{L} + \alpha_\mathbf{H} \rb \right)
    \end{align}
    queries to the (controlled) index oracles $O_r^{(\mathcal{M})},O_c^{(\mathcal{M})}$ for $\mathcal{M} \in \{\mathbf{S}_A, \mathbf{S}_B\}$ and (controlled) $\mathcal{O}(\varepsilon/Q_{\mathbf{G}''})$-accurate phase oracles $O_{\mathbf{M}^{-1/2}}, O_{\mathbf{K}^{1/2}}, O_{\boldsymbol{\eta}_0}, O_{\boldsymbol{\Lambda}}$.
    The algorithm requires $\mathcal{O}\big(\log(N \alpha_{\mathbf{L}} T Q_{\Xi}/\varepsilon)\big)$ ancilla qubits and uses $\mathcal{O}\Big(\big(Q_{\mathbf{G}''} + Q_{\Xi}\log^{5/2}(Q_{\Xi}/\varepsilon)\big)\log(N \alpha_{\mathbf{L}} T Q_{\Xi}/\varepsilon)\Big)$ additional two-qubit gates.
\end{theorem}

In the next section we will analyze the advantage potential of these quantum algorithms compared to classical simulation.

\section{Runtime Analysis} \label{s:runtime}
Above we provide estimates of the scaling of the number of quantum gates needed to estimate the mean energy stored in a subsystem; however, it may not be clear whether this problem is computationally challenging.  
At first glance the resources needed to simulate the problem seems to provide a potential exponential advantage as the scaling in \Cref{th:energy_estimation_applied} does not have explicit polynomial scaling with the dimension of the Hilbert space that we are simulating. 
However, it may not be clear whether or not classical systems can perform competitively to this algorithm.  

Our aim here is to compare and contrast the quantum performance to a new classical simulation algorithm based on information locality arguments for differential equations against which the quantum algorithm can provide only a polynomial advantage for finite dimensional lattices; although for three dimensional problems this advantage is a potentially significant quartic speedup.  
Then we will argue for \BQP-completeness for such dissipative simulations through an argument based on a clock-state that explicitly uses dissipation and so shows hardness even in the regime of non-negligible damping.

\subsection{Classical Bounds}

In this subsection we will derive numerical truncation bounds for classical simulations based on causality arguments, i.e., what time it takes a signal to propagate from an input location to some output location, and dissipation arguments, i.e., after which time an input signal has decayed so significantly that its influence on an output is negligible. 
This allows us to upper bound classical simulation costs for problems that exhibit, for example, finite-dimensional locality of interactions or strong dissipation, which consequently limits the potential for quantum advantage in such settings.

We provide an extensive proof in the appendix that local differential equations with non-positive log-norms lead to dynamics that are confined to a small causal cone and the dependence of a time-evolved local observable outside of that lightcone shrinks exponentially with distance from the light cone (see \Cref{fig:lrb}).  
In this sense, this provides an analogue of a Lieb-Robinson bound for inhomogeneous (source-driven) differential equations under a restricted set of conditions.  
We state the main result for this as the following informal result.  See \Cref{app:LRB} and \Cref{lem:exactPartition} in particular for detailed derivation and statement of information locality arguments for inhomogeneous differential equations.

\begin{proposition}[Information Locality Bound (informal)] \label{lem:exactPartitionInformal}
    Let $\mathbf{Q}$ be a Hermitian operator acting on a finite dimensional Hilbert space with a metric $\mathcal{D}(u,v)$ on it.  Assume that $\mathbf{Q}$ is supported only on a subspace $\mathcal{W}$ of diameter at most $\Delta$ and let $\mathcal{V}^\perp$ be a subspace that is separated from $\mathcal{W}$ by at least distance $d\ge \min_{u\in \mathcal{W},v\in \mathcal{V}^\perp}(\mathcal{D}(u,v))$. Further, let the global initial state be partitioned as $\boldsymbol{\psi}(0) = \boldsymbol{\psi}_{\mathcal{V}}(0) + \boldsymbol{\psi}_{\mathcal{V}^\perp}(0)$, where $\boldsymbol{\psi}_{\mathcal{V}^\perp}(0)$ has strict spatial support within $\mathcal{V}^\perp$. Further, let $$\partial_t \boldsymbol{\psi}(t) = -\mathbf{C} \boldsymbol{\psi}(t) + \boldsymbol{\chi}(t) $$ where $\mathbf{C}$ only couples vectors within distance $\Delta$ of each other.

    Let $\boldsymbol{\Upsilon}_{\mathcal{V}}(t) := \int_0^t e^{\mathbf{C} \tau} \boldsymbol{\chi}_{\mathcal{V}}(\tau) \mathrm{d}\tau$ and let $\boldsymbol{\psi}_{\mathcal{V}}(t) := e^{-\mathbf{C} t} \big( \boldsymbol{\psi}_{\mathcal{V}}(0) + \boldsymbol{\Upsilon}_{\mathcal{V}}(t) \big)$ be the exact component of the state driven solely by the local initial state and local source. 
    We then have for time $t \ge 0$ and $v_I t := 2e\Delta\|\mathbf{C}\|t < d - \Delta$ that the difference in the expectation value of the observable $\mathbf{Q}$ from that due solely to the local state components satisfies
    $$
    \begin{aligned}
    &| {\boldsymbol{\psi}}^\dagger(t) \mathbf{Q} {\boldsymbol{\psi}}(t)-\boldsymbol{\psi}_{\mathcal{V}}^\dagger(t) \mathbf{Q} \boldsymbol{\psi}_{\mathcal{V}}(t)| \nonumber\\
    &\qquad\le 12\|\mathbf{Q}\| C_{\max}^2 e^{-\frac{d -\Delta}{\Delta}\log\left(\frac{d-\Delta}{2e\Delta \|\mathbf{C}\| t} \right)},
    \end{aligned}
    $$
    where $C_{\max} := \max(\|\boldsymbol{\psi}(0)\|, \|\boldsymbol{\chi}\|_{L^1})$.
\end{proposition}

This then leads to our main result involving classical simulation. 
We provide a classical simulation algorithm that uses local Taylor-series expansions and sparse matrix-vector multiplication to evolve the differential equation strictly within the effective light cone of a fixed observable. 
This result is stated below.

\begin{corollary}[Informal Statement of \Cref{cor:classical}] \label{cor:classicalInf}
Let us assume that we are only interested in computing the dynamics of a local observable $\mathbf{Q}$ that has unit norm supported only on a compact subspace $\mathcal{W}$ of constant dimension ${\rm dim}(\mathcal{W})$ embedded in a $D$-dimensional lattice. Let us assume that the following conditions hold:
\begin{enumerate}
    \item The generator $-\mathbf{C}$ is $s$-sparse, couples elements of the space that are strictly within distance $\Delta$ on the lattice, and is bounded by $\|\mathbf{C}\|\le \alpha$.
    \item The logarithmic norm of the generator is non-positive, bounded by $\mu(-\mathbf{C}) \le -\lambda_{\min} \le 0$.
    \item The continuous source term $\boldsymbol{\chi}$ obeys $\sup_{\tau \in [0,t]}\|\partial_\tau^k\boldsymbol{\chi}(\tau)\|\le A_{\boldsymbol{\chi}} \omega_{\boldsymbol{\chi}}^k k!$ for all non-negative integers $k$ and the cost of computing it is proportional to the size of the simulated system.
    \item The cost of simulating the non-Markovian dynamics over the initial time $t^*$ is sub-dominant to the simulation duration $t$.
\end{enumerate}
Under the assumption that the cost of computing the local inhomogeneity is efficient, the number of classical bit operations needed  is bounded by
\begin{equation}
    \begin{aligned}
\tilde{\mathcal{O}}&\left( s \Delta^D (\alpha + \omega_{\boldsymbol{\chi}})t\Biggr(\alpha t + \log\left(\frac{\|\boldsymbol{\psi}(0)\|+A_{\boldsymbol{\chi}}t}{\epsilon} \right)\right)^D\nonumber\\
&\qquad\times \log^3\left(\frac{\|\boldsymbol{\psi}(0)\| + A_{\boldsymbol{\chi}}t}{\epsilon}\right)\Biggr).
    \end{aligned}
\end{equation}
\end{corollary}

In terms of the simulation time $t$, any quantum algorithm that scales at least linearly in $t$ provides a speedup factor of at most $\mathcal{O}(t^D)$ relative to this classical algorithm. 
For a $3$-dimensional lattice ($D=3$), this equates to a classical algorithm scaling as $\mathcal{O}(t^4)$ against a linear quantum algorithm, representing a quartic speedup. 
While potentially substantial, this demonstrates that the exponential advantages one may hope for when simulating linear differential equations may not be realizable when computing local properties of generic local differential equations in finite dimensions. 

Further, previous work suggested that inhomogeneities could allow quantum differential equation algorithms to provide an exponential speedup because steady source terms ensure $\lim_{t\rightarrow \infty} \boldsymbol{\psi}(t) \ne 0$. 
Here we observe that this condition alone is not sufficient to lead to an exponential quantum advantage in finite dimensions, as the spatial locality of the observable still permits an efficient light-cone truncation. 
Either a non-local generator or a combination of a globally driven system with a non-local observable may be needed to provide a super-polynomial advantage over classical algorithms.
The latter might be infeasible due to limited expectation variance in systems that have been evolved over short times \cite{sakamoto2025quantum}.

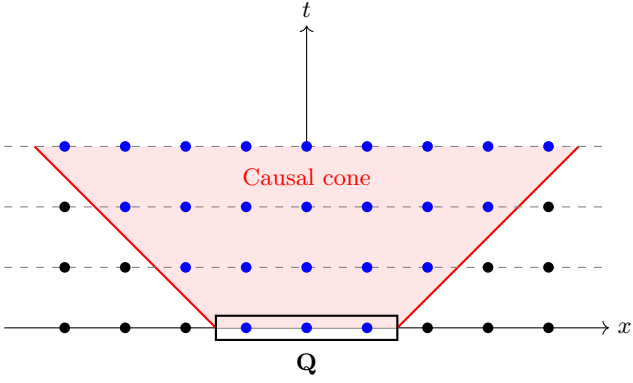
\begin{figure}[t]
\begin{center}
\begin{tikzpicture}[scale=0.8]

    \draw[->] (-5,0) -- (5,0) node[right] {$x$ };
    \draw[->] (0,0) -- (0,5) node[above] {$t$};

    \fill[red!10] (-1.5,0) -- (1.5,0) -- (4.5,3) -- (-4.5,3) -- cycle;

    \draw[thick, red] (1.5,0) -- (4.5,3);
    \draw[thick, red] (-1.5,0) -- (-4.5,3);

    \draw[thick] (-1.5,0.2) rectangle (1.5,-0.2);
    \node at (0,-0.6) {$\mathbf{Q}$};

    \foreach \t in {0,1,2,3} {

        \ifnum\t>0
            \draw[dashed, gray] (-5,\t) -- (5,\t);
        \fi

        \pgfmathsetmacro{\xmin}{-1.5 - \t}
        \pgfmathsetmacro{\xmax}{ 1.5 + \t}

        \foreach \x in {-4,-3,-2,-1,0,1,2,3,4} {

            \pgfmathparse{(\x >= \xmin) && (\x <= \xmax) ? 1 : 0}
            \ifnum\pgfmathresult=1
                \fill[blue] (\x,\t) circle (2.5pt);
            \else
                \fill[black] (\x,\t) circle (2.5pt);
            \fi
        }
    }

    \node[red] at (0.0,2.5) {Causal cone};

\end{tikzpicture}
\end{center}
\caption{Diagram of a light cone for an observable in one-dimension predicted via an information locality bound. The contribution to the evolution of the time-dependent observable $\mathbf{Q}$ from the subsystems outside the shaded area is exponentially small in the distance from the boundary of the light cone.}\label{fig:lrb}
\end{figure}

The polynomial temporal advantage cited above assumes that we have at least one mode in the system that has negligible dissipation. 
However, we must ask whether a polynomial temporal quantum advantage can be retained in the event that the system is strictly dissipative. 
If the damping is extremely small, the temporal memory horizon can exceed the simulation time, allowing the polynomial advantage to persist. 
However, when the damping rate is lower-bounded by a constant ($\lambda_{\min} \in \Omega(1)$), we find that the polynomial temporal scaling fundamentally disappears. 
This occurs because the local dynamics of the system only need to be simulated up to a threshold temporal memory, beyond which the initial states and historical dynamics become irrelevant.

We provide an extensive argument formalizing this in \Cref{app:LRB}. 
By applying Duhamel's principle, the state evaluated at time $t$ can be truncated by discarding history prior to $t_0 = \max(0, t - t^*)$. 
Because the logarithmic norm is bounded by $-\lambda_{\min}$, the error introduced by discarding this history decays exponentially as
\begin{equation}
    \|\boldsymbol{\psi}(t) - \check{\boldsymbol{\psi}}(t)\| \le e^{-\lambda_{\min} t^*} \|\boldsymbol{\psi}(t_0)\| \le 2C_{\max} e^{-\lambda_{\min} t^*},
\end{equation}
where $C_{\max}$ is an upper bound on the global state and source norms. 
Setting this error bound to a fraction of the desired precision $\epsilon$ defines the required memory horizon
\begin{equation}
    t^* = \frac{1}{\lambda_{\min}} \log \left( \frac{4C_{\max}}{\epsilon} \right).
\end{equation}
This implies that instead of tracking the dynamics from the origin, we only need to simulate the interval $[t-t^*, t]$.

Crucially, because information propagates spatially through the lattice at a bounded information velocity $v_I =\mathcal{O}( \Delta\|\mathbf{C}\|)$, a truncated temporal horizon $t^*$ bounds the required spatial simulation volume. 
As detailed in \Cref{thm:stateTruncation}, ensuring a spatial truncation error $\le \epsilon$ requires expanding the simulation radius around the observable to a distance $d_*$ governed by the Lambert $W$ function, which scales asymptotically as $d_* \in \mathcal{O} (\Delta\|\mathbf{C}\| t^*)$.
Thus, the maximum volume of the space that we need to simulate is $\mathcal{O}(d_*^D)$.  The cost of finding the solution is best case scenario for spatially local coupling in the lattice is $\mathcal{O}(d)$ this leads to $\mathcal{O}(d_*^{D+1}) \subseteq \mathcal{O}(t^{D+1})$.  In contrast, the quantum algorithms that we propose can solve the problem in time that scales like $\mathcal{O}(t)$, which implies that there can be as much as quartic advantages in time scaling relative to classical simulations on this reduced space.

Because $C_{\max} \le \|\boldsymbol{\psi}(0)\| + A_{\boldsymbol{\chi}} t$ grows at most linearly with time, for a constant damping rate $\lambda_{\min} \in \Omega(1)$, both the memory horizon $t^*$ and the truncation radius $d_*$ grow only logarithmically with $t$. 
This implies that the classical computational cost to simulate local observables in strictly dissipative systems drops from polynomial to $\tilde{\mathcal{O}}(\log^{D+1}(t))$, retaining only a poly-logarithmic dependence on the evolution time and error tolerance. 
In this regime, the temporal quantum advantage is reduced to at most a poly-logarithmic factor.

Note here that the topology matrix, which describes the coupling between the elements, does not change over time this implies that advantages are unlikely due to modest time-dependent changes.  However, if we considered a model for time-dependence where the topology matrix drifts over time to one that has highly non-local couplings then substantial quantum advantage could be possible.

\subsection{\BQP-Completeness}

Prior work provided strong evidence that simulating exponentially many undamped~\cite{babbush2023exponential} or weakly damped~\cite{krovi2024quantum} coupled classical oscillators is easy for quantum computers but hard for classical computers in the worst case. Specifically, they prove that deciding whether the kinetic energy of a subset of oscillators is exponentially small or polynomially large is \BQP-complete.
Here we extend these hardness results to coupled classical oscillators with strong localized dissipation, where a $1/\mathrm{polylog(M)}$-fraction of the masses is strongly damped and the remaining masses are undamped; though we expect that we can combine our analysis with the one in Ref.~\cite{krovi2024quantum} to extend our result to the case where the remaining masses are weakly damped.

More precisely, we prove \BQP-completeness for the problem of deciding whether the total energy of the oscillator system has decreased significantly after a polynomial amount of time evolution. Below we give a slightly more formal statement of this problem but leave the precise problem statement for Appendix~\ref{app:BQP}.

\begin{problem}[Dissipative oscillators decision problem; simplified]
    Consider a system of $M= 2^n$ coupled classical oscillators as described in~\Cref{sec:embed}. Assume that each mass is connected to at most a constant number of other masses. Further, assume that a $1/\mathrm{polylog(M)}$-fraction of masses is strongly damped with damping coefficients, i.e.~the viscosities of the dashpots, scaling like $\Theta \lb \poly(n) \rb$. The problem is to decide whether after time $T \in \cO \lb \poly(n) \rb$, the total energy of the system is at least $2/3$ or at most $1/3$ of the initial energy, under the promise that one of these holds.
\label{bqp_prob_simplified}
\end{problem}

We prove \BQP-completeness for the above problem by reduction from a standard \BQP-complete problem where we consider a $\poly(q)$-sized quantum circuit $V$ on $q \in \cO \lb n \rb$ qubits with initial state $\ket{0}^{\otimes q}$ and the task is to decide whether the probability of measuring the first qubit in the state $\ket{1}$ is at least $2/3$ or at most $1/3$. We only provide a brief proof sketch here. The detailed proof can be found in Appendix~\ref{app:BQP}. Our overall proof strategy is similar to the original \BQP-completeness proof for undamped oscillators in Ref.~\cite{babbush2023exponential} in the sense that we also utilize the Feynman-Kitaev clock Hamiltonian to reduce the standard circuit \BQP-complete problem to our oscillator problem. To be more specific, consider some quantum circuit $V = U_L \cdots U_1$ where $L \in \cO \lb \poly(q) \rb \subseteq \cO \lb \poly(n) \rb$ and each $U_l$ with $l \in \{1, 2, \dots, L\}$ is a gate from a universal gate set. We choose $\{ X, \text{Hadamard}, \text{Toffoli} \}$ as our universal gate set.
Then the Feynman-Kitaev clock Hamiltonian encoding $V$ is given by
\begin{equation}
    \bH_\cl = \sum_{l=1}^L \lb \ketbra{l+1}{l} + \ketbra{l}{l+1} \rb \otimes U_l,
\end{equation}
where the first qubit register is the clock register and the second qubit register is the computational register. 
Without loss of generality, we can assume that the last gate $U_L$ performs a Toffoli gate targeted onto the last qubit of the computational register, which we will call the ``answer qubit''. Further, we can assume that no other gate acts on this answer qubit. $U_L$ leaves the answer qubit in the state $\ket{0}$ for NO instances (with high probability) and flips it to $\ket{1}$ for YES instances (with high probability). Following Ref.~\cite{babbush2023exponential}, we then show how to encode a slightly modified version of $\bH_\cl$ into a system of coupled classical oscillators. In contrast to Ref.~\cite{babbush2023exponential}, we also add a diagonal dissipator matrix to the generator of time evolution. The dissipator matrix is constructed such that it only damps states whose answer qubit is flipped to $\ket{1}$. It describes purely viscous damping through ground dashpots.
After sufficiently long time evolution, YES instances thus experience significant damping, resulting in
the norm of the evolved vector $\frac{\boldsymbol{\psi}(T)}{\norm{\boldsymbol{\psi}(0)}} = e^{-\bC T} \ket{\boldsymbol{\psi}(0)}$ decaying exponentially to some relatively small constant of our choosing, e.g.~$1/6$.
NO instances on the other hand experience essentially no damping, meaning that the norm of $\frac{\boldsymbol{\psi}(T)}{\norm{\boldsymbol{\psi}(0)}}$ remains close to $1$. Measuring the norm of $\frac{\boldsymbol{\psi}(T)}{\norm{\boldsymbol{\psi}(0)}}$ within constant error can therefore be used to distinguish a YES from a NO instance.

Unlike the undamped case, which involves only normal matrices (such as Hermitian or unitary matrices), we have to analyze the eigenvalues and eigenvectors of non-normal matrices due to the addition of the dissipator matrix which is more challenging. Ultimately, we prove that the real parts of all eigenvalues $\{ \lambda_j\}$ of $-\bC$ are negative in the YES subspace and scale like $\Omega \lb 1/\poly(n) \rb$. This indicates that the required evolution time $T$ needed to distinguish a YES from a NO instance scales only like $\cO \lb \poly(n) \rb$. However, this lower bound on the eigenvalues is not sufficient to guarantee that $T$ does indeed only scale like $\cO \lb \poly(n) \rb$. The issue is that the time evolution operator $e^{-\bC T}$ is not normal which means that its eigenvectors are not orthogonal and can therefore introduce additional $n$-dependent scaling in the spectral decomposition of $e^{-\bC T}$. We therefore also analyze the left and right eigenvectors, $\{\vec{\mathcal{R}}_{j}\}$ and $\{ \vec{\mathcal{L}}_{j}\}$ of $-\bC$ in detail. In particular, we prove lower bounds on $\vec{\mathcal{L}}_j^\top \vec{\mathcal{R}}_j$. Since $\vec{\mathcal{R}}_j$ and $\vec{\mathcal{L}}_j$ are complex valued, the expression $\vec{\mathcal{L}}_j^\top \vec{\mathcal{R}}_j$ is not a proper inner product. However, these quantities appear in the spectral decomposition of $-\bC$ when written in terms of left and right eigenvectors. Specifically,
\begin{equation}
    e^{-\bC t} = \sum_j e^{\lambda_j t} \frac{\vec{\mathcal{R}}_j \vec{\mathcal{L}}_j^\top}{\vec{\mathcal{L}}_j^\top \vec{\mathcal{R}}_j}.
\end{equation}
We pick an easy-to-prepare initial state $\ket{\boldsymbol{\psi}(0)}$ describing a system of coupled oscillators where all masses are initially in their rest positions, the first mass has velocity $+1$, the second mass has velocity $-1$ and all other masses do not have any initial velocity. These are the same initial conditions as in Ref.~\cite{babbush2023exponential}.
We then show that $\norm{e^{-\bC T} \ket{\psi_0}}^2 \leq 1/3$ for YES instances and $\norm{e^{-\bC T} \ket{\psi_0}}^2 \geq 2/3$ for NO instances with $T \in \cO \lb \poly(n) \rb$. Furthermore, it is not difficult to  show that $e^{-\bC T}$ can be efficiently implemented for the $\BQP$-hard oscillator instances, ensuring that the problem is also contained in $\BQP$ and hence is \BQP-complete.

\subsection{Kernel Approximation} \label{s:kernel_approx}

The physical validity of the generalized Maxwell model depends on how well the memory kernel $G(t)$ can be approximated by a finite Prony series (\Cref{eq:prony_series}).
By Bernstein's theorem, any causal, completely monotonic kernel can be represented as a continuous mixture of decaying exponentials, and consequently, it can be systematically approximated by a finite sum of decaying exponentials with strictly positive stiffnesses $k_c > 0$ and relaxation rates $\lambda_c \geq 0$ \cite{bernstein1929fonctions, kammler1976chebyshev, hanyga2005viscous}.
Analogous to our approach, any classical system that simulates this Prony series approximation needs to store $C$ sets of memory variables.
Therefore, based on our results \Cref{th:_lchs_state_prep_applied} and \Cref{th:energy_estimation_applied} one might expect a further up to exponential advantage with respect to the number of parallel channels $C$.
Here we show that a large class of memory kernels can be approximated exceptionally well with only a small number of channels.

Furthermore, there exist other methods of solving the hereditary integral \Cref{eq:law_hereditary_general}, for example, \ac{FFT} based methods \cite{hairer1985fast}.
If the memory kernel is sufficiently smooth, they obtain a cost $\mathcal{O}(T \log^2(T))$ with memory $\mathcal{O}(T)$ \cite{hairer1985fast} instead of Prony's method $\mathcal{O}(T C)$ with memory $\mathcal{O}(C)$.

We restrict our analysis to the class of locally integrable kernels.
Non-integrable kernels, such as $G(t) \propto 1/t^2$, are physically pathological as they yield infinite internal forces and require infinite energy for finite kinematic displacements.
To capture the asymptotic worst-case convergence behavior within this admissible class, we define the kernel $G(t) \propto t^{-\gamma} L(t)$, where $0 < \gamma < 1$, and $L(t)$ is slowly varying \cite{bingham1989regular}, meaning its relative variation is asymptotically slower than any power law $\lim_{t\to\infty} L(at)/L(t) = 1$ for all $a>0$. 
While there is an infinite number of valid choices for $L(t)$, here we discuss two highly relevant cases: fractional kernels with $L(t)=1,\; \gamma \rightarrow 1^-,\; t \rightarrow 0^+$, which physically correspond to constant-$Q$ attenuation where the quality factor is $Q = \cot(\pi\gamma/2)$ \cite{kjartansson1979constant}, and slowly decaying logarithmic kernels with $L(t)=1/\log(j+t/\beta),\; \gamma \rightarrow 0^+,\; t \rightarrow \infty$.
These describe very slow creep processes.
Here $\beta>0$ is a physical short-time cutoff, below which the memory response is effectively instantaneous and is absorbed into the  viscous term $\eta_0 \delta(t)$ as a zeroth-order quasi-static approximation, and $j > 1$ is a dimensionless shift parameter chosen to prevent singularities near $t=0$. 

\paragraph{Fractional Kernels}

Due to the short-time truncation, the convolution is evaluated over the finite temporal range $[\beta, T]$, where $T$ is the total simulation time.
On this domain, the Prony series converges exceptionally fast.
Specifically, the number of terms $C$ required to approximate fractional kernels $t^{-\gamma}$ on the interval $[\beta, T]$ with a maximum relative error $\epsilon_{\text{rel}}$ scales as $C = \mathcal{O}\big(\log^2(1/\epsilon_{\text{rel}}) + \log(1/\epsilon_{\text{rel}}) \log(T/\beta) \big)$ \cite{braess2005approximation, beylkin2005approximation}.
To guarantee a uniform absolute error $\varepsilon$, we relate it to the relative error via the maximum value of the fractional kernel on our interval, which evaluates to $\beta^{-\gamma}$ at the short-time cutoff. 
Thus, the required relative error is $\epsilon_{\text{rel}} = \varepsilon / \beta^{-\gamma} = \varepsilon \beta^{\gamma}$. 
Substituting this into the relative error bound and applying the logarithm yields the required state-space expansion dimension
\begin{equation} \label{eq:braess_bound}
\begin{aligned}
    C = \mathcal{O}\Big(\big[&\log(1/\varepsilon) + \gamma\log(1/\beta)\big]^2 \\+& \big[\log(1/\varepsilon) + \gamma\log(1/\beta)\big] \log(T/\beta) \Big).
\end{aligned}
\end{equation}
We observe that $C$ is at worst polylogarithmic in all parameters. 
This limits the potential for additional quantum advantage to polylogarithmic order.

\paragraph{Logarithmic Kernels}

Applying the universal theoretical bound for the best uniform approximation of completely monotonic functions \cite{kammler1976chebyshev}, the maximum absolute error decays algebraically with the state-space dimension $C$ as $\mathcal{O}(G(\beta)/C)$.
Equating this error bound to the target tolerance $\varepsilon$ and evaluating the shifted origin $G(\beta) \propto 1/\log(j + \beta/\beta)$ yields the required dimension~\cite{kammler1976chebyshev} 
\begin{equation} \label{eq:kammler_bound}
    {C = \mathcal{O}\left(\frac{1}{\varepsilon \log(j+1)} \right).}
\end{equation}
Notably the inverse error is linear in $C$ which might suggest advantages due to a potentially large necessary number of Maxwell bodies for small errors.
However, this bound is potentially very loose, as the constructive proof in \cite{kammler1976chebyshev} relies on a first-order step-function discretization of the underlying spectral measure.
Furthermore, one can always switch to \ac{FFT}-based methods \cite{hairer1985fast} in case $\tfrac{1}{\varepsilon} = \omega(\log^2(T))$ to reduce costs.
As such, the additional runtime advantage is likely at best polylogarithmic. 
In contrast, the same does not trivially hold for memory advantages, which could be exponential.
Algorithms for solving the hereditary integral with logarithmic memory requirements \cite{schadle2006fast} are not trivially applicable due to their requirement of cheap access to the Laplace transform of $G(t)$.

\section{Time-Dependent Materials}

One natural extension of the formalism \ref{s:osc_systems} is to incorporate time-dependent materials.
We show that this leads to changes of the system energy metric that appears equivalently to dissipation or growth.

Let us first consider the fundamental relations that describe the evolution of a single oscillator.
Recall that Newton's law describes the rate of change of momentum of the point mass
\begin{equation}
    \frac{\mathrm{d}(m(t)v(t))}{\mathrm{d} t} = b_0(t)-f(t).
\end{equation}
As $m(t)$ is now time-dependent, it does not commute with the time derivative $\frac{\mathrm{d}}{\mathrm{d} t}$.
Instead, we must use the product rule and rearrange to get the time-dependent form
\begin{equation}
    {\frac{\mathrm{d}v(t)}{\mathrm{d} t} = m^{-1}(t)\Big(b_0(t)-f(t) - \frac{\mathrm{d} m(t)}{\mathrm{d} t}v(t)\Big).}
\end{equation}

Similarly, the internal force $f(t)$ needs to correctly capture time-varying material parameters.
We start with the general time-dependent hereditary integral including body-specific velocity sources $b_c(t)$.
To embed this into a Markovian system, we decompose the kernel into an instantaneous viscous part and a series of relaxation bodies 
\begin{equation}
    G(t, \tau) = \eta_0(t)\delta(t-\tau) + \sum_{c=1}^C G_c(t, \tau).
\end{equation}
This splits the total resisting force into 
\begin{equation}
    f(t) = \eta_0(t)v(t) + \sum_{c=1}^C f_c(t),
\end{equation} 
where
\begin{equation}
    f_c(t) = \int_{-\infty}^{t} G_c(t, \tau) (v(\tau) + b_c(\tau)) \,\mathrm{d}\tau,
\end{equation}
is the force of a single Maxwell body.
Substituting the total force back into Newton's law gives
\begin{equation}
    {\frac{\mathrm{d}v(t)}{\mathrm{d} t} = m^{-1}(t)\bigg(b_0(t) - \hat{\eta}_0(t)v(t) - \sum_{c=1}^Cf_c(t) \bigg),}
\end{equation}
where $\hat{\eta}_0(t) = \eta_0(t) + \frac{\mathrm{d} m(t)}{\mathrm{d} t}$ is the effective viscous drag.

To find the evolution law for the memory terms, we differentiate $f_c(t)$ using the Leibniz integral rule
\begin{equation} \label{eq:leibniz_general}
\begin{aligned}
    \frac{\mathrm{d} f_c(t)}{\mathrm{d} t} &= G_c(t, t) (v(t) + b_c(t)) \\&+ \int_{-\infty}^{t} \frac{\partial G_c(t, \tau)}{\partial t} (v(\tau) + b_c(\tau)) \mathrm{d}\tau.
\end{aligned}
\end{equation}

A valid Markovian embedding exists \emph{only} if the remaining integral can be expressed algebraically in terms of the current state $f_c(t)$. 
This requires the kernel to be separable of the form 
\begin{equation}
    \frac{\partial G_c}{\partial t} = -\hat{\lambda}_c(t) G_c(t, \tau).
\end{equation}
In other words, the evolution of the material's properties must act on the system in a Markovian way.
Otherwise, a nested Markovian embedding is necessary.

Under this Markovian condition, \Cref{eq:leibniz_general} simplifies to the generic evolution equation
\begin{equation}
    \frac{\mathrm{d}f_c(t)}{\mathrm{d} t} = G_c(t, t) (v(t) + b_c(t)) - \hat{\lambda}_c(t) f_c(t).
\end{equation}
For an aging Maxwell element, i.e. a coupling whose stiffnesses and viscosities change over time, 
the form of the kernel is dictated by the constitutive differential equation
\begin{equation}
    v(t) + b_c(t) = \frac{\mathrm{d}}{\mathrm{d}t}\Big(f_c(t)/k_c(t)\Big) + f_c(t)/\eta_c(t).
\end{equation}
Solving this ODE for $f_c$ identifies the unique physically consistent kernel as
\begin{equation}
    G_c(t, \tau) = k_c(t) \exp\left( -\int_{\tau}^{t} \lambda_c(s) \mathrm{d}s \right).
\end{equation}
Differentiating this kernel identifies the effective relaxation rate as $\hat{\lambda}_c(t) = \lambda_c(t) - \frac{\mathrm{d}}{\mathrm{d} t}\ln(k_c(t))$ with $\lambda_c(t)=\tfrac{k_c(t)}{\eta_c(t)}$. 
Here, the intrinsic relaxation is offset by the relative change rate of the stiffness, giving the final form
\begin{equation}
     \frac{\mathrm{d} f_c(t)}{\mathrm{d} t} = k_c(t)(v(t) + b_c(t)) - \hat{\lambda}_c(t) f_c(t).
\end{equation}
The interpretation of this form is revealing.
The time-dependent material properties $m(t),k_c(t)$ introduce relative change rates into the equations of motion that act identical to dissipation or growth, captured by the effective rates $\hat{\eta}_0(t), \hat{\lambda}_c(t)$.
Consequently, a changing material can actively stabilize or destabilize the overall system evolution.

\subsection{Coupled System}

We will now generalize the dynamics of a single oscillator to the simulation of a coupled system of time-dependent damped oscillators using \ac{LCHS}.
We consider the case where the coupling topology does not change with time, but instead assume the material parameters themselves are time dependent.  

We redefine the time-dependent material matrix 
\begin{equation}
    \mathbf B(t) = \diag(\mathbf M^{-1}(t), \mathbf K(t)) \succ 0 \quad \forall t.
\end{equation}
Following the time-dependent single oscillator dynamics, the physical dissipation matrix for the network incorporates the material rates of change
\begin{equation}
    \boldsymbol{\eta}(t) = \diag\left(\boldsymbol{\eta}_{\mathbf{M}}(t) +  \frac{\mathrm{d}\mathbf{M}(t)}{\mathrm{d}t}, \boldsymbol{\eta}_{\mathbf{K}}^{-1}(t) + \frac{\mathrm{d}\mathbf{K}^{-1}(t)}{\mathrm{d}t}\right),
\end{equation}
where $\boldsymbol{\eta}_{\mathbf{M}}(t) = \mathbf{T}_0^\dagger \boldsymbol{\eta}_0(t)\mathbf{T}_0$.

To simulate the time-variant system, we apply the similarity transform $\boldsymbol{\psi}(t)=\mathbf{B}^{-1/2}(t)\boldsymbol{\phi}(t)$.
Unlike the time-invariant case, differentiating $\boldsymbol{\psi}$ requires the product rule, introducing a metric drift term
\begin{equation}
    \frac{\mathrm{d} \boldsymbol{\psi}(t)}{\mathrm{d} t} = \mathbf{B}^{-1/2}(t) \frac{\mathrm{d} \boldsymbol{\phi}(t)}{\mathrm{d} t} + \frac{\mathrm{d}(\mathbf{B}^{-1/2}(t))}{\mathrm{d} t} \boldsymbol{\phi}(t).
\end{equation}
Substituting the physical evolution $\frac{\mathrm{d} \boldsymbol{\phi}}{\mathrm{d} t} = \mathbf{B}(\mathbf{A} - \boldsymbol{\eta})\boldsymbol{\phi} + \mathbf{B}\mathbf{b}$ and expanding the basis derivative $\frac{\mathrm{d}}{\mathrm{d} t}(\mathbf{B}^{-1/2}) = -\frac{1}{2}\mathbf{B}^{-3/2}\frac{\mathrm{d} \mathbf{B}}{\mathrm{d}t}$ (which is valid as $[\mathbf{B},\frac{\mathrm{d}\mathbf{B}}{\mathrm{d}t}]=0$) yields
\begin{equation}
\begin{aligned}
    \frac{\mathrm{d} \boldsymbol{\psi}}{\mathrm{d} t} &= \mathbf{B}^{-1/2} \Big( \mathbf{B}(\mathbf{A}-\boldsymbol{\eta})\boldsymbol{\phi} + \mathbf{B}\mathbf{b} \Big) - \frac{1}{2}\mathbf{B}^{-3/2}\frac{\mathrm{d}\mathbf{B}}{\mathrm{d}t}\boldsymbol{\phi} \\
    &= \left( \mathbf{B}^{1/2}(\mathbf{A} - \boldsymbol{\eta})\mathbf{B}^{1/2} - \frac{1}{2}\frac{\mathrm{d}\ln(\mathbf{B}(t))}{\mathrm{d}t} \right) \boldsymbol{\psi} + \boldsymbol{\chi}.
\end{aligned}
\end{equation}
This preserves the Schrödinger form $\frac{\mathrm{d} \boldsymbol{\psi}}{\mathrm{d} t} = -\mathbf{C}(t)\boldsymbol{\psi} + \boldsymbol{\chi}$, with $\mathbf{C}(t) = i\mathbf{H}(t) + \mathbf{L}(t)$. 

The Hamiltonian $\mathbf{H}(t) = i\mathbf{B}(t)^{1/2}\mathbf{A}\mathbf{B}(t)^{1/2}$ is Hermitian, while the dissipator becomes
\begin{equation}
    \mathbf{L}(t) = \mathbf{B}^{1/2}(t)\boldsymbol{\eta}(t)\mathbf{B}^{1/2}(t) + \frac{1}{2}\frac{\mathrm{d}}{\mathrm{d}t}\ln(\mathbf{B}(t)).
\end{equation}
By substituting the explicit block forms of $\mathbf{B}(t)$ and $\boldsymbol{\eta}(t)$, the full material derivatives partially cancel with the similarity transform's metric drift, elegantly decoupling the intrinsic material dissipation from the geometric changes
\begin{equation} \label{eq:time_dep_L_blocks}
    \mathbf{L}(t) = 
    \begin{bmatrix}
        \mathbf{M}^{-1/2}\boldsymbol{\eta}_{\mathbf{M}}\mathbf{M}^{-1/2} + \frac{1}{2}\mathbf{M}^{(1)} & \mathbf{0} \\
        \mathbf{0} & \boldsymbol{\Lambda} - \frac{1}{2}\mathbf{K}^{(1)}
    \end{bmatrix},
\end{equation}
where $\boldsymbol{\Lambda}(t) = \mathbf{K}(t)\boldsymbol{\eta}_{\mathbf{K}}^{-1}(t)$ is the intrinsic relaxation rate of the Maxwell bodies and where we have defined $\mathbf{M}^{(1)}(t) = \frac{\mathrm{d}\ln(\mathbf{M}(t))}{\mathrm{d}t}$ with diagonal elements $[m^{(1)}(t)]_{pp}$ and $\mathbf{K}^{(1)}(t) = \frac{\mathrm{d}\ln(\mathbf{K}(t))}{\mathrm{d}t}$ with diagonal elements $[k^{(1)}(t)]_{qq}$ as the metric drift terms. 
Crucially, the \ac{LCHS} framework requires this final effective dissipator to be positive semi-definite ($\mathbf{L}(t) \succeq 0$). 
This physically restricts the maximum rate at which the network's materials can lose mass or gain stiffness before the metric drift destabilizes the system.

\subsection{Access Model} \label{s:access_td}

To derive the query complexity of the time-dependent case we need to generalize the oracle access model and change the \ac{LCHS} algorithm to its time-dependent variant \cite{low_optimal_2025}. 
Specifically, the \ac{LCHS} formula needs to simulate a linear combination of evolutions of the form $\mathcal{T}e^{-i \int_0^T (\mathbf{H}(s) + k\mathbf{L}(s)) \mathrm{d}s }$ for a range of $k$.  
To construct the time-ordered operator exponential we need to construct block encodings of $\mathbf{H}(t)$ and $\mathbf{L}(t)$.
As a first step we generalize the access model to the time dependent case.

As derived in the previous section, the Hamiltonian $\mathbf{H}(t)$ preserves its time-independent structural form and only requires a generalization of its diagonal material matrices $\mathbf{M}^{-1/2}(t)$ and $\mathbf{K}^{1/2}(t)$.
However, the dissipator $\mathbf{L}(t)$ (\Cref{eq:time_dep_L_blocks}) acquires new metric drift terms $\frac{1}{2}\mathbf{M}^{(1)}(t)$ and $-\frac{1}{2}\mathbf{K}^{(1)}(t)$ that do not fit into the original structure.
Consequently, we aim to reformulate the system to match the time-independent $\mathbf{L}$.
This transformation allows us to apply only minor modifications to the access model.

First, by recognizing that both the intrinsic relaxation rates $\boldsymbol{\Lambda}(t)$ and the stiffness drift $-\frac{1}{2}\mathbf{K}^{(1)}(t)$ are strictly diagonal, we introduce the combined effective relaxation rate 
\begin{equation}
    \tilde{\boldsymbol{\Lambda}}(t) = \boldsymbol{\Lambda}(t) - \frac{1}{2}\mathbf{K}^{(1)}(t).
\end{equation}
Second, without loss of generality, we can assume that all $M$ masses are connected to the ground padding the topology with zero-viscosity ground connections $\eta_{g,0}=0$ and zero-stiffness Maxwell bodies $K_{g,c}=0$ where necessary (see Appendix \ref{s:be_hamiltonian}). 
By ordering the ground edges to match the mass indices, $N_g=M$ and the ground selection matrix becomes the identity $\mathbf{S}_0 = \mathbf{I}$. 
Because the mass drift $\frac{1}{2}\mathbf{M}^{(1)}(t)$ acts as a purely diagonal geometric dissipation on the nodes, we can directly absorb it into the ground viscosity by defining the effective  ground viscosities
\begin{equation}
\begin{aligned}
    \tilde{\boldsymbol{\eta}}_{g,0}(t) &= \boldsymbol{\eta}_{g,0}(t) + \mathbf{M}^{1/2}(t)\left(\frac{1}{2}\mathbf{M}^{(1)}(t)\right)\mathbf{M}^{1/2}(t)
    \\&= \boldsymbol{\eta}_{g,0}(t) + \frac{1}{2}\mathbf{M}(t)\mathbf{M}^{(1)}(t).
\end{aligned}
\end{equation}
By collecting this into the effective  viscosities $\tilde{\boldsymbol{\eta}}_0(t) = \diag(\boldsymbol{\eta}_{e,0}(t), \tilde{\boldsymbol{\eta}}_{g,0}(t))$, the time-dependent dissipator collapses into the same algebraic form as the time-independent case
\begin{equation} \label{eq:time_dep_L_final}
    \mathbf{L}(t) = 
    \begin{bmatrix}
        \mathbf{M}^{-1/2}\mathbf{T}_0^\dagger \tilde{\boldsymbol{\eta}}_0(t)\mathbf{T}_0\mathbf{M}^{-1/2} & \mathbf{0} \\
        \mathbf{0} & \tilde{\boldsymbol{\Lambda}}(t)
    \end{bmatrix}.
\end{equation}

Because we assume the network coupling topology does not change over time, the structural access model remains identical to the time-independent case. 
Specifically, the index oracles $O_r^{(\mathcal{M})}$ and $O_c^{(\mathcal{M})}$ corresponding to the constant topological selection matrices $\mathcal{M} \in \{\mathbf{S}_A, \mathbf{S}_B\}$ are unchanged. 
Similarly, the state preparation oracles $O_{\boldsymbol{\psi}_0}$, $O_{\boldsymbol{\chi}_\tau}$, and $O_{\boldsymbol{\chi}_x}$ already support temporal discretization to capture time-varying source terms and can be used as defined previously.

The only modifications required are to the phase oracles.
To accommodate the multiplexed block-encoding required for time-dependent \ac{LCHS} \cite{low_optimal_2025}, these oracles must be generalized to controlled unitaries over a discrete time grid. 
We discretize the continuous time interval $t \in [0,T]$ with time stepsize $\Delta t > 0$, where $T/\Delta t \in \mathbb{Z}_{\ge 0}$, yielding $N_t = T/\Delta t + 1$ time points such that $t = j\Delta t$, and the phase oracles are controlled by a discrete clock register $\ket{j}$. 

\begin{definition} [Time-Dependent Phase Oracles]
Let $\mathbf{R}(t) \in \{\mathbf{M}^{-1/2}(t), \mathbf{K}^{1/2}(t), \tilde{\boldsymbol{\eta}}_0(t), \tilde{\boldsymbol{\Lambda}}(t)\}$ be a time-dependent diagonal matrix of dimension $N_{\mathbf{R}} \times N_{\mathbf{R}}$, discretized over time steps $t = j\Delta t$. 
We define the time-dependent phase vector $\boldsymbol{\theta}_{\mathbf{R}}(j\Delta t) \in \mathbb{R}^{N_{\mathbf{R}}}$, where 
\begin{equation}
    (\theta_{\mathbf{R}}(j\Delta t))_k = \arccos\left(\frac{\mathbf{R}_{kk}(j\Delta t)}{\alpha_{\mathbf{R}}}\right), 
\end{equation}
and $\alpha_{\mathbf{R}} \geq \max_{t \in [0,T]} \|\mathbf{R}(t)\|$.

A $\varepsilon'$-accurate \emph{time-dependent phase oracle} for $\mathbf{R}(t)$ is a unitary $O_{\mathbf{R}}$ acting on an $n_t=\lceil\log_2(T/\Delta t+1)\rceil$ qubit clock register and an $n_{\mathbf{R}}=\lceil\log_2 N_{\mathbf{R}}\rceil$ qubit system register, such that
\begin{equation}
\bigg\| O_{\mathbf{R}} - \sum_{j=0}^{T/\Delta t} \sum_{k=0}^{N_{\mathbf{R}}-1} e^{-i (\theta_{\mathbf{R}}(j\Delta t))_k} \ketbra{j}{j} \otimes \ketbra{k}{k} \bigg\| \le \varepsilon',
\end{equation}
where $\ket{j}$ encodes the discrete time step and $\|\cdot\|$ denotes the operator norm.
\end{definition}

\subsection{Algorithms}

Our quantum algorithms for simulating the non-Markovian system of oscillators reduces to the problem of taking these block encodings for the differential equations and then simulating them on a quantum computer using the \ac{LCHS} paradigm.  
The block-encoding of $\mathbf{H}(t)$ and $\mathbf{L}(t)$ extends the time-independent constructions to the time-dependent case via the multiplexed block-encoding formalism required by the \ac{LCHS} framework \cite[Definition 15]{low_optimal_2025}. 
Because we have absorbed the metric drift terms into the material matrices, the time-dependence is fully localized within the diagonal material operators. 
We establish that the temporal discretization is structurally isolated in a clock register $\ket{j}$, which dictates the time step $t = j\Delta t$. 
This allows the static topology block-encoding to be reused unaltered.
As in the time-independent case, all results are stated in the target-error convention: each result takes a single target error $\varepsilon$ as input and specifies the oracle accuracies and the temporal grid resolution that suffice to achieve it.

We first encode the time-dependent Hamiltonian $\mathbf{H}(t)$:
\begin{lemma}[Hamiltonian Encoding, Proof in \Cref{apx:td_block_encodings}] \label{lem:td_hamiltonian_encoding}
    For any target error $\varepsilon > 0$ and simulation duration $T$, the time-dependent Hamiltonian $\mathbf{H}(t)$ admits an $(\alpha_{\mathbf{H}}, a_{\mathbf{H}}, \varepsilon)$ block-encoding $\mathcal{U}_{\mathbf{H}(t)}$ of the form $\sum_{j} \ketbra{j}{j} \otimes \mathrm{BE}[\mathbf{H}(j\Delta t)/\alpha_{\mathbf{H}}]$, whose induced time-discretization error over $[0,T]$ is also at most $\varepsilon$, with normalization factor
    \begin{equation}
        {\alpha_{\mathbf{H}} = (\sqrt{d_{\operatorname{in}}}+\sqrt{d_{\operatorname{out}}})\sqrt{\frac{\max_{t \in [0,T]} k_{\max}(t)}{\min_{t \in [0,T]} m_{\min}(t)}C}.}
    \end{equation}
    Assuming the material parameters $m_p(t)$ and $k_q(t)$ are $C^1$-smooth with a maximum absolute metric drift  $\Gamma_{\mathbf{H}} = \frac{1}{2}\max_{t \in [0,T]}\big(|k^{(1)}(t)|_{\max} + |m^{(1)}(t)|_{\max} \big)$, it suffices to choose the temporal grid step size $\Delta t = \mathcal{O}\left(\frac{\varepsilon}{T \alpha_{\mathbf{H}} (\Gamma_{\mathbf{H}} + \alpha_{\mathbf{H}})}\right)$.
    The encoding requires $a_{\mathbf{H}} = \mathcal{O}\left(\log\left(\frac{NT \alpha_{\mathbf{H}} (\Gamma_{\mathbf{H}} + \alpha_{\mathbf{H}})}{\varepsilon}\right)\right)$ ancilla qubits, $\mathcal{O}(\log N)$ additional two-qubit gates, and $\mathcal{O}(1)$ queries to the (controlled) index oracles $O_r^{(\mathcal{M})},O_c^{(\mathcal{M})}$ for $\mathcal{M} \in \{\mathbf{S}_A, \mathbf{S}_B\}$ and to the (controlled) $\mathcal{O}(\varepsilon/\alpha_{\mathbf{H}})$-accurate time-dependent phase oracles $O_{\mathbf{M}^{-1/2}}, O_{\mathbf{K}^{1/2}}$.
\end{lemma}

Similarly, we encode the time-dependent Dissipator $\mathbf{L}(t)$. 
For this we define the effective dissipation change rates $\tilde{\boldsymbol{\eta}}_0^{(1)}(t) = \frac{\mathrm{d}\ln\tilde{\boldsymbol{\eta}}_0(t)}{\mathrm{d}t}$ with diagonal elements $[\tilde{\eta}_0^{(1)}(t)]_{pp}$ and $\tilde{\boldsymbol{\Lambda}}^{(1)}(t) = \frac{\mathrm{d}\ln\tilde{\boldsymbol{\Lambda}}(t)}{\mathrm{d}t}$ with diagonal elements $[\tilde{\lambda}^{(1)}(t)]_{qq}$:
\begin{lemma}[Dissipator Encoding, Proof in \Cref{apx:td_block_encodings}] \label{lem:td_dissipator_encoding}
    For any target error $\varepsilon \in (0,1)$ and simulation duration $T$, the time-dependent Dissipator $\mathbf{L}(t)$ admits an $(\alpha_{\mathbf{L}}, a_{\mathbf{L}}, \varepsilon)$ block-encoding $\mathcal{U}_{\mathbf{L}(t)}$ of the form $\sum_{j} \ketbra{j}{j} \otimes \mathrm{BE}[\mathbf{L}(j\Delta t)/\alpha_{\mathbf{L}}]$, whose induced time-discretization error over $[0,T]$ under the \ac{LCHS} multiplier $|k| \le R = \mathcal{O}(\log(1/\varepsilon))$ is also at most $\varepsilon$, with normalization factor
    \begin{equation}
    \begin{aligned}
        \alpha_{\mathbf{L}} = \max\bigg( &(\sqrt{d_{\operatorname{in}}}+\sqrt{d_{\operatorname{out}}})^2 \frac{\max_{t \in [0,T]} \tilde{\eta}_{0,\max}(t)}{\min_{t \in [0,T]} m_{\min}(t)}, \\ &\max_{t \in [0,T]} \tilde{\lambda}_{\max}(t) \bigg).
    \end{aligned}
    \end{equation}
    Assuming the effective parameters $m_p(t)$, $\tilde{\eta}_{0,p}(t)$, and $\tilde{\Lambda}_q(t)$ are in the continuity class $C^1$, we define the maximum absolute drift rate $\Gamma_{\mathbf{L}} = \max_{t \in [0,T]} \big(|m^{(1)}(t)|_{\max} + |\tilde{\eta}_0^{(1)}(t)|_{\max} + |\tilde{\lambda}^{(1)}(t)|_{\max} \big)$; it then suffices to choose the temporal grid step size $\Delta t = \mathcal{O}\left(\frac{\varepsilon}{T \alpha_{\mathbf{L}}\log(1/\varepsilon) (\Gamma_{\mathbf{L}} + \alpha_{\mathbf{L}}\log(1/\varepsilon))}\right)$.
    The encoding requires $a_{\mathbf{L}} = \mathcal{O}\left(\log\left(\frac{NT \alpha_{\mathbf{L}} (\Gamma_{\mathbf{L}} + \alpha_{\mathbf{L}})}{\varepsilon}\right)\right)$ ancilla qubits, $\mathcal{O}(\log N)$ additional two-qubit gates, and $\mathcal{O}(1)$ queries to the (controlled) index oracles $O_r^{(\mathcal{M})},O_c^{(\mathcal{M})}$ for $\mathcal{M} \in \{\mathbf{S}_A, \mathbf{S}_B\}$ and to the (controlled) $\mathcal{O}(\varepsilon/\alpha_{\mathbf{L}})$-accurate time-dependent phase oracles $O_{\mathbf{M}^{-1/2}}, O_{\tilde{\boldsymbol{\eta}}_0}, O_{\tilde{\boldsymbol{\Lambda}}}$.
\end{lemma}

Note that framing the algorithmic complexity in terms of the $C^1$-smoothness of the effective parameters $\tilde{\eta}_{0,p}(t)$ and $\tilde{\Lambda}_q(t)$ decouples the simulation bounds from the underlying physical parameterizations.
Because these effective parameters natively incorporate the first logarithmic derivatives $m_p^{(1)}(t)$ and $k_q^{(1)}(t)$, ensuring that $\tilde{\eta}_{0,p}(t)$  and $\tilde{\Lambda}_q(t)$ remain $C^1$-smooth generally requires a $C^2$-smoothness for the underlying physical mass $m_p(t)$ and stiffness $k_q(t)$ profiles, unless discontinuities in the material derivatives are perfectly compensated.

With the multiplexed block encodings $\mathrm{BE}[\mathbf{H}(t)]$ and $\mathrm{BE}[\mathbf{L}(t)]$ satisfying the prerequisites of the time-dependent \ac{LCHS} framework \cite{low_optimal_2025}, we block-encode the system evolution. 
Note that moving to a time-dependent generator inherently introduces a logarithmic penalty to the query complexity inherited from time-dependent Hamiltonian simulation \cite{low_optimal_2025}, which we now explicitly include to reflect the exact scaling.
We define the combined drift-normalization parameter $\zeta=\alpha_{\mathbf{H}}(\Gamma_{\mathbf{H}}+\alpha_{\mathbf{H}}) + \alpha_{\mathbf{L}}(\Gamma_{\mathbf{L}}+\alpha_{\mathbf{L}})$.

\begin{lemma}[LCHS Lemma] \label{th:_td_lchs_op_applied}
    Consider the \ac{ODE} system
    \begin{equation}
        \frac{\mathrm{d} \mathbf{G}(t)}{\mathrm{d} t} = -\mathbf{C}(t)\mathbf{G}(t), \quad \mathbf{G}(0)=\mathbf{I},
    \end{equation}
    where $\mathbf{C}(t) = i\mathbf{H}(t) + \mathbf{L}(t)$, with $\mathbf{H}(t)$ (\Cref{lem:td_hamiltonian_encoding}) and $\mathbf{L}(t)$ (\Cref{lem:td_dissipator_encoding}) being Hermitian and $\mathbf{L}(t) \succeq 0 \,\,\forall t$.
    
    The two-parameter time-ordered evolution operator is defined as $\mathbf{G}(t_2, t_1) = \mathcal{T} e^{-\int_{t_1}^{t_2} \mathbf{C}(s) \mathrm{d}s}$.
    For any target error $\varepsilon \in (0,1)$, the homogeneous propagator $\mathbf{G}(T) \equiv \mathbf{G}(T, 0)$ admits an $(\alpha_{\mathbf{G}}, a_{\mathbf{G}}, \varepsilon)$-\ac{BE} $\mathcal{U}_{\mathbf{G}}$ with normalization factor $\alpha_{\mathbf{G}}=\mathcal{O}(1)$.
    Defining the effective normalization $\alpha_{\operatorname{eff}} = \alpha_{\mathbf{H}} + \alpha_{\mathbf{L}}\log(1/\varepsilon)$, the block encoding can be constructed using
    \begin{equation}
    {Q_{\mathbf{G}}=\mathcal{O}\left( \alpha_{\operatorname{eff}} T \frac{\log(\alpha_{\operatorname{eff}} T / \varepsilon)}{\log\log(\alpha_{\operatorname{eff}} T / \varepsilon)} \right)}
    \end{equation}
    queries to the (controlled) index oracles $O_r^{(\mathcal{M})},O_c^{(\mathcal{M})}$ for $\mathcal{M} \in \{\mathbf{S}_A, \mathbf{S}_B\}$ and to the (controlled) $\mathcal{O}(\varepsilon/Q_{\mathbf{G}})$-accurate time-dependent phase oracles $O_{\mathbf{M}^{-1/2}}, O_{\mathbf{K}^{1/2}}, O_{\tilde{\boldsymbol{\eta}}_0}, O_{\tilde{\boldsymbol{\Lambda}}}$.
    The encoding requires $a_{\mathbf{G}} = \mathcal{O}\left(\log\left(\frac{NT \zeta}{\varepsilon}\right)\right)$ ancilla qubits and uses $\mathcal{O}\Big( \big(Q_\mathbf{G} + \log^{5/2}(1/\varepsilon)\big)\log\big(\tfrac{NT \zeta}{\varepsilon}\big) \Big)$ additional two-qubit gates.
\end{lemma}
\begin{proof}
    Proof follows by substitution of variables into the claims of~\cite{low_optimal_2025}, allocating constant fractions of the error budget to the \ac{LCHS} algorithmic and quadrature error, the time-discretization error of \Cref{lem:td_hamiltonian_encoding,lem:td_dissipator_encoding}, and the accumulated phase oracle errors, and merging all logarithmic factors up to constants in the nontrivial regime $\alpha_{\mathbf{L}}T \geq 1$.
\end{proof}

This evolution operator solves the homogeneous dynamics. 
Following \cite{low_optimal_2025}, we obtain the solution for the inhomogeneous system by applying the \ac{LCHS} protocol on the initial state and using \ac{LCU} to integrate over the source terms.
\begin{theorem}[State Preparation, see \cite{low_optimal_2025}] \label{th:_td_lchs_state_prep_applied}
    Let $O_{\boldsymbol{\psi}_0}$ be a state preparation oracle for the initial spatial state $\ket{\boldsymbol{\psi}(0)}\in \mathbb{C}^N$. 
    Let $O_{\boldsymbol{\chi}_\tau}$ be a state preparation oracle for the time-discretized coefficients proportional to $\sum_j \sqrt{\|\boldsymbol{\chi}(\tau_j)\|} \ket{\tau_j}$, and let $O_{\boldsymbol{\chi}_x}$ be a state preparation oracle for the spatial state $\ket{\bar{\boldsymbol{\chi}}(\tau_j)} = \boldsymbol{\chi}(\tau_j)/\|\boldsymbol{\chi}(\tau_j)\|$ controlled on $\ket{\tau_j}$.
    The normalized state 
    \begin{equation}
    \begin{aligned}
        \ket{\boldsymbol{\psi}(T)} &\propto \|\boldsymbol{\psi}(0)\|\mathbf{G}(T) \ket{\boldsymbol{\psi}(0)} \\&+ \int_{0}^{T} \|\boldsymbol{\chi}(\tau)\| \mathbf{G}(T, \tau) \ket{\bar{\boldsymbol{\chi}}(\tau)} \mathrm{d}\tau,
    \end{aligned}
    \end{equation}
    can be prepared to any additive target error $\varepsilon \in (0,1)$ with constant success probability using
    \begin{equation}
        {Q_{\boldsymbol{\psi}} = \mathcal{O}\left(\frac{\|\boldsymbol{\psi}(0)\| + \|\boldsymbol{\chi}\|_{L^1}}{\|\boldsymbol{\psi}(T)\|}\right)}
    \end{equation}
    queries to the $\mathcal{O}(\varepsilon/Q_{\boldsymbol{\psi}})$-accurate state preparation oracles $O_{\boldsymbol{\psi}_0}$, $O_{\boldsymbol{\chi}_\tau}$, and $O_{\boldsymbol{\chi}_x}$, where $\|\boldsymbol{\chi}\|_{L^1} = \int_{0}^{T}\|\boldsymbol{\chi}(\tau)\|\mathrm{d}\tau$, and, defining the effective normalization $\alpha_{\operatorname{eff}} = \alpha_{\mathbf{H}} + \alpha_{\mathbf{L}}\log(Q_{\boldsymbol{\psi}}/\varepsilon)$,
    \begin{equation}
        {Q_{\mathbf{G}'} = \mathcal{O}\left( Q_{\boldsymbol{\psi}} \alpha_{\operatorname{eff}} T \frac{\log(\alpha_{\operatorname{eff}} T Q_{\boldsymbol{\psi}} / \varepsilon)}{\log\log(\alpha_{\operatorname{eff}} T Q_{\boldsymbol{\psi}} / \varepsilon)} \right) }
    \end{equation}
    queries to the (controlled) index oracles $O_r^{(\mathcal{M})},O_c^{(\mathcal{M})}$ and the (controlled) $\mathcal{O}(\varepsilon/Q_{\mathbf{G}'})$-accurate time-dependent phase oracles $O_{\mathbf{M}^{-1/2}}, O_{\mathbf{K}^{1/2}}, O_{\tilde{\boldsymbol{\eta}}_0}, O_{\tilde{\boldsymbol{\Lambda}}}$.
    The algorithm requires $\mathcal{O}\left(\log\left(\frac{NT \zeta Q_{\boldsymbol{\psi}}}{\varepsilon}\right)\right)$ ancilla qubits and uses 
    \begin{equation}
        \mathcal{O}\Bigg(\Big(Q_{\mathbf{G}'} + Q_{\boldsymbol{\psi}}\log^{5/2}(Q_{\boldsymbol{\psi}}/\varepsilon)\Big)\log\left(\frac{NT \zeta Q_{\boldsymbol{\psi}}}{\varepsilon}\right)\Bigg) 
    \end{equation}
    additional two-qubit gates.
\end{theorem}
Results on energy estimation of time-varying systems can be constructed in an analog way using the same approach as in \Cref{th:energy_estimation_applied} and will not be discussed in detail.

\section{Extensions} \label{s:extensions}
\subsection{Active Systems} \label{s:act_systems}

Assume the dissipator $\mathbf{L}$ is not positive semi-definite but instead has a smallest eigenvalue $\lambda_{\min} < 0$, satisfying $\mathbf{L} \succeq \lambda_{\min}\mathbf{I}$.
Such a system supports at least transient unstable growth, preventing direct simulation through \ac{LCHS} which requires $\mathbf{L} \succeq 0$.
However, following \cite{an2023quantumalgorithmlinearnonunitary, low_optimal_2025}, we can stabilize the simulation by uniformly shifting the spectrum. 
We define a stabilized dissipator $\mathbf{L}_{\operatorname{shifted}}$ by subtracting the minimum eigenvalue:
\begin{equation}
    {\mathbf{L}} = {\mathbf{L}}_{\operatorname{shifted}} + \lambda_{\min}\mathbf{I},
    \qquad
    {\mathbf{L}}_{\operatorname{shifted}} := {\mathbf{L}} - \lambda_{\min}\mathbf{I} \succeq 0.
\end{equation}
Substituting this into the evolution equation yields:
\begin{equation}
    \frac{\mathrm{d}{\boldsymbol\psi}(t)}{\mathrm{d}t}
    = -(i{\mathbf H} + {\mathbf{L}}_{\operatorname{shifted}} + \lambda_{\min}\mathbf{I}){\boldsymbol\psi}(t) + {\boldsymbol{\chi}}(t).
\end{equation}
To remove the scalar growth term $\lambda_{\min}\mathbf{I}$, we introduce the rescaled state and source:
\begin{equation}
    \hat{\boldsymbol\psi}(t) := e^{\lambda_{\min} t}{\boldsymbol\psi}(t),
    \qquad
    \hat{\boldsymbol{\chi}}(t) := e^{\lambda_{\min} t}{\boldsymbol{\chi}}(t).
\end{equation}
Since $\lambda_{\min} < 0$, the simulation variable $\hat{\boldsymbol\psi}$ is exponentially damped relative to the physical state $\boldsymbol\psi$.
Differentiating $\hat{\boldsymbol\psi}(t)$ and substituting the dynamics of $\boldsymbol\psi(t)$ recovers the standard Schrödinger form
\begin{equation} \label{eq:spectral_shift}
    {
    \frac{\mathrm{d}\hat{\boldsymbol\psi}(t)}{\mathrm{d} t} = -(i{\mathbf H} + {\mathbf{L}}_{\operatorname{shifted}})\hat{\boldsymbol\psi}(t) + \hat{\boldsymbol{\chi}}(t).
    }
\end{equation}
The solution is given by:
\begin{equation}
\begin{aligned}
    \hat{\boldsymbol\psi}(T)
    &= e^{-(i{\mathbf H} + {\mathbf{L}}_{\operatorname{shifted}}) T}\hat{\boldsymbol\psi}(0) \\
    &\quad + \int_{0}^{T}e^{-(i{\mathbf H} + {\mathbf{L}}_{\operatorname{shifted}}) (T-\tau)}\hat{\boldsymbol{\chi}}(\tau)\mathrm{d}\tau.
\end{aligned}
\end{equation}
This confirms that unstable systems can be simulated with \ac{LCHS} for short time scales.
However, we incur an additional renormalization factor $e^{-\lambda_{\min} T}$.
Note that in case of time-dependent systems such a corrective shift can be applied time-resolved \cite{low_hamiltonian_2019}, which allows to compensate for jumps in material properties that induce growth.
Consequently, brief episodes of instability with bounded integrated magnitude are simulable at essentially no extra cost.

Note, that if $\mathbf{L}$ has a smallest eigenvalue $\lambda_{\min} > 0$, and this magnitude is known or can be lower bounded, we can shift the spectrum in reverse direction.
This factors out the uniform global dissipation and exponentially increases the norm of the simulated state $\hat{\boldsymbol\psi}(t)$ by $e^{\lambda_{\min} T}$.
If the spectral shift is exact, this exponentially reduces the query complexity required to achieve a target error of e.g. energy estimates.

\subsection{Probabilistic Systems} \label{s:probabilistic_systems}

We extend our framework to simulate an ensemble of $K$ discrete material configurations with prior probabilities $p_k$. 
Introducing an ancillary configuration register $\ket{k}$, we define the joint initial state
\begin{equation}
    \boldsymbol{\Psi}(0) = \sum_{k=1}^K \sqrt{p_k} \ket{k} \otimes \boldsymbol{\psi}_k(0).
\end{equation}
Even if the physical initial condition $\boldsymbol{\phi}(0)$ is identical across the ensemble, the encoded states $\boldsymbol{\psi}_k(0) = \mathbf{B}_k^{-1/2}\boldsymbol{\phi}(0)$ are configuration-dependent. 
As $\mathbf{B}_k$ is strictly diagonal, the normalized joint initial state can be prepared efficiently using quantum arithmetic and inequality testing \cite{babbush2023exponential}, with algorithmic costs limited only by the worst-case configuration in the ensemble.

Since the configurations do not interact, the joint dynamics are governed by the block-diagonal operators
\begin{equation}
    \mathbf{H} = \sum_{k=1}^K \ketbra{k}{k} \otimes \mathbf{H}_k, \quad \mathbf{L} = \sum_{k=1}^K \ketbra{k}{k} \otimes \mathbf{L}_k.
\end{equation}
Using a controlled quantum oracle, this superposition requires $\lceil \log_2 K \rceil$ additional ancilla qubits. 
Crucially, the block-encoding normalizations are bounded by the worst-case configurations, $\alpha_{\mathbf{H,L}} = \max_k \alpha_{\mathbf{H}_k,\mathbf{L}_k}$. 
Consequently, the \ac{LCHS} query complexity to evolve the entire ensemble asymptotically matches the single most demanding deterministic configuration.

Under \ac{LCHS} evolution, the system assumes the state $\boldsymbol{\Psi}(t) = \sum_{k=1}^K \sqrt{p_k} \ket{k} \otimes \boldsymbol{\psi}_k(t)$. 
Because the squared norm of each subspace reflects its instantaneous physical energy $E_k(t) = \frac{1}{2}\|\boldsymbol{\psi}_k(t)\|^2$, the global state norm exactly captures the ensemble's expected energy: 
\begin{equation}
    \|\boldsymbol{\Psi}(t)\|^2 = 2 \mathbb{E}[E(t)].
\end{equation} 
For a source-free simulation, the \ac{LCHS} state preparation cost of the re-normalized final ensemble evaluates to
\begin{equation}
    Q_{\boldsymbol{\Psi}} = \sqrt{\frac{\mathbb{E}[E(0)]}{\mathbb{E}[E(t)]}}.
\end{equation}
This intrinsically grants the global algorithm immunity against highly dissipative outlier configurations, which simply decay within the superposition without significantly impacting the global state norm.

Macroscopic expected properties can be extracted from the final ensemble $\boldsymbol{\Psi}(t)$ with an evolution query complexity independent of $K$.
For example, $\mathbb{E}[E(t)]$ is proportional to the unnormalized state's squared norm.
Measuring a spatial subsystem projector (as in \Cref{prob:subspace_energy}) directly on $\boldsymbol{\Psi}(t)$ yields the expected energy fraction $\mathbb{E}[E_{\mathcal{S}}(t)]/\mathbb{E}[E(t)]$, using the same method as in \Cref{th:energy_estimation_applied}.

\subsection{Fast-Forwarding} \label{s:fast_forward}

Note, that we can quadratically reduce the evolution cost for purely dissipative problems where $\mathbf{H}=0$ without violating no-fast-forwarding \cite{gilyen_quantum_2019, an2022theory, low_optimal_2025}.
So far it is not established if this quadratic cost reduction of the dissipative component is possible for mixed oscillatory-dissipative systems. 
However, here we show that such an advancement could be fully utilized within our framework.

For a quadratic cost reduction of the dissipative component we need to establish oracle access to a rectangular factorization $\mathbf{F}$, such that
\begin{equation}
    \mathbf{L} = \mathbf{F}^\dagger\mathbf{F}.
\end{equation}
Let us expand the dissipator
\begin{equation}
    \mathbf{L} = \begin{bmatrix}
        \mathbf{M}^{-1/2}\mathbf{T}_0^\dagger \boldsymbol{\eta}_0\mathbf{T}_0\mathbf{M}^{-1/2} & \mathbf{0} \\
        \mathbf{0} & \mathbf{K}^{1/2}\boldsymbol{\eta}^{-1}_{\mathbf{K}}\mathbf{K}^{1/2}
    \end{bmatrix}.
\end{equation}
As $\boldsymbol{\eta}_0\succeq0,\boldsymbol{\eta}^{-1}_{\mathbf{K}}\succeq0$ are diagonal we can factorize
\begin{equation}
\begin{aligned}
    \mathbf{F} &= 
    \begin{bmatrix}
        \boldsymbol{\eta}^{1/2}_0\mathbf{T}_0\mathbf{M}^{-1/2} & \mathbf{0} \\
        \mathbf{0} & \boldsymbol{\eta}^{-1/2}_{\mathbf{K}}\mathbf{K}^{1/2}.
    \end{bmatrix} \\ 
    &= \begin{bmatrix}
        \hat{\mathbf{T}}_0 & \mathbf{0} \\
        \mathbf{0} & \boldsymbol{\Lambda}^{1/2}.
    \end{bmatrix},
    \end{aligned}
\end{equation}
where $\hat{\mathbf{T}}_0=\boldsymbol{\eta}^{1/2}_0\mathbf{T}_0\mathbf{M}^{-1/2}$ is the viscosity-weighted topology matrix.
Note that this enforces the stricter condition $\boldsymbol{\eta}_0 \succeq 0$.
This demonstrates that our framework is fully prepared to leverage any future algorithmic advancements for fast-forwarding mixed systems.

\subsection{Grounding Edges} \label{s:ground_edges}

A key advantage of explicitly separating the ground connectors $\mathbf{S}_0$ from the inter-mass connectors $\mathbf{D}_0$ in our topology formulation arises in the block-encoding normalization. 
Physically, adjacent masses connected by a spring can oscillate entirely out of phase. 
This anti-cyclic behavior effectively doubles the relative displacement and quadruples the maximum energetic contribution compared to a single moving mass. 
In contrast, grounding springs are fixed at one end and cannot support this anti-cyclic amplification. 

Because our framework isolates these grounding edges, we can can optimally capture both characteristics. 
As a result, the simulated system can incorporate ground connections that are up to four times stiffer than the maximum inter-mass spring stiffness without incurring any penalty to the normalization factors $\alpha_{\mathbf{H}}$ or $\alpha_{\mathbf{L}}$. 
However, under this constructive addition, the mapping of the Euclidean inner product to the physical energies becomes non-uniform across the computational basis.
Thus, measuring subsystem energies (\Cref{prob:subspace_energy}) requires weighted projectors. 
Hence, we choose the sub-optimal joint normalization in our base framework.

\section{Discussion} \label{s:discussion}

According to our previous analysis we can point out different regimes of potential quantum advantage over standard classical methods, as well as restrictions of the approach.
First, our \BQP-completeness result suggests that one can obtain exponential advantages even in regimes of strong localized dissipation under the assumption of non-local coupling. 
This extends previous results that showed \BQP-completeness for undamped or weakly damped systems and opens new application ideas, for example, the estimation of global energy loss after a certain evolution time $t$~\cite{babbush2023exponential}.

Second, this result also suggests that in regimes of strong dissipation the information that composes the observed quantities must be sufficiently contained within the weakly damped eigenmodes of the systems.
Otherwise, the system might admit a classical reduction to a smaller system size, or sampling and renormalization costs may prevent quartic speedups.
We expect such hard problems to emerge, for example, in linear time-invariant systems with constant dissipation per oscillation cycle for all modes (constant-Q) with sufficient eigenmode complexity.

Third, we prove that in a topology with only nearest-neighbor interactions, such as one obtains in finite dimensional meshes or lattices, one can not expect exponential runtime advantages due to the slow expansion of causal influence over time.
This holds even if inhomogeneities are considered and can get worsened due to dissipation, which allows to truncate past influence if the dissipation is sufficiently strong.

Damped oscillator models are of particular relevance in describing differential equations: discretizing a continuous space-time damped wave equation using, e.g., the Finite Element or Finite Difference method gives rise to a discrete system that is equivalent to a damped mass spring model \cite{bosch_quantum_2025}. 
In this context, the number of masses is set by the spatial resolution which in turn is dictated by the highest frequency components of the wave field \cite{fichtner2010full}. 
For periodic homogeneous constant-Q material models the effective damping of the eigenmodes scales exponentially with the temporal frequency, as the frequency- and eigenspectrum coincide.
Hence, the highest frequency components are exponentially damped with respect to the system size and estimating them
introduces exponential amplification or estimation cost, which classical computers do not pay.
We see two regimes where quantum advantage might still emerge: (i) the model is sufficiently far from homogeneous and the relevant eigenmodes are essentially undamped (ii) the overall damping is sufficiently weak with respect to the modeled system size so that the amplitude loss is insignificant but the damping is still required to model wave propagation correctly. 
When working with observational data, we expect both cases to be relevant since strong damping of all modes would preclude meaningful measurements. Hence, the presence of high frequency components in measurements acts as a guarantor of some minimum transfer strength of those modes.

In summary, our framework extends prior work on simulating coupled classical oscillators to non-Markovian dissipation, time-dependent materials and external sources. We prove \BQP-completeness and thus up-to exponential quantum advantages in the the case of strong localized dissipation.
In case of local couplings this quantum advantage potential reduces to a polynomial one based on the derived Lieb-Robinson type arguments.
For example, given the presence of sufficiently weakly damped modes, the largest advantage one could expect in a three dimensional lattice topology is quartic, preserving the undamped result \cite{babbush2023exponential, bosch_quantum_2025}.
Note, that interestingly we obtain an additional logarithmic advantage with respect to the number of parallel channels $C$ compared to classical methods that use the same Prony series method. 
Furthermore, our results suggest exponential memory advantages with respect to the number of Maxwell bodies $C$.
However, it is unclear if this advantage is of relevance for practical problems.

Our work leaves open a number of possibilities for future work.
While our framework captures a significant amount of physical mechanisms that govern the dynamical evolution of time-dependent non-Markovian classical oscillator dynamics, there remain mechanisms we can not currently address. 
First our framework does not capture non-linear dynamics. 
While linearity is often a sufficiently accurate approximation for a large set of practical applications, this assumption can severely break down, for example, in the high-energy limit.
Second, we currently do not capture rotating reference frames. 
These will cause fictitious forces such as the Coriolis force, which we can not embed into our current structure of the system operator $\mathbf{C}$.
Third, we are currently limiting ourselves to completely monotonic kernels that do not possess local oscillation. 
For problems that demand complex kernels a generalization of our approach beyond the Maxwell model is necessary.
Last, it would be valuable to understand if \BQP-completeness prevails in the case where a constant fraction of the system size is strongly damped. 

Our current time-integration approach has some known sub-optimalities that demand further research.
First, literature shows that while the \ac{LCHS} algorithm is asymptotically optimal in terms of its query cost to $\mathbf{H}$ and $\mathbf{L}$ it is not clear if a quadratic fast-forwarding of the dissipative dynamics in mixed systems might be possible \cite{low_optimal_2025}.
Our algorithm could directly exploit such advancements, as we have shown that query access to a factorization $\mathbf{F}$ satisfying $\mathbf{L} = \mathbf{F}^\dagger\mathbf{F}$ can be easily established.
Also it is not clear if the simulation complexity of time-dependent systems can be improved \cite{low_hamiltonian_2019, low_optimal_2025} in the general case.
Furthermore, it is not known if the optimality guarantees of \ac{LCHS} translate to our specific application, specifically in the context of real-world realizations.
Another open question is the optimal approximation of arbitrary memory kernels with the smallest number of Maxwell bodies $C$.
As the error might decay only linear with respect to the number of Maxwell bodies $C$ for certain kernels, this might be a major cost determining factor of classical and quantum simulation.

While we show that our framework can efficiently estimate subspace energies of evolved non-Markovian oscillator networks and that these estimates are in the worst case classically hard, it remains to find practical end-to-end applications where a quantum advantage can be established.
We have shown that in a low-dimensional space such advantages can at best be polynomial and are highly dependent on the dissipative regime. 
However, even establishing a super-quadratic end-to-end advantage for practical use-cases in computational mechanics would be of scientific and commercial interest.
Such a discovery would warrant further constant factor analysis and circuit optimization, including application-specific oracle implementations.

\begin{acknowledgments}
We thank Guang-Hao Low and Rolando Somma for valuable feedback on this work.  NW and SS acknowledge support from DOE, Office of Science, National Quantum Information Science Research Centers, Co-design Center for Quantum Advantage (C2QA) under Contract No.~DE-SC0012704 (Basic Energy Sciences, PNNL FWP 76274) and Pacific Northwest National Laboratory's Quantum Algorithms and Architecture for Domain Science (QuAADS) Laboratory Directed Research and Development (LDRD) Initiative as well as from Canada's National Research Council. C.B. was supported by the Swiss National Science Foundation (SNSF) through a Postdoc.Mobility fellowship (P500PT 217673/1). We thank Google Quantum AI (Google LLC) for providing intellectual and monetary support. 
\end{acknowledgments}

\noindent \textbf{AI statement}: We declare the use of LLMs for text revision as well as proof assistance and verification.

\noindent \textbf{Author contribution statement}: study conception: M.~Schade; Markovian embedding: M.~Schade, S.~Keating and C.~Bösch; quantum simulation: M.~Schade, S.~Simon, N.~Wiebe, classical bounds: N.~Wiebe, S.~Simon, M.~Schade, BQP-completeness: S.~Simon, N.~Wiebe, M.~Schade, Time-Dependent materials: M.~Schade, S.~Keating, and C.~Bösch; manuscript preparation: M.~Schade, S.~Simon, N.~Wiebe, S.~Keating, C.~Bösch. All authors reviewed the results and approved the final version of the manuscript.

\bibliography{apssamp}

@article{low2017hamiltonian,
  title={Hamiltonian simulation by uniform spectral amplification},
  author={Low, Guang Hao and Chuang, Isaac L},
  journal={arXiv preprint arXiv:1707.05391},
  year={2017}
}

@article{nachtergaele2006lieb,
  title={Lieb-Robinson bounds and the exponential clustering theorem},
  author={Nachtergaele, Bruno and Sims, Robert},
  journal={Communications in mathematical physics},
  volume={265},
  number={1},
  pages={119--130},
  year={2006},
  publisher={Springer}
}

@article{chen2023speed,
  title={Speed limits and locality in many-body quantum dynamics},
  author={Chen, Chi-Fang and Lucas, Andrew and Yin, Chao},
  journal={Reports on Progress in Physics},
  volume={86},
  number={11},
  pages={116001},
  year={2023},
  publisher={IOP Publishing}
}

@article{simon2024amplified,
  title={Amplified amplitude estimation: exploiting prior knowledge to improve estimates of expectation values},
  author={Simon, Sophia and Degroote, Matthias and Moll, Nikolaj and Santagati, Raffaele and Streif, Michael and Wiebe, Nathan},
  journal={arXiv preprint arXiv:2402.14791},
  year={2024}
}

@article{king2026quantum,
  title={Quantum simulation with sum-of-squares spectral amplification},
  author={King, Robbie and Low, Guang Hao and Babbush, Ryan and Somma, Rolando D and Rubin, Nicholas C},
  journal={Physical Review Letters},
  volume={136},
  number={11},
  pages={110601},
  year={2026},
  publisher={APS}
}

@inproceedings{berry2015hamiltonian,
  title={Hamiltonian simulation with nearly optimal dependence on all parameters},
  author={Berry, Dominic W. and Childs, Andrew M. and Kothari, Robin},
  booktitle={2015 IEEE 56th annual symposium on foundations of computer science},
  pages={792--809},
  year={2015},
  organization={IEEG}
}

@article{aharonov2003simple,
  title={A simple proof that Toffoli and Hadamard are quantum universal},
  author={Aharonov, Dorit},
  journal={arXiv preprint quant-ph/0301040},
  year={2003}
}

@article{sakamoto2025quantum,
  title={On the quantum computational complexity of classical linear dynamics with geometrically local interactions: Dequantization and universality},
  author={Sakamoto, Kazuki and Fujii, Keisuke},
  journal={arXiv preprint arXiv:2505.10445},
  year={2025}
}

@article{babbush2023exponential,
	title = {Exponential quantum speedup in simulating coupled classical oscillators},
	volume = {13},
	issn = {2160-3308},
	url = {http://arxiv.org/abs/2303.13012},
	doi = {10.1103/PhysRevX.13.041041},
	number = {4},
	urldate = {2024-03-15},
	journal = {Physical Review X},
	author = {Babbush, Ryan and Berry, Dominic W. and Kothari, Robin and Somma, Rolando D. and Wiebe, Nathan},
	month = dec,
	year = {2023},
	pages = {041041},
}

@article{costa2019quantum,
	title = {Quantum algorithm for simulating the wave equation},
	volume = {99},
	url = {https://link.aps.org/doi/10.1103/PhysRevA.99.012323},
	doi = {10.1103/PhysRevA.99.012323},
	number = {1},
	urldate = {2024-10-24},
	journal = {Physical Review A},
	author = {Costa, Pedro C. S. and Jordan, Stephen and Ostrander, Aaron},
	month = jan,
	year = {2019},
	pages = {012323},
}

@article{low2019hamiltonian,
	title = {Hamiltonian {Simulation} by {Qubitization}},
	volume = {3},
	url = {https://quantum-journal.org/papers/q-2019-07-12-163/},
	doi = {10.22331/q-2019-07-12-163},
	urldate = {2024-10-24},
	journal = {Quantum},
	author = {Low, Guang Hao and Chuang, Isaac L.},
	month = jul,
	year = {2019},
	pages = {163},
}

@article{zhang2022quantum,
	title = {Quantum {State} {Preparation} with {Optimal} {Circuit} {Depth}: {Implementations} and {Applications}},
	volume = {129},
	issn = {0031-9007, 1079-7114},
	shorttitle = {Quantum {State} {Preparation} with {Optimal} {Circuit} {Depth}},
	url = {https://link.aps.org/doi/10.1103/PhysRevLett.129.230504},
	doi = {10.1103/PhysRevLett.129.230504},
	number = {23},
	urldate = {2024-05-06},
	journal = {Physical Review Letters},
	author = {Zhang, Xiao-Ming and Li, Tongyang and Yuan, Xiao},
	month = nov,
	year = {2022},
	pages = {230504},
}

@article{knill2007optimal,
	title = {Optimal quantum measurements of expectation values of observables},
	volume = {75},
	issn = {1050-2947, 1094-1622},
	url = {https://link.aps.org/doi/10.1103/PhysRevA.75.012328},
	doi = {10.1103/PhysRevA.75.012328},
	number = {1},
	urldate = {2024-03-15},
	journal = {Physical Review A},
	author = {Knill, Emanuel and Ortiz, Gerardo and Somma, Rolando D.},
	month = jan,
	year = {2007},
	pages = {012328},
}

@book{fichtner2010full,
	series = {Advances in {Geophysical} and {Environmental} {Mechanics} and {Mathematics}},
	title = {Full {Seismic} {Waveform} {Modelling} and {Inversion}},
	copyright = {https://www.springernature.com/gp/researchers/text-and-data-mining},
	isbn = {978-3-642-15806-3 978-3-642-15807-0},
	url = {https://link.springer.com/10.1007/978-3-642-15807-0},
	urldate = {2024-10-25},
	publisher = {Springer},
	author = {Fichtner, Andreas},
	year = {2011},
	doi = {10.1007/978-3-642-15807-0},
}

@article{an2022theory,
  title={A theory of quantum differential equation solvers: limitations and fast-forwarding},
  author={An, Dong and Liu, Jin-Peng and Wang, Daochen and Zhao, Qi},
  journal={arXiv preprint arXiv:2211.05246},
  year={2022}
}

@misc{krovi2024quantum,
	title = {Quantum algorithms to simulate quadratic classical {Hamiltonians} and optimal control},
	url = {http://arxiv.org/abs/2404.07303},
	doi = {10.48550/arXiv.2404.07303},
	urldate = {2024-10-25},
	publisher = {arXiv},
	author = {Krovi, Hari},
	month = {04},
	year = {2024},
	note = {arXiv:2404.07303}
}

@unpublished{sato2024quantum,
	title = {Quantum algorithm for partial differential equations of non-conservative systems with spatially varying parameters},
	url = {http://arxiv.org/abs/2407.05019},
	doi = {10.48550/arXiv.2407.05019},
	urldate = {2024-10-24},
	publisher = {arXiv},
	author = {Sato, Yuki and Tezuka, Hiroyuki and Kondo, Ruho and Yamamoto, Naoki},
	month = jul,
	year = {2024},
	note = {arXiv:2407.05019},
}

@inproceedings{gleinig2021efficient,
	title = {An {Efficient} {Algorithm} for {Sparse} {Quantum} {State} {Preparation}},
	url = {https://ieeexplore.ieee.org/abstract/document/9586240},
	doi = {10.1109/DAC18074.2021.9586240},
	urldate = {2024-10-24},
	booktitle = {2021 58th {ACM}/{IEEE} {Design} {Automation} {Conference} ({DAC})},
	author = {Gleinig, Niels and Hoefler, Torsten},
	month = dec,
	year = {2021},
	pages = {433--438},
}

@article{low2017optimal,
	title = {Optimal {Hamiltonian} {Simulation} by {Quantum} {Signal} {Processing}},
	volume = {118},
	issn = {0031-9007, 1079-7114},
	url = {https://link.aps.org/doi/10.1103/PhysRevLett.118.010501},
	doi = {10.1103/PhysRevLett.118.010501},
	number = {1},
	urldate = {2023-10-23},
	journal = {Phys. Rev. Lett.},
	author = {Low, Guang Hao and Chuang, Isaac L.},
	month = jan,
	year = {2017},
	pages = {010501},
}

@unpublished{childs2012hamiltonian,
	title = {Hamiltonian {Simulation} {Using} {Linear} {Combinations} of {Unitary} {Operations}},
	url = {http://arxiv.org/abs/1202.5822},
	doi = {10.48550/arXiv.1202.5822},
	urldate = {2024-10-25},
	publisher = {arXiv},
	author = {Childs, Andrew M. and Wiebe, Nathan},
	month = feb,
	year = {2012},
	note = {arXiv:1202.5822},
}

@article{de2022double,
	title = {Double sparse quantum state preparation},
	volume = {21},
	issn = {1573-1332},
	url = {https://doi.org/10.1007/s11128-022-03549-y},
	doi = {10.1007/s11128-022-03549-y},
	number = {6},
	urldate = {2024-05-06},
	journal = {Quantum Information Processing},
	author = {de Veras, Tiago M. L. and da Silva, Leon D. and da Silva, Adenilton J.},
	month = jun,
	year = {2022},
	pages = {204},
}

@article{koukoutsis2023,
	title = {Dyson maps and unitary evolution for {Maxwell} equations in tensor dielectric media},
	volume = {107},
	url = {https://link.aps.org/doi/10.1103/PhysRevA.107.042215},
	doi = {10.1103/PhysRevA.107.042215},
	number = {4},
	urldate = {2024-12-19},
	journal = {Physical Review A},
	author = {Koukoutsis, Efstratios and Hizanidis, Kyriakos and Ram, Abhay K. and Vahala, George},
	month = apr,
	year = {2023},
	pages = {042215},
}

@unpublished{bosch_quantum_2025,
	title = {Quantum {Wave} {Simulation} with {Sources} and {Loss} {Functions}},
	url = {http://arxiv.org/abs/2411.17630},
	doi = {10.48550/arXiv.2411.17630},
	urldate = {2025-02-27},
	publisher = {arXiv},
	author = {Bösch, Cyrill and Schade, Malte and Aloisi, Giacomo and Keating, Scott D. and Fichtner, Andreas},
	month = feb,
	year = {2025},
	note = {arXiv:2411.17630 [quant-ph]},
}

@inproceedings{gilyen_quantum_2019,
    address = {New York, NY, USA},
    series = {{STOC} 2019},
    title = {Quantum singular value transformation and beyond: exponential improvements for quantum matrix arithmetics},
    isbn = {978-1-4503-6705-9},
    shorttitle = {Quantum singular value transformation and beyond},
    url = {https://dl.acm.org/doi/10.1145/3313276.3316366},
    doi = {10.1145/3313276.3316366},
    urldate = {2024-10-24},
    booktitle = {Proceedings of the 51st {Annual} {ACM} {SIGACT} {Symposium} on {Theory} of {Computing}},
    publisher = {Association for Computing Machinery},
    author = {Gilyén, András and Su, Yuan and Low, Guang Hao and Wiebe, Nathan},
    month = jun,
    year = {2019},
    pages = {193--204},
}

@article{breuer2016colloquium,
  title={Colloquium: Non-Markovian dynamics in open quantum systems},
  author={Breuer, Heinz-Peter and Laine, Elsi-Mari and Piilo, Jyrki and Vacchini, Bassano},
  journal={Reviews of Modern Physics},
  volume={88},
  number={2},
  pages={021002},
  year={2016},
  publisher={APS}
}

@misc{low_hamiltonian_2019,
    title = {Hamiltonian {Simulation} in the {Interaction} {Picture}},
    url = {http://arxiv.org/abs/1805.00675},
    doi = {10.48550/arXiv.1805.00675},
    urldate = {2025-03-03},
    publisher = {arXiv},
    author = {Low, Guang Hao and Wiebe, Nathan},
    month = jun,
    year = {2019},
    note = {arXiv:1805.00675 [quant-ph]},
}

@inproceedings{berry_exponential_2014,
    title = {Exponential improvement in precision for simulating sparse {Hamiltonians}},
    url = {http://arxiv.org/abs/1312.1414},
    doi = {10.1145/2591796.2591854},
    urldate = {2024-03-31},
    booktitle = {Proceedings of the forty-sixth annual {ACM} symposium on {Theory} of computing},
    author = {Berry, Dominic W. and Childs, Andrew M. and Cleve, Richard and Kothari, Robin and Somma, Rolando D.},
    month = may,
    year = {2014},
    note = {arXiv:1312.1414 [quant-ph]},
    pages = {283--292},
}

@misc{rosenkranz_quantum_2024,
    title = {Quantum state preparation for multivariate functions},
    url = {http://arxiv.org/abs/2405.21058},
    doi = {10.48550/arXiv.2405.21058},
    urldate = {2024-10-24},
    publisher = {arXiv},
    author = {Rosenkranz, Matthias and Brunner, Eric and Marin-Sanchez, Gabriel and Fitzpatrick, Nathan and Dilkes, Silas and Tang, Yao and Kikuchi, Yuta and Benedetti, Marcello},
    month = may,
    year = {2024},
    note = {arXiv:2405.21058},
}

@article{zylberman_efficient_2024,
    title = {Efficient quantum state preparation with {Walsh} series},
    volume = {109},
    url = {https://link.aps.org/doi/10.1103/PhysRevA.109.042401},
    doi = {10.1103/PhysRevA.109.042401},
    number = {4},
    urldate = {2024-10-24},
    journal = {Physical Review A},
    author = {Zylberman, Julien and Debbasch, Fabrice},
    month = apr,
    year = {2024},
    note = {Publisher: American Physical Society},
    pages = {042401},
}

@misc{an2023quantumalgorithmlinearnonunitary,
      title={Quantum algorithm for linear non-unitary dynamics with near-optimal dependence on all parameters}, 
      author={Dong An and Andrew M. Childs and Lin Lin},
      year={2023},
      eprint={2312.03916},
      archivePrefix={arXiv},
      primaryClass={quant-ph},
      url={https://arxiv.org/abs/2312.03916}, 
}

@article{robertsson_viscoelastic_1994,
    title = {Viscoelastic finite‐difference modeling},
    volume = {59},
    issn = {0016-8033},
    url = {https://library.seg.org/doi/abs/10.1190/1.1443701},
    doi = {10.1190/1.1443701},
    number = {9},
    urldate = {2025-03-05},
    journal = {GEOPHYSICS},
    author = {Robertsson, Johan O. A. and Blanch, Joakim O. and Symes, William W.},
    month = sep,
    year = {1994},
    note = {Publisher: Society of Exploration Geophysicists},
    pages = {1444--1456},
}

@article{coleman_foundations_1961,
    title = {Foundations of {Linear} {Viscoelasticity}},
    volume = {33},
    url = {https://link.aps.org/doi/10.1103/RevModPhys.33.239},
    doi = {10.1103/RevModPhys.33.239},
    number = {2},
    urldate = {2025-05-28},
    journal = {Reviews of Modern Physics},
    author = {Coleman, Bernard D. and Noll, Walter},
    month = apr,
    year = {1961},
    note = {Publisher: American Physical Society},
    pages = {239--249},
}

@article{blanch_modeling_1995,
    title = {Modeling of a constant {Q}; methodology and algorithm for an efficient and optimally inexpensive viscoelastic technique},
    volume = {60},
    issn = {0016-8033},
    url = {https://doi.org/10.1190/1.1443744},
    doi = {10.1190/1.1443744},
    number = {1},
    urldate = {2025-03-05},
    journal = {Geophysics},
    author = {Blanch, Joakim O. and Robertsson, Johan O. A. and Symes, William W.},
    month = feb,
    year = {1995},
    pages = {176--184},
}

@book{tschoegl_phenomenological_2012,
    title = {The {Phenomenological} {Theory} of {Linear} {Viscoelastic} {Behavior}: {An} {Introduction}},
    isbn = {978-3-642-73602-5},
    shorttitle = {The {Phenomenological} {Theory} of {Linear} {Viscoelastic} {Behavior}},
    publisher = {Springer Science \& Business Media},
    author = {Tschoegl, Nicholas W.},
    month = dec,
    year = {2012},
    note = {Google-Books-ID: 7Kf7CAAAQBAJ},
}

@article{carcione_seismic_1993,
    title = {Seismic modeling in viscoelastic media},
    volume = {58},
    issn = {0016-8033},
    url = {https://library.seg.org/doi/abs/10.1190/1.1443340},
    doi = {10.1190/1.1443340},
    number = {1},
    urldate = {2025-05-28},
    journal = {GEOPHYSICS},
    author = {Carcione, José M.},
    month = jan,
    year = {1993},
    note = {Publisher: Society of Exploration Geophysicists},
    pages = {110--120},
}

@book{lakes_viscoelastic_2009,
    title = {Viscoelastic {Materials}},
    isbn = {978-0-521-88568-3},
    publisher = {Cambridge University Press},
    author = {Lakes, Roderic S.},
    month = apr,
    year = {2009},
    note = {Google-Books-ID: BH6f2hWWBkAC},
}

@article{alves_numerical_2021,
    title = {Numerical {Methods} for {Viscoelastic} {Fluid} {Flows}},
    volume = {53},
    issn = {0066-4189, 1545-4479},
    url = {https://www.annualreviews.org/content/journals/10.1146/annurev-fluid-010719-060107},
    doi = {10.1146/annurev-fluid-010719-060107},
    number = {Volume 53, 2021},
    urldate = {2025-05-28},
    journal = {Annual Review of Fluid Mechanics},
    author = {Alves, M. A. and Oliveira, P. J. and Pinho, F. T.},
    month = jan,
    year = {2021},
    note = {Publisher: Annual Reviews},
    pages = {509--541},
}

@misc{low_optimal_2025,
    title = {Optimal quantum simulation of linear non-unitary dynamics},
    url = {http://arxiv.org/abs/2508.19238},
    doi = {10.48550/arXiv.2508.19238},
    urldate = {2025-09-29},
    publisher = {arXiv},
    author = {Low, Guang Hao and Somma, Rolando D.},
    month = sep,
    year = {2025},
    note = {arXiv:2508.19238 [quant-ph]},
}

@article{gaur_novel_2024,
    title = {Novel {Optimized} {Designs} of {Modulo} 2n+1 {Adder} for {Quantum} {Computing}},
    volume = {32},
    issn = {1557-9999},
    url = {https://ieeexplore.ieee.org/document/10599288/},
    doi = {10.1109/TVLSI.2024.3418930},
    number = {9},
    urldate = {2025-09-30},
    journal = {IEEE Transactions on Very Large Scale Integration (VLSI) Systems},
    author = {Gaur, Bhaskar and Thapliyal, Himanshu},
    month = sep,
    year = {2024},
    pages = {1759--1763},
}

@inproceedings{gaur_logarithmic_2023,
    address = {Knoxville TN USA},
    title = {A {Logarithmic} {Depth} {Quantum} {Carry}-{Lookahead} {Modulo} (2n - 1) {Adder}},
    isbn = {979-8-4007-0125-2},
    url = {https://dl.acm.org/doi/10.1145/3583781.3590205},
    doi = {10.1145/3583781.3590205},
    urldate = {2025-09-30},
    booktitle = {Proceedings of the {Great} {Lakes} {Symposium} on {VLSI} 2023},
    publisher = {ACM},
    author = {Gaur, Bhaskar and Muñoz-Coreas, Edgard and Thapliyal, Himanshu},
    month = jun,
    year = {2023},
    pages = {125--130},
}

@article{gidney_halving_2018,
    title = {Halving the cost of quantum addition},
    volume = {2},
    url = {https://quantum-journal.org/papers/q-2018-06-18-74/},
    doi = {10.22331/q-2018-06-18-74},
    urldate = {2025-09-30},
    journal = {Quantum},
    author = {Gidney, Craig},
    month = jun,
    year = {2018},
    note = {Publisher: Verein zur Förderung des Open Access Publizierens in den Quantenwissenschaften},
    pages = {74},
}

@misc{takahashi_quantum_2009,
    title = {Quantum {Addition} {Circuits} and {Unbounded} {Fan}-{Out}},
    url = {http://arxiv.org/abs/0910.2530},
    doi = {10.48550/arXiv.0910.2530},
    urldate = {2025-10-03},
    publisher = {arXiv},
    author = {Takahashi, Yasuhiro and Tani, Seiichiro and Kunihiro, Noboru},
    month = oct,
    year = {2009},
    note = {arXiv:0910.2530 [quant-ph]},
}

@article{shukla_efficient_2024,
    title = {An efficient quantum algorithm for preparation of uniform quantum superposition states},
    volume = {23},
    issn = {1573-1332},
    url = {http://arxiv.org/abs/2306.11747},
    doi = {10.1007/s11128-024-04258-4},
    number = {2},
    urldate = {2025-10-06},
    journal = {Quantum Information Processing},
    author = {Shukla, Alok and Vedula, Prakash},
    month = jan,
    year = {2024},
    note = {arXiv:2306.11747 [quant-ph]},
    pages = {38},
}

@article{zylberman_efficient_2025,
    title = {Efficient {Quantum} {Circuits} for {Non}-{Unitary} and {Unitary} {Diagonal} {Operators} with {Space}-{Time}-{Accuracy} {Trade}-{Offs}},
    volume = {6},
    url = {https://dl.acm.org/doi/10.1145/3718348},
    doi = {10.1145/3718348},
    number = {2},
    urldate = {2025-10-20},
    journal = {ACM Transactions on Quantum Computing},
    author = {Zylberman, Julien and Nzongani, Ugo and Simonetto, Andrea and Debbasch, Fabrice},
    month = apr,
    year = {2025},
    pages = {15:1--15:43},
}

@misc{christensen2025nonmarkovian,
      title={Ancilla-train quantum algorithm for simulating non-Markovian open quantum systems}, 
      author={Hans Michael Christensen and Johannes Agerskov and Frederik Nathan},
      year={2025},
      eprint={2509.12717},
      archivePrefix={arXiv},
      primaryClass={quant-ph},
      url={https://arxiv.org/abs/2509.12717}, 
}

@article{walters2024nonmarkovian,
  title = {Path integral quantum algorithm for simulating non-Markovian quantum dynamics in open quantum systems},
  author = {Walters, Peter L. and Wang, Fei},
  journal = {Phys. Rev. Res.},
  volume = {6},
  issue = {1},
  pages = {013135},
  numpages = {11},
  year = {2024},
  month = {Feb},
  publisher = {American Physical Society},
  doi = {10.1103/PhysRevResearch.6.013135},
  url = {https://link.aps.org/doi/10.1103/PhysRevResearch.6.013135}
}

@article{li2204nonmarkovian,
  title = {Toward quantum simulation of non-Markovian open quantum dynamics: A universal and compact theory},
  author = {Li, Xiang and Lyu, Su-Xiang and Wang, Yao and Xu, Rui-Xue and Zheng, Xiao and Yan, YiJing},
  journal = {Phys. Rev. A},
  volume = {110},
  issue = {3},
  pages = {032620},
  numpages = {13},
  year = {2024},
  month = {Sep},
  publisher = {American Physical Society},
  doi = {10.1103/PhysRevA.110.032620},
  url = {https://link.aps.org/doi/10.1103/PhysRevA.110.032620}
}

@article{brassard2000quantum,
  title={Quantum amplitude amplification and estimation},
  author={Brassard, Gilles and Hoyer, Peter and Mosca, Michele and Tapp, Alain},
  journal={arXiv preprint quant-ph/0005055},
  year={2000}
}

@misc{villanyi2025oscillators_dissipation,
      title={Exponential Quantum Advantage for Simulating Open Classical Systems}, 
      author={Agi Villanyi and Yariv Yanay and Ari Mizel},
      year={2025},
      eprint={2503.11483},
      archivePrefix={arXiv},
      primaryClass={quant-ph},
      url={https://arxiv.org/abs/2503.11483}, 
}

@article{low2018hamiltonian,
  title={Hamiltonian simulation in the interaction picture},
  author={Low, Guang Hao and Wiebe, Nathan},
  journal={arXiv preprint arXiv:1805.00675},
  year={2018}
}

@article{kammler1976chebyshev,
  title={Chebyshev approximation of completely monotonic functions by sums of exponentials},
  author={Kammler, David W},
  journal={SIAM Journal on Numerical Analysis},
  volume={13},
  number={5},
  pages={761--774},
  year={1976},
  publisher={SIAM}
}

@article{braess2005approximation,
  title={Approximation of 1/x by exponential sums in [1,∞},
  author={Braess, Dietrich and Hackbusch, Wolfgang},
  journal={IMA journal of numerical analysis},
  volume={25},
  number={4},
  pages={685--697},
  year={2005},
  publisher={Oxford University Press}
}

@article{beylkin2005approximation,
  title={On approximation of functions by exponential sums},
  author={Beylkin, Gregory and Monz{\'o}n, Lucas},
  journal={Applied and Computational Harmonic Analysis},
  volume={19},
  number={1},
  pages={17--48},
  year={2005},
  publisher={Elsevier}
}

@article{bernstein1929fonctions,
  title={Sur les fonctions absolument monotones},
  author={Bernstein, Serge},
  journal={Acta Mathematica},
  volume={52},
  number={1},
  pages={1--66},
  year={1929},
  publisher={Springer}
}

@article{hanyga2005viscous,
  title={Viscous dissipation and completely monotonic relaxation moduli},
  author={Hanyga, Andrzej},
  journal={Rheologica Acta},
  volume={44},
  number={6},
  pages={614--621},
  year={2005},
  publisher={Springer}
}

@article{yano2026quantum,
  title={Quantum framework for parameterizing partial differential equations via diagonal block-encoding},
  author={Yano, Hiroshi and Sato, Yuki},
  journal={arXiv preprint arXiv:2603.01358},
  year={2026}
}

@article{sun_asymptotically_2023,
    title = {Asymptotically {Optimal} {Circuit} {Depth} for {Quantum} {State} {Preparation} and {General} {Unitary} {Synthesis}},
    volume = {42},
    issn = {1937-4151},
    url = {https://ieeexplore.ieee.org/abstract/document/10044235},
    doi = {10.1109/TCAD.2023.3244885},
    number = {10},
    urldate = {2024-10-24},
    journal = {IEEE Transactions on Computer-Aided Design of Integrated Circuits and Systems},
    author = {Sun, Xiaoming and Tian, Guojing and Yang, Shuai and Yuan, Pei and Zhang, Shengyu},
    month = oct,
    year = {2023},
    note = {Conference Name: IEEE Transactions on Computer-Aided Design of Integrated Circuits and Systems},
    pages = {3301--3314},
}

@book{bingham1989regular,
  title={Regular variation},
  author={Bingham, Nicholas H and Goldie, Charles M and Teugels, Jef L},
  volume={27},
  year={1989},
  publisher={Cambridge university press}
}

@article{kjartansson1979constant,
  title={Constant Q-wave propagation and attenuation},
  author={Kjartansson, Einar},
  journal={Journal of Geophysical Research: Solid Earth},
  volume={84},
  number={B9},
  pages={4737--4748},
  year={1979},
  publisher={Wiley Online Library}
}

@article{hairer1985fast,
  title={Fast numerical solution of nonlinear Volterra convolution equations},
  author={Hairer, Ernst and Lubich, Ch and Schlichte, M},
  journal={SIAM journal on scientific and statistical computing},
  volume={6},
  number={3},
  pages={532--541},
  year={1985},
  publisher={SIAM}
}

@article{schadle2006fast,
  title={Fast and oblivious convolution quadrature},
  author={Sch{\"a}dle, Achim and L{\'o}pez-Fern{\'a}ndez, Mar{\'\i}a and Lubich, Christian},
  journal={SIAM Journal on Scientific Computing},
  volume={28},
  number={2},
  pages={421--438},
  year={2006},
  publisher={SIAM}
}

@book{griewank2008evaluating,
  title={Evaluating derivatives: principles and techniques of algorithmic differentiation},
  author={Griewank, Andreas and Walther, Andrea},
  year={2008},
  publisher={SIAM}
}

@article{suau_practical_2021,
    title = {Practical {Quantum} {Computing}: {Solving} the {Wave} {Equation} {Using} a {Quantum} {Approach}},
    volume = {2},
    shorttitle = {Practical {Quantum} {Computing}},
    url = {https://dl.acm.org/doi/10.1145/3430030},
    doi = {10.1145/3430030},
    number = {1},
    urldate = {2024-10-24},
    journal = {ACM Transactions on Quantum Computing},
    author = {Suau, Adrien and Staffelbach, Gabriel and Calandra, Henri},
    month = feb,
    year = {2021},
    pages = {2:1--2:35},
}

@article{nachtergaele2007lieb,
  title={Lieb-Robinson bounds for harmonic and anharmonic lattice systems},
  author={Nachtergaele, Bruno and Raz, Hillel and Schlein, Benjamin and Sims, Robert},
  journal={arXiv preprint arXiv:0712.3820},
  year={2007}
}

@article{ameri2026quantum,
  title={Quantum simulation of non-Markovian dynamical systems},
  author={Ameri, Abtin and Dutt, Arkopal and Krovi, Hari},
  journal={arXiv preprint arXiv:2608.13533},
  year={2026}
}

\onecolumngrid
\appendix

\section{Prerequisites} \label{apx:prerequisites}

The aim of this appendix is to review background material that is necessary to construct our simulation scheme for damped harmonic oscillators.
Block encoding is a standard technique used in quantum algorithm design wherein a general matrix is embedded as a sub-space of a higher dimensional matrix.  
For brevity, we will at times refer to a block encoding as a $\ac{BE}$.
A \ac{BE} of a matrix $\mathbf{C}\in \mathbb{C}^{N\times N}$ is a unitary matrix $\mathcal{U}_\mathbf{C}$ acting on an enlarged Hilbert space, such that
\begin{equation} \label{eq:block_encoding}
    {\mathcal{U}_\mathbf{C}
    = \begin{bmatrix}
    \mathbf{C}/\alpha_{\mathbf{C}} & \cdot \\
    \cdot & \cdot
    \end{bmatrix},}
\end{equation}
where $\alpha_{\mathbf{C}}$ is a normalization constant satisfying $\|\mathbf{C}/\alpha_{\mathbf{C}}\|\leq1$. 
Here $\|\cdot\|$ denotes the spectral norm. 
Equivalently, this relation can be written in terms of projection onto $\ketbra{0}{0}$, namely
\begin{equation} \label{eq:block_encoding_proj}
    \mathbf{C}/\alpha_{\mathbf{C}} = \bigl(\bra{0}^{\otimes a} \otimes \mathbf{I}_N\bigr)\mathcal{U}_\mathbf{C}\bigl(\ket{0}^{\otimes a} \otimes \mathbf{I}_N\bigr),
\end{equation}
where $\mathbf{I}_N$ is the $N \times N$ identity.
When $\mathcal{U}_\mathbf{C}$ uses $a$ ancillary qubits, we call it an $a$-qubit \ac{BE} of $\mathbf{C}$.

The probability of successfully measuring $\ket{0}^{\otimes a}$ for an input state $\ket{\psi}$ is
\begin{equation}
    \Pr[a=0] = \|\mathbf{C}\ket{\psi}\|^2/\alpha_{\mathbf{C}}^2.
\end{equation}
One can employ amplitude amplification techniques \cite{berry_exponential_2014} to boost this success probability to near unity using $\mathcal{O}(\alpha_{\mathbf{C}}/\|\mathbf{C}\ket{\psi}\|)$ queries. 
We define a general \ac{BE} notation:
\begin{definition}[Block encoding; see \cite{gilyen_quantum_2019}]
    Suppose that $\mathbf{C} \in \mathbb{C}^{N \times N}$ is an operator, $\alpha_{\mathbf{C}}, \varepsilon \in \mathbb{R}^+$, and $a\in\mathbb{N}$.
    Then we say that the $(a + \lceil\log_2 N\rceil)$-qubit unitary $\mathcal{U}_\mathbf{C}$ is an $(\alpha_{\mathbf{C}},a,\varepsilon)$-\ac{BE} of $\mathbf{C}$, if 
    \begin{equation}
        {\|\mathbf{C} - \alpha_{\mathbf{C}}\bigl(\bra{0}^{\otimes a} \otimes \mathbf{I}_N\bigr)\mathcal{U}_\mathbf{C}\bigl(\ket{0}^{\otimes a} \otimes \mathbf{I}_N\bigr) \| \leq \varepsilon.}
    \end{equation}
\end{definition}

One prominent algorithm that implements block encodings is the \ac{LCU} framework \cite{childs2012hamiltonian}.
The implementation of an operator $\mathbf{C}$ through \ac{LCU} requires a decomposition as a linear combination of unitary matrices
\begin{equation} \label{eq:lcu}
    {\mathbf{C}=\sum_{j=1}^J\alpha_j \mathcal{U}_j,}
\end{equation}
where $\alpha_j > 0$ are real, positive coefficients and $\mathcal{U}_{j}$ are unitaries. 
We can formalize the \ac{BE} construction as follows:

\begin{lemma}[LCU Block Encoding; see \cite{childs2012hamiltonian}] \label{lem:lcu}
    Let $\mathbf{C}$ be decomposed as in \Cref{eq:lcu} with $L^1$-norm $\alpha_{\mathbf{C}} = \sum_j \alpha_j$. 
    Given a prepare unitary $\operatorname{PREP}$ acting on $a = \lceil \log_2 J \rceil$ ancilla qubits such that
    \begin{equation}
        \operatorname{PREP} \ket{0}^{\otimes a} = \sum_{j=1}^{J} \sqrt{\frac{\alpha_j}{\alpha_{\mathbf{C}}}} \ket{j},
    \end{equation}
    and a select unitary $\operatorname{SEL} = \sum_{j=1}^{J} \ketbra{j}{j} \otimes \mathcal{U}_j$, 
    the operator 
    \begin{equation}
        {\mathcal{U}_{\mathbf{C}} = (\operatorname{PREP}^\dagger \otimes \mathbf{I}_N) \operatorname{SEL} (\operatorname{PREP} \otimes \mathbf{I}_N),}
    \end{equation}
    is an $(\alpha_{\mathbf{C}}, a, 0)$-\ac{BE} of $\mathbf{C}$.
\end{lemma}

\section{Block Encodings} \label{apx:block_encodings}

Next we use these definitions to construct block encodings of the matrices that appear in the damped oscillator problem under consideration.
We need to block encode diagonal operators as an important subroutine in our algorithms.  
Our first algorithmic claim is a specific statement about the fact that one-sparse (i.e. diagonal) matrices can be block encoded at potentially low cost (assuming each diagonal element is efficiently computable). 
Similar results have been given extensively in the literature, such as in the logarithmic block encoding presented in~\cite{gilyen_quantum_2019}, but we provide a concise statement below for simplicity

\begin{lemma}[Real Diagonal Matrix Encoding] \label{lemma:real_diag}
    Any real diagonal matrix $\mathbf{R}=\diag(R_1,\dots,R_N) \in \mathbb{R}^{N \times N}$ with normalization factor $\alpha_{\mathbf{R}} = \|\mathbf{R}\|_{\infty}$ admits an $(\alpha_{\mathbf{R}}, 1, \alpha_{\mathbf{R}}\varepsilon')$-\ac{BE}, denoted $\mathcal{U}_{\mathbf{R}}$, given access to an $\varepsilon'$-approximate phase oracle $O_{\mathbf{R}}$ of the form $O_{\mathbf{R}}:\ket{i} \mapsto e^{-i\arccos(R_i/\alpha_{\bf R})}\ket{i}$ using $\mathcal{O}(1)$ applications of this subroutine and additional two-qubit gates.
\end{lemma}

\begin{proof}
    We choose target phases $\boldsymbol{\theta}$ such that $\theta_j = \arccos(R_j/\alpha_{\mathbf{R}})$.
    The \ac{LCU} construction yields $\mathcal{U}_{\mathbf{R}}$ using the oracle $O_{\mathbf{R}}$ \cite{zylberman_efficient_2024, zylberman_efficient_2025}.
    The precision of this encoding is bounded by the oracle error
    \begin{equation} \label{eq:real_diag}
    {
    \bigg\| \mathbf{R}/\alpha_{\mathbf{R}} - \frac{O_{\mathbf{R}}+O_{\mathbf{R}}^\dagger}{2} \bigg\| \le \varepsilon'.}
    \end{equation}
\end{proof}

Next in order to achieve some of the quadratic speedups seen in this paper we use optimizations similar to the sum of squares formalism~\cite{king2026quantum,simon2024amplified,babbush2023exponential}.  A key idea behind these works is the idea that projector like matrices can be fast-square-rooted by extending them to a rectangular matrix.  Specifically if we take $\mathbf{A}$ to be a square matrix and $\mathbf{B}$ to be a rectangular matrix such that $\mathbf{B}^\dagger \mathbf{B} = \mathbf{A}$ then we can use a block encoding of $\mathbf{B}$ to reduce the overall block encoding constant.  This result is stated below for completeness for the special case of a selection matrix (which can be thought of as a linear isometry).
\begin{lemma}
    Let $\mathbf{S}$ be a selection matrix with row sparsity $1$ such that 
    \begin{equation} \label{eq:selection_mat}
        \mathbf{S} = \sum_{i=0}^{N_{\operatorname{row}}-1} \ket{i}\bra{j_i} \in \mathbb{R}^{N_{\operatorname{row}} \times M},
    \end{equation}
    where $j_i$ is the column index of the single nonzero matrix element in row $i$ and $M \leq N_{\operatorname{row}}$.
    Let $s_c$ denote the maximum column sparsity of $\mathbf{S}$, and let these matrix elements be yielded by the row and column oracles $O_r,O_c$.
    Further, define the ancilla register size as $a_{\mathbf{S}} := \lceil \log_2 N_{\operatorname{row}} \rceil$.
    Then a $\lb \sqrt{s_c}, a_{\mathbf{S}}, 0 \rb$-\ac{BE} of $\mathbf{S}$ can be constructed using one query to $O_r^{-1} = O_r^\dagger$ and one query to $O_c$, requiring a total of $\mathcal{O}(\log N_{\operatorname{row}} )$ two-qubit gates.
\end{lemma}

\begin{proof}
Since $\mathbf{S}$ is a selection matrix, its row sparsity is strictly $s_r = 1$. 
Define the state preparation unitary $W_c$ on the column ancilla:
\begin{equation}
W_c \ket{0} = \frac{1}{\sqrt{s_c}}\sum_{l=0}^{s_c - 1} \ket{l}.
\end{equation}
Preparing this uniform superposition requires $ \mathcal{O}(a_{\bf s})$ two-qubit gates \cite{shukla_efficient_2024}.
We then use this preparation routine to construct the following unitary $\mathcal{U}_S \in \mathbb{C}^{2^{a_{\bf s}}\times 2^{a_{\bf s}}}\otimes \mathbb{C}^{2^{a_{\bf s}}\times 2^{a_{\bf s}}}$:
\begin{equation}
\mathcal{U}_S := \mathtt{SWAP} \cdot O_r^\dagger \cdot O_c \cdot \lb W_c \otimes \one \rb
\end{equation}
Here $\mathcal{U}_S$ is a $\lb \sqrt{s_c}, a_{\mathbf{S}}, 0 \rb$-block-encoding of $\mathbf{S}$.
This can be verified by direct computation.
Let $\ket{\psi} = \sum_{j=0}^{M-1} \alpha_j \ket{j}$ be an arbitrary pure quantum state in the domain of $S$. Note that while $\ket{\psi}$ is in $\mathbb{C}^{2^{a_{\bf S}}}$, we do not care about its support outside of this subspace because of the definition of the block encoding.  By observing that $s_r = 1$ implies the row index $k=0$, which means that we do not have to sum over all other non-zero matrix elements in our results.  With this simplification, we have
\begin{equation}
\begin{split}
\lb \bra{0} \otimes \one \rb \mathcal{U}_S \lb \ket{0} \otimes \ket{\psi} \rb &= \lb \bra{0} \otimes \one \rb \mathtt{SWAP} \cdot O_r^\dagger \cdot O_c \cdot \lb W_c \otimes \one \rb \lb \ket{0} \otimes \ket{\psi} \rb \\
&= \lb \bra{0} \otimes \sum_{i=0}^{2^{a_{\bf s}}-1} \ketbra{i}{i} \rb \mathtt{SWAP} \cdot O_r^\dagger \cdot O_c \frac{1}{\sqrt{s_c}} \sum_{l=0}^{s_c-1} \ket{l} \otimes \sum_{j=0}^{2^{a_{\bf s}}-1} \alpha_j \ket{j} \\
&= \lb \sum_{i=0}^{2^{a_{\bf s}}-1} \bra{i} \otimes \ketbra{i}{0} \rb O_r^\dagger \cdot O_c \frac{1}{\sqrt{s_c}} \sum_{l=0}^{s_c-1} \ket{l} \otimes \sum_{j=0}^{2^{a_{\bf s}}-1} \alpha_j \ket{j} \\
&= \frac{1}{\sqrt{s_c}} \sum_{i=0}^{2^{a_{\bf s}}-1} \sum_{l=0}^{s_c-1} \sum_{j=0}^{2^{a_{\bf s}}-1} \alpha_j \lb \bra{i} \otimes \ketbra{i}{r_{i0}} \rb \lb \ket{c_{lj}} \otimes \ket{j} \rb \\
&= \frac{1}{\sqrt{s_c}} \sum_{i=0}^{2^{a_{\bf s}}-1} \alpha_{r_{i0}} \ket{i} \\
&= \frac{1}{\sqrt{s_c}} \mathbf{S} \ket{\psi}.
\end{split}
\end{equation}
As this operation is a linear operator and it has the correct action on an arbitrary input state, it follows by linearity that it must have the correct action on any input.  Thus the block encoding is correct.

The runtime analysis for the algorithm is simple. The $\mathtt{SWAP}$ operation requires additional $\mathcal{O}(a_{\bf S})$ two-qubit gates to implement and only two queries are made to the oracles $O_r$ and $O_c$ in the procedure, which proves our claim.
\end{proof}

With these results we can proceed to block encode the topology matrix, which serves to provide the structure of the couplings in our graph.

\begin{lemma}[Topology Matrix Encoding]
    The topology matrix $\mathbf{T}_0=[\mathbf{D}_0;\mathbf{S}_0] \in \{-1,0,1\}^{(N_e+N_g) \times M}$, composed of the directed incidence matrix $\mathbf{D}_0\in \{-1,0,1\}^{N_e \times M}$ and ground selector $\mathbf{S}_0\in \{0,1\}^{N_g \times M}$, admits a $(\alpha_{\mathbf{T}_0}, 
     \mathcal{O}(\log_2(N_e + N_g)),0)$-\ac{BE} with normalization factor
    \begin{equation} \label{eq:triangle}
        {\alpha_{\mathbf{T}_0} = \sqrt{d_{\operatorname{in}}}+\sqrt{d_{\operatorname{out}}},}
    \end{equation}
    where $d_{\operatorname{in}}=\max_j(s^+_{j}+s_{0,j})$ and $d_{\operatorname{out}}=\max_j(s^-_{j}+s_{0,j})$ are the maximum input and output connectivities, respectively.
    The circuit uses $\mathcal{O}(\log(N_e + N_g))$ two-qubit gates.
\end{lemma}

\begin{proof}
    We decompose $\mathbf{D}_0$ using spring heads $\mathbf{S}^+\in \{0,1\}^{N_e \times M}$ and tails $\mathbf{S}^-\in \{0,1\}^{N_e \times M}$ selection matrices.
    Using \ac{LCU} with one auxiliary qubit, we compose the two selection matrices $\mathbf{S}_A = [\mathbf{S}^+; \mathbf{S}_0]\in \{0,1\}^{(N_e+N_g) \times M}$ and $\mathbf{S}_B = [\mathbf{S}^-; \mathbf{S}_0]\in \{0,1\}^{(N_e+N_g) \times M}$ provided  as follows:
    \begin{equation} \label{eq:single_topology}
    \mathbf{T}_0 =
    \begin{bmatrix}
    \mathbf{I}_{N_e} & \mathbf{0} \\ \mathbf{0} & e^{i\pi/3}\mathbf{I}_{N_g}
    \end{bmatrix}
    \underbrace{
    \begin{bmatrix}
    \mathbf{S}^+ \\ \mathbf{S}_0
    \end{bmatrix}}_{\mathbf{S}_A} +
    \begin{bmatrix}
    -\mathbf{I}_{N_e} & \mathbf{0} \\ \mathbf{0} & e^{-i\pi/3}\mathbf{I}_{N_g}
    \end{bmatrix}
    \underbrace{
    \begin{bmatrix}
    \mathbf{S}^- \\ \mathbf{S}_0
    \end{bmatrix}}_{\mathbf{S}_B}
    =
    \begin{bmatrix}
    \mathbf{D}_0 \\ \mathbf{S}_0
    \end{bmatrix}.
    \end{equation}
    Here, the matrix $\mathbf{S}_A$ is constructed by vertically stacking $\mathbf{S}^+$ and $\mathbf{S}_0$. 
    Because both $\mathbf{S}^+$ and $\mathbf{S}_0$ are selection matrices containing exactly a single $1$ per row, the rows of $\mathbf{S}_A$ are standard basis vectors. 
    Consequently, the product $\mathbf{S}_A^\dagger \mathbf{S}_A$ is a diagonal matrix whose entries are the column sums of $\mathbf{S}_A$. 
    Since the maximum column sum of $\mathbf{S}_A$ is $\max_j (s_j^+ + s_{0,j})$ by definition, its spectral norm evaluates to $\|\mathbf{S}_A\| = \sqrt{\|\mathbf{S}_A^\dagger \mathbf{S}_A\|} = \sqrt{\max_j (s_j^+ + s_{0,j})}$. 
    Then, invoking the previous lemma alongside the techniques of~\cite{low2017hamiltonian}, a block encoding of $\mathbf{S}_A$ can be constructed with normalization constant $\sqrt{\max_j (s_j^+ + s_{0,j})}$.
    
    The factor multiplying $\mathbf S_A$ in \eqref{eq:single_topology} is a unitary matrix.  
    Thus we can block encode it with a normalization constant of $1$.  Then using the multiplicative property of block encodings~\cite{gilyen_quantum_2019}, the product of the block encodings has the product of the block encoding constants.  The sum of block encodings can be implemented using the sum of the block encoding constants, and thus by repeating the same argument for $\mathbf S_B$ we have that ${\bf T_0}$ has a $\sqrt{\max_j(s_j^++s_{0,j})}+\sqrt{\max_j(s_j^-+s_{0,j})}$ block encoding as claimed.
    
    We define $\alpha_{\mathbf{S}_A}=\sqrt{\max_j(s^+_{j}+s_{0,j})}=\sqrt{d_{\operatorname{in}}}$ and $\alpha_{\mathbf{S}_B}=\sqrt{\max_j(s^-_{j}+s_{0,j})}=\sqrt{d_{\operatorname{out}}}$.
    The total normalization $\alpha_{\mathbf{T}_0}$ follows from the \ac{LCU} sum of the components and the construction introduces $\mathcal{O}(1)$ single qubit phase gates and inherits $\mathcal{O}(\log(N_e + N_g))$ two-qubit gates from $\mathcal{U}_{\mathbf{S}}$.  The block encoding construction requires a single ancillary qubit to perform the rotations, but the block encoding itself of each of the summands requires a further $\lceil\log(N_e+N_g)\rceil$ qubits to hold the output of the column index for the sparse Hamiltonian oracle.  This is the dominant contribution to the space overheads which justifies the $\mathcal{O}(\log(N_e+N_g))$ scaling.
\end{proof}

In \Cref{eq:single_topology}, we chose the sub-optimal phase cancellation implemented by two phase gates for notational simplicity.
We point out though that by adding the terms constructively, i.e. $\mathbf{T}^\prime_0=[\mathbf{D}_0;2\mathbf{S}_0]$, one can implement grounding springs with a stiffness $\mathbf{K}^\prime_{g,c}$ up to $4$ times larger without negatively impacting the normalization constant.
This accounts for the lower maximum eigen-frequencies of springs that are fixed to ground.
However, in this notion the clear interpretation of equally weighted energies in the computational basis vanishes.

The normalization in \Cref{eq:triangle} is optimal for $d$-regular bipartite graphs (for even $d$) under a balanced edge orientation gauge. 
By orienting half of the incident edges inward and half outward at every node, we achieve $d_{\operatorname{in}} = d_{\operatorname{out}} = d/2$. 
This yields $\alpha_{\mathbf{T}_0} = \sqrt{d/2} + \sqrt{d/2} = 2\sqrt{d/2} = \sqrt{2d} = \|\mathbf{T}_0\|$.
Similarly, the normalization is asymptotically optimal for star graphs with $N$ leaves.
For such graphs, we have $d_{\operatorname{in}}=N$ and $d_{\operatorname{out}}=1$, yielding $\alpha_{\mathbf{T}_0} = \sqrt{N}+1$.
Given the spectral norm $\|\mathbf{T}_0\|=\sqrt{N+1}$, we observe that the normalization tightens relatively in the large limit:
\begin{equation}
\lim_{N\rightarrow \infty} \frac{\alpha_{\mathbf{T}_0}}{\|\mathbf{T}_0\|} = \lim_{N\rightarrow \infty} \frac{\sqrt{N}+1}{\sqrt{N+1}} = 1.
\end{equation}
This encoding provides a normalization advantage of up to $\sqrt{2}$ for certain topologies compared to previous work \cite{babbush2023exponential}.

\begin{lemma}[Body Matrix Encoding]
    Assuming all coupling blocks, as described in \Cref{fig:general_visco_2mass}, possess exactly $C$ parallel springs, the multi-body topology matrix $\mathbf{T} = \mathbf{1}_C \otimes \mathbf{T}_0$ admits a $(\alpha_{\mathbf{T}}, a_{\mathbf{T}}, 0)$-\ac{BE} with normalization factor
    \begin{equation}
    {\alpha_{\mathbf{T}} = \alpha_{\mathbf{T}_0}\sqrt{C},}
    \end{equation}
    and with $a_{\mathbf{T}}=a_{\mathbf{T}_0}+\lceil \log_2(C)\rceil$.
    The circuit requires a single query to $\mathcal{U}_{\mathbf{T}_0}$ and uses $\mathcal{O}(\log C + \log(N_e + N_g))$ two-qubit gates.
\end{lemma}
\begin{proof}
    Assuming all coupling blocks possess $C$ parallel springs, the multi-body topology matrix is defined by
    \begin{equation}
        \mathbf{T} = \mathbf{1}_C \otimes \mathbf{T}_0 \in \{-1,0,1\}^{C(N_e+N_g)\times M},
    \end{equation}
    where $\mathbf{1}_C$ denotes the column vector of all ones. By observing that $\mathbf{1}_C$ corresponds to the scaled uniform superposition state, i.e., $\mathbf{1}_C = \sqrt{C}\ket{+_C}$ with $\ket{+_C} = \frac{1}{\sqrt{C}}\sum_{j=0}^{C-1}\ket{j}$, we may express the expansion as
    \begin{equation}
        \mathbf{T} = \sqrt{C}\ket{+_C} \otimes \mathbf{T}_0.
    \end{equation}
    The implementation utilizes a body register $\mathsf{R}_{A}$ and an auxiliary register $\mathsf{R}_{B}$, each of size $c= \lceil \log_2(C)\rceil$. 
    To block-encode the expansion, we apply a transversal $\mathtt{CNOT}$ operation $\mathcal{U}_{\mathtt{CNOT}} = \bigotimes_{k=1}^c \mathtt{CNOT}^{(k)}_{\mathsf{R}_{A} \to \mathsf{R}_{B}}$.
    We then prepare the uniform superposition $\ket{+_C}$ on $\mathsf{R}_{A}$ with $\mathcal{U}_{\operatorname{prep}}$ using $\mathcal{O}(c)$ gates \cite{shukla_efficient_2024} followed by the query to $\mathcal{U}_{\mathbf{T}_0}$ in the data register $\mathsf{R}_{D}$ to block encode $\mathbf{T}$ in the $\ket{0}\!\bra{0}_{\mathsf{R}_{B}}$ block.
    More specifically
    \begin{equation}
        \begin{aligned}
            &(\bra{0}_{\mathsf{R}_B} \otimes \mathbb{1}_{\mathsf{R}_A} \otimes \mathbb{1}_{\mathsf{R}_D}) (\mathbb{1}_{\mathsf{R}_B} \otimes \mathbb{1}_{\mathsf{R}_A} \otimes \mathcal{U}_{\mathbf{T}_0}) (\mathbb{1}_{\mathsf{R}_B} \otimes \mathcal{U}_{\operatorname{prep}} \otimes \mathbb{1}_{\mathsf{R}_D}) (\mathcal{U}_{\mathtt{CNOT}}^{\mathsf{R}_A \to \mathsf{R}_B} \otimes \mathbb{1}_{\mathsf{R}_D}) (\ket{0}_{\mathsf{R}_B} \otimes \mathbb{1}_{\mathsf{R}_A} \otimes \mathbb{1}_{\mathsf{R}_D}) \\
            &= (\mathbb{1}_{\mathsf{R}_A} \otimes \mathcal{U}_{\mathbf{T}_0}) (\mathcal{U}_{\operatorname{prep}} \otimes \mathbb{1}_{\mathsf{R}_D}) \left( \sum_{k=0}^{2^c-1} \braket{0|k} \ket{k}\!\bra{k}_{\mathsf{R}_{A}} \otimes \mathbb{1}_{\mathsf{R}_D} \right) \\
            &= \mathcal{U}_{\operatorname{prep}} \ket{0}\!\bra{0}_{\mathsf{R}_{A}} \otimes \mathcal{U}_{\mathbf{T}_0} \\
            &= \ket{+_C}\!\bra{0}_{\mathsf{R}_{A}} \otimes \mathcal{U}_{\mathbf{T}_0} \\
            &= \frac{1}{\sqrt{C}} \sum_{j=0}^{C-1} \ket{j}\!\bra{0}_{\mathsf{R}_{A}} \otimes \mathcal{U}_{\mathbf{T}_0}
        \end{aligned}
    \end{equation}
    Since $\mathcal{U}_{\mathbf{T}_0}$ block-encodes $\mathbf{T}_0$ with normalization $\alpha_{\mathbf{T}_0}$, this final operator encodes $\mathbf{1}_C \otimes \mathbf{T}_0 = \mathbf{T}$ with the overall normalization factor $\alpha_{\mathbf{T}} = \alpha_{\mathbf{T}_0}\sqrt{C}$.
    The circuit requires a single query to $\mathcal{U}_{\mathbf{T}_0}$, introduces $2c$ additional qubits where $c$ are auxiliaries, and introduces $\mathcal{O}(\log C)$ two-qubit gates, while inheriting $\mathcal{O}(\log(N_e + N_g))$ two-qubit gates.
\end{proof}

\subsection{Hamiltonian} \label{s:be_hamiltonian}

\begin{lemma}[Weighted Topology Matrix Encoding]\label{lem:weightTopo}
    The weighted topology matrix $\hat{\mathbf{T}} = \mathbf{K}^{1/2}\mathbf{T}\mathbf{M}^{-1/2}$ admits a $(\alpha_{\hat{\mathbf{T}}}, a_{\hat{\mathbf{T}}}, \varepsilon)$-\ac{BE} with normalization factor
    \begin{equation}
    {\alpha_{\hat{\mathbf{T}}} = (\sqrt{d_{\operatorname{in}}}+\sqrt{d_{\operatorname{out}}})\sqrt{\frac{k_{\max}}{m_{\min}}C}.}
    \end{equation}
    The encoding requires one query to $\mathcal{U}_\mathbf{T}$, and two queries each to the $\varepsilon'$-accurate oracles $O_{\mathbf{K}^{1/2}}$, and $O_{\mathbf{M}^{-1/2}}$ and uses $a_{\hat{\mathbf{T}}} = 3 + \lceil \log_2(C)\rceil +\lceil\log_2(N_e + N_g)\rceil$ ancillae, $\mathcal{O}(\log C + \log(N_e + N_g))$ two-qubit gates, and has additive error $\varepsilon = 2\alpha_{\hat{\mathbf{T}}}\varepsilon'$.
\end{lemma}

\begin{proof}
    We introduce the weighted topology matrix
    \begin{equation} \label{eq:weighted_topology}
    \hat{\mathbf{T}} = \mathbf{K}^{1/2}\mathbf{T}\mathbf{M}^{-1/2},
    \end{equation}
    where $\mathbf{K} = \diag\bigl(\mathbf{K}_{1}, \dots, \mathbf{K}_C)\succ0$ with $\mathbf{K}_c=\diag\bigl(\mathbf{K}_{e,c},\mathbf{K}_{g,c}\bigr)$ is the diagonal matrix of spring constants and $\mathbf{M}\succ0$ is the diagonal matrix of masses.
    
    As $\mathbf{K}^{1/2},\mathbf{M}^{-1/2}$ are diagonal and positive definite, we avoid computing the inverse square-root, for example, through inequality testing \cite{babbush2023exponential}, and directly access the real diagonal matrices $\mathcal{U}_{\mathbf{R}}\rightarrow \mathcal{U}_{\mathbf{K}^{1/2}}$ and $\mathcal{U}_{\mathbf{R}} \rightarrow \mathcal{U}_{\mathbf{M}^{-1/2}}$ defined in \Cref{eq:real_diag} through their respective phase oracles.
    By the multiplicative property of block encodings \cite{gilyen_quantum_2019} the product of two block encodings can be implemented with multiplicative normalization constant and summed number of ancilla qubits.
    Consequently, applied to the product $\mathbf{K}^{1/2}\mathbf{T}\mathbf{M}^{-1/2}$, we obtain the stated overall normalization constant. 
    Because the block encoding of $\mathbf{T}$ is exact, the additive errors of the diagonal oracles cross-multiply with the respective normalizations, yielding a total encoding error of 
    \begin{equation}
        \varepsilon = \alpha_{\mathbf{M}^{-1/2}}\alpha_{\mathbf{T}}(\alpha_{\mathbf{K}^{1/2}}\varepsilon') + (\alpha_{\mathbf{K}^{1/2}}\alpha_{\mathbf{T}})(\alpha_{\mathbf{M}^{-1/2}}\varepsilon') = 2\alpha_{\hat{\mathbf{T}}}\varepsilon' = \mathcal{O}(\alpha_{\hat{\mathbf{T}}}\varepsilon').
    \end{equation}
    The circuit cost is from the queried \acp{BE}: $\mathcal{O}(1)+\mathcal{O}(\log C + \log(N_e + N_g)) +\mathcal{O}(1)$.
\end{proof}
This block encoding of the weighted topology matrix allows us to immediate construct a block encoding of the Hamiltonian ${\bf H}$ which is defined in~\Cref{eq:hamiltonian_dissipator}.

\begin{lemma}[Hamiltonian Encoding (Induced Error)]
    The Hamiltonian $\mathbf{H}$ admits a $(\alpha_{\mathbf{H}}, a_{\mathbf{H}}, \varepsilon)$ Hermitian block encoding with normalization factor
    \begin{equation}
    {\alpha_{\mathbf{H}} =(\sqrt{d_{\operatorname{in}}}+\sqrt{d_{\operatorname{out}}})\sqrt{\frac{k_{\max}}{m_{\min}}C},}
    \end{equation}
    and $a_{\mathbf{H}} = 3 + \lceil \log_2(C)\rceil +\lceil\log_2(N_e + N_g)\rceil$, where $N=M+C(N_e+N_g)$.
    The encoding requires four queries each to the $\varepsilon'$-accurate phase oracles $O_{\mathbf{K}^{1/2}}$, $O_{\mathbf{M}^{-1/2}}$, and one query each to the index oracles $O_r^{(\mathcal{M})},O_c^{(\mathcal{M})}$ for $\mathcal{M} \in \{\mathbf{S}_A, \mathbf{S}_B\}$ and to each of their adjoints.
    The circuit uses $\mathcal{O}(\log C + \log(N_e + N_g))$ two-qubit gates and has
    additive error $\varepsilon = 2\alpha_{\mathbf{H}} \varepsilon'$.
\end{lemma}
\begin{proof}
    We construct $\check{\mathbf{H}}\in \mathbb{C}^{2C(N_e+N_g) \times 2C(N_e+N_g)}$ as the Hermitian dilation of $\hat{\mathbf{T}}$ to embed the physical Hamiltonian $\mathbf{H} \in \mathbb{C}^{N \times N}$  \cite{babbush2023exponential}.
    By definition in~\Cref{eq:hamiltonian_dissipator}, this matrix is 
    \begin{equation}
    {\bf H} = i\sqrt{\bf B} \begin{bmatrix}
        {\bf 0} & -{\bf T}^\dagger\\
        \bf{T} & {\bf 0}
    \end{bmatrix}    \sqrt{\bf B}
    \end{equation}
    This can be done unitarily with constant gate overhead by introducing a single state qubit, one call to controlled-$\mathcal{U}_{\hat{\mathbf{T}}}$ and its adjoint, and one Pauli-$\mathbf{Y}$ gate as
    \begin{equation}
    \check{\mathbf{H}} = (\mathbf{Y}\otimes \mathbf{I})\begin{pmatrix} \hat{\mathbf{T}} & \mathbf{0} \\ \mathbf{0} & \hat{\mathbf{T}}^\dagger \end{pmatrix}
    = i\begin{bmatrix}
    \mathbf{0}&-\hat{\mathbf{T}}^{\dagger}\\
    \hat{\mathbf{T}}&\mathbf{0}
    \end{bmatrix},
    \end{equation}
    where $\mathbf{Y}= \begin{bmatrix} 0 & -i \\ i & 0 \end{bmatrix}$ and $\hat{\mathbf{T}}$ is the weighted topology matrix defined in~\Cref{eq:weighted_topology}.

     \Cref{lem:weightTopo} provides a bound we use for the unitary block encoding for $\hat{\bf T}$, and $\hat{\bf T}^\dagger$ can be generated by changing our oracle to generate the adjoint of $\hat{\bf T}$.
     For this we invoke the adjoint identity 
    \begin{equation}
     (\hat{\mathbf{T}})^\dagger = (\mathbf{K}^{1/2}\mathbf{T}\mathbf{M}^{-1/2})^\dagger = (\mathbf{M}^{-1/2})^\dagger(\mathbf{T})^\dagger(\mathbf{K}^{1/2})^\dagger = \mathbf{M}^{-1/2}(\mathbf{T})^\dagger\mathbf{K}^{1/2},
     \end{equation}
     which follows from the material matrices being real and diagonal.
    As we assume index oracles, we can implement $(\hat{\mathbf{T}})^\dagger$ using an adjoint implementation $\mathbf{S}^\dagger$, followed by the previously described topology matrix encoding.
    The unitary block encoding for $\hat{\bf T}$ and $\hat{\bf T}^\dagger$ can thus be achieved using at most twice the resources as generating $\hat{\bf T}$ because we can build a controlled adjoint circuit and a controlled ordinary unitary circuit to implement $\ketbra{0}{0}\otimes \hat{\bf T}+\ketbra{1}{1}\otimes \hat{\bf T}^\dagger$. 
    
    The block encoding constant for such a direct sum is simply the maximum of the two block encoding constants.  We have from~\Cref{lem:weightTopo}
    that the normalization factor $\alpha_{\mathbf{H}}$ is  $\alpha_{\hat{\mathbf{T}}}$.
    The block encoding unitary also requires $3 + \lceil \log_2(C)\rceil +\lceil\log_2(N_e + N_g)\rceil$ ancillae and uses $\mathcal{O}(\log C + \log(N_e + N_g))$ two-qubit gates from \Cref{lem:weightTopo}.
    
    Each query to the controllect block encoding, controlled-$\mathcal{U}_{\hat{\mathbf{T}}}$, decomposes into two queries to $O_{\mathbf{K}^{1/2}}$, $O_{\mathbf{M}^{-1/2}}$, $O_c$, $O_r^\dagger$ from~\Cref{lem:weightTopo}, and one query to controlled-$\mathcal{U}_{\hat{\mathbf{T}}^\dagger}$, which decomposes as two queries to $O_{\mathbf{K}^{1/2}}$, $O_{\mathbf{M}^{-1/2}}$, $O_c^\dagger$, $O_r$.

    The block encoding constant of a matrix does not change when multiplying by an exact unitary from the block encoding lemma of~\cite{gilyen_quantum_2019}.  Thus the block encoding constant using this approach is simply the block encoding constant of $\hat{\bf T}$.  This block encoding constant is from~\Cref{lem:weightTopo}:
    \begin{equation}
         {\alpha_{\mathbf{H}} =(\sqrt{d_{\operatorname{in}}}+\sqrt{d_{\operatorname{out}}})\sqrt{\frac{k_{\max}}{m_{\min}}C},}
    \end{equation}
    Similarly, the error in the product of block encoding and a fixed unitary matrix with normalization is simply the error in the block encoding of $\hat{\bf T}$ using the results of~\cite{gilyen_quantum_2019}.
    Thus the total error in the Hermitian block encoding $\hat{\bf H}$ obeys $\varepsilon = 2\alpha_{\hat{\mathbf{T}}}\varepsilon'$ is inherited from $\mathcal{U}_{\hat{\mathbf{T}}}$.
\end{proof}

As $M \geq 0$ and $N = M + C(N_e+N_g)$, we inherently have $C(N_e+N_g) \leq N$, which yields $\log C + \log(N_e+N_g) \leq \log N$.
Furthermore, as $M \leq N_e + N_g$, we have by substitution
\begin{equation}
    N \leq (C+1)(N_e+N_g),
\end{equation}
and by using the product rule of the logarithm
\begin{equation}
    \log N \leq \log \Big((C+1)(N_e+N_g) \Big) = \log (C+1) + \log (N_e+N_g).
\end{equation}
These bounds establish that $3 + \lceil \log_2(C)\rceil +\lceil\log_2(N_e + N_g)\rceil = \mathcal{O}(\log C + \log(N_e + N_g)) = \Theta(\log N)$, and \Cref{lem:hamiltonian_encoding} follows.
Note, that as an edge case for ungrounded global tree topologies we have $M = N_e + N_g + 1$.
Such a graph can still be simulated by introducing a single dummy ground connection with $k=\eta=0$ (and appropriate relaxation of the oracle to include dummy rows), which requires relaxing the parameter constraint to $k \geq0$ and to restrict the corresponding oracles and inverses to the non-zero subspace.
By attaching the ground connection to one of the leaf nodes, the normalization remains unaffected.
The same idea can be generalized to multiple disconnected trees.
Also note, that the 1-ancilla embedding requires a padding to $\check{\mathbf{H}} \in i\mathbb{R}^{2C(N_e+N_g) \times 2C(N_e+N_g)}$.
For $M=C(N_e+N_g)$ this is optimal.
Otherwise it can be a qubit overhead of up to $1$.
In the interest of notational brevity we do not explicitly denote this padding in the main text.

The previous result fixes the oracle accuracy and reports the induced encoding error. 
Inverting the relation gives the equivalent statement in terms of a target encoding error.
\begin{corollary} [Hamiltonian Encoding (Target Error)]
    The Hamiltonian $\mathbf{H}$ admits a $(\alpha_{\mathbf{H}}, a_{\mathbf{H}}, \varepsilon)$ Hermitian block encoding with normalization factor
    \begin{equation}
    {\alpha_{\mathbf{H}} =(\sqrt{d_{\operatorname{in}}}+\sqrt{d_{\operatorname{out}}})\sqrt{\frac{k_{\max}}{m_{\min}}C},}
    \end{equation}
    and $a_{\mathbf{H}} = 3 + \lceil \log_2(C)\rceil +\lceil\log_2(N_e + N_g)\rceil$, where $N=M+C(N_e+N_g)$ and $\varepsilon \in (0,1)$ is the target error.
    The encoding requires four queries each to the $\varepsilon/(2\alpha_{\mathbf{H}})$-accurate phase oracles $O_{\mathbf{K}^{1/2}}$, $O_{\mathbf{M}^{-1/2}}$, and one query each to the index oracles $O_r^{(\mathcal{M})},O_c^{(\mathcal{M})}$ for $\mathcal{M} \in \{\mathbf{S}_A, \mathbf{S}_B\}$ and to each of their adjoints.
    The circuit uses $\mathcal{O}(\log C + \log(N_e + N_g))$ two-qubit gates.
\end{corollary}
\begin{proof}
    Trivially by inversion of error relation.
\end{proof}

\subsection{Dissipator} \label{s:be_dissipator}

\begin{lemma}[Mass Dissipation Encoding]
    The mass-weighted dissipation operator $\hat{\boldsymbol{\eta}}_{\mathbf{M}} = \mathbf{M}^{-1/2} \mathbf{T}_0^\dagger \boldsymbol{\eta}_0 \mathbf{T}_0 \mathbf{M}^{-1/2}$ admits a $(\alpha_{\hat{\boldsymbol{\eta}}_{\mathbf{M}}}, a_{\hat{\boldsymbol{\eta}}_{\mathbf{M}}}, \varepsilon)$-\ac{BE} with normalization factor
    \begin{equation}
    {\alpha_{\hat{\boldsymbol{\eta}}_{\mathbf{M}}}=(\sqrt{d_{\operatorname{in}}}+\sqrt{d_{\operatorname{out}}})^2\frac{\|{\boldsymbol{\eta}}_0\|_{\max}}{m_{\min}}.}
    \end{equation}
    The encoding requires four queries to the $\varepsilon'$-accurate phase oracle $O_{\mathbf{M}^{-1/2}}$, two queries to the $\varepsilon'$-accurate phase oracle $O_{\boldsymbol{\eta}_0}$, and one query each to the index oracles $O_r^{(\mathcal{M})},O_c^{(\mathcal{M})}$ for $\mathcal{M} \in \{\mathbf{S}_A, \mathbf{S}_B\}$ and to each of their adjoints.
    It uses $a_{\hat{\boldsymbol{\eta}}_{\mathbf{M}}} = 5 + 2\lceil\log_2(N_e + N_g)\rceil$ ancillae, $\mathcal{O}(\log(N_e + N_g))$ two-qubit gates, and has additive error $\varepsilon = 3\alpha_{\hat{\boldsymbol{\eta}}_{\mathbf{M}}}\varepsilon'$.
\end{lemma}

\begin{proof}
    We express the mass dissipation operator as
    \begin{equation}
    \boldsymbol{\eta}_{\mathbf{M}} = \mathbf{D}_0^{\dagger}\boldsymbol{\eta}_{e,0}\mathbf{D}_0 + \mathbf{S}_0^{\dagger}\boldsymbol{\eta}_{g,0}\mathbf{S}_0
    = \mathbf{T}_0^\dagger \boldsymbol{\eta}_0\mathbf{T}_0,
    \end{equation}
    where $\boldsymbol{\eta}_0=\diag(\boldsymbol{\eta}_{e,0},\boldsymbol{\eta}_{g,0})\succeq0$.
    The corresponding mass-weighted operator is
    \begin{equation}
    \hat{\boldsymbol{\eta}}_{\mathbf{M}} = \mathbf{M}^{-1/2} \boldsymbol{\eta}_{\mathbf{M}} \mathbf{M}^{-1/2} = \mathbf{M}^{-1/2} \mathbf{T}_0^\dagger \boldsymbol{\eta}_0 \mathbf{T}_0 \mathbf{M}^{-1/2} \succeq 0.
    \end{equation}
    As $\mathbf{M}^{-1/2}$ and $\boldsymbol{\eta}_0$ are diagonal and $\mathbf{M}^{-1/2}$ is positive definite, we avoid computing the inverse square-root through arithmetic and directly access the real diagonal matrices $\mathcal{U}_{\mathbf{R}}\rightarrow \mathcal{U}_{\mathbf{M}^{-1/2}}$ and $\mathcal{U}_{\mathbf{R}}\rightarrow \mathcal{U}_{\boldsymbol{\eta}_0}$ defined in \Cref{eq:real_diag} through their respective phase oracles.
    By the multiplicative property of block encodings \cite{gilyen_quantum_2019} applied to the sequential product $\mathbf{M}^{-1/2} \mathbf{T}_0^\dagger \boldsymbol{\eta}_0 \mathbf{T}_0 \mathbf{M}^{-1/2}$, we obtain the stated overall normalization constant $\alpha_{\hat{\boldsymbol{\eta}}_{\mathbf{M}}} = \alpha_{\mathbf{M}^{-1/2}}^2 \alpha_{\mathbf{T}_0}^2 \alpha_{\boldsymbol{\eta}_0}$. Substituting $\alpha_{\mathbf{M}^{-1/2}} = m_{\min}^{-1/2}$, $\alpha_{\boldsymbol{\eta}_0} = \|{\boldsymbol{\eta}}_0\|_{\max}$, and $\alpha_{\mathbf{T}_0} = \sqrt{d_{\operatorname{in}}}+\sqrt{d_{\operatorname{out}}}$ directly yields the stated $\alpha_{\hat{\boldsymbol{\eta}}_{\mathbf{M}}}$.
    
    Because the block encodings of $\mathbf{T}_0$ and $\mathbf{T}_0^\dagger$ are exact ($\varepsilon = 0$), the additive errors originate exclusively from the three diagonal phase oracles. Following the sub-multiplicativity of block-encoding errors for products, the total unnormalized additive error evaluates to the sum of the relative errors scaled by the total normalization:
    \begin{equation}
        \varepsilon = \alpha_{\hat{\boldsymbol{\eta}}_{\mathbf{M}}} \left( \frac{\alpha_{\mathbf{M}^{-1/2}}\varepsilon'}{\alpha_{\mathbf{M}^{-1/2}}} + 0 + \frac{\alpha_{\boldsymbol{\eta}_0}\varepsilon'}{\alpha_{\boldsymbol{\eta}_0}} + 0 + \frac{\alpha_{\mathbf{M}^{-1/2}}\varepsilon'}{\alpha_{\mathbf{M}^{-1/2}}} \right) = 3\alpha_{\hat{\boldsymbol{\eta}}_{\mathbf{M}}}\varepsilon'.
    \end{equation}
    
    The circuit cost is inherited from the queried \acp{BE}, requiring $\mathcal{O}(\log(N_e + N_g))$ two-qubit gates primarily from the topology components.
    It combines the disjoint ancillae of the constituent encodings: $1$ (from $\mathcal{U}_{\mathbf{M}^{-1/2}}$) $+ (1 + \lceil\log_2(N_e + N_g)\rceil)$ (from $\mathcal{U}_{\mathbf{T}_0^\dagger}$) $+ 1$ (from $\mathcal{U}_{\boldsymbol{\eta}_0}$) $+ (1 + \lceil\log_2(N_e + N_g)\rceil)$ (from $\mathcal{U}_{\mathbf{T}_0}$) $+ 1$ (from $\mathcal{U}_{\mathbf{M}^{-1/2}}$) $= 5 + 2\lceil\log_2(N_e + N_g)\rceil$ ancillae.
    The query counts exactly decompose into two queries per real diagonal phase oracle evaluation, and two queries per sparse oracle evaluation for each of the $\mathbf{T}_0$ and $\mathbf{T}_0^\dagger$ components, yielding the stated quantities.
\end{proof}

\begin{lemma}[Spring Dissipation Encoding]
    The material-weighted spring dissipation operator $\boldsymbol{\Lambda} = \mathbf{K}^{1/2}\boldsymbol{\eta}^{-1}_{\mathbf{K}}\mathbf{K}^{1/2}$ admits a $(\alpha_{\boldsymbol{\Lambda}}, a_{\boldsymbol{\Lambda}}, \varepsilon)$-\ac{BE} with normalization factor
    \begin{equation}
    {\alpha_{\boldsymbol{\Lambda}} = \|\boldsymbol{\Lambda}\|_{\infty}.}
    \end{equation}
    The encoding requires two queries to the $\varepsilon'$-accurate phase oracle $O_{\boldsymbol{\Lambda}}$, uses $a_{\boldsymbol{\Lambda}} = 1$ ancilla qubit, $\mathcal{O}(1)$ additional two-qubit gates, and has additive error $\varepsilon = \alpha_{\boldsymbol{\Lambda}}\varepsilon'$.
\end{lemma}

\begin{proof}
    The spring dissipation operator is diagonal and positive definite
    \begin{equation}
    \boldsymbol{\eta}_{\mathbf{K}} = \operatorname{diag}(\boldsymbol{\eta}_1, \dots, \boldsymbol{\eta}_C),
    \end{equation}
    with $\boldsymbol{\eta}_c=\diag(\boldsymbol{\eta}_{e,c},\boldsymbol{\eta}_{g,c})$.
    Since $\mathbf{K}\succ 0$ is diagonal, it strictly commutes with $\boldsymbol{\eta}^{-1}_{\mathbf{K}} \succeq 0$.
    We can therefore write the material-weighted operator exactly as
    \begin{equation}
    \boldsymbol{\Lambda} = \mathbf{K}^{1/2}\boldsymbol{\eta}^{-1}_{\mathbf{K}}\mathbf{K}^{1/2} = \mathbf{K}\boldsymbol{\eta}^{-1}_{\mathbf{K}}.
    \end{equation}
    Because the relaxation rate matrix $\boldsymbol{\Lambda}$ is a real diagonal matrix with elements $\lambda_c = k_c/\eta_c$, we directly access it as $\mathcal{U}_{\mathbf{R}}\rightarrow \mathcal{U}_{\boldsymbol{\Lambda}}$ via the real diagonal matrix encoding scheme defined in \Cref{eq:real_diag}. 
    The normalization factor natively evaluates to $\alpha_{\boldsymbol{\Lambda}} = \|\boldsymbol{\Lambda}\|_{\infty}$, and the precision is strictly bounded by $\alpha_{\boldsymbol{\Lambda}}\varepsilon'$.
    The circuit utilizes exactly $1$ ancilla qubit, $\mathcal{O}(1)$ additional two-qubit gates to evaluate the 1-qubit \ac{LCU}, and decomposes identically into two queries to $O_{\boldsymbol{\Lambda}}$.
\end{proof}

\begin{lemma}[Dissipator Encoding (Induced Error)]
    The Dissipator $\mathbf{L}$ admits an $(\alpha_{\mathbf{L}}, a_{\mathbf{L}}, \varepsilon)$-\ac{BE} with normalization factor
    \begin{equation}
    {\alpha_{\mathbf{L}} =\max\Big((\sqrt{d_{\operatorname{in}}}+\sqrt{d_{\operatorname{out}}})^2\frac{\|{\boldsymbol{\eta}}_0\|_{\max}}{m_{\min}}, \|\boldsymbol{\Lambda}\|_{\infty}\Big),}
    \end{equation}
    and $a_{\mathbf{L}} = 5 + 2\lceil\log_2(N_e + N_g)\rceil$, where $N=M+C(N_e+N_g)$.
    The encoding requires four queries to the $\varepsilon'$-accurate phase oracle $O_{\mathbf{M}^{-1/2}}$, two queries each to the $\varepsilon'$-accurate phase oracles $O_{\boldsymbol{\eta}_0}$ and $O_{\boldsymbol{\Lambda}}$, and one query each to the index oracles $O_r^{(\mathcal{M})},O_c^{(\mathcal{M})}$ for $\mathcal{M} \in \{\mathbf{S}_A, \mathbf{S}_B\}$ and to each of their adjoints.
    The circuit uses $\mathcal{O}(\log(N_e + N_g))$ two-qubit gates and has additive error $\varepsilon = 3\alpha_{\mathbf{L}} \varepsilon'$.
\end{lemma}

\begin{proof}
    We use a block-dilation $\check{\mathbf{L}} \in \mathbb{R}^{2C(N_e+N_g) \times 2C(N_e+N_g)}$ to encode the physical dissipator
    \begin{equation}
    \mathbf{L} =
    \begin{bmatrix}
    \hat{\boldsymbol{\eta}}_{\mathbf{M}} & \mathbf{0} \\ \mathbf{0} & \boldsymbol{\Lambda}
    \end{bmatrix} \in \mathbb{R}^{N \times N}.
    \end{equation}
    To implement this, we introduce a single system qubit. 
    To ensure a uniform block-encoding normalization, we construct the sub-block encodings to natively equal the joint maximum normalization $\alpha_{\mathbf{L}} = \max(\alpha_{\hat{\boldsymbol{\eta}}_{\mathbf{M}}}, \alpha_{\boldsymbol{\Lambda}})$.
    Conditioned on the qubit being $\ket{0}$, we apply $\mathcal{U}_{\hat{\boldsymbol{\eta}}_{\mathbf{M}}}$, and conditioned on $\ket{1}$, we apply $\mathcal{U}_{\boldsymbol{\Lambda}}$. 
    The block-encoding constant is the maximum of the two and both are queried once.

    The total ancilla requirement scales as the maximum of the disjoint sub-block ancillae, giving $a_{\mathbf{L}} = \max(a_{\hat{\boldsymbol{\eta}}_{\mathbf{M}}}, a_{\boldsymbol{\Lambda}}) = 5 + 2\lceil\log_2(N_e+N_g)\rceil$. 
    The total two-qubit gate complexity inherits the component topologies $\mathcal{O}(\log(N_e + N_g))$.
    Finally, the total additive error is bounded by the maximum error of the rescaled sub-blocks: $\varepsilon \leq \max(3\alpha_{\mathbf{L}}\varepsilon', \alpha_{\mathbf{L}}\varepsilon') \leq 3\alpha_{\mathbf{L}}\varepsilon'$.
\end{proof}

As established in \Cref{s:be_hamiltonian}, bounding $N \leq (C+1)(N_e+N_g)$ inherently establishes $a_{\mathbf{L}} = \mathcal{O}(\log N)$.
Hence, the previous proof can be directly reduced to the asymptotic bounds stated in \Cref{lem:dissipator_encoding}. 
Note that we can furthermore halve the number of ancillae used for the mass dissipation encoding.
We do this by eliminating the double registers in the mass dissipation through establishing query access to the factorization block encoding $\mathcal{U}_{\mathbf{F}}$ (see \Cref{s:fast_forward}) and constructing the quantum walk operator $\mathcal{W} = \mathcal{U}_{\mathbf{F}}^\dagger (2\ketbra{0}{0} - \mathbf{I}) \mathcal{U}_{\mathbf{F}}$. 
Together with a $1$-qubit \ac{LCU}, this natively extracts $\mathbf{F}^\dagger \mathbf{F}$ \cite{gilyen_quantum_2019}.
This approach introduces only a logarithmic gate overhead and preserves the normalization.
However, as such a change requires switching the access model $O_{\boldsymbol{\eta}_0} \rightarrow O_{\boldsymbol{\eta}_0^{{1/2}}}$ and brings no asymptotic benefit, we choose to prove the simpler approach.
Similar to the Hamiltonian, the 1-ancilla embedding requires a padding to $\check{\mathbf{L}} \in \mathbb{R}^{2C(N_e+N_g) \times 2C(N_e+N_g)}$, with the same associated overhead. 
In the interest of notational brevity we do not explicitly denote this padding in the main text.

The previous result fixes the oracle accuracy and reports the induced encoding error. 
Inverting the relation gives the equivalent statement in terms of a target encoding error.
\begin{corollary} [Dissipator Encoding (Target Error)]
    The Dissipator $\mathbf{L}$ admits an $(\alpha_{\mathbf{L}}, a_{\mathbf{L}}, \varepsilon)$-\ac{BE} with normalization factor
    \begin{equation}
    {\alpha_{\mathbf{L}} =\max\Big((\sqrt{d_{\operatorname{in}}}+\sqrt{d_{\operatorname{out}}})^2\frac{\|{\boldsymbol{\eta}}_0\|_{\max}}{m_{\min}}, \|\boldsymbol{\Lambda}\|_{\infty}\Big),}
    \end{equation}
    and $a_{\mathbf{L}} = 5 + 2\lceil\log_2(N_e + N_g)\rceil$, where $N=M+C(N_e+N_g)$ and $\varepsilon \in (0,1)$ is the target error.
    The encoding requires four queries to the $\varepsilon/(3\alpha_{\mathbf{L}})$-accurate phase oracle $O_{\mathbf{M}^{-1/2}}$, two queries each to the $\varepsilon/(3\alpha_{\mathbf{L}})$-accurate phase oracles $O_{\boldsymbol{\eta}_0}$ and $O_{\boldsymbol{\Lambda}}$, and one query each to the index oracles $O_r^{(\mathcal{M})},O_c^{(\mathcal{M})}$ for $\mathcal{M} \in \{\mathbf{S}_A, \mathbf{S}_B\}$ and to each of their adjoints.
    The circuit uses $\mathcal{O}(\log(N_e + N_g))$ two-qubit gates.
\end{corollary}
\begin{proof}
    Trivially by inversion of error relation.
\end{proof}
Note that the error budgeting can be relaxed to non-uniformity. 
However in the interest of notational brevity we enforce slighter tighter error budgets by constant factors.

\section{Energy Estimation} \label{apx:measurements}

For general systems where the energy is proportional to the squared Euclidean norm and where the energy metric is diagonal in the computational basis, we prove the following result.
This directly extends prior work \cite{babbush2023exponential} to dissipative systems:

\begin{lemma}[Subspace Energy Estimation (Target Error); see \cite{babbush2023exponential}] \label{th:energy_estimation}
    Let $\varepsilon \in (0,1)$ be a target error and $\delta \in (0,1)$ a failure probability.
    Then \Cref{prob:subspace_energy} can be solved with probability at least $1-\delta$ using
    \begin{equation}
        {Q_{\Xi} = \mathcal{O}\bigg(\frac{\log(1/\delta)}{\varepsilon}\bigg)}
    \end{equation}
    queries to the (controlled) marking oracle $O_{\mathcal{S}}$ and to the $\varepsilon/(6Q_{\Xi})$-accurate state preparation oracles $O_{\boldsymbol{\psi}_0}$, $O_{\boldsymbol{\chi}_\tau}$, and $O_{\boldsymbol{\chi}_x}$, and
    \begin{equation}
        {Q_{\mathbf{G}''} = \mathcal{O}\Big(Q_{\Xi} \big(T(\alpha_\mathbf{L}\log(Q_{\Xi}/\varepsilon) + \alpha_\mathbf{H}) + \log(Q_{\Xi}/\varepsilon) \big) \Big)}
    \end{equation}
    queries to the (controlled) index oracles $O_r^{(\mathcal{M})},O_c^{(\mathcal{M})}$ for $\mathcal{M} \in \{\mathbf{S}_A, \mathbf{S}_B\}$ and to the (controlled) $\varepsilon/(6Q_{\mathbf{G}''})$-accurate phase oracles $O_{\mathbf{M}^{-1/2}}, O_{\mathbf{K}^{1/2}}, O_{\boldsymbol{\eta}_0}, O_{\boldsymbol{\Lambda}}$.
    The algorithm requires $\mathcal{O}\big(\log(N \alpha_{\mathbf{L}} T Q_{\Xi}/\varepsilon)\big)$ ancilla qubits and uses
    \begin{equation}
        \mathcal{O}\Big(\big(Q_{\mathbf{G}''} + Q_{\Xi}\log^{5/2}(Q_{\Xi}/\varepsilon)\big)\log(N \alpha_{\mathbf{L}} T Q_{\Xi}/\varepsilon)\Big)
    \end{equation}
    additional two-qubit gates.
\end{lemma}

\begin{proof}
    We allocate a statistical budget of $\varepsilon/2$ to \ac{QAE} and a systematic budget of $\varepsilon/2$ to coherent implementation errors.
    The systematic budget is split evenly over the three error sources by choosing, up to constants fixed at the end of the proof, the \ac{LCHS} algorithmic and time-discretization error $\bar{\varepsilon} = \varepsilon/(6Q_{\Xi})$, the state preparation accuracy $\hat{\varepsilon} = \varepsilon/(6Q_{\Xi})$, and the phase oracle accuracy $\varepsilon' = \varepsilon/(6Q_{\mathbf{G}''})$.

    Using the state preparation oracles $O_{\boldsymbol{\psi}_0}$, $O_{\boldsymbol{\chi}_\tau}$, and $O_{\boldsymbol{\chi}_x}$, we construct the unnormalized \ac{LCHS} state preparation unitary \cite{low_optimal_2025}.
    It prepares the global \ac{LCHS} state $\ket{\Phi}$, avoiding the renormalization cost $Q_{\boldsymbol{\psi}}$.
    Projecting $\ket{\Phi}$ onto the all-zero success ancilla state $\ket{0}_a$ extracts the unnormalized physical state $\boldsymbol{\psi}(T)$ encoded in the system register $s$
    \begin{equation}
        (\bra{0}_a \otimes \mathbf{I}_s) \ket{\Phi} = \frac{1}{\alpha_{\mathbf{G}}(\|\boldsymbol{\psi}(0)\| + \|\boldsymbol{\chi}\|_{L^1})} \boldsymbol{\psi}(T),
    \end{equation}
    where $\alpha_{\mathbf{G}} = \mathcal{O}(1)$ is the normalization constant of the \ac{LCHS} \ac{BE}.

    To estimate the energy, we define the \ac{LCHS} success projector $\boldsymbol{\Pi}_a = \ketbra{0}{0}_a \otimes \mathbf{I}_s$ and the diagonal system subspace projector $\boldsymbol{\Pi}_s = \mathbf{I}_a \otimes \mathbf{P}_{\mathcal{S}}$.
    Because they act on disjoint registers, they strictly commute, and we define the combined projector as their product $\boldsymbol{\Pi} = \boldsymbol{\Pi}_a \boldsymbol{\Pi}_s = \ketbra{0}{0}_a \otimes \mathbf{P}_{\mathcal{S}}$.
    \ac{QAE} \cite{brassard2000quantum, knill2007optimal} requires the Grover reflection operator over the target subspace, given by $\mathbf{I}_{a \otimes s} - 2\boldsymbol{\Pi}$.
    Because the underlying projectors commute, this reflection is exactly implemented by applying the exact marking oracle $O_{\mathcal{S}} = \mathbf{I}_s - 2\mathbf{P}_{\mathcal{S}}$ strictly controlled on the ancilla register $a$ being in the state $\ket{0}_a$.

    The probability $p$ of measuring $\ket{\Phi}$ in this target subspace is the expectation value of the combined projector.
    Substituting the subsystem energy $E_{\mathcal{S}}(T) = \frac{1}{2}\|\mathbf{P}_{\mathcal{S}}\boldsymbol{\psi}(T)\|^2$ and the maximum theoretical energy $E_{\operatorname{tot}} = \frac{1}{2}(\|\boldsymbol{\psi}(0)\| + \|\boldsymbol{\chi}\|_{L^1})^2$ yields the exact relation
    \begin{equation}
        p = \braket{\Phi|\boldsymbol{\Pi}|\Phi} = \frac{E_{\mathcal{S}}(T)}{\alpha_{\mathbf{G}}^2 E_{\operatorname{tot}}}.
    \end{equation}
    To estimate the ratio $E_{\mathcal{S}}(T)/E_{\operatorname{tot}}$ within the statistical budget $\varepsilon/2$, we must estimate $p$ to additive error $\tilde{\varepsilon} = \varepsilon/(2\alpha_{\mathbf{G}}^2)$.
    Because our system is passive ($\mathbf{L} \succeq 0$), the energy in any subsystem can never exceed the total injected energy, enforcing $E_{\mathcal{S}}(T) \leq E_{\operatorname{tot}}$ and hence $\sqrt{p} \leq 1/\alpha_{\mathbf{G}}$.
    Substituting into the \ac{QAE} complexity $\mathcal{O}\big(\tfrac{\sqrt{p}}{\tilde{\varepsilon}}\log(1/\delta)\big)$ bounds the number of iterations by
    \begin{equation}
        Q_{\Xi} = \mathcal{O}\left( \frac{\alpha_{\mathbf{G}}\log(1/\delta)}{\varepsilon} \right) = \mathcal{O}\left(\frac{\log(1/\delta)}{\varepsilon}\right).
    \end{equation}
    Each iteration applies the reflection once and the \ac{LCHS} state preparation unitary and its inverse, which requires $\mathcal{O}(1)$ queries to the state preparation oracles, giving the stated query counts for $O_{\mathcal{S}}$, $O_{\boldsymbol{\psi}_0}$, $O_{\boldsymbol{\chi}_\tau}$, and $O_{\boldsymbol{\chi}_x}$.

    A single application of the \ac{LCHS} state preparation with internal accuracy $\bar{\varepsilon}$ requires $Q_{\mathbf{G}} = \mathcal{O}\big(T(\alpha_\mathbf{L}\log(1/\bar{\varepsilon}) + \alpha_\mathbf{H}) + \log(1/\bar{\varepsilon})\big)$ queries to the (controlled) index and phase oracles by \Cref{lem:hamiltonian_encoding,lem:dissipator_encoding,th:_lchs_op_applied}.
    Substituting $\bar{\varepsilon} = \varepsilon/(6Q_{\Xi})$ yields $\log(1/\bar{\varepsilon}) = \mathcal{O}(\log(Q_{\Xi}/\varepsilon))$, and the total number of oracle queries is the product $Q_{\mathbf{G}''} = Q_{\Xi} Q_{\mathbf{G}}$ as stated.
    Each iteration accumulates the coherent error $\bar{\varepsilon} + \hat{\varepsilon} + \mathcal{O}(Q_{\mathbf{G}}\varepsilon')$.
    Over $Q_{\Xi}$ iterations, the total systematic deviation is $\mathcal{O}\big(Q_{\Xi}(\bar{\varepsilon} + \hat{\varepsilon}) + Q_{\mathbf{G}''}\varepsilon'\big)$, which the chosen accuracies bound by $\varepsilon/2$ for suitable constants.
    Combined with the statistical budget, the total estimation error is at most $\varepsilon$.
    Finally, a single application of the \ac{LCHS} preparation requires $\mathcal{O}(\log N + V + \log(T/\bar{\varepsilon}))$ ancilla qubits, where $V = \mathcal{O}(\log(\alpha_{\mathbf{L}}T + \log(1/\bar{\varepsilon})))$, which are reused across all \ac{QAE} iterations, and $\mathcal{O}\big(Q_{\mathbf{G}}(\log N + V + \log(T/\bar{\varepsilon})) + V\log^{5/2}(1/\bar{\varepsilon})\big)$ two-qubit gates.
    Substituting $\bar{\varepsilon} = \varepsilon/(6Q_{\Xi})$ and absorbing all logarithmic factors into the single term $\log(N \alpha_{\mathbf{L}} T Q_{\Xi}/\varepsilon)$, which is valid up to constant factors in the nontrivial regime $\alpha_{\mathbf{L}}T \geq 1$ (and otherwise with the arguments of the logarithms offset by constants), yields the stated ancilla count and, after multiplying the per-iteration gate cost by $Q_{\Xi}$, the stated gate count.
\end{proof}

\section{Time-Dependent Block Encodings} \label{apx:td_block_encodings}
Here we provide proofs for the \ac{BE} of $\mathbf{H}(t)$ and $\mathbf{L}(t)$ for time-varying materials.
We use the access model specified in \Cref{s:access_td}.
Consistent with the target-error convention of the main text, each proof splits the single target error $\varepsilon$ evenly between the per-step encoding error, which fixes the phase oracle accuracies, and the time-discretization error, which fixes the grid step size $\Delta t$; the factor-$2$ overheads are absorbed into the $\mathcal{O}(\cdot)$ notation.

\begin{proof}[Proof of \Cref{lem:td_hamiltonian_encoding}]
    In the time-independent case (\Cref{s:be_hamiltonian}), the weighted topology matrix is constructed via the product $\mathcal{U}_{\hat{\mathbf{T}}} = \mathcal{U}_{\mathbf{K}^{1/2}} \cdot \mathcal{U}_{\mathbf{T}} \cdot \mathcal{U}_{\mathbf{M}^{-1/2}}$. 
    For the time-dependent case, the time-dependent phase oracles natively produce the multiplexed block encodings $\mathcal{U}_{\mathbf{R}} = \sum_{j} \ketbra{j}{j} \otimes \mathcal{U}_{\mathbf{R}(j\Delta t)}$. 
    Because the topology unitary $\mathcal{U}_{\mathbf{T}}$ encodes a static, time-invariant graph, its natural extension acts as the identity on the clock register $\mathbf{I} \otimes \mathcal{U}_{\mathbf{T}}$.
    
    Evaluating the product of these block encodings over the joint space yields
    \begin{equation}
    \begin{aligned}
        \mathcal{U}_{\hat{\mathbf{T}}(t)} &= \mathcal{U}_{\mathbf{K}^{1/2}(t)} \cdot (\mathbf{I} \otimes \mathcal{U}_{\mathbf{T}}) \cdot \mathcal{U}_{\mathbf{M}^{-1/2}(t)} \\
        &= \left( \sum_{j} \ketbra{j}{j} \otimes \mathcal{U}_{\mathbf{K}^{1/2}(j\Delta t)} \right) (\mathbf{I} \otimes \mathcal{U}_{\mathbf{T}}) \left( \sum_{l} \ketbra{l}{l} \otimes \mathcal{U}_{\mathbf{M}^{-1/2}(l\Delta t)} \right) \\
        &= \sum_{j} \ketbra{j}{j} \otimes \left( \mathcal{U}_{\mathbf{K}^{1/2}(j\Delta t)} \mathcal{U}_{\mathbf{T}} \mathcal{U}_{\mathbf{M}^{-1/2}(j\Delta t)} \right).
    \end{aligned}
    \end{equation}
    Because the clock register $\ket{j}$ acts on a disjoint Hilbert space from the \ac{LCU} ancillae, the target block-encoding projection $\bra{0}^{\otimes a}$ commutes with the temporal projectors $\ketbra{j}{j}$. 
    Thus, the matrix multiplication distributes perfectly, yielding the multiplexed weighted topology matrix $\sum_j \ketbra{j}{j} \otimes \mathrm{BE}[\hat{\mathbf{T}}(j\Delta t)/\alpha_{\hat{\mathbf{T}}}]$.
    
    The subsequent 1-ancilla Hermitian dilation $\check{\mathbf{H}}(t) = (\mathbf{Y} \otimes \mathbf{I}) \operatorname{diag}(\hat{\mathbf{T}}(t), \hat{\mathbf{T}}^\dagger(t))$ introduces an external control qubit synchronous to the time-independent case. 
    This static control operation trivially commutes with the internal clock multiplexing.
    The overall normalization $\alpha_{\mathbf{H}}$ is governed by the product of the individual block-encoding maximums across all time steps. 
    Additive error sub-multiplicativity holds block-wise, so the per-step encoding error is uniformly bounded by the phase oracle errors scaled by $\alpha_{\mathbf{H}}$.
    Choosing the phase oracle accuracy as $\mathcal{O}(\varepsilon/\alpha_{\mathbf{H}})$, in exact analogy to the target-error corollary of the time-independent case (\Cref{s:be_hamiltonian}), bounds the per-step encoding error by $\varepsilon$.

    Finally, for a time-dependent generator $G(t)$, the grid stepsize $\Delta t$ required to bound the time-discretization error to $\varepsilon$ over simulation duration $T$ is governed by the truncated Dyson series discretization requirements \cite{low2018hamiltonian, low_hamiltonian_2019}, satisfying $\Delta t = \mathcal{O}\left(\frac{\varepsilon}{T \big( \max_t \|\frac{\mathrm{d}G(t)}{\mathrm{d} t}\| + \max_t \|G(t)\|^2 \big)}\right)$. 
    By applying the product rule to the weighted topology matrix $\hat{\mathbf{T}}(t) = \mathbf{K}^{1/2} \mathbf{T} \mathbf{M}^{-1/2}$, we obtain
    \begin{equation}
        \frac{\mathrm{d}\hat{\mathbf{T}}}{\mathrm{d}t} = \frac{\mathrm{d}\mathbf{K}^{1/2}}{\mathrm{d}t} \mathbf{T} \mathbf{M}^{-1/2} + \mathbf{K}^{1/2} \mathbf{T} \frac{\mathrm{d}\mathbf{M}^{-1/2}}{\mathrm{d}t}.
    \end{equation}
    Recall that the metric drift terms are defined as $\mathbf{K}^{(1)} = \frac{\mathrm{d} \ln \mathbf{K}}{\mathrm{d}t}$, which implies $\frac{\mathrm{d}\mathbf{K}^{1/2}}{\mathrm{d}t} = \frac{1}{2}\mathbf{K}^{(1)} \mathbf{K}^{1/2}$. 
    Similarly, $\mathbf{M}^{(1)} = \frac{\mathrm{d} \ln \mathbf{M}}{\mathrm{d}t}$, giving $\frac{\mathrm{d}\mathbf{M}^{-1/2}}{\mathrm{d}t} = - \frac{1}{2}\mathbf{M}^{-1/2} \mathbf{M}^{(1)}$. 
    Substituting these identities and applying the triangle inequality alongside the sub-multiplicativity of the spectral norm yields
    \begin{equation}
    \begin{aligned}
        \Big\|\frac{\mathrm{d}\hat{\mathbf{T}}}{\mathrm{d}t}\Big\| &\le \frac{1}{2}\|\mathbf{K}^{(1)}\| \|\mathbf{K}^{1/2} \mathbf{T} \mathbf{M}^{-1/2}\| + \frac{1}{2}\|\mathbf{M}^{(1)}\| \|\mathbf{K}^{1/2} \mathbf{T} \mathbf{M}^{-1/2}\| \\
        &\le \frac{1}{2}\big(|k^{(1)}|_{\max} + |m^{(1)}|_{\max}\big) \alpha_{\mathbf{H}}.
    \end{aligned}
    \end{equation}
    Since $\mathbf{H}(t)$ is the Hermitian dilation of $\hat{\mathbf{T}}(t)$, the derivative is bounded as $\|\frac{\mathrm{d}\mathbf{H}(t)}{\mathrm{d} t}\| \le \Gamma_{\mathbf{H}} \alpha_{\mathbf{H}}$, while its norm is bounded as $\|\mathbf{H}(t)\| \le \alpha_{\mathbf{H}}$. 
    Substituting $G(t) = \mathbf{H}(t)$ into the discretization step requirement yields $\Delta t = \mathcal{O}\left(\frac{\varepsilon}{T \alpha_{\mathbf{H}} (\Gamma_{\mathbf{H}} + \alpha_{\mathbf{H}})}\right)$, which dictates that the clock register size scales as $n_t = \mathcal{O}\left(\log\left(\frac{T \alpha_{\mathbf{H}} (\Gamma_{\mathbf{H}} + \alpha_{\mathbf{H}})}{\varepsilon}\right)\right)$.
    Combining $n_t$ with the $\mathcal{O}(\log N)$ ancillae of the time-independent case gives the stated ancilla count $a_{\mathbf{H}}$, and the remaining costs are inherited unchanged.
\end{proof}

\begin{proof}[Proof of \Cref{lem:td_dissipator_encoding}]
    The time-dependent dissipator $\mathbf{L}(t)$ is formed by the block-diagonal composition of $\hat{\boldsymbol{\eta}}_{\mathbf{M}}(t) = \mathbf{M}^{-1/2}(t) \mathbf{T}_0^\dagger \tilde{\boldsymbol{\eta}}_0(t) \mathbf{T}_0 \mathbf{M}^{-1/2}(t)$ and $\tilde{\boldsymbol{\Lambda}}(t)$ (\Cref{s:be_dissipator}).
    As established in \Cref{lem:td_hamiltonian_encoding}, accessing the real diagonal matrices through their time-dependent phase oracles natively produces the multiplexed block encodings, for instance $\mathcal{U}_{\mathbf{M}^{-1/2}(t)} = \sum_{j} \ketbra{j}{j} \otimes \mathcal{U}_{\mathbf{M}^{-1/2}(j\Delta t)}$. 
    Because the nested applications of the static topology unitary $\mathcal{U}_{\mathbf{T}_0}$ act as the identity on the clock register, they commute with the temporal projectors $\ketbra{j}{j}$.
    Evaluating the product for the mass dissipation block over the joint space yields the multiplexed encoding $\mathcal{U}_{\hat{\boldsymbol{\eta}}_{\mathbf{M}}(t)} = \sum_j \ketbra{j}{j} \otimes \mathrm{BE}[\hat{\boldsymbol{\eta}}_{\mathbf{M}}(j\Delta t)/\alpha_{\hat{\boldsymbol{\eta}}_{\mathbf{M}}}]$.
    Likewise, the time-dependent phase oracle $O_{\tilde{\boldsymbol{\Lambda}}}$ directly implements the multiplexed encoding of the spring dissipation body $\mathcal{U}_{\tilde{\boldsymbol{\Lambda}}(t)}$.
    
    The final \ac{LCU} operation that embeds these two sub-blocks onto the diagonal of $\check{\mathbf{L}}(t)$ introduces an external control qubit synchronous to the time-independent case. 
    As this operation acts on a disjoint subspace, it commutes with the internal clock multiplexing. 
    Thus, the block-wise sum is preserved across all time steps simultaneously, generating $\sum_j \ketbra{j}{j} \otimes \mathrm{BE}[\mathbf{L}(j\Delta t)/\alpha_{\mathbf{L}}]$.
    The overall normalization $\alpha_{\mathbf{L}}$ is governed by the global supremum of the individual block-encoding maximums across all time steps $t \in [0,T]$.
    Additive error sub-multiplicativity holds block-wise, so choosing the phase oracle accuracy as $\mathcal{O}(\varepsilon/\alpha_{\mathbf{L}})$, in exact analogy to the target-error corollary of the time-independent case (\Cref{s:be_dissipator}), bounds the per-step encoding error by $\varepsilon$.

    Finally, inside the LCHS simulation routine, the simulated generator is $G(t) = k \mathbf{L}(t)$ with $|k| \le R = \mathcal{O}(\log(1/\varepsilon))$. By the truncated Dyson series discretization requirements \cite{low2018hamiltonian, low_optimal_2025}, the required grid stepsize is $\Delta t = \mathcal{O}\left(\frac{\varepsilon}{T \big( \max_t \|\frac{\mathrm{d}G(t)}{\mathrm{d} t}\| + \max_t \|G(t)\|^2 \big)}\right)$.
    By applying the product rule to the mass dissipation block, we obtain
    \begin{equation}
    \begin{aligned}
        \frac{\mathrm{d}\hat{\boldsymbol{\eta}}_{\mathbf{M}}}{\mathrm{d}t} &= \frac{\mathrm{d}\mathbf{M}^{-1/2}}{\mathrm{d}t} \mathbf{T}_0^\dagger \tilde{\boldsymbol{\eta}}_0 \mathbf{T}_0 \mathbf{M}^{-1/2} +\mathbf{M}^{-1/2} \mathbf{T}_0^\dagger \frac{\mathrm{d}\tilde{\boldsymbol{\eta}}_0}{\mathrm{d}t} \mathbf{T}_0 \mathbf{M}^{-1/2} +\mathbf{M}^{-1/2} \mathbf{T}_0^\dagger \tilde{\boldsymbol{\eta}}_0 \mathbf{T}_0 \frac{\mathrm{d}\mathbf{M}^{-1/2}}{\mathrm{d}t}.
    \end{aligned}
    \end{equation}
    Recall that $\frac{\mathrm{d}\mathbf{M}^{-1/2}}{\mathrm{d}t} = - \frac{1}{2} \mathbf{M}^{-1/2} \mathbf{M}^{(1)}$ and the relative rates $\tilde{\boldsymbol{\eta}}_0^{(1)} = \frac{\mathrm{d}\ln\tilde{\boldsymbol{\eta}}_0}{\mathrm{d}t}$ and $\tilde{\boldsymbol{\Lambda}}^{(1)} = \frac{\mathrm{d}\ln\tilde{\boldsymbol{\Lambda}}}{\mathrm{d}t}$. 
    Substituting $\frac{\mathrm{d}\tilde{\boldsymbol{\eta}}_0}{\mathrm{d}t} = \tilde{\boldsymbol{\eta}}_0^{(1)}\tilde{\boldsymbol{\eta}}_0$, and noting that $\tilde{\boldsymbol{\eta}}_0^{(1)}$ is a diagonal matrix bounded by its maximum entry, we can factor its norm out of the positive semi-definite product.
    Applying the triangle inequality alongside the sub-multiplicativity of the spectral norm yields
    \begin{equation}
    \begin{aligned}
        \Big\|\frac{\mathrm{d}\hat{\boldsymbol{\eta}}_{\mathbf{M}}}{\mathrm{d}t}\Big\| &\le 2 \Big\|\frac{1}{2}\mathbf{M}^{(1)}\Big\| \|\hat{\boldsymbol{\eta}}_{\mathbf{M}}\| + \|\tilde{\boldsymbol{\eta}}_0^{(1)}\| \|\mathbf{M}^{-1/2} \mathbf{T}_0^\dagger \tilde{\boldsymbol{\eta}}_0 \mathbf{T}_0 \mathbf{M}^{-1/2}\| \\
        &\le \big(|m^{(1)}|_{\max} + |\tilde{\eta}_0^{(1)}|_{\max}\big) \alpha_{\hat{\boldsymbol{\eta}}_{\mathbf{M}}}.
    \end{aligned}
    \end{equation}
    Similarly, the derivative of the spring dissipation block is bounded by $\|\frac{\mathrm{d}\tilde{\boldsymbol{\Lambda}}}{\mathrm{d}t}\| \le \|\tilde{\boldsymbol{\Lambda}}^{(1)}\| \alpha_{\tilde{\boldsymbol{\Lambda}}} \le |\tilde{\lambda}^{(1)}|_{\max} \alpha_{\tilde{\boldsymbol{\Lambda}}}$.
    Since $\mathbf{L}$ is block-diagonal, its derivative is bounded by $\|\frac{\mathrm{d}\mathbf{L}(t)}{\mathrm{d} t}\| \le \alpha_{\mathbf{L}} \Gamma_{\mathbf{L}}$, where $\Gamma_{\mathbf{L}} = \max_{t \in [0,T]} \big(|m^{(1)}(t)|_{\max} + |\tilde{\eta}_0^{(1)}(t)|_{\max} + |\tilde{\lambda}^{(1)}(t)|_{\max} \big)$.
    Consequently, for $G(t) = k \mathbf{L}(t)$, we have $\|\frac{\mathrm{d}G(t)}{\mathrm{d} t}\| \le R \alpha_{\mathbf{L}} \Gamma_{\mathbf{L}}$ and $\|G(t)\| \le R \alpha_{\mathbf{L}}$. Substituting $R = \mathcal{O}(\log(1/\varepsilon))$ into the discretization step requirement yields $\Delta t = \mathcal{O}\left(\frac{\varepsilon}{T \alpha_{\mathbf{L}}\log(1/\varepsilon) (\Gamma_{\mathbf{L}} + \alpha_{\mathbf{L}}\log(1/\varepsilon))}\right)$, which dictates that the clock register size scales as $n_t = \mathcal{O}\left(\log\left(\frac{T \alpha_{\mathbf{L}}\log(1/\varepsilon) (\Gamma_{\mathbf{L}} + \alpha_{\mathbf{L}}\log(1/\varepsilon))}{\varepsilon}\right)\right)$.
    Because $\log\log(1/\varepsilon) = \mathcal{O}(\log(1/\varepsilon))$, this simplifies to $n_t = \mathcal{O}\left(\log\left(\frac{T \alpha_{\mathbf{L}} (\Gamma_{\mathbf{L}} + \alpha_{\mathbf{L}})}{\varepsilon}\right)\right)$ up to constant factors.
    Combining $n_t$ with the $\mathcal{O}(\log N)$ ancillae of the time-independent case gives the stated ancilla count $a_{\mathbf{L}}$, and the remaining costs are inherited unchanged.
\end{proof}

\section{Information Locality Bounds for Differential Equations}\label{app:LRB}

Our aim in this section is to build an understanding of when dissipative classical dynamics of a local observable cannot be simulated exponentially faster using a quantum computer than a classical computer by providing efficient classical simulation algorithms in such cases.  We will also provide a discussion of similar conditions for the simulation of a low-dimensional subspace of the solution.
We provide analogues of Lieb-Robinson bounds for generic linear time-invariant inhomogeneous differential equations and thereby show that we can simulate a small subspace of a system, such as our coupled oscillators, without needing to simulate the entire vector space because of information locality arguments.  
This is challenging for three reasons.  
First, commutators are no longer relevant for dynamics for generic differential equations.  
This means that other notions of independence of observables will be needed in this context.
Second, inhomogeneities can create the illusion of signaling through correlations between their time dependent values and this effect needs to be addressed in the analysis.
Third, unlike the Lieb-Robinson bounds, the Hilbert space under consideration may not necessarily have a natural sub-system decomposition.
For example, for the case of the damped non-Markovian oscillators considered above, each oscillator is effectively a one-dimensional subsystem.
This makes the analysis techniques used in Lieb-Robinson analysis not directly applicable. 

We consider inhomogeneous linear differential equations of the following form:
\begin{align}
    \frac{\mathrm{d} {\boldsymbol{\psi}(t)}}{\mathrm{d}t} = -\mathbf{C} {\boldsymbol{\psi}(t)} + {\boldsymbol{\chi}(t)},
\end{align}
where $\mathbf{C} \in \mathbb{C}^{N \times N}$.
The solution to this is
\begin{equation}
{\boldsymbol{\psi}(t)} = e^{-\mathbf{C}t} {\boldsymbol{\psi}(0)}+e^{-\mathbf{C}t} \int_0^t e^{\mathbf{C}\tau} {\boldsymbol{\chi}(\tau)} \mathrm{d}\tau:=e^{-\mathbf{C}t} {\boldsymbol{\psi}(0)} + e^{-\mathbf{C}t} {\boldsymbol{\Upsilon}(t)}.\label{eq:diffDef}
\end{equation}
Here we further have that
\begin{equation}
    \frac{\mathrm{d} }{\mathrm{d}  t}{\boldsymbol{\Upsilon}(t)} ~= e^{\mathbf{C}t} {\boldsymbol{\chi}(t)},
\end{equation}
which immediately allows one to verify by substitution that the above solution is valid.


Our aim is to now examine how a classical simulation could be performed for this differential equation under the assumption that the generator $-\mathbf{C}$ has a non-positive log-norm, i.e. the system is strictly dissipative or conservative.
For simplicity, we evaluate both the inhomogeneity and the function over a uniformly spaced grid; however, other choices may be made given prior knowledge of the distribution.
We provide this result in the following technical lemma.

\begin{lemma}[Classical Simulation]\label{lem:classicalSim}
    Let $-\mathbf{C}\in \mathbb{C}^{N\times N}$ be an $s$-sparse matrix with a non-positive log-norm and spectral norm $\|\mathbf{C}\|\le \alpha$, and let $\epsilon\in (0,1]$ be a desired error tolerance.
    Assume we are provided a list of size $sN$ of the matrix elements of $\mathbf{C}$ and 
    {that $\boldsymbol{\chi}:[0,T] \rightarrow \mathbb{C}^{N}$ is an analytic function that obeys $\sup_{\tau \in [0,T]}\|\partial_{\tau}^q \boldsymbol{\chi}(\tau)\| \le A_{\boldsymbol{\chi}} \omega_{\boldsymbol{\chi}}^q q!$}
    for an amplitude bound $A_{\boldsymbol{\chi}} \ge 0$ and frequency bound $\omega_{\boldsymbol{\chi}} \ge 0$.
    Assume we are provided a subroutine that computes $\boldsymbol{\chi}$ and its first $K_{\boldsymbol{\chi}}-1$ derivatives at any given time $t$ to the working precision using $\mathsf{T}_{\boldsymbol{\chi}}$ arithmetic operations, where $K_{\boldsymbol{\chi}} = \left\lceil \log_2\left( \max\left(2, \frac{4 e A_{\boldsymbol{\chi}} T}{\epsilon} \right) \right) \right\rceil$.
    Then a solution $\tilde{\boldsymbol{\psi}}(T) \in \mathbb{C}^{N}$ to the inhomogeneous differential equation $\frac{\mathrm{d}\boldsymbol{\psi}(t)}{\mathrm{d}t} = -\mathbf{C}\boldsymbol{\psi}(t) + \boldsymbol{\chi}(t)$ with initial condition $\boldsymbol{\psi}(0)$ can be computed on a classical computer such that $\|\tilde{\boldsymbol{\psi}}(T) - \boldsymbol{\psi}(T)\|\le \epsilon$ using a number of bit operations that scales as
    $$
    \tilde{\mathcal{O}}\left( (Ns + \mathsf{T}_{\boldsymbol{\chi}}) (\alpha+\omega_{\boldsymbol{\chi}}) T \log^2\left(\frac{\|\boldsymbol{\psi}(0)\| + A_{\boldsymbol{\chi}} T}{\epsilon}\right) \right).
    $$
\end{lemma}

\begin{proof}
Our classical algorithm steps forward in time by partitioning the integration interval $[0, T]$ into $N_t = \max(1, \lceil (\alpha + 2\omega_{\boldsymbol{\chi}}) T \rceil)$ sub-intervals of length $\Delta t = T / N_t$. 
By construction, this ensures that $\alpha \Delta t \le 1$ and $\omega_{\boldsymbol{\chi}} \Delta t \le 1/2$.

The exact solution to the inhomogeneous differential equation over a single time step from $t_i$ to $t_{i+1} = t_i + \Delta t$ is given by Duhamel's principle
\begin{equation} \label{eq:duhamel_step}
    \boldsymbol{\psi}(t_{i+1}) = e^{-\mathbf{C}\Delta t} \boldsymbol{\psi}(t_i) + \int_0^{\Delta t} e^{-\mathbf{C}(\Delta t - \tau)} \boldsymbol{\chi}(t_i + \tau) \mathrm{d}\tau.
\end{equation}
Because $-\mathbf{C}$ has a non-positive log-norm, the homogeneous propagator is bounded by $\|e^{-\mathbf{C}\tau}\| \le 1$ for all $\tau \ge 0$. 
This implies the state norm grows at most linearly 
due to the source $\|\boldsymbol{\psi}(t_i)\| \le \|\boldsymbol{\psi}(0)\| + T \sup_{\tau \in [0, T]}\|\boldsymbol{\chi}(\tau)\| \le \|\boldsymbol{\psi}(0)\| + A_{\boldsymbol{\chi}} T$.

We expand the analytic source function $\boldsymbol{\chi}(t_i + \tau)$ in a Taylor series up to degree $K_{\boldsymbol{\chi}}-1$. 
The remainder is bounded by 
\begin{equation}
    \epsilon_{\boldsymbol{\chi}}(\tau):=\sup_{h \in[0, \tau]} \|\boldsymbol{\chi}^{(K_{\boldsymbol{\chi}})}(t_i + h)\| \frac{\tau^{K_{\boldsymbol{\chi}}}}{K_{\boldsymbol{\chi}}!} \le A_{\boldsymbol{\chi}} \omega_{\boldsymbol{\chi}}^{K_{\boldsymbol{\chi}}} \tau^{K_{\boldsymbol{\chi}}}.
\end{equation}

Substituting the remainder of the Taylor series into the integral in \Cref{eq:duhamel_step}, taking its norm, and applying the triangle inequality, the local source truncation error at a single time step is bounded by
\begin{equation}
    \int_0^{\Delta t} A_{\boldsymbol{\chi}} \omega_{\boldsymbol{\chi}}^{K_{\boldsymbol{\chi}}} \tau^{K_{\boldsymbol{\chi}}} \mathrm{d}\tau = \frac{A_{\boldsymbol{\chi}} \omega_{\boldsymbol{\chi}}^{K_{\boldsymbol{\chi}}} (\Delta t)^{K_{\boldsymbol{\chi}}+1}}{K_{\boldsymbol{\chi}}+1}.
\end{equation}
If we then define $\tilde{\boldsymbol{\chi}}$ to be the piecewise local Taylor approximation of the inhomogeneity at each time step, we have that
\begin{align}
    \epsilon_{\Upsilon}:=\left\|\int_0^t e^{-\mathbf{C}(t-\tau)}\boldsymbol{\chi}(\tau)\mathrm{d}\tau - \int_0^t e^{-\mathbf{C}(t-\tau)}\tilde{\boldsymbol{\chi}}(\tau)\mathrm{d}\tau\right\| &\le \sum_{i}\|e^{-\mathbf{C}(t-t_{i+1})}\|\int_{t_i}^{t_i +\Delta t} \|e^{-\mathbf{C}(t_{i+1}-\tau)}\| \|\boldsymbol{\chi}(\tau)-\tilde{\boldsymbol{\chi}}(\tau)\| \mathrm{d}\tau \nonumber\\
    &\le \sum_i \int_0^{\Delta t} A_{\boldsymbol{\chi}} \omega_{\boldsymbol{\chi}}^{K_{\boldsymbol{\chi}}} \tau^{K_{\boldsymbol{\chi}}} \mathrm{d}\tau \nonumber\\
    &\le 
    \frac{A_{\boldsymbol{\chi}} t}{K_{\boldsymbol{\chi}}+1} (\omega_{\boldsymbol{\chi}} \Delta t)^{K_{\boldsymbol{\chi}}} \nonumber\\&\le \frac{A_{\boldsymbol{\chi}} T}{K_{\boldsymbol{\chi}}+1} (\omega_{\boldsymbol{\chi}} \Delta t)^{K_{\boldsymbol{\chi}}}.
\end{align}

Because we chose $\Delta t$ such that $\omega_{\boldsymbol{\chi}} \Delta t \le 1/2$, we can bound $\epsilon_{\Upsilon}$  by $\frac{\epsilon}{4e}$ by selecting a truncation degree
\begin{equation}
    K_{\boldsymbol{\chi}} = \left\lceil \log_2\left( \max\left(2, \frac{4 e A_{\boldsymbol{\chi}} T}{\epsilon} \right) \right) \right\rceil \in \mathcal{O}\left( \log\left(\frac{A_{\boldsymbol{\chi}} T}{\epsilon}\right) \right).
\end{equation}

We extract the updated state at a fixed time step $t_i$ by defining the sequence of vectors $\boldsymbol{u}_k$ via the recurrence relation
\begin{align}
    \boldsymbol{u}_0 &= \tilde{\boldsymbol{\psi}}(t_i) \\
    \boldsymbol{u}_k &= -\mathbf{C} \boldsymbol{u}_{k-1} + \boldsymbol{\chi}^{(k-1)}(t_i) \quad \text{for } 1 \le k \le K_{\boldsymbol{\chi}} \\
    \boldsymbol{u}_k &= -\mathbf{C} \boldsymbol{u}_{k-1} \qquad\qquad\qquad\quad \text{for } k > K_{\boldsymbol{\chi}}.
\end{align}
For the purpose of analyzing the local truncation error of one step of the recurrence, we assume an exact initial condition, letting $\boldsymbol{u}_0 = \boldsymbol{\psi}(t_i)$.
By applying a truncated Taylor series of degree $K_{\boldsymbol{\psi}}$, the updated state evaluates to
\begin{equation} \label{eq:beta_sum}
    \tilde{\boldsymbol{\psi}}(t_{i+1}) = \sum_{k=0}^{K_{\boldsymbol{\psi}}} \frac{(\Delta t)^k}{k!} \boldsymbol{u}_k.
\end{equation}

Let $\tilde{\mathbf{G}} = \sum_{k=0}^{K_{\boldsymbol{\psi}}} \frac{(-\mathbf{C}\Delta t)^k}{k!}$ be the numerical propagation operator. 
To guarantee a global discretization error $\le \epsilon/2$ without exponential amplification of local errors, we enforce an operator distance $\|e^{-\mathbf{C}\Delta t} - \tilde{\mathbf{G}}\| \le \sum_{k=K_{\boldsymbol{\psi}}+1}^\infty \frac{(\|\mathbf{C}\|\Delta t)^k}{k!} \le \frac{1}{N_t}$. 
Since $-\mathbf{C}$ has a non-positive log-norm this bounds the numerical operator norm by $\|\tilde{\mathbf{G}}\| \le 1 + \frac{1}{N_t}$, capping the maximum global amplification over $N_t$ steps to $(1 + 1/N_t)^{N_t} \le e$. 
Therefore, to guarantee a global discretization error $\le \epsilon/2$, it suffices to bound the local series truncation error as $\|\tilde{\boldsymbol{\psi}}(t_{i+1}) - \hat{\boldsymbol{\psi}}(t_{i+1})\| \le \epsilon / (2 e N_t)$, where $\hat{\boldsymbol{\psi}}(t_{i+1})$ is the exact solution, starting from $\hat{\boldsymbol{\psi}}(t_i) = \boldsymbol{\psi}(t_i)$, of the ODE driven by the truncated source $\tilde{\boldsymbol{\chi}}$. 
Expanding the recurrence analytically, this local discretization error $\epsilon_L$ decomposes into a homogeneous $\epsilon_H$ and an inhomogeneous component $\epsilon_{I}$:
\begin{align}
    \epsilon_L &\le \max_i \underbrace{\sum_{k=K_{\boldsymbol{\psi}}+1}^\infty \frac{(\alpha\Delta t)^k}{k!} \|\boldsymbol{\psi}(t_i)\|}_{\epsilon_H} + \underbrace{\sum_{k=K_{\boldsymbol{\psi}}+1}^\infty \frac{(\Delta t)^k}{k!} \sum_{q=0}^{\min(k-1, K_{\boldsymbol{\chi}}-1)} \alpha^{k-1-q} A_{\boldsymbol{\chi}} \omega_{\boldsymbol{\chi}}^q q!}_{\epsilon_{I}}.
\end{align}

For the inhomogeneous error, we simplify notation by introducing a new dummy variable $m$ such that $m = k - 1 - q$, or equivalently, $k = m + q + 1$. 
Choosing $K_{\boldsymbol{\psi}} \ge K_{\boldsymbol{\chi}} - 1$ ensures $m \ge 0$, allowing us to use the factorial bound $\frac{q!}{k!} = \frac{1}{(q+1)\cdots(q+m+1)} \le \frac{1}{m!}$. 
Thus, the inhomogeneous error is bounded by
\begin{equation}
    \epsilon_I \le A_{\boldsymbol{\chi}} \Delta t \sum_{q=0}^{K_{\boldsymbol{\chi}}-1} (\omega_{\boldsymbol{\chi}}\Delta t)^q \left( \sum_{m=K_{\boldsymbol{\psi}}-q}^\infty \frac{(\alpha\Delta t)^m}{m!} \right).
\end{equation}
Since $q \le K_{\boldsymbol{\chi}} - 1$, the inner sum over $m$ starts at minimum $m = K_{\boldsymbol{\psi}} - K_{\boldsymbol{\chi}} + 1$. 
Because $\alpha\Delta t \le 1$, the inner Taylor remainder sum is bounded by $\frac{e}{(K_{\boldsymbol{\psi}} - K_{\boldsymbol{\chi}} + 1)!}$. 
The remaining sum over $q$ is a geometric series bounded by $2$, as $\omega_{\boldsymbol{\chi}}\Delta t \le 1/2$. 
Thus, the total inhomogeneous error is bounded by
\begin{equation}
    \epsilon_{I} \le \frac{2 e A_{\boldsymbol{\chi}} \Delta t}{(K_{\boldsymbol{\psi}} - K_{\boldsymbol{\chi}} + 1)!}.
\end{equation}
To ensure this error is at most $\frac{\epsilon}{4 \|\tilde{\mathbf{G}}\|^{N_t} N_t}\ge \frac{\epsilon}{4 e N_t}$ locally, we require $(K_{\boldsymbol{\psi}} - K_{\boldsymbol{\chi}} + 1)! \ge \frac{8 e^2 A_{\boldsymbol{\chi}} T}{\epsilon}$, meaning $K_{\boldsymbol{\psi}} - K_{\boldsymbol{\chi}} \in \mathcal{O}\left(\frac{\log(A_{\boldsymbol{\chi}} T / \epsilon)}{\log\log(A_{\boldsymbol{\chi}} T / \epsilon)}\right)$.

By the Taylor remainder theorem, the homogeneous portion is similarly bounded by
\begin{equation}
    \epsilon_H \le \frac{e (\|\boldsymbol{\psi}(0)\| + A_{\boldsymbol{\chi}} T)}{(K_{\boldsymbol{\psi}}+1)!}.
\end{equation}
Setting this to be $ \epsilon_H\le \frac{\epsilon}{4 e N_t}$ alongside the operator bound $\|e^{-\mathbf{C}\Delta t} - \tilde{\mathbf{G}}\| \le \frac{1}{N_t}$ dictates 
\begin{equation}
    K_{\boldsymbol{\psi}} \in \mathcal{O}\left(\frac{\log(N_t)}{\log\log(N_t)} + \frac{\log\left(N_t (\|\boldsymbol{\psi}(0)\| + A_{\boldsymbol{\chi}} T) / \epsilon \right)}{\log\log\left(N_t (\|\boldsymbol{\psi}(0)\| + A_{\boldsymbol{\chi}} T) / \epsilon \right)}\right).
\end{equation}
Since $\log(N_t)$ is logarithmic in the simulation parameters, both errors and the operator norm are suitably suppressed by choosing a truncation order $K_{\boldsymbol{\psi}} \in \mathcal{O}\left(\log\big((\alpha + \omega_{\boldsymbol{\chi}}) T\big) + \log\left(\frac{\|\boldsymbol{\psi}(0)\| + A_{\boldsymbol{\chi}} T}{\epsilon}\right)\right)$.
Because the numerical operator amplifies prior errors by at most $\|\tilde{\mathbf{G}}\|^{N_t} \le e$, the unamplified global source error contributes at most $e \times \frac{\epsilon}{4e} = \frac{\epsilon}{4}$ to the final state, which, when combined with the $\frac{\epsilon}{2}$ global discretization error, strictly bounds the total error by $\epsilon$.

At each of the $N_t$ time steps, generating the sequence $\boldsymbol{u}_k$ up to $K_{\boldsymbol{\psi}}$ requires exactly $K_{\boldsymbol{\psi}}$ sparse matrix-vector multiplications with $\mathbf{C}$ and the evaluation of the source derivatives. 
Therefore, the computational cost to step forward in time is $\mathcal{O}(K_{\boldsymbol{\psi}} N s + \mathsf{T}_{\boldsymbol{\chi}})$.

Multiplying this step cost by the total number of time steps $N_t = \max(1, \lceil (\alpha + 2\omega_{\boldsymbol{\chi}}) T \rceil)$, the global arithmetic complexity evaluates to
\begin{equation}
    \mathsf{T}_{\text{ari}} = \tilde{\mathcal{O}}\left( (Ns + \mathsf{T}_{\boldsymbol{\chi}})(\alpha+\omega_{\boldsymbol{\chi}})T \log\left(\frac{\|\boldsymbol{\psi}(0)\| + A_{\boldsymbol{\chi}} T}{\epsilon}\right) \right),
\end{equation}
where we assume a non-trivial regime.
Bounding the accumulated round-off error to $\mathcal{O}(\epsilon)$ across the simulation requires executing the arithmetic operations with a floating-point precision of $\mathcal{O}\big(\log\big(\frac{\|\boldsymbol{\psi}(0)\| + A_{\boldsymbol{\chi}} T}{\epsilon}\big)\big)$ bits.
Since the classical cost of basic arithmetic operations scales quasi-linearly with the number of bits, the overall bit complexity incurs an additional logarithmic factor, yielding a final Boolean complexity of
\begin{equation}
    \mathsf{T}_{\text{bit}} = \tilde{\mathcal{O}}\left( (Ns + \mathsf{T}_{\boldsymbol{\chi}})(\alpha+\omega_{\boldsymbol{\chi}})T\log^2\left(\frac{\|\boldsymbol{\psi}(0)\| + A_{\boldsymbol{\chi}} T}{\epsilon}\right) \right).
\end{equation}
\end{proof}

The arithmetic cost $\mathsf{T}_{\boldsymbol{\chi}}$ to evaluate the $K_{\boldsymbol{\chi}}$ derivatives depends on the nature of the source. 
For most physical systems driven by standard analytic forms, the closed-form derivatives are known and can be evaluated in $\mathcal{O}(1)$ operations per matrix dimension, yielding a scaling of $\mathsf{T}_{\boldsymbol{\chi}} = \mathcal{O}(K_{\boldsymbol{\chi}} N)$. 
If the source is instead defined by an arbitrary analytic algorithmic procedure, its derivatives can be computed automatically using Taylor-mode Automatic Differentiation \cite{griewank2008evaluating}. 
In this algorithmic regime, the cost evaluates to $\mathsf{T}_{\boldsymbol{\chi}} = \mathcal{O}(K_{\boldsymbol{\chi}}^2 T_0)$ using standard polynomial arithmetic, where $T_0$ is the classical arithmetic cost of evaluating the un-differentiated function once. 
In either case, the overall complexity remains efficient.


\begin{figure}
    \centering

\begin{tikzpicture}[thick]

\draw (-4,-3) rectangle (4,3);
\node[anchor=north west] at (-3.8,2.8) {$\mathcal{H}$};

\draw (0,0) ellipse (2.6cm and 1.7cm);
\node at (1.6,1.0) {$\mathcal{V}$};

\draw (0,0) ellipse (1.1cm and 0.7cm);
\node at (0,0) {$\mathcal{W}$};

\node at (-3.0,-2.1) {$\mathcal{V}^{\perp}$};

\end{tikzpicture}
    \caption{Schematic diagram showing the ambient Hilbert space $\mathcal{H}$ along with the subspace of support for the observable $\mathbf{Q}$, $\mathcal{W}$ the region included in the simulation $\mathcal{V}$ and the region excluded from the simulation $\mathcal{V}^\perp$}
    \label{fig:regions}
\end{figure}
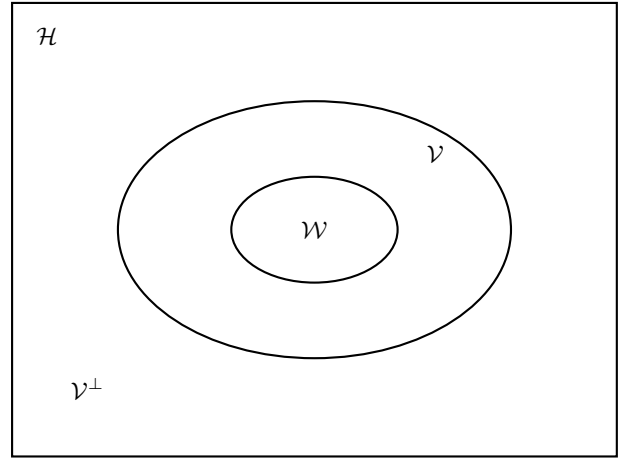

\subsection{Information Locality and Truncation Bounds}

The above result shows that we can efficiently compute the solution to a dissipative dynamical system in time that is quasi-linear in the Hilbert space dimension $N$.  
Our next step involves introducing a variant of a Lieb-Robinson bound that is appropriate for inhomogeneous differential equations which will then allow us to truncate the value of $N$ for a short-time evolution of the dynamical system or for strictly dissipative systems.

A Lieb-Robinson bound provides a bound for the commutator norm between a time-evolved observable $e^{i\mathbf{H}t}\mathbf{Q} e^{-i\mathbf{H}t}$ and another observable $\mathbf{V}$ such that the subsystems that support the observables $\mathbf{Q}$ and $\mathbf{V}$ are separated by a distance of at least 
$d$.  
It shows that the commutator between these observables, for a Hamiltonian that is sufficiently local \cite{nachtergaele2006lieb,chen2023speed}, is of the form
\begin{equation}
    \|[e^{i\mathbf{H}t}\mathbf{Q} e^{-i\mathbf{H}t},\mathbf{V}]\| \le b \|\mathbf{Q}\| \|\mathbf{V}\| e^{-c(d -v_It)}
\end{equation}
for constants $b,c$ and information propagation velocity $v_I$, which is often called the Lieb-Robinson velocity.  
As a result two commuting observables remain approximately commuting if the evolution time is sufficiently small and they are spatially separated. 
Consequently we can treat the evolution of $\mathbf{Q}$ to be independent of $\mathbf{V}$ up to exponentially small error.

The Lieb-Robinson bound has proven to be a major leap for classical simulations of quantum dynamics because it shows that we do not necessarily need to simulate the entirety of the Hilbert space to understand the dynamics of a local observable: we only need to evolve the portions of it that are important for the dynamics.  
Specifically, the dynamical equation for an observable without explicit, inherent time-dependence is
\begin{equation}
    \partial_t (e^{i\mathbf{H}t}\mathbf{Q} e^{-i\mathbf{H}t}) = i[\mathbf{H},e^{i\mathbf{H}t}\mathbf{Q} e^{-i\mathbf{H}t}],
\end{equation}
where $\mathbf{H}$ is the Hamiltonian of the system.
Thus, if $\mathbf{H}=\mathbf{H}_1+\mathbf{H}_2$ and $\sup_{\tau \in [0, t]}\|[\mathbf{H}_1,e^{i\mathbf{H}\tau}\mathbf{Q} e^{-i\mathbf{H}\tau}]\|\le \epsilon/t$, then, from the fundamental theorem of calculus applied to the interaction picture evolution of $e^{i\mathbf{H}t}\mathbf{Q} e^{-i\mathbf{H}t}$, removing $\mathbf{H}_1$ from the Hamiltonian yields an error of at most $\epsilon$.

Our aim here is to generalize the idea behind a Lieb-Robinson bound to apply to differential equations.  
This will require a number of changes to the setting and further involves a fundamentally new derivation.  
The central idea behind our approach is to consider expectation values of the form 
\begin{equation}
    \mathbf{u}^\dagger(t) \mathbf{Q}\mathbf{w}(t) = \mathbf{u}^\dagger e^{-\mathbf{C}^\dagger t}\mathbf{Q} e^{-\mathbf{C} t} \mathbf{w},
\end{equation}
where $\mathbf{u}(t)=e^{-\mathbf{C} t} \mathbf{u}$ and $\mathbf{w}(t) = e^{-\mathbf{C}t} \mathbf{w}$ are different input vectors in $\mathbb{C}^N$.   
By showing that such an expectation value remains exponentially small when the initial vectors $\mathbf{u}$ and $\mathbf{w}$ are supported on spatially separated subsystems, we can prove that information propagates at a bounded velocity.

\subsubsection{Observable Expansions and Locality}

The first step in this derivation involves finding differential equations for the analogues of expectation values of the observables under the evolution.
These expressions, as we will see, resemble the commutators that appear in the Heisenberg equation of motion but take the form of more general brackets here because the operator $\mathbf{C}$ is not necessarily anti-Hermitian.  
We define the generalized bracket operation for two operators as $(\mathbf{X},\mathbf{Y}) := \mathbf{Y}^\dagger \mathbf{X} + \mathbf{XY}$. 
For notational convenience, we denote the $k$-fold nested bracket recursively as $(\mathbf{X}, \mathbf{Y})^{(k)} := \big( (\mathbf{X}, \mathbf{Y})^{(k-1)}, \mathbf{Y} \big)$ for any integer $k \ge 1$, with the base case defined as the bare operator $(\mathbf{X}, \mathbf{Y})^{(0)} := \mathbf{X}$.

\begin{lemma}[Generalized Bracket Derivative]\label{lem:diffLem}
Let $\mathbf{Q}, \mathbf{C} \in \mathbb{C}^{N \times N}$ be linear operators acting on an $N$-dimensional Hilbert space. 
We then have that
\begin{align}
    \frac{\mathrm{d}}{\mathrm{d}t} \left( e^{-\mathbf{C}^\dagger t} \mathbf{Q} e^{-\mathbf{C}t} \right) &= -\mathbf{C}^\dagger e^{-\mathbf{C}^\dagger t} \mathbf{Q} e^{-\mathbf{C}t} - e^{-\mathbf{C}^\dagger t} \mathbf{Q} e^{-\mathbf{C}t} \mathbf{C} \nonumber \\
    &= e^{-\mathbf{C}^\dagger t} (\mathbf{Q}, -\mathbf{C}) e^{-\mathbf{C}t}.
\end{align}
\end{lemma}

\begin{proof}
    This follows directly from the product rule of differentiation and the commutativity of the matrix exponentials with their respective generators.
\end{proof}

This notation allows us to express differentiation in a form that is similar to a commutator structure, but does not have the same anti-symmetry properties as the standard commutator.
It further allows us to compactly represent the expectation value of any time-evolved observable for a generic homogeneous differential equation in an analogous manner to the commutator expansion that appears due to the Schr\"{o}dinger equation.  
Specifically, the following gives a version of a Taylor series expansion that is analogous to the commutator series $e^{\mathbf{C}t} \mathbf{Q} e^{-\mathbf{C}t} = \mathbf{Q} + t[\mathbf{C}, \mathbf{Q}] + \frac{t^2}{2!} [\mathbf{C},[\mathbf{C}, \mathbf{Q}]]+\cdots$, and reduces to it in the case where $\mathbf{C}$ is anti-Hermitian.
\begin{corollary}[Homogeneous Observable Expansion] \label{cor:homo} For any positive integer $K$ and operator $\mathbf{Q} \in \mathbb{C}^{N \times N}$,
\begin{align}
e^{-\mathbf{C}^\dagger t} \mathbf{Q} e^{-\mathbf{C} t} =& \sum_{p=0}^{K-1} (\mathbf{Q},-\mathbf{C})^{(p)} \frac{t^{p}}{{p!}} \nonumber\\
&+\int_0^t\int_0^{\tau_1}\cdots \int_0^{\tau_{K-1}}{ e^{-\mathbf{C}^\dagger \tau_K} (\mathbf{Q},-\mathbf{C})^{(K)} e^{-\mathbf{C} \tau_K}} \mathrm{d}^K \tau.
\end{align}
where $(\mathbf{Q},-\mathbf{C})^{(p)} = (\mathbf{Q},(\mathbf{Q},\cdots (\mathbf{Q},\mathbf{C})\cdots))$
\end{corollary}

\begin{proof}
    By applying the fundamental theorem of calculus, we have
    \begin{equation}
        e^{-\mathbf{C}^\dagger t} \mathbf{Q} e^{-\mathbf{C} t} = \mathbf{Q} + \int_0^t \frac{\mathrm{d}}{\mathrm{d}\tau} \left(e^{-\mathbf{C}^\dagger \tau} \mathbf{Q} e^{-\mathbf{C} \tau}\right) \mathrm{d}\tau = \mathbf{Q} + \int_0^t e^{-\mathbf{C}^\dagger \tau}(\mathbf{Q}, -\mathbf{C}) e^{-\mathbf{C}\tau} \mathrm{d} \tau.
    \end{equation}
    This demonstrates the claim for $K=1$.  
    Now assume that the expression holds for some $K\ge 1$.  
    By expanding the integrand of the remainder term using the fundamental theorem of calculus and Lemma~\ref{lem:diffLem}, we have
    \begin{align} \label{eq:importantLemma}
        e^{-\mathbf{C}^\dagger t} \mathbf{Q} e^{-\mathbf{C} t} =& \sum_{p=0}^{K-1} (\mathbf{Q},-\mathbf{C})^{(p)}\frac{t^{p}}{{p!}} \nonumber\\
        &+\int_0^t\int_0^{\tau_1}\cdots \int_0^{\tau_{K-1}} { e^{-\mathbf{C}^\dagger \tau_K} (\mathbf{Q},-\mathbf{C})^{(K)} e^{-\mathbf{C} \tau_K}} \mathrm{d}^K \tau\nonumber\\
        =& \sum_{p=0}^{K-1} (\mathbf{Q},-\mathbf{C})^{(p)} \frac{t^{p}}{{p!}} \nonumber\\
        &+\int_0^t\int_0^{\tau_1}\cdots \int_0^{\tau_{K-1}} (\mathbf{Q},-\mathbf{C})^{(K)} \mathrm{d}^K \tau \nonumber\\
        &+\int_0^t\int_0^{\tau_1}\cdots \int_0^{\tau_{K}}  e^{-\mathbf{C}^\dagger \tau_{K+1}} (\mathbf{Q},-\mathbf{C})^{(K+1)} e^{-\mathbf{C} \tau_{K+1}} \mathrm{d}^{K+1} \tau\nonumber\\
        =& \sum_{p=0}^{K} (\mathbf{Q},-\mathbf{C})^{(p)} \frac{t^{p}}{{p!}} \nonumber\\
        &+\int_0^t\int_0^{\tau_1}\cdots \int_0^{\tau_{K}} { e^{-\mathbf{C}^\dagger \tau_{K+1}} (\mathbf{Q},-\mathbf{C})^{(K+1)} e^{-\mathbf{C} \tau_{K+1}}} \mathrm{d}^{K+1} \tau. 
    \end{align}
    This demonstrates the induction step and the proof then follows. 
\end{proof}

This gives us an expression that can readily be used to evaluate the case of homogeneous differential equations.
The inhomogeneous case is a little more involved because of the presence of derivatives of the state.

\begin{corollary}[Inhomogeneous Cross-Term Expansion] \label{cor:inhomo}
    For any positive integer $K$, bounded operator $\mathbf{Q} \in \mathbb{C}^{N \times N}$, and $\boldsymbol{\psi}(0) \in \mathbb{C}^N$, let $\boldsymbol{\psi}(t)$ be the solution to the differential equation $\partial_t \boldsymbol{\psi}(t) = -\mathbf{C}\boldsymbol{\psi}(t) +\boldsymbol{\chi}(t)$. 
    Defining $\boldsymbol{\Upsilon}(t) := \int_0^t e^{\mathbf{C}\tau}\boldsymbol{\chi}(\tau)\mathrm{d}\tau$, we have
    \begin{align}
         \boldsymbol{\psi}(0)^\dagger e^{-\mathbf{C}^\dagger t} \mathbf{Q} e^{-\mathbf{C}t}\boldsymbol{\Upsilon}(t)=& \sum_{p=1}^{K}\int_0^t \int_0^{\tau_1}\cdots \int_0^{\tau_{p-1}} \boldsymbol{\psi}(0)^\dagger e^{-\mathbf{C}^\dagger \tau_{p}} (\mathbf{Q},-\mathbf{C})^{(p-1)} \boldsymbol{\chi}(\tau_p)\mathrm{d}^p\tau\nonumber\\
        &+\int_0^t\int_0^{\tau_1}\cdots \int_0^{\tau_{K-1}} \boldsymbol{\psi}(0)^\dagger e^{-\mathbf{C}^\dagger \tau_K} (\mathbf{Q},-\mathbf{C})^{(K)} e^{-\mathbf{C}\tau_K}\boldsymbol{\Upsilon}(\tau_K)\mathrm{d}^K\tau,\\
        \boldsymbol{\Upsilon}(t)^\dagger e^{-\mathbf{C}^\dagger t} \mathbf{Q} e^{-\mathbf{C} t}\boldsymbol{\psi}(0)=& \sum_{p=1}^{K}\int_0^t \int_0^{\tau_1}\cdots \int_0^{\tau_{p-1}} \boldsymbol{\chi}(\tau_{p})^\dagger (\mathbf{Q},-\mathbf{C})^{(p-1)} e^{-\mathbf{C}\tau_{p}}\boldsymbol{\psi}(0)\mathrm{d}^p\tau\nonumber\\
        &+\int_0^t\int_0^{\tau_1}\cdots \int_0^{\tau_{K-1}} \boldsymbol{\Upsilon}(\tau_{K})^\dagger e^{-\mathbf{C}^\dagger \tau_K} (\mathbf{Q},-\mathbf{C})^{(K)} e^{-\mathbf{C}\tau_K}\boldsymbol{\psi}(0)\mathrm{d}^K\tau.
    \end{align}
\end{corollary}

\begin{proof}
    Following the argument given previously, we can use the fundamental theorem of calculus and the product rule. Using the fact that $\boldsymbol{\Upsilon}(0)=0$, we write:
    \begin{align}
        \boldsymbol{\psi}(0)^\dagger e^{-\mathbf{C}^\dagger t} \mathbf{Q} e^{-\mathbf{C} t}\boldsymbol{\Upsilon}(t) &= \int_0^t \frac{\mathrm{d}}{\mathrm{d}\tau} \left( \boldsymbol{\psi}(0)^\dagger e^{-\mathbf{C}^\dagger \tau} \mathbf{Q} e^{-\mathbf{C} \tau}\boldsymbol{\Upsilon}(\tau) \right) \mathrm{d}\tau \nonumber \\
        &= \int_0^t \boldsymbol{\psi}(0)^\dagger e^{-\mathbf{C}^\dagger \tau} (\mathbf{Q},-\mathbf{C}) e^{-\mathbf{C} \tau}\boldsymbol{\Upsilon}(\tau)\mathrm{d}\tau \nonumber \\
        &\quad + \int_0^t \boldsymbol{\psi}(0)^\dagger e^{-\mathbf{C}^\dagger \tau} \mathbf{Q} e^{-\mathbf{C} \tau} \frac{\mathrm{d}\boldsymbol{\Upsilon}(\tau)}{\mathrm{d}\tau}\mathrm{d}\tau.
    \end{align}
    From the definition of $\boldsymbol{\Upsilon}(\tau)$, we have $\frac{\mathrm{d}}{\mathrm{d}\tau}\boldsymbol{\Upsilon}(\tau) = e^{\mathbf{C}\tau}\boldsymbol{\chi}(\tau)$. Substituting this into the second term, the forward and backward propagators acting on the source state cancel ($e^{-\mathbf{C}\tau} e^{\mathbf{C}\tau} = \mathbf{I}$). This allows the operator $\mathbf{Q}$ to act directly on the bare source $\boldsymbol{\chi}(\tau)$:
    \begin{equation}
        \boldsymbol{\psi}(0)^\dagger e^{-\mathbf{C}^\dagger t} \mathbf{Q} e^{-\mathbf{C} t}\boldsymbol{\Upsilon}(t)= \int_0^t \boldsymbol{\psi}(0)^\dagger e^{-\mathbf{C}^\dagger \tau} \mathbf{Q} \boldsymbol{\chi}(\tau)\mathrm{d}\tau+\int_0^t \boldsymbol{\psi}(0)^\dagger e^{-\mathbf{C}^\dagger \tau} (\mathbf{Q},-\mathbf{C}) e^{-\mathbf{C} \tau}\boldsymbol{\Upsilon}(\tau)\mathrm{d}\tau.
    \end{equation}
    This demonstrates the base case ($K=1$) for our argument. Next, let us assume that the result holds for some $K\ge 1$:
    \begin{align}
        \boldsymbol{\psi}(0)^\dagger e^{-\mathbf{C}^\dagger t} \mathbf{Q} e^{-\mathbf{C} t}\boldsymbol{\Upsilon}(t)=& \sum_{p=1}^{K}\int_0^t \int_0^{\tau_1}\cdots \int_0^{\tau_{p-1}} \boldsymbol{\psi}(0)^\dagger e^{-\mathbf{C}^\dagger \tau_{p}} (\mathbf{Q},-\mathbf{C})^{(p-1)} \boldsymbol{\chi}(\tau_p)\mathrm{d}^p\tau\nonumber\\
        &+\int_0^t\int_0^{\tau_1}\cdots \int_0^{\tau_{K-1}} \boldsymbol{\psi}(0)^\dagger e^{-\mathbf{C}^\dagger \tau_K} (\mathbf{Q},-\mathbf{C})^{(K)} e^{-\mathbf{C}\tau_K}\boldsymbol{\Upsilon}(\tau_K)\mathrm{d}^K\tau.
    \end{align}
    By applying Lemma~\ref{lem:diffLem} and the fundamental theorem of calculus, the exact same cancellation of propagators occurs on the derivative of $\boldsymbol{\Upsilon}(\tau_K)$ in the $K$-th order remainder term, yielding:
    \begin{align}
        &\int_0^t\cdots \int_0^{\tau_{K-1}} \boldsymbol{\psi}(0)^\dagger e^{-\mathbf{C}^\dagger \tau_K} (\mathbf{Q},-\mathbf{C})^{(K)} e^{-\mathbf{C}\tau_K}\boldsymbol{\Upsilon}(\tau_K)\mathrm{d}^K\tau \nonumber\\
        =& \int_0^t\cdots \int_0^{\tau_{K}} \boldsymbol{\psi}(0)^\dagger e^{-\mathbf{C}^\dagger \tau_{K+1}} (\mathbf{Q},-\mathbf{C})^{(K)} \boldsymbol{\chi}(\tau_{K+1})\mathrm{d}^{K+1}\tau \nonumber\\
        &+ \int_0^t\cdots \int_0^{\tau_{K}} \boldsymbol{\psi}(0)^\dagger e^{-\mathbf{C}^\dagger \tau_{K+1}} (\mathbf{Q},-\mathbf{C})^{(K+1)} e^{-\mathbf{C}\tau_{K+1}}\boldsymbol{\Upsilon}(\tau_{K+1})\mathrm{d}^{K+1}\tau,
    \end{align}
    which demonstrates the inductive step.

    Next, note that for any operator $\mathbf{X}$, $(\mathbf{X},-\mathbf{C})^\dagger=-\mathbf{C}^\dagger \mathbf{X}^\dagger - \mathbf{X}^\dagger \mathbf{C} = (\mathbf{X}^\dagger,-\mathbf{C})$.  
    Further, we have that 
    \begin{equation}
        \boldsymbol{\Upsilon}(t)^\dagger e^{-\mathbf{C}^\dagger t} \mathbf{Q} e^{-\mathbf{C} t}\boldsymbol{\psi}(0) = \left(\boldsymbol{\psi}(0)^\dagger e^{-\mathbf{C}^\dagger t} \mathbf{Q}^\dagger e^{-\mathbf{C} t}\boldsymbol{\Upsilon}(t) \right)^\dagger.
    \end{equation}
    Thus, by recursively applying the conjugation property of adjoints through the nested brackets, 
     \begin{align}
        \boldsymbol{\Upsilon}(t)^\dagger e^{-\mathbf{C}^\dagger t} \mathbf{Q} e^{-\mathbf{C} t}\boldsymbol{\psi}(0)=& \sum_{p=1}^K\int_0^t \int_0^{\tau_1}\cdots \int_0^{\tau_{p-1}} \boldsymbol{\chi}(\tau_{p})^\dagger (\mathbf{Q},-\mathbf{C})^{(p-1)} e^{-\mathbf{C} \tau_p}\boldsymbol{\psi}(0)\mathrm{d}^p\tau\nonumber\\
        &+\int_0^t\int_0^{\tau_1}\cdots \int_0^{\tau_{K-1}} \boldsymbol{\Upsilon}(\tau_{K})^\dagger e^{-\mathbf{C}^\dagger \tau_K} (\mathbf{Q},-\mathbf{C})^{(K)} e^{-\mathbf{C}\tau_K}\boldsymbol{\psi}(0)\mathrm{d}^K\tau.
    \end{align}
\end{proof}

In order to understand how information propagates through some differential equation system we will need to make some assumptions about the generator $\mathbf{C}$.  
The key assumptions qualitatively are that the dynamics generated by $\mathbf{C}$ is dissipative and that the operator is local with respect to a lattice.  
We summarize the key assumptions that we need to make below.

\begin{definition}[$\Delta$-locality]
    Let $-\mathbf{C} \in \mathbb{C}^{N \times N}$ be a matrix representing an operator acting on a $D$-dimensional lattice, such that the total system dimension is $N = L^D$.  
    We say that $\mathbf{C}$ is physically $\Delta$-local with respect to an $L_p$-norm distance metric $\|x-y\|_p$ on this lattice if $\bra{x} \mathbf{C}\ket{y}=0$ for all computational basis states $\ket{x},\ket{y}$ such that $\|x-y\|_p > \Delta$ with $p \geq 1$, where the state labels $x,y$ are interpreted as coordinate vectors in $\mathbb{Z}^{D}$.
\end{definition}

We can then define the homogeneous information locality bound:
\begin{lemma}[Homogeneous Information Locality Bound] \label{lem:xterms}
    Under the assumptions of \Cref{lem:diffLem}, further assume that $-\mathbf{C}$ is physically $\Delta$-local with respect to a $D$-dimensional lattice and assume that the logarithmic norm of $-\mathbf{C}$ is at most $-\lambda_{\min}\le 0$.
    Further, let $\mathcal{W},\mathcal{V}^\perp \subset (\mathbb{C}^{L})^{\otimes D}$ be disjoint subspaces such that for all basis states $\ket{x}\in \mathcal{W}, \ket{y}\in \mathcal{V}^\perp$ we have $\|x-y\|_p \ge d$, and let the initial state $\boldsymbol{\psi}(0)\in \mathcal{V}^\perp$ and $\mathbf{Q}$ be an operator such that both $\mathbf{Q} \ket{v} =0$ and $\mathbf{Q}^\dagger \ket{v} =0$ for all basis states $\ket{v} \not\in \mathcal{W}$.  
    We then have for time $t\ge 0$ and $v_I t := 2e\Delta \|\mathbf{C}\| t < d -\Delta$ that
$$
|\boldsymbol{\psi}(0)^\dagger e^{-\mathbf{C}^\dagger t} \mathbf{Q} e^{-\mathbf{C} t} \boldsymbol{\psi}(0)|\le \|\boldsymbol{\psi}(0)\|^2 \|\mathbf{Q}\| \min\left(e^{-\frac{d-\Delta}{\Delta} \log\left(\frac{d-\Delta}{2e\Delta \|\mathbf{C}\| t} \right)}, e^{-2\lambda_{\min} t}\right).
$$  
\end{lemma}

\begin{proof}
  From~\Cref{cor:homo} we have that for any $K\ge 1$ 
\begin{align}
\boldsymbol{\psi}(0)^\dagger e^{-\mathbf{C}^\dagger t} \mathbf{Q} e^{-\mathbf{C} t} \boldsymbol{\psi}(0) =& \sum_{p=0}^{K-1} \boldsymbol{\psi}(0)^\dagger (\mathbf{Q},-\mathbf{C})^{(p)} \boldsymbol{\psi}(0) \frac{t^{p}}{{p!}} \nonumber\\
&+\int_0^t\int_0^{\tau_1}\cdots \int_0^{\tau_{K-1}}{ \boldsymbol{\psi}(0)^\dagger e^{-\mathbf{C}^\dagger \tau_K} (\mathbf{Q},-\mathbf{C})^{(K)} e^{-\mathbf{C} \tau_K}}\boldsymbol{\psi}(0) \mathrm{d}^K \tau.
\end{align}
Next consider $(\mathbf{Q},-\mathbf{C})\boldsymbol{\psi}(0)$ and assume for contradiction that the result is non-zero while $d >\Delta$.  
If this is true then either $\mathbf{C} \boldsymbol{\psi}(0)$ has non-zero support on $\mathcal{W}$ or $\mathbf{Q} \boldsymbol{\psi}(0) \ne 0$.  
We assumed that $\boldsymbol{\psi}(0) \in \mathcal{V}^\perp$ (which is disjoint from $\mathcal{W}$) and hence by assumption $\mathbf{Q} \boldsymbol{\psi}(0) =0$. 
Thus $\mathbf{C} \boldsymbol{\psi}(0)$ has non-zero support on $\mathcal{W}$ and if this is true then $\mathcal{V}^\perp$ and $\mathcal{W}$ must contain two basis states that are distance at most $\Delta$ from each other.  
This is impossible if $d > \Delta$ so by contradiction this cannot hold.  
Now let us assume that for some integer $k$, the $k$-fold nested bracket satisfies $(\mathbf{Q},-\mathbf{C})^{(k)}\boldsymbol{\phi}=0$ for any state $\boldsymbol{\phi}$ whose support is separated from $\mathcal{W}$ by a distance strictly greater than $k\Delta$. 
Because $\mathbf{C}$ is $\Delta$-local, the expanded state $\mathbf{C}\boldsymbol{\psi}(0)$ is separated from $\mathcal{W}$ by at least $d-\Delta$.
If $d > (k+1)\Delta$, then $d-\Delta > k\Delta$, meaning by the inductive hypothesis the $(k+1)$-fold nested bracket acting on $\boldsymbol{\psi}(0)$ must also evaluate to zero.
Therefore by induction if we set $K = \lfloor d/\Delta \rfloor$ then
\begin{align}
\boldsymbol{\psi}(0)^\dagger e^{-\mathbf{C}^\dagger t} \mathbf{Q} e^{-\mathbf{C} t} \boldsymbol{\psi}(0) =&\int_0^t\int_0^{\tau_1}\cdots \int_0^{\tau_{K-1}}{ \boldsymbol{\psi}(0)^\dagger e^{-\mathbf{C}^\dagger \tau_K} (\mathbf{Q},-\mathbf{C})^{(K)} e^{-\mathbf{C} \tau_K}}\boldsymbol{\psi}(0) \mathrm{d}^K \tau.
\end{align}
Then from the triangle inequality
\begin{align}
    |\boldsymbol{\psi}(0)^\dagger e^{-\mathbf{C}^\dagger t} \mathbf{Q} e^{-\mathbf{C} t} \boldsymbol{\psi}(0)| \le \|\boldsymbol{\psi}(0)\|^2 \|\mathbf{Q}\| \frac{(2\|\mathbf{C}\|t)^{K}}{K!}.
\end{align}
Using Stirling's approximation it is easy to see that $K! \ge (K/e)^K = e^{K \log(K/e)}$ and thus for $K \ge (d-\Delta)/\Delta$
\begin{align}
    |\boldsymbol{\psi}(0)^\dagger e^{-\mathbf{C}^\dagger t} \mathbf{Q} e^{-\mathbf{C} t} \boldsymbol{\psi}(0)| \le \|\boldsymbol{\psi}(0)\|^2 \|\mathbf{Q}\|  e^{ - \frac{(d-\Delta)}{\Delta} \log\left(\frac{d-\Delta}{2e\Delta \|\mathbf{C}\| t} \right)}.
\end{align}
Next we can argue that the following trivial bound holds for the expectation value due to the dissipative dynamics:
\begin{equation}
    |\boldsymbol{\psi}(0)^\dagger e^{-\mathbf{C}^\dagger t} \mathbf{Q} e^{-\mathbf{C} t} \boldsymbol{\psi}(0)| \le \|\boldsymbol{\psi}(0)\|^2 \|\mathbf{Q}\| e^{-2\lambda_{\min} t},
\end{equation}
from which the claim of the lemma follows by choosing the upper bound to be the minimum of these two bounds.
\end{proof}
This shows that for the homogeneous case, information can only travel at a finite velocity in these systems, which leads to the expectation value of the operator $\mathbf{Q}$ outside a ball of radius $d > v_I t$ shrinking exponentially with $d$.
In subsequent results, we will invert these bounds to determine the truncation distance $d$ required to guarantee a target simulation error $\epsilon$.

Next we will provide analogous bounds for the remainder of the terms that appear in expectation values for the inhomogeneous differential equation $\partial_t \boldsymbol{\psi}(t) = -\mathbf{C} \boldsymbol{\psi}(t) +\boldsymbol{\chi}(t)$ with solution $\boldsymbol{\psi}(t) = e^{-\mathbf{C} t} \boldsymbol{\psi}(0) + e^{-\mathbf{C} t} \int_0^t e^{\mathbf{C} \tau} \boldsymbol{\chi}(\tau) \mathrm{d}\tau=e^{-\mathbf{C} t} \boldsymbol{\psi}(0) + e^{-\mathbf{C} t} \boldsymbol{\Upsilon}(t)$.
For notational convenience in the following bounds, we define the time-integrated $L^1$-norm of the source vector as $\|\boldsymbol{\chi}\|_{L^1} := \int_0^t \|\boldsymbol{\chi}(\tau)\| \mathrm{d}\tau$.

\begin{lemma} [Mixed-Term Information Locality Bound] \label{lem:mixedTerms}
Under the assumptions of \Cref{lem:diffLem}, further assume that $-\mathbf{C}$ is physically $\Delta$-local with respect to a $D$-dimensional lattice and assume that the logarithmic norm of $-\mathbf{C}$ is at most $-\lambda_{\min}\le 0$.  Further, let $\mathcal{W},\mathcal{V}^\perp \subset (\mathbb{C}^{L})^{\otimes D}$ be disjoint sets of vectors such that for all $\ket{x}\in \mathcal{W}, \ket{y}\in \mathcal{V}^\perp$ we have $\|x-y\|_p \ge d$. Let the initial state $\boldsymbol{\psi}(0)\in \mathcal{V}^\perp$, assume the source vector $\boldsymbol{\chi}(t)$ has strict spatial support within $\mathcal{V}^\perp$ for all $t$, and let $\mathbf{Q}$ be an operator such that both $\mathbf{Q} \ket{v} =0$ and $\mathbf{Q}^\dagger \ket{v} =0$ for all $\ket{v} \not\in \mathcal{W}$. We then have for time $t\ge 0$ and $v_I t := 2e\Delta \|\mathbf{C}\| t < d -\Delta$ that
\begin{align*}
|\boldsymbol{\Upsilon}^\dagger(t) e^{-\mathbf{C}^\dagger t} \mathbf{Q} e^{-\mathbf{C} t} \boldsymbol{\psi}(0)|&\le \|\mathbf{Q}\|\|\boldsymbol{\psi}(0)\| \|\boldsymbol{\chi}\|_{L^1} \min\left(e^{- \frac{d-\Delta}{\Delta} \log\left( \frac{d-\Delta}{2e \Delta \|\mathbf{C}\| t}\right)},e^{-\lambda_{\min} t}\right)\\
|\boldsymbol{\psi}(0)^\dagger e^{-\mathbf{C}^\dagger t} \mathbf{Q} e^{-\mathbf{C} t} \boldsymbol{\Upsilon}(t)|&\le \|\mathbf{Q}\|\|\boldsymbol{\psi}(0)\| \|\boldsymbol{\chi}\|_{L^1} \min \left(e^{- \frac{d-\Delta}{\Delta} \log\left( \frac{d-\Delta}{2e \Delta \|\mathbf{C}\| t}\right)},e^{-\lambda_{\min} t}\right).
\end{align*}
\end{lemma}

\begin{proof}
From \Cref{cor:inhomo} and the triangle inequality:
    \begin{align}
        |\boldsymbol{\psi}(0)^\dagger e^{-\mathbf{C}^\dagger t} \mathbf{Q} e^{-\mathbf{C} t} \boldsymbol{\Upsilon}(t)|\le & \sum_{p=1}^{K}\int_0^t \int_0^{\tau_1}\cdots \int_0^{\tau_{p-1}} \left| \boldsymbol{\psi}(0)^\dagger e^{-\mathbf{C}^\dagger \tau_{p}} (\mathbf{Q},-\mathbf{C})^{(p-1)} \boldsymbol{\chi}(\tau_p) \right|\mathrm{d}^p\tau\nonumber\\
        &+\int_0^t\int_0^{\tau_1}\cdots \int_0^{\tau_{K-1}} \left| \boldsymbol{\psi}(0)^\dagger e^{-\mathbf{C}^\dagger \tau_K} (\mathbf{Q},-\mathbf{C})^{(K)} e^{-\mathbf{C}\tau_K} \boldsymbol{\Upsilon}(\tau_K) \right|\mathrm{d}^K\tau.
    \end{align}
    Because $\boldsymbol{\chi}(\tau_p)$ has strict support in $\mathcal{V}^\perp$, the nested bracket acting on it evaluates exactly to the zero vector for all terms in the sum, provided we choose $K = \lfloor d/\Delta \rfloor$. Thus, assuming that we make this choice and that $d -\Delta > 2e\Delta \|\mathbf{C}\| t$:
\begin{align}
        |\boldsymbol{\psi}(0)^\dagger e^{-\mathbf{C}^\dagger t} \mathbf{Q} e^{-\mathbf{C} t} \boldsymbol{\Upsilon}(t)|\le 
        &\int_0^t\int_0^{\tau_1}\cdots \int_0^{\tau_{K-1}} \left| \boldsymbol{\psi}(0)^\dagger e^{-\mathbf{C}^\dagger \tau_K} (\mathbf{Q},-\mathbf{C})^{(K)} e^{-\mathbf{C}\tau_K} \boldsymbol{\Upsilon}(\tau_K) \right|\mathrm{d}^K\tau \nonumber\\
        &\le \|\mathbf{Q}\|\|\boldsymbol{\psi}(0)\| \max_{\tau\in [0,t]} \|e^{-\mathbf{C} \tau} \boldsymbol{\Upsilon}(\tau)\|  \frac{(2\|\mathbf{C}\| t)^K}{K!}\nonumber\\
        &\le \|\mathbf{Q}\|\|\boldsymbol{\psi}(0)\| \|\boldsymbol{\chi}\|_{L^1} \left(\frac{(2e\|\mathbf{C}\| t)}{K}\right)^K\nonumber\\
        &= \|\mathbf{Q}\|\|\boldsymbol{\psi}(0)\| \|\boldsymbol{\chi}\|_{L^1} e^{- K \log\left( \frac{K}{2e \|\mathbf{C}\| t}\right)}\nonumber\\
        &\le \|\mathbf{Q}\|\|\boldsymbol{\psi}(0)\| \|\boldsymbol{\chi}\|_{L^1} e^{ - \frac{d-\Delta}{\Delta} \log\left( \frac{d-\Delta}{2e \Delta \|\mathbf{C}\| t}\right)}.
    \end{align}
    Then by the exact same argument it follows that:
    \begin{align}
        |\boldsymbol{\Upsilon}^\dagger(t) e^{-\mathbf{C}^\dagger t} \mathbf{Q} e^{-\mathbf{C} t} \boldsymbol{\psi}(0)|\le \|\mathbf{Q}\| \|\boldsymbol{\psi}(0)\| \|\boldsymbol{\chi}\|_{L^1} e^{ - \frac{d-\Delta}{\Delta} \log\left( \frac{d-\Delta}{2e \Delta \|\mathbf{C}\| t}\right)}.
    \end{align}
    Then in both cases we can also see that the alternative trivial bound holds. For example:
    \begin{equation}
        |\boldsymbol{\Upsilon}^\dagger(t) e^{-\mathbf{C}^\dagger t} \mathbf{Q} e^{-\mathbf{C} t} \boldsymbol{\psi}(0)|\le \|\mathbf{Q}\|\|\boldsymbol{\psi}(0)\| \|\boldsymbol{\chi}\|_{L^1} e^{-\lambda_{\min} t}.
    \end{equation}
    The conjugate case holds using the same reasoning, and just as in the prior lemma, we can find the smallest upper bound by choosing the minimum of the trivial bound and the information locality bound.
\end{proof}
While \Cref{lem:mixedTerms} assumes the source $\boldsymbol{\chi}(t)$ is strictly localized to $\mathcal{V}^\perp$, this naturally extends to generic global sources. 
Because the differential equation is linear, any global source can be decomposed into a sum of local sources. 
The components of the source within the effective light-cone of $\mathcal{W}$ will influence the observable dynamically, while the components outside the light-cone are exponentially suppressed by the bound above.

Finally, we will show the analogous result for expectation values involving two copies of $\boldsymbol{\Upsilon}$.  
In order to apply this to local simulation, we need to ultimately argue that we can safely truncate the inhomogeneity to a finite region around the observable $\mathbf{Q}$, because the contribution from the source in any distant region is exponentially suppressed.  
Specifically, let us partition the total source into $\boldsymbol{\chi}(t)= \boldsymbol{\chi}_{\mathcal V}(t) + \boldsymbol{\chi}_{\mathcal{V}^\perp}(t)$, where $\boldsymbol{\chi}_{\mathcal{V}^\perp}(t)$ is the component strictly supported on a distant subspace $\mathcal{V}^\perp$ for all $t$ and $\boldsymbol{\chi}_{\mathcal{V}}(t)$ is supported on its complement.  
We then provide the following result which shows that the cross-term expectation value shrinks exponentially with distance for any term that involves at least one copy of the distant source $\boldsymbol{\chi}_{\mathcal{V}^\perp}(t)$.

\begin{lemma}[Inhomogeneous Source Locality Bound] \label{lem:doubleTerm}
Under the assumptions of \Cref{lem:diffLem}, further assume that $-\mathbf{C}$ is physically $\Delta$-local with respect to a $D$-dimensional lattice and assume that the logarithmic norm of $-\mathbf{C}$ is at most $-\lambda_{\min}\le 0$. Further, let $\mathcal{W},\mathcal{V}^\perp \subset (\mathbb{C}^{L})^{\otimes D}$ be disjoint subspaces such that for all basis states $\ket{x}\in \mathcal{W}, \ket{y}\in \mathcal{V}^\perp$ we have $\|x-y\|_p \ge d$, and let $\mathbf{Q}$ be an operator such that both $\mathbf{Q} \ket{v} =0$ and $\mathbf{Q}^\dagger \ket{v} =0$ for all $\ket{v} \not\in \mathcal{W}$. We then have for time $t\ge 0$ and $v_I t := 2e\Delta \|\mathbf{C}\| t < d -\Delta$ that
$$
|\boldsymbol{\Upsilon}^\dagger(t) e^{-\mathbf{C}^\dagger t} \mathbf{Q} e^{-\mathbf{C} t} \boldsymbol{\Upsilon}_{\mathcal{V}^\perp}(t)|\le \|\mathbf{Q}\|\min  \left(\|\boldsymbol{\chi}\|_{L^1} \|\boldsymbol{\chi}_{\mathcal{V}^\perp}\|_{L^1} e^{-\frac{d -\Delta}{\Delta}\log\left(\frac{d-\Delta}{2e\Delta \|\mathbf{C}\| t} \right)}, A_{\boldsymbol{\chi}} A_{\boldsymbol{\chi}_{\mathcal{V}^\perp}} \left( \frac{1 - e^{-\lambda_{\min} t}}{\lambda_{\min}} \right)^2 \right)
$$
where $\boldsymbol{\Upsilon}_{\mathcal{V}^\perp}(t) = \int_0^t e^{\mathbf{C} \tau} \boldsymbol{\chi}_{\mathcal{V}^\perp}(\tau) \mathrm{d}\tau$, $\boldsymbol{\chi}_{\mathcal{V}^\perp}$ is the projection of the inhomogeneity $\boldsymbol{\chi}(t)$ onto the subspace $\mathcal{V}^\perp$, and $A_{\boldsymbol{\chi}}, A_{\boldsymbol{\chi}_{\mathcal{V}^\perp}}$ are the respective amplitude bounds defined as in \Cref{lem:classicalSim}.
\end{lemma}

\begin{proof}
To evaluate the expectation value using the generalized brackets, we analyze the inner product of the time-dependent states $e^{-\mathbf{C}t}\boldsymbol{\Upsilon}(t)$ and $e^{-\mathbf{C}t}\boldsymbol{\Upsilon}_{\mathcal{V}^\perp}(t)$. 
By the product rule and \Cref{lem:diffLem}, the time derivative of the expectation value of the $k$-fold nested bracket between these states evaluates exactly to:
\begin{align} \label{eq:state_deriv}
    \frac{\mathrm{d}}{\mathrm{d}\tau} \left[ \boldsymbol{\Upsilon}^\dagger(\tau) e^{-\mathbf{C}^\dagger \tau} (\mathbf{Q},-\mathbf{C})^{(k)} e^{-\mathbf{C}\tau} \boldsymbol{\Upsilon}_{\mathcal{V}^\perp}(\tau) \right] =& \boldsymbol{\Upsilon}^\dagger(\tau) e^{-\mathbf{C}^\dagger \tau} (\mathbf{Q},-\mathbf{C})^{(k+1)} e^{-\mathbf{C}\tau} \boldsymbol{\Upsilon}_{\mathcal{V}^\perp}(\tau) \nonumber\\
    &+ \boldsymbol{\chi}^\dagger(\tau) (\mathbf{Q},-\mathbf{C})^{(k)} e^{-\mathbf{C}\tau} \boldsymbol{\Upsilon}_{\mathcal{V}^\perp}(\tau) \nonumber\\
    &+ \boldsymbol{\Upsilon}^\dagger(\tau) e^{-\mathbf{C}^\dagger \tau} (\mathbf{Q},-\mathbf{C})^{(k)} \boldsymbol{\chi}_{\mathcal{V}^\perp}(\tau).
\end{align}
Because $\boldsymbol{\chi}_{\mathcal{V}^\perp}(\tau)$ is strictly supported on $\mathcal{V}^\perp$ and $\mathbf{C}$ is $\Delta$-local, for any $k < K = \lfloor d / \Delta \rfloor$, the expanded support of $(\mathbf{Q},-\mathbf{C})^{(k)}$ is separated from $\mathcal{V}^\perp$ by at least $d - (K-1)\Delta \ge \Delta > 0$. 
Because $\mathbf{Q}$ acts as the zero operator outside $\mathcal{W}$, $(\mathbf{Q},-\mathbf{C})^{(k)} \boldsymbol{\chi}_{\mathcal{V}^\perp}(\tau) = 0$ exactly. 

Using $\boldsymbol{\Upsilon}(0) = \boldsymbol{\Upsilon}_{\mathcal{V}^\perp}(0) = 0$, we integrate \Cref{eq:state_deriv} from $0$ to $t$ and apply it recursively $K$ times:
\begin{align}
    \boldsymbol{\Upsilon}^\dagger(t) e^{-\mathbf{C}^\dagger t} \mathbf{Q} e^{-\mathbf{C}t} \boldsymbol{\Upsilon}_{\mathcal{V}^\perp}(t) =& \int_0^t\int_0^{\tau_1}\cdots \int_0^{\tau_{K-1}} \boldsymbol{\Upsilon}^\dagger(\tau_K) e^{-\mathbf{C}^\dagger \tau_K} (\mathbf{Q},-\mathbf{C})^{(K)} e^{-\mathbf{C}\tau_K} \boldsymbol{\Upsilon}_{\mathcal{V}^\perp}(\tau_K) \mathrm{d}^K\tau \nonumber\\
    &+ \sum_{p=1}^K\int_0^t\int_0^{\tau_1}\cdots \int_0^{\tau_{p-1}} \boldsymbol{\chi}^\dagger(\tau_{p}) (\mathbf{Q},-\mathbf{C})^{(p-1)} e^{-\mathbf{C}\tau_{p}} \boldsymbol{\Upsilon}_{\mathcal{V}^\perp}(\tau_{p})\mathrm{d}^p\tau.
\end{align}
We bound the remainder term using the uniform bounds $\|e^{-\mathbf{C}\tau}\boldsymbol{\Upsilon}(\tau)\| \le \|\boldsymbol{\chi}\|_{L^1}$ and $\|e^{-\mathbf{C}\tau}\boldsymbol{\Upsilon}_{\mathcal{V}^\perp}(\tau)\| \le \|\boldsymbol{\chi}_{\mathcal{V}^\perp}\|_{L^1}$. 
Since $\|(\mathbf{Q},-\mathbf{C})^{(K)}\| \le 2^K \|\mathbf{C}\|^K \|\mathbf{Q}\|$, integrating over the $K$-simplex yields a bound of $2^K \|\mathbf{Q}\| \|\boldsymbol{\chi}\|_{L^1} \|\boldsymbol{\chi}_{\mathcal{V}^\perp}\|_{L^1} \frac{(\|\mathbf{C}\| t)^K}{K!}$.

To bound the sum terms, we expand the distant vector $e^{-\mathbf{C}\tau_p} \boldsymbol{\Upsilon}_{\mathcal{V}^\perp}(\tau_{p}) = \int_0^{\tau_{p}} e^{-\mathbf{C}(\tau_{p}-\tau)} \boldsymbol{\chi}_{\mathcal{V}^\perp}(\tau) \mathrm{d}\tau$. 
Because $(\mathbf{Q},-\mathbf{C})^{(p-1)}$ is separated from $\mathcal{V}^\perp$ by $d - (p-1)\Delta$, we expand the single propagator $e^{-\mathbf{C}(\tau_{p}-\tau)}$ to order $K-(p-1)$. 
The polynomial terms vanish because the total expanded support is at most $(p-1+K-p)\Delta = (K-1)\Delta < d$. 
Writing the remainder of this single propagator in its integral form, $e^{-\mathbf{C}s} - \sum_{k=0}^{M-1}\frac{(-\mathbf{C}s)^k}{k!} = \int_0^s \frac{(s-u)^{M-1}}{(M-1)!}(-\mathbf{C})^M e^{-\mathbf{C}u}\,\mathrm{d}u$ with $M = K-(p-1)$ and $s = \tau_p - \tau \le \tau_p$, and bounding $\|e^{-\mathbf{C}u}\| \le 1$ by the non-positive logarithmic norm, its norm is at most $\frac{(\|\mathbf{C}\|s)^M}{M!}$. This yields:
\begin{equation}
    \left|\boldsymbol{\chi}^\dagger(\tau_{p}) (\mathbf{Q},-\mathbf{C})^{(p-1)} e^{-\mathbf{C}\tau_p} \boldsymbol{\Upsilon}_{\mathcal{V}^\perp}(\tau_{p})\right| \le \|\boldsymbol{\chi}(\tau_{p})\| \|(\mathbf{Q},-\mathbf{C})^{(p-1)}\| \|\boldsymbol{\chi}_{\mathcal{V}^\perp}\|_{L^1} \frac{(\|\mathbf{C}\|\tau_{p})^{K-(p-1)}}{(K-(p-1))!}.
\end{equation}
Substituting $\|(\mathbf{Q},-\mathbf{C})^{(p-1)}\| \le 2^{p-1} \|\mathbf{C}\|^{p-1} \|\mathbf{Q}\|$, the geometric coefficient of each sum term is exactly $2^{p-1} \|\mathbf{C}\|^K \|\mathbf{Q}\|$. 
Integrating out $\tau_1,\ldots,\tau_{p-1}$ over the simplex leaves $\int_0^t \frac{(t-\tau)^{p-1}}{(p-1)!} \frac{\tau^{K-p+1}}{(K-p+1)!} \|\boldsymbol{\chi}(\tau)\| \mathrm{d}\tau$; substituting $j = p-1$ and $q = \tau/t$, the polynomial factor equals $\frac{t^K}{K!}\binom{K}{j} q^{K-j}(1-q)^{j} \le \frac{t^K}{K!}$, since the binomial factor is a probability mass. 
The remaining nested time integrals are therefore bounded by $\frac{t^K}{K!} \|\boldsymbol{\chi}\|_{L^1}$.
Summing the geometric sequence $\sum_{p=1}^{K} 2^{p-1} = 2^K - 1$ and adding the $K$-th order remainder (which contributes $2^K$) bounds the total expectation value by:
\begin{align}
    |\boldsymbol{\Upsilon}^\dagger(t) e^{-\mathbf{C}^\dagger t} \mathbf{Q} e^{-\mathbf{C} t} \boldsymbol{\Upsilon}_{\mathcal{V}^\perp}(t)| &\le (2^{K+1}-1) \|\mathbf{Q}\| \|\boldsymbol{\chi}\|_{L^1} \|\boldsymbol{\chi}_{\mathcal{V}^\perp}\|_{L^1} \frac{(\|\mathbf{C}\| t)^K}{K!} \nonumber \\
    &< 2 \|\mathbf{Q}\| \|\boldsymbol{\chi}\|_{L^1} \|\boldsymbol{\chi}_{\mathcal{V}^\perp}\|_{L^1} \frac{(2\|\mathbf{C}\| t)^K}{K!}.
\end{align}
By Stirling's approximation, $K! \ge \sqrt{2\pi K} (K/e)^K$. 
Since $\sqrt{2\pi K} > 2$ for all $K \ge 1$, the mathematical inequality $\frac{2}{K!} \le (e/K)^K$ holds globally. 
We absorb the constant $2$ entirely to drop the factorial and substitute $K \ge (d-\Delta)/\Delta$, yielding the stated information locality bound:
\begin{equation}
    |\boldsymbol{\Upsilon}^\dagger(t) e^{-\mathbf{C}^\dagger t} \mathbf{Q} e^{-\mathbf{C} t} \boldsymbol{\Upsilon}_{\mathcal{V}^\perp}(t)| \le \|\mathbf{Q}\| \|\boldsymbol{\chi}\|_{L^1} \|\boldsymbol{\chi}_{\mathcal{V}^\perp}\|_{L^1} e^{-\frac{d -\Delta}{\Delta}\log\left(\frac{d-\Delta}{2e\Delta \|\mathbf{C}\| t} \right)}.
\end{equation}
The trivial bound is evaluated by bounding the continuous drive directly via the exact norm integrals established in \Cref{lem:classicalSim}:
\begin{align}
     |\boldsymbol{\Upsilon}^\dagger(t) e^{-\mathbf{C}^\dagger t} \mathbf{Q} e^{-\mathbf{C} t} \boldsymbol{\Upsilon}_{\mathcal{V}^\perp}(t)| &\le \left(\int_0^t e^{-\lambda_{\min}(t-\tau_1)} \|\boldsymbol{\chi}(\tau_1)\| \mathrm{d}\tau_1\right) \|\mathbf{Q}\| \left(\int_0^t e^{-\lambda_{\min}(t-\tau_2)} \|\boldsymbol{\chi}_{\mathcal{V}^\perp}(\tau_2)\| \mathrm{d}\tau_2\right) \nonumber\\
     &\le \|\mathbf{Q}\| A_{\boldsymbol{\chi}} A_{\boldsymbol{\chi}_{\mathcal{V}^\perp}} \left( \frac{1 - e^{-\lambda_{\min} t}}{\lambda_{\min}} \right)^2.
\end{align}
As both bounds are upper bounds, the tightest achievable bound is simply their minimum.
\end{proof}

Now we can simply combine these technical lemmas to arrive at the following result, where we assume a non-trivial parameter regime.
\begin{lemma}[Partitioned Observable Bound] \label{lem:exactPartition}
    Assume that the preconditions of~\Cref{lem:doubleTerm} hold. In particular, let $\mathbf{Q}$ be a physically $\Delta$-local observable with support only on a subspace $\mathcal{W}$ and let $\mathcal{V}^\perp$ be a subspace that is separated from $\mathcal{W}$ by at least distance $d$. Further, let the global initial state be partitioned as $\boldsymbol{\psi}(0) = \boldsymbol{\psi}_{\mathcal{V}}(0) + \boldsymbol{\psi}_{\mathcal{V}^\perp}(0)$, where $\boldsymbol{\psi}_{\mathcal{V}^\perp}(0)$ has strict spatial support within $\mathcal{V}^\perp$. 
    Let $\boldsymbol{\Upsilon}_{\mathcal{V}}(t) := \int_0^t e^{\mathbf{C} \tau} \boldsymbol{\chi}_{\mathcal{V}}(\tau) \mathrm{d}\tau$ and let $\boldsymbol{\psi}_{\mathcal{V}}(t) := e^{-\mathbf{C} t} \big( \boldsymbol{\psi}_{\mathcal{V}}(0) + \boldsymbol{\Upsilon}_{\mathcal{V}}(t) \big)$ be the exact component of the state driven solely by the local initial state and local source. 
    We then have for time $t \ge 0$ and $v_I t := 2e\Delta\|\mathbf{C}\|t < d - \Delta$ that the difference in the expectation value of the observable $\mathbf{Q}$ from that due solely to the local state components satisfies
    $$
    | {\boldsymbol{\psi}}^\dagger(t) \mathbf{Q} {\boldsymbol{\psi}}(t)-\boldsymbol{\psi}_{\mathcal{V}}^\dagger(t) \mathbf{Q} \boldsymbol{\psi}_{\mathcal{V}}(t)| \le 12\|\mathbf{Q}\| C_{\max}^2 e^{-\frac{d -\Delta}{\Delta}\log\left(\frac{d-\Delta}{2e\Delta \|\mathbf{C}\| t} \right)},
    $$
    where $C_{\max} := \max(\|\boldsymbol{\psi}(0)\|, \|\boldsymbol{\chi}\|_{L^1})$.
\end{lemma}
\begin{proof}
From the partitioned definitions of the initial state and source, we can split the exact forward evolution into distant and local components. 
Let $\boldsymbol{\Upsilon}(t) := \boldsymbol{\Upsilon}_{\mathcal{V}}(t) + \boldsymbol{\Upsilon}_{\mathcal{V}^\perp}(t)$ follow the identical spatial partition. 
By expanding the equation of motion of $\boldsymbol{\psi}$ in \eqref{eq:diffDef}, we have $\boldsymbol{\psi}(t) = e^{-\mathbf{C}t}\big(\boldsymbol{\psi}_{\mathcal{V}}(0) + \boldsymbol{\psi}_{\mathcal{V}^\perp}(0)\big) + e^{-\mathbf{C}t}\big(\boldsymbol{\Upsilon}_{\mathcal{V}}(t) + \boldsymbol{\Upsilon}_{\mathcal{V}^\perp}(t)\big)$.
Substituting this exact state into the expectation value distributes the difference ${\boldsymbol{\psi}}^\dagger(t) \mathbf{Q} {\boldsymbol{\psi}}(t) -\boldsymbol{\psi}_{\mathcal{V}}^\dagger(t) \mathbf{Q} \boldsymbol{\psi}_{\mathcal{V}}(t)$ into 12 cross-terms. 
Specifically,
\begin{align}
    \left|{\boldsymbol{\psi}}^\dagger(t) \mathbf{Q} {\boldsymbol{\psi}}(t) -\boldsymbol{\psi}_{\mathcal{V}}^\dagger(t) \mathbf{Q} \boldsymbol{\psi}_{\mathcal{V}}(t)\right| &\le 2\left| \boldsymbol{\psi}_{\mathcal{V}}(0)^\dagger e^{-\mathbf{C}^\dagger t} \mathbf{Q} e^{-\mathbf{C} t} \boldsymbol{\psi}_{\mathcal{V}^\perp}(0) \right| + \left| \boldsymbol{\psi}_{\mathcal{V}^\perp}(0)^\dagger e^{-\mathbf{C}^\dagger t} \mathbf{Q} e^{-\mathbf{C} t} \boldsymbol{\psi}_{\mathcal{V}^\perp}(0) \right|\nonumber\\
    &\quad+ 2\left| \boldsymbol{\Upsilon}_{\mathcal{V}}(t)^\dagger e^{-\mathbf{C}^\dagger t} \mathbf{Q} e^{-\mathbf{C} t} \boldsymbol{\psi}_{\mathcal{V}^\perp}(0) \right|+2\left| \boldsymbol{\Upsilon}_{\mathcal{V}^\perp}(t)^\dagger e^{-\mathbf{C}^\dagger t} \mathbf{Q} e^{-\mathbf{C} t} \boldsymbol{\psi}_{\mathcal{V}^\perp}(0) \right|
    \nonumber\\
    &\quad+2\left| \boldsymbol{\psi}_{\mathcal{V}}(0)^\dagger e^{-\mathbf{C}^\dagger t} \mathbf{Q} e^{-\mathbf{C} t} \boldsymbol{\Upsilon}_{\mathcal{V}^\perp}(t) \right|+2\left| \boldsymbol{\Upsilon}_{\mathcal{V}}(t)^\dagger e^{-\mathbf{C}^\dagger t} \mathbf{Q} e^{-\mathbf{C} t} \boldsymbol{\Upsilon}_{\mathcal{V}^\perp}(t) \right|\nonumber\\
    &\quad + \left| \boldsymbol{\Upsilon}_{\mathcal{V}^\perp}(t)^\dagger e^{-\mathbf{C}^\dagger t} \mathbf{Q} e^{-\mathbf{C} t} \boldsymbol{\Upsilon}_{\mathcal{V}^\perp}(t) \right|
\end{align}
Because $\boldsymbol{\psi}_{\mathcal{V}^\perp}(0)$ and $\boldsymbol{\Upsilon}_{\mathcal{V}^\perp}(t)$ are derived from vectors strictly supported on $\mathcal{V}^\perp$, every single term in this 12-term expansion contains at least one vector belonging to the distant subspace.  

For the fully distant cross-terms, where both vectors belong to $\mathcal{V}^\perp$, applying \Cref{lem:xterms,lem:mixedTerms,lem:doubleTerm} bounds them exponentially. 
Specifically, from \Cref{lem:xterms} we have that
\begin{align}
    \left|{\boldsymbol{\psi}}^\dagger(t) \mathbf{Q} {\boldsymbol{\psi}}(t) -\boldsymbol{\psi}_{\mathcal{V}}^\dagger(t) \mathbf{Q} \boldsymbol{\psi}_{\mathcal{V}}(t)\right| &\le 3 \|\mathbf{Q}\| C_{\max}^2 e^{-\frac{d -\Delta}{\Delta}\log\left(\frac{d-\Delta}{2e\Delta \|\mathbf{C}\| t} \right)}\nonumber\\
    &\quad+ 2\left| \boldsymbol{\Upsilon}_{\mathcal{V}}(t)^\dagger e^{-\mathbf{C}^\dagger t} \mathbf{Q} e^{-\mathbf{C} t} \boldsymbol{\psi}_{\mathcal{V}^\perp}(0) \right|+2\left| \boldsymbol{\Upsilon}_{\mathcal{V}^\perp}(t)^\dagger e^{-\mathbf{C}^\dagger t} \mathbf{Q} e^{-\mathbf{C} t} \boldsymbol{\psi}_{\mathcal{V}^\perp}(0) \right|
    \nonumber\\
    &\quad+2\left| \boldsymbol{\psi}_{\mathcal{V}}(0)^\dagger e^{-\mathbf{C}^\dagger t} \mathbf{Q} e^{-\mathbf{C} t} \boldsymbol{\Upsilon}_{\mathcal{V}^\perp}(t) \right|+2\left| \boldsymbol{\Upsilon}_{\mathcal{V}}(t)^\dagger e^{-\mathbf{C}^\dagger t} \mathbf{Q} e^{-\mathbf{C} t} \boldsymbol{\Upsilon}_{\mathcal{V}^\perp}(t) \right|\nonumber\\
    &\quad + \left| \boldsymbol{\Upsilon}_{\mathcal{V}^\perp}(t)^\dagger e^{-\mathbf{C}^\dagger t} \mathbf{Q} e^{-\mathbf{C} t} \boldsymbol{\Upsilon}_{\mathcal{V}^\perp}(t) \right|
\end{align}
Then invoking \Cref{lem:mixedTerms} we have 
\begin{align}
    \left|{\boldsymbol{\psi}}^\dagger(t) \mathbf{Q} {\boldsymbol{\psi}}(t) -\boldsymbol{\psi}_{\mathcal{V}}^\dagger(t) \mathbf{Q} \boldsymbol{\psi}_{\mathcal{V}}(t)\right| &\le 11 \|\mathbf{Q}\| C_{\max}^2 e^{-\frac{d -\Delta}{\Delta}\log\left(\frac{d-\Delta}{2e\Delta \|\mathbf{C}\| t} \right)}\nonumber\\
    &\quad + \left| \boldsymbol{\Upsilon}_{\mathcal{V}^\perp}(t)^\dagger e^{-\mathbf{C}^\dagger t} \mathbf{Q} e^{-\mathbf{C} t} \boldsymbol{\Upsilon}_{\mathcal{V}^\perp}(t) \right|
\end{align}
The claim then immediately follows by substituting in the result of~\Cref{lem:doubleTerm} to bound the double inhomogeneous term.
\end{proof}

\subsubsection{Observable Truncation Bounds}

The previous lemma establishes that the exact dynamics of a local observable are exponentially insensitive to distant initial states and sources. To translate this physical property into an efficient classical algorithm, we must truncate the vector space. We partition the generator as $\mathbf{C} = \mathbf{C}_{\mathcal{V}} + \mathbf{C}_I$, where $\mathbf{C}_{\mathcal{V}}$ is the local generator strictly acting within $\mathcal{V}$, and $\mathbf{C}_I$ contains the severed boundary interactions coupling $\mathcal{V}$ to $\mathcal{V}^\perp$. 
Let $\tilde{\boldsymbol{\psi}}_{\mathcal{V}}(t)$ be the simulated state evolved only within $\mathcal{V}$ by $\mathbf{C}_{\mathcal{V}}$. 
We now bound the algorithmic truncation error for a local observable by directly mapping the boundary errors onto the exact partitioned expectation value bounds derived previously.

\begin{lemma}[Observable Truncation Bound] \label{thm:obsTruncation}
    Let $\mathbf{Q}$ be an 
    observable with spatial support strictly within a subspace $\mathcal{W}$ that is within distance at least $d$ from a subspace $\mathcal{V}^\perp$. Let the generator be physically $\Delta$-local and partitioned as $\mathbf{C} = \mathbf{C}_{\mathcal{V}} + \mathbf{C}_I$, where the truncated generator $\mathbf{C}_{\mathcal{V}}$ strictly acts within the local subspace $\mathcal{V}$ and $\mathbf{C}_I$ contains the interaction terms coupling $\mathcal{V}^\perp$ to $\mathcal{V}$ and further assume
    \begin{enumerate}
        \item for any error tolerance $\epsilon > 0$, define the temporal memory horizon $t^* = \frac{1}{\lambda_{\min}} \log\left(\frac{16 \|\mathbf{Q}\| C_{\max}^2}{\epsilon}\right)$
        \item let  the simulation start time be $t_0 = \max(0, t - t^*)$
        \item if $\boldsymbol{\psi}(t)$ is the exact global state and $\tilde{\boldsymbol{\psi}}_{\mathcal{V}}(t)$ is the approximate state evolved entirely within the truncated space $\mathcal{V}$ by the truncated differential equation $\partial_\tau \tilde{\boldsymbol{\psi}}_{\mathcal{V}}(\tau) = -\mathbf{C}_{\mathcal{V}}\tilde{\boldsymbol{\psi}}_{\mathcal{V}}(\tau) + \boldsymbol{\chi}_{\mathcal{V}}(\tau)$ over the interval $[t_0, t]$, with initial condition $\tilde{\boldsymbol{\psi}}_{\mathcal{V}}(t_0) = \boldsymbol{\psi}_{\mathcal{V}}(0)$ if $t_0 = 0$, and $\tilde{\boldsymbol{\psi}}_{\mathcal{V}}(t_0) = \mathbf{0}$ if $t_0 > 0$
        \item the assumptions of~\Cref{lem:exactPartition} hold,
    \end{enumerate} 
     then there exists a truncation distance $d_* \ge 0$ such that for all $d \ge d_*$, the error in the expectation value obeys
    \begin{equation}
        \left| \boldsymbol{\psi}^\dagger(t) \mathbf{Q} \boldsymbol{\psi}(t) - \tilde{\boldsymbol{\psi}}_{\mathcal{V}}^\dagger(t) \mathbf{Q} \tilde{\boldsymbol{\psi}}_{\mathcal{V}}(t) \right| \le \epsilon,
    \end{equation}
    for a spatial truncation distance $d_*$ satisfying
    \begin{equation}
    \begin{gathered}
        d_* = 2\Delta + \Delta \frac{c}{W(c/a)},
    \end{gathered}
    \end{equation}
    where $a = 2e\|\mathbf{C}\| (t-t_0)$, $c = \log\left(\frac{24 \|\mathbf{Q}\| C_{\max}^2 (1 + 2(t-t_0)\|\mathbf{C}_I\|)^2}{\epsilon}\right)$, and $W$ denotes the principal branch of the Lambert $W$ function.
\end{lemma}
\begin{proof}
    By Duhamel's principle, the exact state at time $t$ evaluated from $t_0 = \max(0, t - t^*)$ is $\boldsymbol{\psi}(t) = e^{-\mathbf{C}(t-t_0)}\boldsymbol{\psi}(t_0) + \int_{t_0}^t e^{-\mathbf{C}(t-\tau)}\boldsymbol{\chi}(\tau)\mathrm{d}\tau$. 
    Let us define a temporally truncated exact state $\check{\boldsymbol{\psi}}(t)$ that ignores the history prior to $t_0$. 
    If $t_0 = 0$, $\check{\boldsymbol{\psi}}(t) = \boldsymbol{\psi}(t)$. 
    If $t_0 > 0$, $\check{\boldsymbol{\psi}}(t) = \int_{t_0}^t e^{-\mathbf{C}(t-\tau)}\boldsymbol{\chi}(\tau)\mathrm{d}\tau$. 
    The temporal error state is $\boldsymbol{\psi}(t) - \check{\boldsymbol{\psi}}(t) = e^{-\mathbf{C}(t-t_0)}\boldsymbol{\psi}(t_0)$ for $t_0>0$. 
    Because the logarithmic norm is bounded by $-\lambda_{\min}$, this difference decays as $\|\boldsymbol{\psi}(t) - \check{\boldsymbol{\psi}}(t)\| \le 2C_{\max} e^{-\lambda_{\min} t^*}$. 
    Applying the triangle inequality, the temporal error in the expectation value is bounded by 
    \begin{equation}
        |\boldsymbol{\psi}^\dagger(t) \mathbf{Q} \boldsymbol{\psi}(t) - \check{\boldsymbol{\psi}}^\dagger(t) \mathbf{Q} \check{\boldsymbol{\psi}}(t)| \le 8\|\mathbf{Q}\| C_{\max}^2 e^{-\lambda_{\min} t^*}.
    \end{equation}
    Setting this temporal error bound to $\epsilon/2$ yields the defined memory horizon $t^* = \frac{1}{\lambda_{\min}} \log\left(\frac{16 \|\mathbf{Q}\| C_{\max}^2}{\epsilon}\right)$.
    
    Over the integration window $[t_0, t]$ of duration $(t-t_0)$, the exact differential equation for $\check{\boldsymbol{\psi}}$ can be rewritten algebraically as $\partial_\tau \check{\boldsymbol{\psi}}(\tau) = -\mathbf{C}_{\mathcal{V}}\check{\boldsymbol{\psi}}(\tau) + \boldsymbol{\chi}_{\mathcal{V}}(\tau) + \boldsymbol{\chi}_{\text{eff}}(\tau)$, where $\boldsymbol{\chi}_{\text{eff}}(\tau) = \boldsymbol{\chi}_{\mathcal{V}^\perp}(\tau) - \mathbf{C}_I \check{\boldsymbol{\psi}}(\tau)$. 
    Because $\mathbf{C}_I$ consists exclusively of boundary interactions coupling $\mathcal{V}^\perp$ to $\mathcal{V}$, the effective source $\boldsymbol{\chi}_{\text{eff}}(\tau)$ structurally acts from the distant region, supported at a distance of at least $d-\Delta$ from the observable $\mathbf{Q}$. 
    Therefore, the spatial truncation error $| \check{\boldsymbol{\psi}}^\dagger(t) \mathbf{Q} \check{\boldsymbol{\psi}}(t) - \tilde{\boldsymbol{\psi}}_{\mathcal{V}}^\dagger(t) \mathbf{Q} \tilde{\boldsymbol{\psi}}_{\mathcal{V}}(t) |$ evaluates identically to the exact partitioned observable error bounded in \Cref{lem:exactPartition}, driven by this effective distant source.
    
    Over this window, the local initial state at $t_0$ is $\check{\boldsymbol{\psi}}(t_0)$, which by definition evaluates exactly to $\boldsymbol{\psi}(0)$ or $\mathbf{0}$, guaranteeing its norm is bounded strictly by $C_{\max}$.
    The effective source is bounded by an integrated norm of $\int_{t_0}^t \|\boldsymbol{\chi}_{\text{eff}}(\tau)\| \mathrm{d}\tau \le \|\boldsymbol{\chi}\|_{L^1} + \int_{t_0}^t \|\mathbf{C}_I\| \|\check{\boldsymbol{\psi}}(\tau)\| \mathrm{d}\tau \le C_{\max} + (t-t_0)\|\mathbf{C}_I\| (2C_{\max}) = C_{\max}(1 + 2(t-t_0)\|\mathbf{C}_I\|)$. 
    Because the initial state bound $C_{\max}$ is bounded by this total source bound, the elevated maximum constant evaluates exactly to $\hat{C}_{\max} = C_{\max}(1 + 2(t-t_0)\|\mathbf{C}_I\|)$.
    Substituting $\hat{C}_{\max}$ into \Cref{lem:exactPartition} guarantees a spatial truncation error bounded by
    \begin{equation}
        12 \|\mathbf{Q}\| \hat{C}_{\max}^2 \exp\left( -\frac{d - 2\Delta}{\Delta}\log\left(\frac{d - 2\Delta}{2e\Delta \|\mathbf{C}\| (t-t_0)} \right) \right)\le \epsilon'.\label{eq:epsilonPrime}
    \end{equation}
    Setting $\epsilon'$ to $\epsilon/2$ in~\eqref{eq:epsilonPrime} and defining the dimensionless variables $x = \frac{d - 2\Delta}{\Delta}$, $a = 2e\|\mathbf{C}\|(t-t_0)$, and $c = \log\left(\frac{24 \|\mathbf{Q}\| \hat{C}_{\max}^2}{\epsilon}\right)$, the inequality reduces to $x \log(x/a) \ge c$. 
    Utilizing the principal branch of the Lambert $W$ function, the exact boundary is given by $x = \frac{c}{W(c/a)}$. 
    Substituting $x$ into the distance relation $d_* = 2\Delta + \Delta x$ yields the exact spatial truncation radius
    \begin{equation}
        d_* = 2\Delta + \Delta \frac{c}{W(c/a)},
    \end{equation}
    which bounds the spatial projection error by $\epsilon/2$. 
    By the triangle inequality, the total error evaluates to $\le \epsilon/2 + \epsilon/2 = \epsilon$.
\end{proof}

\subsubsection{State Vector Truncation Bounds}

The previous discussion examined the error on local observables in region $\mathcal{W}$ and showed that any $\Delta$-local observable can be $\epsilon$ approximated using a restricted simulation.  However, in some applications one may wish to simulate vectors in the entire space rather than simply observables.  The results that we find in this setting are similar and we see that for fixed dimension of the subspace is qualitatively the same.  However, in order to use observables to ensure that the projected dynamics is close to the actual dynamics we need to ensure that the solution vector is bounded away from zero as part of our argument.  For this reason, the following theorem requires a slightly stronger assumption of a lower bound on the solution norm as well as the upper bound postulated earlier.
Furthermore, the result is only resolved up to a global phase $\theta \in [0, 2\pi)$.

\begin{theorem}[Solution Vector Error Bound]\label{thm:solnVector}
    Let $\Pi_{\mathcal{W}}$ be a projector onto a subspace $\mathcal{W}$ that has power of $2$ dimension and is within distance at least $d$ from a subspace $\mathcal{V}^\perp$. Let the physically $\Delta$-local generator be partitioned as $\mathbf{C} = \mathbf{C}_{\mathcal{V}} + \mathbf{C}_I$, where the truncated generator $\mathbf{C}_{\mathcal{V}}$ strictly acts within the local subspace $\mathcal{V}$ and $\mathbf{C}_I$ contains the interaction terms coupling $\mathcal{V}^\perp$ to $\mathcal{V}$ and further assume
    \begin{enumerate}
        \item there exists $C_{\min}>0$ such that for all $t$ in the domain of simulation  $\|\Pi_{\mathcal{W}} \boldsymbol{\psi}(t)\| \ge C_{\min}$,
        \item for any error tolerance $\epsilon > 0$, define the temporal memory horizon $t^* \in \mathcal{O}\left( \frac{1}{\lambda_{\min}} \log\left(\frac{C_{\max}|\mathcal{W}|}{C_{\min}\epsilon}\right)\right)$
        \item let the simulation start time be $t_0 = \max(0, t - t^*)$
        \item if $\boldsymbol{\psi}(t)$ is the exact global state and $\tilde{\boldsymbol{\psi}}_{\mathcal{V}}(t)$ is the approximate state evolved entirely within the truncated space $\mathcal{V}$ by the truncated differential equation $\partial_\tau \tilde{\boldsymbol{\psi}}_{\mathcal{V}}(\tau) = -\mathbf{C}_{\mathcal{V}}\tilde{\boldsymbol{\psi}}_{\mathcal{V}}(\tau) + \boldsymbol{\chi}_{\mathcal{V}}(\tau)$ over the interval $[t_0, t]$, with initial condition $\tilde{\boldsymbol{\psi}}_{\mathcal{V}}(t_0) = \boldsymbol{\psi}_{\mathcal{V}}(0)$ if $t_0 = 0$, and $\tilde{\boldsymbol{\psi}}_{\mathcal{V}}(t_0) = \mathbf{0}$ if $t_0 > 0$
        \item all additional assumptions of~\Cref{lem:exactPartition} hold.
    \end{enumerate} 
    then there exists a truncation distance $d_* \ge 0$ such that for all $d \ge d_*$, there exists a phase $\theta \in [0, 2\pi)$ such that the error in the projected state relative to the restricted dynamics on $\mathcal{V}$ obeys
    \begin{equation}
        \left\| \Pi_{\mathcal{W}}\boldsymbol{\psi}(t) - e^{i\theta}\Pi_{\mathcal{W}}\tilde{\boldsymbol{\psi}}_{\mathcal{V}}(t) \right\| \le \epsilon,
    \end{equation}
    for a spatial truncation distance $d_*$ satisfying
    $$
    \begin{gathered}
    d_*\in \mathcal{O}\left( \Delta\|\mathbf{C}\|t + \Delta \log\left(\frac{C_{\max}|\mathcal{W}|}{C_{\min}\epsilon}\right) \right).
    \end{gathered}
    $$
\end{theorem}
\begin{proof}
    Define the projected state vectors $\mathbf{u}(t) := \Pi_{\mathcal{W}}\boldsymbol{\psi}(t)$ and $\mathbf{v}(t) := \Pi_{\mathcal{W}}\tilde{\boldsymbol{\psi}}_{\mathcal{V}}(t)$ on $\mathcal{W}$. From Lemma~\ref{thm:obsTruncation} we have that provided $d\ge d_*$ that for any $\mathbf{Q}$ supported only on $\mathcal{W}\subset \mathcal{V}$ with $\|\mathbf{Q}\|\le 1$,
    \begin{equation}
        \left| \boldsymbol{\psi}^\dagger(t) \mathbf{Q} \boldsymbol{\psi}(t) - \tilde{\boldsymbol{\psi}}_{\mathcal{V}}^\dagger(t) \mathbf{Q} \tilde{\boldsymbol{\psi}}_{\mathcal{V}}(t) \right| = \left| \mathbf{u}^\dagger(t) \mathbf{Q} \mathbf{u}(t) - \mathbf{v}^\dagger(t) \mathbf{Q} \mathbf{v}(t) \right| \le \frac{\epsilon}{|\mathcal{W}|}.
    \end{equation}
    This implies that since $|\mathcal{W}|$ is a power of $2$ we can take $\mathbf{Q}$ to be elements of a set of unitary and Hermitian operators forming a complete operator basis of $\mathcal{W}$, where $\mathbf{P}_i$ ($1 \le i \le |\mathcal{W}|^2$) is chosen to be a particular element of the basis.  For a generic operator basis of dimension $|\mathcal{W}|^2$ we see that by choosing $d\ge d_*$
    \begin{align}
        \left| \mathbf{u}^\dagger(t) \mathbf{P}_i \mathbf{u}(t) - \mathbf{v}^\dagger(t) \mathbf{P}_i \mathbf{v}(t) \right|&= \left|{\rm Tr}(\mathbf{u}(t)\mathbf{u}^\dagger(t) \mathbf{P}_i ) - {\rm Tr}(\mathbf{v}(t) \mathbf{v}^\dagger(t) \mathbf{P}_i) \right|\nonumber\\
        &=|\mathcal{W}| \left|\langle \mathbf{P}_i, \mathbf{u}(t) \mathbf{u}^\dagger(t)\rangle - \langle \mathbf{P}_i, \mathbf{v}(t) \mathbf{v}^\dagger(t)\rangle \right| \le \frac{\epsilon}{|\mathcal{W}|}
    \end{align}
    where $\langle\cdot ,\cdot\rangle$ is the normalized Hilbert-Schmidt inner product on matrices.  
    As $\mathbf{P}_i$ forms a Hermitian orthonormal operator basis for matrices on $\mathcal{W}$, this implies that
    \begin{equation}
        \mathbf{u}(t) \mathbf{u}^\dagger(t) = \sum_{i=1}^{|\mathcal{W}|^2} \mathbf{P}_i \langle \mathbf{P}_i, \mathbf{u}(t) \mathbf{u}^\dagger(t) \rangle.
    \end{equation}
    Next, this implies from the fact that $\|\mathbf{P}_i\| \le 1$ that
    \begin{align}
        \| \mathbf{u}(t) \mathbf{u}^\dagger(t) - \mathbf{v}(t) \mathbf{v}^\dagger(t)\| &=\left\| \sum_{i=1}^{|\mathcal{W}|^2} \mathbf{P}_i\left(\langle \mathbf{P}_i, \mathbf{u}(t) \mathbf{u}^\dagger(t)\rangle - \langle \mathbf{P}_i, \mathbf{v}(t) \mathbf{v}^\dagger(t)\rangle\right) \right\|\nonumber\\
        &\le \sum_{i=1}^{|\mathcal{W}|^2} \left\| \mathbf{P}_i\left(\langle \mathbf{P}_i, \mathbf{u}(t) \mathbf{u}^\dagger(t)\rangle - \langle \mathbf{P}_i, \mathbf{v}(t) \mathbf{v}^\dagger(t)\rangle\right) \right\|\nonumber\\
        &\le \sum_{i=1}^{|\mathcal{W}|^2} \left| \langle \mathbf{P}_i, \mathbf{u}(t) \mathbf{u}^\dagger(t)\rangle - \langle \mathbf{P}_i, \mathbf{v}(t) \mathbf{v}^\dagger(t)\rangle \right|\nonumber\\
        &\le \sum_{i=1}^{|\mathcal{W}|^2} \frac{\epsilon}{|\mathcal{W}|^2} = \epsilon.
    \end{align}
    This implies that the density matrix analogues of both of these solutions are $\epsilon$-close to each other in the spectral norm.  From the matrix norm inequality, the norm of $\mathbf{v}(t)$ satisfies $\|\mathbf{v}(t)\|^2 = \|\mathbf{v}(t)\mathbf{v}^\dagger(t)\| \ge \|\mathbf{u}(t)\mathbf{u}^\dagger(t)\| - \|\mathbf{u}(t)\mathbf{u}^\dagger(t) - \mathbf{v}(t)\mathbf{v}^\dagger(t)\| \ge C_{\min}^2 - \epsilon \ge \frac{1}{2}C_{\min}^2$ for sufficiently small $\epsilon$. The eigenvalue gap of the rank-$1$ operator $\mathbf{v}(t) \mathbf{v}^\dagger(t)$ is its non-zero eigenvalue $\|\mathbf{v}(t)\|^2 \ge \frac{1}{2} C_{\min}^2$. Invoking the Davis-Kahan theorem, we have that there exists a phase $\theta \in [0, 2\pi)$ such that
    \begin{equation}
        \left\| \frac{\mathbf{u}(t)}{\|\mathbf{u}(t)\|} - e^{i\theta}\frac{\mathbf{v}(t)}{\|\mathbf{v}(t)\|} \right\| \le \frac{2\sqrt{2} \epsilon}{C_{\min}^2}. 
    \end{equation}
    Further, the corresponding eigenvalues of the two principal eigenvectors differ by at most $\epsilon$, so $\left|\|\mathbf{u}(t)\| - \|\mathbf{v}(t)\|\right| \le \frac{\epsilon}{C_{\min}}$. Combining these with $\|\mathbf{u}(t)\| \le C_{\max}$ yields
    \begin{align}
        \left\| \Pi_{\mathcal{W}} \boldsymbol{\psi}(t) - e^{i\theta}\Pi_{\mathcal{W}} \tilde{\boldsymbol{\psi}}_{\mathcal{V}}(t) \right\| = \| \mathbf{u}(t) - e^{i\theta} \mathbf{v}(t) \| \in \mathcal{O}\left(\frac{\epsilon C_{\max}}{C_{\min}^2} \right).  
    \end{align}
    We then achieve the final results by substituting $\epsilon \to \epsilon \frac{C_{\min}^2}{(2\sqrt2+1)C_{\max} |\mathcal{W}|}$ into the results of \Cref{thm:obsTruncation}, shifting the parameters inside the logarithmic terms by constant factors that are absorbed into $\mathcal{O}(d_*)$.
\end{proof}

While \Cref{thm:solnVector} establishes a state vector bound via observable tomography, it relies on a non-zero lower bound $C_{\min}$ and holds only up to a global phase factor $e^{i\theta}$. To remove these assumptions for our classical algorithm (\Cref{lem:classicalSim}), we now present a direct proof that bounds the $L_2$-norm error of the locally projected state vector via submatrix Taylor expansion remainders, requiring no lower bounds or phase alignments.

\begin{theorem}[State Vector Truncation Bound] \label{thm:stateTruncation}
    Under the assumptions of~\Cref{lem:exactPartition}, let $\mathcal{W}$ be the local spatial support of the subsystem of interest. Let $\mathcal{V}^\perp$ be the distant region separated from $\mathcal{W}$ by a distance $d$. Let the physically $\Delta$-local  and $s$-sparse generator acting on $\mathbb{C}^N$ be partitioned as $\mathbf{C} = \mathbf{C}_{\mathcal{V}} + \mathbf{C}_I$, where the truncated generator $\mathbf{C}_{\mathcal{V}}$ strictly acts within the local subspace $\mathcal{V}$ and $\mathbf{C}_I$ contains the interaction terms coupling $\mathcal{V}^\perp$ to $\mathcal{V}$. 
    For any error tolerance $\epsilon > 0$, define the temporal memory horizon $t^* = \frac{1}{\lambda_{\min}} \log\left(\frac{4 C_{\max}}{\epsilon}\right)$ and the simulation start time $t_0 = \max(0, t - t^*)$. 
    If $\boldsymbol{\psi}(t)$ is the exact global state and $\tilde{\boldsymbol{\psi}}_{\mathcal{V}}(t)$ is the approximate state evolved entirely within the truncated space $\mathcal{V}$ by the truncated differential equation $\partial_\tau \tilde{\boldsymbol{\psi}}_{\mathcal{V}}(\tau) = -\mathbf{C}_{\mathcal{V}}\tilde{\boldsymbol{\psi}}_{\mathcal{V}}(\tau) + \boldsymbol{\chi}_{\mathcal{V}}(\tau)$ over the interval $[t_0, t]$, with initial condition $\tilde{\boldsymbol{\psi}}_{\mathcal{V}}(t_0) = \boldsymbol{\psi}_{\mathcal{V}}(0)$ if $t_0 = 0$, and $\tilde{\boldsymbol{\psi}}_{\mathcal{V}}(t_0) = \mathbf{0}$ if $t_0 > 0$, then there exists a truncation distance $d_* \ge 0$ such that for all $d \ge d_*$, the $L_2$-norm error of the locally projected state vector obeys
    \begin{equation}
        \|\Pi_{\mathcal{W}}\boldsymbol{\psi}(t) - \Pi_{\mathcal{W}}\tilde{\boldsymbol{\psi}}_{\mathcal{V}}(t)\| \le \epsilon,
    \end{equation}
    where $\Pi_{\mathcal{W}}$ is the orthogonal projector onto $\mathcal{W}$, for a spatial truncation distance $d_*$ satisfying
    \begin{equation}
    \begin{gathered}
    d_* = 2\Delta + \Delta \frac{c}{W(c/a)},
    \end{gathered}
    \end{equation}
    where $a = e\|\mathbf{C}\| (t-t_0)$, $c = \log\left(\frac{4 C_{\max} (1 + (t-t_0)\|\mathbf{C}_I\|)}{\epsilon}\right)$, and $W$ denotes the principal branch of the Lambert $W$ function.
\end{theorem}

\begin{proof}
    By Duhamel's principle, the exact state at time $t$ evaluated from $t_0 = \max(0, t - t^*)$ is $\boldsymbol{\psi}(t) = e^{-\mathbf{C}(t-t_0)}\boldsymbol{\psi}(t_0) + \int_{t_0}^t e^{-\mathbf{C}(t-\tau)}\boldsymbol{\chi}(\tau)\mathrm{d}\tau$. 
    Let us define a temporally truncated exact state $\check{\boldsymbol{\psi}}(t)$ that ignores the history prior to $t_0$. 
    If $t_0 = 0$, $\check{\boldsymbol{\psi}}(t) = \boldsymbol{\psi}(t)$. 
    If $t_0 > 0$, $\check{\boldsymbol{\psi}}(t) = \int_{t_0}^t e^{-\mathbf{C}(t-\tau)}\boldsymbol{\chi}(\tau)\mathrm{d}\tau$. 
    The difference between the exact state and this temporally truncated state evaluates to
    \begin{equation}
        \boldsymbol{\psi}(t) - \check{\boldsymbol{\psi}}(t) = e^{-\mathbf{C}(t-t_0)}\boldsymbol{\psi}(t_0)
    \end{equation}
    for $t_0>0$. 
    Because the logarithmic norm is bounded by $-\lambda_{\min}$, this difference decays as
    \begin{equation}
        \|\boldsymbol{\psi}(t) - \check{\boldsymbol{\psi}}(t)\| \le e^{-\lambda_{\min} t^*} \|\boldsymbol{\psi}(t_0)\| \le 2C_{\max} e^{-\lambda_{\min} t^*}.
    \end{equation}
    Setting this temporal error bound to $\epsilon/2$ yields the memory horizon $t^* = \frac{1}{\lambda_{\min}} \log\left(\frac{4 C_{\max}}{\epsilon}\right)$.
    
    Over the integration window $[t_0, t]$ of duration $(t-t_0)$, we track the spatial error $\check{\boldsymbol{\psi}}(\tau) - \tilde{\boldsymbol{\psi}}_{\mathcal{V}}(\tau)$. 
    Because the locally truncated approximation shares the identical local initial conditions with $\check{\boldsymbol{\psi}}(t_0)$, the initial spatial error evaluates to $\check{\boldsymbol{\psi}}(t_0) - \tilde{\boldsymbol{\psi}}_{\mathcal{V}}(t_0) = \boldsymbol{\psi}_{\mathcal{V}^\perp}(0)$ if $t_0 = 0$, and $\check{\boldsymbol{\psi}}(t_0) - \tilde{\boldsymbol{\psi}}_{\mathcal{V}}(t_0) = \mathbf{0}$ if $t_0 > 0$. 
    By evaluating the time derivative and applying Duhamel's principle, this spatial difference evolves according to the truncated generator $-\mathbf{C}_{\mathcal{V}}$ driven by the distant initial state, the distant source, and the boundary interactions
    \begin{equation}
        \check{\boldsymbol{\psi}}(t) - \tilde{\boldsymbol{\psi}}_{\mathcal{V}}(t) = e^{-\mathbf{C}_{\mathcal{V}}(t-t_0)}\big(\check{\boldsymbol{\psi}}(t_0) - \tilde{\boldsymbol{\psi}}_{\mathcal{V}}(t_0)\big) + \int_{t_0}^t e^{-\mathbf{C}_{\mathcal{V}}(t-\tau)} \big[ \boldsymbol{\chi}_{\mathcal{V}^\perp}(\tau) - \mathbf{C}_I \check{\boldsymbol{\psi}}(\tau) \big] \mathrm{d}\tau.
    \end{equation}
    Because $-\mathbf{C}_{\mathcal{V}}$ is a principal submatrix of $-\mathbf{C}$, it inherits a non-positive log-norm, guaranteeing that the homogeneous propagator does not amplify the norm. 
    
    To bound the spatial error that reaches the local subsystem, we multiply by the local spatial projector $\Pi_{\mathcal{W}}$. 
    Every error source component, i.e. $\check{\boldsymbol{\psi}}(t_0) - \tilde{\boldsymbol{\psi}}_{\mathcal{V}}(t_0)$, $\boldsymbol{\chi}_{\mathcal{V}^\perp}(\tau)$, and $\mathbf{C}_I \check{\boldsymbol{\psi}}(\tau)$, is supported exclusively in the distant region $\mathcal{V}^\perp$ and its boundary layer, which are separated from $\mathcal{W}$ by a distance of at least $d - \Delta$. 
    By expanding the exact Taylor series of the propagator $e^{-\mathbf{C}_{\mathcal{V}}\tau}$ up to degree $K-1$, where $K = \lfloor (d - \Delta)/\Delta \rfloor$, the polynomial terms $\mathbf{C}_{\mathcal{V}}^k$ expand the spatial support of these distant vectors by at most $k\Delta \le (K-1)\Delta \le d - 2\Delta$. 
    Because this is strictly less than the separation distance, the local projector $\Pi_{\mathcal{W}}$ identically annihilates the polynomial portion of the evolution.
    
    Therefore, the effect of the propagator on these distant components is entirely bounded by the norm of its analytical remainder. 
    Using $K! \ge (K/e)^K$ and substituting the continuous lower bound $K \ge (d - 2\Delta)/\Delta$ yields the spatial attenuation factor
    \begin{equation}
        \frac{(\|\mathbf{C}_{\mathcal{V}}\|\tau)^K}{K!} \le \left(\frac{e\|\mathbf{C}\|\tau}{K}\right)^K \le \exp\left(-\frac{d - 2\Delta}{\Delta}\log\left(\frac{d- 2\Delta}{e\Delta \|\mathbf{C}\| \tau} \right)\right).
    \end{equation}
    Applying the triangle inequality to the locally projected spatial error equation, uniformly upper-bounding the monotonically increasing exponential at $\tau=t-t_0$, and utilizing the global norm bounds $\|\check{\boldsymbol{\psi}}(t_0) - \tilde{\boldsymbol{\psi}}_{\mathcal{V}}(t_0)\| \le C_{\max}$, $\int_{t_0}^t \|\boldsymbol{\chi}_{\mathcal{V}^\perp}(\tau)\|\mathrm{d}\tau \le C_{\max}$, and $\max_\tau \|\check{\boldsymbol{\psi}}(\tau)\| \le 2C_{\max}$ established previously, we obtain
    \begin{align}
        \left\|\Pi_{\mathcal{W}} \big(\check{\boldsymbol{\psi}}(t) - \tilde{\boldsymbol{\psi}}_{\mathcal{V}}(t)\big)\right\| &\le \left( \left\|\check{\boldsymbol{\psi}}(t_0) - \tilde{\boldsymbol{\psi}}_{\mathcal{V}}(t_0)\right\| + \int_{t_0}^t \big( \|\boldsymbol{\chi}_{\mathcal{V}^\perp}(\tau)\| + \|\mathbf{C}_I\| \|\check{\boldsymbol{\psi}}(\tau)\| \big) \mathrm{d}\tau \right) \nonumber \\
        &\quad \times \exp\left(-\frac{d - 2\Delta}{\Delta}\log\left(\frac{d- 2\Delta}{e\Delta \|\mathbf{C}\| (t-t_0)} \right)\right) \nonumber \\
        &\le 2 C_{\max} (1 + (t-t_0)\|\mathbf{C}_I\|) \exp\left(-\frac{d - 2\Delta}{\Delta}\log\left(\frac{d- 2\Delta}{e\Delta \|\mathbf{C}\| (t-t_0)} \right)\right).
    \end{align}
    Setting the spatial error budget $\left\|\Pi_{\mathcal{W}} \big(\check{\boldsymbol{\psi}}(t) - \tilde{\boldsymbol{\psi}}_{\mathcal{V}}(t)\big)\right\| \le \epsilon/2$, we define the dimensionless variables $x = \frac{d - 2\Delta}{\Delta}$, $a = e\|\mathbf{C}\|(t-t_0)$, and $c = \log\left(\frac{4 C_{\max}(1+(t-t_0)\|\mathbf{C}_I\|)}{\epsilon}\right)$. 
    For $a > 0$, this bounds the error via the transcendental inequality $(x/a)\log(x/a) \ge c/a$. 
    Utilizing the principal branch of the Lambert $W$ function, defined by $W(z)e^{W(z)} = z$, the boundary satisfying the inequality is given by $x = \frac{c}{W(c/a)}$. 
    Substituting $x$ back into the distance relation $d_* = \Delta x + 2\Delta$ yields the required spatial truncation radius
    \begin{equation}
    \begin{gathered}
    d_* = 2\Delta + \Delta \frac{c}{W(c/a)},
    \end{gathered}
    \end{equation}
    which bounds the local spatial projection error by $\epsilon/2$.
    
    By the triangle inequality, the total bounded error evaluates to
    \begin{align}
        \|\Pi_{\mathcal{W}}\boldsymbol{\psi}(t) - \Pi_{\mathcal{W}}\tilde{\boldsymbol{\psi}}_{\mathcal{V}}(t)\| &\le \|\Pi_{\mathcal{W}}\big(\boldsymbol{\psi}(t) - \check{\boldsymbol{\psi}}(t)\big)\| + \|\Pi_{\mathcal{W}}\big(\check{\boldsymbol{\psi}}(t) - \tilde{\boldsymbol{\psi}}_{\mathcal{V}}(t)\big)\| \nonumber \\
        &\le \frac{\epsilon}{2} + \frac{\epsilon}{2} = \epsilon.
    \end{align}
\end{proof}

We use this result to restrict the complexity of classical local subsystem simulation: 
\begin{theorem}[Classical Complexity of Local Observable Simulation] \label{cor:classical}
    Under the assumptions of \Cref{lem:classicalSim} and \Cref{thm:stateTruncation}, assume that 
    \begin{enumerate}
        \item $\mathbf{Q}$ is a physically $\Delta$-local operator acting on a subspace $\mathcal{W}$ corresponding to a ball of radius $\mathcal{O}(|\mathcal{W}|^{1/D})$ in a $D$-dimensional lattice for constant $D$ and has unit norm, and the subsystem dimension $|\mathcal{W}|$ is $\mathcal{O}(1)$.
        \item $\mathbf{C}$ is the generator for the differential equation $\partial_t \boldsymbol{\psi}(t) = -\mathbf{C} \boldsymbol{\psi}(t)+\boldsymbol{\chi}(t)$ on $\mathbb{C}^N$ where $\mathbf{C}$ is $s$-sparse and physically $\Delta$ local with $\|\mathbf{C}\|\le \alpha$.
        \item the number of arithmetic operations $\mathsf{T}_{\boldsymbol{\chi}}$ needed to compute the local source and its derivatives scale at most linearly with the active simulated dimension $N_*$, such that $\mathsf{T}_{\boldsymbol{\chi}} \in \mathcal{O}\big(N_* \log((\|\boldsymbol{\psi}(0)\|+A_{\boldsymbol{\chi}} t)/\epsilon)\big)$.
        \item the local source obeys $\sup_{\tau\in [0,t]}\|\partial_\tau^p \boldsymbol{\chi}(\tau)\|\le A_{\boldsymbol{\chi}} \omega_{\boldsymbol{\chi}}^p p!$ for an amplitude bound $A_{\boldsymbol{\chi}}$ and frequency bound $\omega_{\boldsymbol{\chi}}$.
    \end{enumerate} 
    
     There exists a classical algorithm that can provide an $\epsilon$-approximate estimate of $\boldsymbol{\psi}^\dagger(t) \mathbf{Q} \boldsymbol{\psi}(t)$ 
    using a number of bit operations that is in
    \begin{equation}
        \tilde{\mathcal{O}}\left( s \Delta^D (\alpha + \omega_{\boldsymbol{\chi}})t\left(\alpha t + \log\left(\frac{\|\boldsymbol{\psi}(0)\|+A_{\boldsymbol{\chi}}t}{\epsilon} \right)\right)^D \log^3\left(\frac{\|\boldsymbol{\psi}(0)\| + A_{\boldsymbol{\chi}}t}{\epsilon}\right)\right).
    \end{equation}
\end{theorem}

\begin{proof}
    Our strategy is to use the truncation bounds provided previously to show that the expectation value of $\mathbf{Q}$ can be computed over a lower-dimensional space of dimension $d$.  We will then truncate the space to the minimum dimension and solve the differential equation directly in the space.  Our algorithm then computes the expectation value of $\mathbf{Q}$ from the resulting state $\boldsymbol{\psi}(t)$.
    
    Let $\tilde{\epsilon} = \sqrt{C_{\max}^2 + \epsilon} - C_{\max}$. We allocate a target error of $\tilde{\epsilon}/2$ to the spatiotemporal state truncation governed by \Cref{thm:stateTruncation} and $\tilde{\epsilon}/2$ to the numerical discretization error of the classical solver governed by \Cref{lem:classicalSim}, where $C_{\max} = \|\boldsymbol{\psi}(0)\| + A_{\boldsymbol{\chi}} t \ge \max_\tau \|\boldsymbol{\psi}(\tau)\|$.
    
    From \Cref{thm:stateTruncation}, ensuring a spatiotemporal state truncation error of $\tilde{\epsilon}/2$ requires defining the memory horizon $t^* = \frac{1}{\lambda_{\min}} \log\big(\frac{8 C_{\max}}{\tilde{\epsilon}}\big)$ and simulating a subspace $\mathcal{V}$ defined by a distance $d_*$ surrounding $\mathcal{W}$, where
    \begin{equation}
        d_* = 2\Delta + \Delta \frac{c}{W(c/a)},
    \end{equation}
    with $a = e\|\mathbf{C}\|(t-t_0)$ and $c = \log\left(\frac{8 C_{\max} (1 + (t-t_0)\|\mathbf{C}_I\|)}{\tilde{\epsilon}}\right)$. 
    For strict error tolerances $\epsilon \to 0$, the parameter $c$ grows large. 
    For $a,c>0$, the identity $c=x\log(x/a)$ with $x=c/W(c/a)$ and the inequality $\log(x/a)\ge 1-a/x$ imply $c/W(c/a)\le a+c$. 
    Using $\|\mathbf{C}_I\|\le 2\alpha$, $\log(1+2\alpha(t-t_0))\le 2\alpha(t-t_0)$, and the definition of $\tilde{\epsilon}$ gives a uniform coarse bound.
    Therefore, the asymptotic truncation radius obeys
    \begin{equation}
    \begin{gathered}
        d_* \in \mathcal{O}\left( \Delta \alpha t + \Delta \log\left(\frac{C_{\max}}{\epsilon}\right) \right).
    \end{gathered}
    \end{equation}
    Substituting the global norm bound $C_{\max} \le \|\boldsymbol{\psi}(0)\| + A_{\boldsymbol{\chi}} t$, the asymptotic radius scales as
    \begin{equation}
    \begin{gathered}
    d_* \in \mathcal{O}\left( \Delta \alpha t + \Delta \log\left(\frac{\|\boldsymbol{\psi}(0)\| + A_{\boldsymbol{\chi}} t}{\epsilon}\right) \right).
    \end{gathered}
    \end{equation}
    
    Assuming the local subsystem $\mathcal{W}$ occupies a ball of constant radius $r_0 \in \mathcal{O}(|\mathcal{W}|^{1/D})$, the simulated region $\mathcal{V}$ has a total radius of $r_0 + d_*$. 
    The effective dimension of this truncated metric space evaluates to
    \begin{equation}
        N_{*} \in \mathcal{O}\big((r_0 + d_*)^D\big) = \mathcal{O}\big(|\mathcal{W}| + d_*^D\big).
    \end{equation}
    Because $|\mathcal{W}|$ is constant, the active simulated dimension scales as $N_{*} \in \mathcal{O}(d_*^D)$.
    
    We now evaluate the cost to classically simulate the exact projected state $\tilde{\boldsymbol{\psi}}_{\mathcal{V}}(t)$ over the interval $[t_0, t]$ of duration $(t-t_0)$ on this $N_{*}$-dimensional space to within an accuracy of $\tilde{\epsilon}/2$. 
    By \Cref{lem:classicalSim}, utilizing the global bounds to uniformly bound the state at $t_0$ via $\|\tilde{\boldsymbol{\psi}}_{\mathcal{V}}(t_0)\| \le \|\boldsymbol{\psi}(0)\| + A_{\boldsymbol{\chi}} t$ and $A_{\boldsymbol{\chi}_{\mathcal{V}}} \le A_{\boldsymbol{\chi}}$, the bit complexity is
    \begin{equation}
        \mathsf{T}_{\text{bit}} = \tilde{\mathcal{O}}\left( (N_{*} s + \mathsf{T}_{\boldsymbol{\chi}}) (\alpha + \omega_{\boldsymbol{\chi}}) (t-t_0) \log^2\left(\frac{\|\boldsymbol{\psi}(0)\| + A_{\boldsymbol{\chi}} t}{\tilde{\epsilon}}\right) \right).
    \end{equation}
    Furthermore, under the assumption that the cost to evaluate the local source and its derivatives scale at most linearly with the active dimension, we have $\mathsf{T}_{\boldsymbol{\chi}}\in \Theta\big(N_* \log\big((\|\boldsymbol{\psi}(0)\| + A_{\boldsymbol{\chi}} t)/\epsilon\big)\big)$.  Thus $N_*s + \mathsf{T}_{\boldsymbol{\chi}}\in \mathcal{O}\big(N_* s \log\big((\|\boldsymbol{\psi}(0)\| + A_{\boldsymbol{\chi}} t)/\epsilon\big)\big)$.
    
    Factoring out the parameter $\Delta^D$ from the exponentiated radius and substituting $N_{*} \in \mathcal{O}(d_*^D)$ yields a final boolean scaling of
    \begin{equation}
        \mathsf{T}_{\text{bit}} = \tilde{\mathcal{O}}\left( s \Delta^D (\alpha + \omega_{\boldsymbol{\chi}}) t \left( \alpha t + \log\left(\frac{\|\boldsymbol{\psi}(0)\| + A_{\boldsymbol{\chi}} t}{\epsilon}\right) \right)^D \log^3\left(\frac{\|\boldsymbol{\psi}(0)\| + A_{\boldsymbol{\chi}} t}{\epsilon}\right) \right),
    \end{equation}
    which bounds the global computational complexity to estimate the state within an error of at most $\tilde{\epsilon}/2$.

    If the error in the resultant state $\tilde{\boldsymbol{\psi}}(t)$ compared to the exact state $\boldsymbol{\psi}(t)$ accumulates at most $\tilde{\epsilon}/2$ from spatiotemporal truncation and $\tilde{\epsilon}/2$ from numerical discretization, the projected state error $E = \|\Pi_{\mathcal{W}}(\tilde{\boldsymbol{\psi}}(t) - \boldsymbol{\psi}(t))\|$ by the triangle inequality is bounded by $\tilde{\epsilon} = \sqrt{C_{\max}^2 + \epsilon} - C_{\max}$. Because $\mathbf{Q} = \Pi_{\mathcal{W}} \mathbf{Q} \Pi_{\mathcal{W}}$ and $\|\mathbf{Q}\|\le 1$, the final observable error evaluates to 
    \begin{align}
        |\tilde{\boldsymbol{\psi}}^\dagger(t) \mathbf{Q} \tilde{\boldsymbol{\psi}}(t) - \boldsymbol{\psi}^\dagger(t) \mathbf{Q} \boldsymbol{\psi}(t)| &\le \|\Pi_{\mathcal{W}}(\tilde{\boldsymbol{\psi}}(t) - \boldsymbol{\psi}(t))\| \big(\|\Pi_{\mathcal{W}}\tilde{\boldsymbol{\psi}}(t)\| + \|\Pi_{\mathcal{W}}\boldsymbol{\psi}(t)\|\big) \nonumber \\
        &\le E \big(2\|\boldsymbol{\psi}(t)\| + E\big) \nonumber \\
        &\le \tilde{\epsilon} \big(2 C_{\max} + \tilde{\epsilon}\big) = \epsilon. 
    \end{align}
\end{proof}

\section{\BQP-completeness}
\label{app:BQP}

Here we consider the following decision problem for a system of dissipative coupled classical oscillators which aims to determine whether most of the energy of a system of coupled oscillators has dissipated or not.  While \BQP-completeness has been demonstrated for undamped systems of harmonic oscillators~\cite{babbush2023exponential}, such arguments are sensitive to the encoding of the solution used which differs in our setting.  Further, relying on such a proof strategy would at best show classical hardness of such simulations in cases where the dynamics is non-dissipative and as such any simulation problem that contains non-negligible damping would not necessarily be seen to be computationally hard. Here we provide a decision problem that will be appropriate for the \BQP-completeness argument that is not only adapted to our oscillator encoding but also includes non-negligible damping.

\begin{problem}[Dissipative oscillators decision problem]
    Using the notation and encoding introduced in~\Cref{sec:embed}, consider a system of $M = 2^n$ coupled classical oscillators with $1$ Maxwell body per connection, i.e.~$C=1$, and no body force sources, i.e.~$\mathbf{b}(t) = \bzero$. Each mass has mass equal to $1$ and is coupled to at most $4$ other masses. Each spring constant is further promised to be  upper bounded by $4$. 
    All springs are undamped in the sense that springs and dashpots only ever appear in parallel but not in series. However, there are no dashpots between pairs of masses. All masses have a spring connecting them to ground and a $1/\mathrm{polylog(M)}$-fraction of the ground connections also feature a dashpot of viscosity $\eta \in \Theta \lb \poly(n) \rb$ in parallel with the spring. 
    We are given $\poly(n)$-sized quantum circuits to efficiently implement the oracles for the spring constants, the topology matrix, and the ground-dashpot viscosities.
    The initial conditions are such that all masses are initially in their rest positions, the first mass has velocity $+1$, the second mass has velocity $-1$ and all other masses do not have any initial velocity.
    The problem is to decide whether after time $T \in \cO \lb \poly(n) \rb$, $\norm{e^{-\bC T}\ket{\psi_0}}^2$ is at least $2/3$ or at most $1/3$, under the promise that one of these holds.
\label{bqp_problem_diss}
\end{problem}

We show \BQP-completeness for the above problem by reduction from a standard \BQP-complete problem where we consider a $\poly(n)$-sized quantum circuit on $n$ qubits with initial state $\ket{0}^{\otimes n}$ and the task is to decide whether the probability of measuring the first qubit in the state $\ket{1}$ is at least $2/3$ or at most $1/3$.

\begin{theorem}[\BQP-completeness of dissipative coupled classical oscillators]
    The dissipative oscillators decision problem described in Problem~\ref{bqp_problem_diss} is \BQP-complete.
\end{theorem}

\begin{proof}
    This is a proof in two parts. In the first part, we show how to map the Feynman-Kitaev clock Hamiltonian to a system of coupled oscillators. This allows us to reduce the standard \BQP-complete problem of deciding whether the probability of measuring the first qubit of a poly-sized quantum circuit in the state $\ket{1}$ is least $2/3$ or at most $1/3$, to the problem of deciding whether the total energy of a system of coupled classical oscillators has fallen below a fixed threshold.

    In the second part, we prove bounds on the eigenvalues of a system of dissipative coupled classical oscillators. Additionally, we prove bounds on the left and right eigenvectors of the non-normal matrix governing the time evolution of the dissipative oscillators. These bounds are needed in order to argue about the simulation time, and hence computational complexity, required to distinguish a YES from a NO instance.

    \textbf{Part I: Feynman-Kitaev clock Hamiltonian and mapping to oscillators}

    This part of our proof is similar to that of Ref.~\cite{babbush2023exponential}. In particular, we also utilize the Feynman-Kitaev construction for embedding a quantum circuit in a Hamiltonian. We consider a universal gate set consisting of the Hadamard and the Toffoli gate (see~\cite{aharonov2003simple} for an elementary proof of universality). Following Ref.~\cite{babbush2023exponential}, we also include the $X$ gate in our gate set for simplicity. Since Toffoli and $X$ can be combined to construct $\mathtt{SWAP}$ gates, we can then assume, without loss of generality, that all Hadamard gates act on the same qubit. Further, we may assume without loss of generality that the quantum circuit we wish to simulate circuit does not contain consecutive Hadamard gates since the Hadamard gate is self-inverse, i.e.~$H^2 = \one$. Note that all three gates, Hadamard, Toffoli and $X$, are Hermitian. Now consider a quantum circuit on $q$ qubits,
    \begin{equation}
        V := U_L \cdots U_1,
    \end{equation}
    where each $U_l$ with $l \in \{1, 2, \dots, L\}$ is either a Hadamard, a Toffoli or an $X$ gate and the total number of gates is polynomial in the number of qubits, i.e. $L \in \cO \lb \poly(q) \rb$.
    The success probability of an arbitrary \BQP circuit, which accepts a YES instance with probability at least $2/3$, can always be boosted to $1 - 2^{-\Theta(k)} =  1 - \epsilon'$ for some $\epsilon' \geq 0$ by running $k$ copies of the circuit in parallel and taking a majority vote on the output of the first qubits of all copies. For notational simplicity, we will assume that $V$ is already such a boosted circuit. For our purposes, it will suffice to have $\epsilon' \leq 1/6$.
    Recall the standard Feynman-Kitaev clock Hamiltonian for encoding the quantum circuit $V$:
    \begin{equation}
        \bH_\cl = \sum_{l=1}^L \lb \ketbra{l+1}{l} + \ketbra{l}{l+1} \rb \otimes U_l,
    \end{equation}
    where the first qubit register is the clock register and the second qubit register is the computational register.
    Without loss of generality, we can assume that the last gate $U_L$ performs a Toffoli gate targeted onto the last qubit of the computational register, which we will call the ``answer qubit''. Further, we can assume that no other gate acts on this answer qubit. $U_L$ leaves the answer qubit in the state $\ket{0}$ for NO instances with probability at least $1 - \epsilon' \geq 5/6$ and flips it to $\ket{1}$ with probability at least $1 - \epsilon' \geq 5/6$ for YES instances. 

    Let us now discuss how to map the clock Hamiltonian $\bH_\cl$ to coupled classical oscillators using the encoding from the main text, which differs somewhat from the original encoding used in Ref.~\cite{babbush2023exponential}.
    Note that all matrix elements of the clock Hamiltonian $\bH_\cl$ are in the set $\{-1/\sqrt{2}, 0, 1/\sqrt{2}, 1 \}$. We cannot directly encode $\bH_\cl$ into a system of coupled oscillators due to two hurdles: $\bH_\cl$ is not positive semi-definite and some of its off-diagonal entries are positive. A direct encoding would therefore require negative spring constants. The first issue can be fixed by considering the following shifted clock-Hamiltonian instead: 
    \begin{equation}
        \bH_\cl' := 4\one - \bH_{\cl},
    \end{equation}
    which is positive definite and 4-sparse. It is 4-sparse as opposed to 5-sparse since all our Hadamard gates act on the same qubit and we do not have consecutive Hadamard gates in our circuit as explained above. Despite being positive definite now, $\bH_\cl'$ can still have positive off-diagonal entries due to some of the $U_l$'s being Hadamard gates.
    Following Ref.~\cite{babbush2023exponential} we can circumvent this problem by utilizing a resource state. Specifically, we add one ancilla qubit initialized to $\ket{-} = \frac{1}{\sqrt{2}} \lb \ket{0} - \ket{1} \rb$ to our circuit. Whenever we would like to implement a Hadamard gate, we instead apply the following matrix to the target qubit and the ancilla qubit:
    \begin{equation}
        H_{\mathrm{enc}} := \frac{1}{\sqrt{2}} 
        \begin{pmatrix}
            \one_2 & \one_2 \\
            \one_2 & X
        \end{pmatrix}.
    \end{equation}
    This matrix is not unitary but on the subspace where the additional ancilla qubit is in the state $\ket{-}$, it acts like a Hadamard gate.
    Thus, we ultimately work with the following modified clock Hamiltonian:
    \begin{equation}
        \bH_\cl'' := 4\one - \sum_{l=1}^L \lb \ketbra{l+1}{l} + \ketbra{l}{l+1} \rb \otimes W_l,
    \end{equation}
    where $W_l$ is in the set $\{ X \otimes \one, \mathtt{Toff} \otimes \one,  H_{\mathrm{enc}} \}$, where the identity attached to $X$ and $\mathtt{Toff}$ indicates that they act trivially on the ancilla qubit used for $H_{\mathrm{enc}}$. By construction, applying $\bH_\cl''$ to a state of the form $\ket{\psi}_{q}\ket{-}$ is identical to applying $\bH_\cl' \otimes \one$ to the same state $\ket{\psi}_{q}\ket{-}$. 

    Let us now discuss how to map $\bH_\cl''$ to a system of coupled oscillators. First, we set all masses equal to $1$ for simplicity, so the mass matrix $\mathbf{M}$ is just an identity matrix.
    Following the notation in the main text, we will then represent $\bH_\cl''$ as $\mathbf{T}^\dagger \mathbf{K} \mathbf{T}$. It may be helpful to recall that $\mathbf{T}^\dagger \mathbf{K} \mathbf{T}$ is just a convenient decomposition of a weighted undirected graph Laplacian. In particular, all off-diagonal elements of a graph Laplacian describing a graph with positive edge weights are non-positive.
    In our case, we consider only one Maxwell body per connection, so $\bT = \bT_0$ and $\bK = \bK_1$. Note that $\mathbf{T}^\dagger \mathbf{K} \mathbf{T}$ is the upper left block of $\mathbf{H}^2$ where $\bH$, given in Eq.~\eqref{H_expanded}, captures the oscillatory part of the dynamics. For convenience, we give the block form of $\bH$ again below:
    \begin{equation}
        \mathbf{H} = 
        \begin{bmatrix}
            \mathbf{0} & -i\mathbf{T}^\dagger\mathbf{K}^{1/2} \\
            i\mathbf{K}^{1/2}\mathbf{T} & \mathbf{0}
        \end{bmatrix}.
    \end{equation}
    This means that we do not directly map $\bH_\cl''$ to the oscillatory part, $\bH$, of the generator matrix $-\bC$. The reason is that the form of $\bH$ itself is too restricted to be able to directly emulate $\bH_\cl''$. For example, note that $\bH$ has a very specific block-off-diagonal structure which does not directly capture the structure of $\bH_\cl''$. Instead, one can think of $\bH$ as a square root of $\bH_\cl''$.

    Now, to map $\bH_\cl''$ to $\mathbf{T}^\dagger \mathbf{K} \mathbf{T}$, note that $\bH''_\cl$ acts on a space of dimension $(L+1) \times 2^{q+1}$ where the extra qubit is coming from the resource state $\ket{-}$ used for encoding the Hadamard gate. Let us now show how $\bH_\cl''$ can be encoded in $M = 2^n$ oscillators where
    \begin{equation}
        n = \left\lceil \log_2 \lb (L+1) \times 2^{q+1} \rb \right\rceil.
    \end{equation}
    If $M>(L+1)\times2^{q+1}$, we can pad the oscillator system with stationary oscillators. Specifically, we add $M-(L+1)\times2^{q+1}$ oscillators of mass $1$, each connected only to ground by a spring of stiffness $4$, with no dashpot and zero initial displacement and velocity. These added oscillators remain at rest. All subsequent matrices, vectors, and oscillator labels in the reduction are restricted to the original invariant component of $(L+1)\times2^{q+1}$ oscillators.
    It will be useful to label the oscillators by a tuple $(l,r)$ with $l \in \{1, 2, \dots, L+1\}$ and $r \in \{1, 2, \dots, 2^{q+1} \}$.
    We start with the off-diagonal elements. Due to the use of the encoded Hadamard gate, $H_{\mathrm{enc}}$, all off-diagonal elements of $\bH''_\cl$ are non-positive. Thus, any off-diagonal matrix element $h''_{mn}$ of $\bH''_\cl$ can be directly encoded as an inter-mass spring with stiffness $-h''_{mn} \geq 0$ between masses $m$ and $n$. 
    Next, we deal with the diagonal terms $h''_{mm}$ which are all equal to $4$. This can be achieved by connecting the mass labeled $m$ with a spring of stiffness $4 - \sum_{n\neq m} |h''_{mn}| \geq 0$ to the ground.
    Note that $\sum_{n\neq m} |h''_{mn}| \leq 1 + 2/\sqrt{2} < 4$ due to $\bH''_\cl$ being 4-sparse as discussed before.
    With this mapping, all masses are connected to ground via one spring each. 
    
    So far, all connections between masses or masses and ground consist only of single springs without any dashpots.
    The last step is to add dissipation via carefully chosen dashpots. We only place dashpots in parallel to certain ground springs. Specifically, we place dashpots on masses with label $l=L+1$, corresponding to oscillators who are representing the last state of the clock register, and simultaneously have $r$ corresponding to the answer qubit being in the state $\ket{1}$.
    These viscosities can be computed efficiently by testing whether the clock label is $L+1$ and the answer bit is $1$, returning $\eta$ when both conditions hold and zero otherwise.

    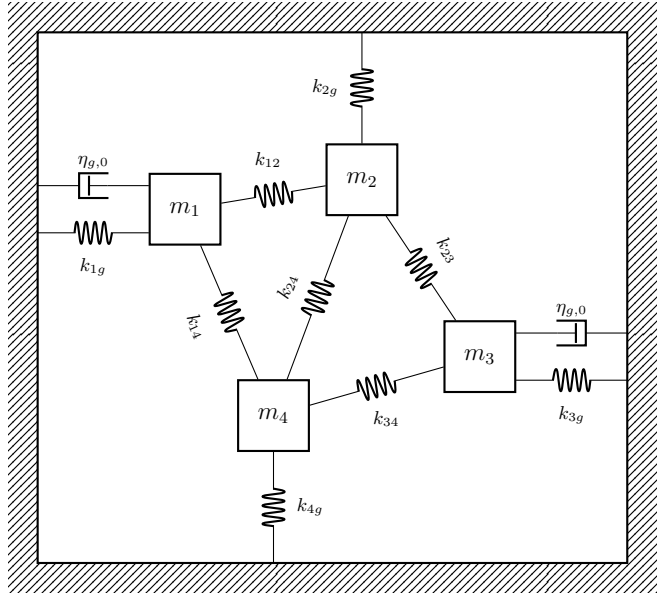
\begin{figure}[tbp]
    \centering
    \begin{tikzpicture}[scale=0.78, transform shape]
    
    \ctikzset{
        bipoles/length=1.2cm,
        bipoles/generic/width=0.60,
        bipoles/generic/height=0.25
    }
    
    \draw[pattern=north east lines, draw=none] (-0.5, -0.5) rectangle (10.5, 9.5);
    \fill[white] (0, 0) rectangle (10, 9);
    \draw[thick] (0, 0) rectangle (10, 9);
    
    \node[minimum size=1.2cm, draw, thick, fill=white] (m1) at (2.5, 6.0) {\large $m_1$};
    \node[minimum size=1.2cm, draw, thick, fill=white] (m2) at (5.5, 6.5) {\large $m_2$};
    \node[minimum size=1.2cm, draw, thick, fill=white] (m3) at (7.5, 3.5) {\large $m_3$};
    \node[minimum size=1.2cm, draw, thick, fill=white] (m4) at (4.0, 2.5) {\large $m_4$};
    
    \draw ([yshift=0.4cm]m1.west) to[damper, l_=$\eta_{g,0}$] (0, 6.4);
    \draw ([yshift=-0.4cm]m1.west) to[spring=$k_{1g}$] (0, 5.6);
    
    \draw (m2.north) to[spring=$k_{2g}$] (5.5, 9.0);
    
    \draw ([yshift=0.4cm]m3.east) to[damper=$\eta_{g,0}$] (10.0, 3.9);
    \draw ([yshift=-0.4cm]m3.east) to[spring, l_=$k_{3g}$] (10.0, 3.1);
    
    \draw (m4.south) to[spring=$k_{4g}$] (4.0, 0);

    \draw (m1) to[spring=$k_{12}$] (m2);
    \draw (m2) to[spring=$k_{23}$] (m3);
    \draw (m3) to[spring=$k_{34}$] (m4);
    \draw (m4) to[spring=$k_{14}$] (m1);
    
    \draw (m4) to[spring=$k_{24}$, pos=0.4] (m2);
    
    \end{tikzpicture}
    \caption{Example mapping of \BQP-hard instances to a network of dissipative oscillators. Dissipation only happens as viscous damping through ground connectors.}
    \label{fig:bqp-hard_networks}
    \end{figure}

    \textbf{Part II: Spectral analysis}
    
    To simplify the analysis, we  consider the second-order differential equation for the position/displacement vector $\mathbf{x} \in \mathbb{C}^{(L+1)2^{q+1}}$ which can be derived from Eqs.~\eqref{velocity_de} and \eqref{eq:maxwell_element_short} displayed below again for convenience:
    \begin{align}
        \frac{\mathrm{d} \mathbf{v}(t)}{\mathrm{d} t} &= \mathbf{M}^{-1}\Big(\mathbf{b}_0(t) - \mathbf{T}_0^{\dagger}\big(\boldsymbol{\eta}_{0}\mathbf{T}_0\mathbf{v}(t) + \sum_{c=1}^{C}{\mathbf{f}}_{c}(t)\big)\Big)  \\
        \frac{\mathrm{d} {\mathbf{f}}_{c}(t)}{\mathrm{d} t} &= {\mathbf{K}}_{c}  (\mathbf{T}_0 \mathbf{v}(t) + {\mathbf{b}}_{c}(t)) - \boldsymbol{\Lambda}_c{\mathbf{f}}_{c}(t).
    \end{align}
    In our case, all masses are equal to $1$, $\mathbf{b}_0 = \mathbf{b}_1 = \mathbf{0}$ (no source terms), $C=1$ (only one Maxwell element) and $\boldsymbol{\Lambda_{c}} = \mathbf{0}$ (no dashpots in series with springs). For simplicity, we will also assume that $\bx(0) = \bzero$ and $\mathbf{f}(0) = \bzero$.
    Integrating the second differential equation with these simplifications, we then have that
    \begin{equation}
        \mathbf{f}(t) = \bK \bT \bx(t).
    \end{equation}
    Plugging this into the first differential equation, we obtain the following second-order differential equation for the position vector $\bx$:
    \begin{align}
         \frac{\mathrm{d}^2 \mathbf{x}(t)}{\mathrm{d} t^2} &= \frac{\mathrm{d} \mathbf{v}(t)}{\mathrm{d} t} =  -\mathbf{T}^{\dagger} \boldsymbol{\eta}_{0}\mathbf{T}\mathbf{v}(t) -  \mathbf{T}^{\dagger}{\mathbf{f}}(t) \nn
         &= -\boldsymbol{\eta}_{\mathbf{M}}  \frac{\mathrm{d} \mathbf{x}(t)}{\mathrm{d} t} - \mathbf{T}^{\dagger} \mathbf{K}  \mathbf{T} \mathbf{x}(t).\label{eq:2orderdiffeq}
    \end{align}
    Ultimately, the goal is to find the eigenvalues associated with the above second-order differential equation. As we will show, all eigenvalues in the YES sector have a negative real part. We are interested in determining the largest real part (i.e.~closest to $0$) of any eigenvalue in this sector since this determines the slowest decay rate which in turn dictates the amount of time evolution necessary to distinguish a YES from a NO instance.

    To simplify the problem, first recall the specific form of $\boldsymbol{\eta}_{\mathbf{M}}$ from~\eqref{eq:etaM} and $\mathbf{T}^{\dagger} \mathbf{K} \mathbf{T}$ corresponding to the encoding of the clock Hamiltonian~$\bH_{\cl}''$:
    \begin{align}
        \boldsymbol{\eta}_{\mathbf{M}} &= \eta \ketbra{L+1}{L+1} \otimes \one_{2^{q-1}} \otimes \ketbra{1}{1} \otimes \one_2 \\
        \mathbf{T}^{\dagger} \mathbf{K}  \mathbf{T} &= 4 \one - \sum_{l=1}^L \lb \ketbra{l+1}{l} + \ketbra{l}{l+1} \rb \otimes W_l.
    \end{align}
    Also recall that the $W_l$ are not necessarily unitary on the entire space due to the encoded Hadamard gate. They are unitary however on the subspace where the ancilla qubit is in the state $\ket{-}$. For the spectral analysis, it does not matter whether we analyze $\bH_{\cl}'$ or $\bH_{\cl}''$ since they have the same spectrum on the relevant subspace. For simplicity we therefore choose to work with $\bH_{\cl}'$ (which acts on one qubit less than $\bH_{\cl}''$) going forward. Note that
    \begin{align}
        \norm{\boldsymbol{\eta}_{\mathbf{M}}} &= \eta \in \Theta \lb \poly(n) \rb \\
          \norm{\bH_{\cl}'} &= \norm{4 \one - \sum_{l=1}^L \lb \ketbra{l+1}{l} + \ketbra{l}{l+1} \rb \otimes U_l}\le 6.
    \end{align}
    In particular, we consider $\eta \gg \|\mathbf{H}'_{\rm cl}\|$.

    Further, define the following matrices:
    \begin{align}
        \bS &:= \sum_{l=0}^L \ketbra{l+1}{l+1} \otimes U_l \cdots U_0  \\
        \bJ &:= \sum_{l=1}^L \lb \ketbra{l+1}{l} + \ketbra{l}{l+1} \rb,
    \end{align}
    with $U_0 = \one$.
    Rotating $\bx$ by $\bS^\dagger$ and defining $\tbx := \bS^\dagger \bx$, we then have from~\eqref{eq:2orderdiffeq} that
    \begin{align}
         \frac{\mathrm{d}^2 \tbx(t)}{\mathrm{d} t^2} = - \eta \lb \bPi_{L+1} \otimes \bPi_y \rb \frac{\mathrm{d} \tbx(t)}{\mathrm{d} t} - \lb 4 \one - \bJ \otimes \one \rb \tbx(t),
    \label{second_order_ODE}
    \end{align}
    where
    \begin{align}
        \bPi_{L+1} &:= \ketbra{L+1}{L+1} \\
        \bPi_y &:= U_1 \cdots U_L \lb \one_{2^{q-1}} \otimes \ketbra{1}{1} \rb U_L \cdots U_1,
    \end{align}
    i.e.~$\bPi_{L+1}$ is the projector onto the last clock state and $\bPi_y$ is the projector onto the set of initial states whose answer qubit is set to $\ket{1}$ at the end of the computation (i.e.~the YES sector). Since the YES and NO sectors remain invariant under time evolution, we can analyze both cases separately. More specifically, $\bPi_y$ and $\one$ commute, so we can choose the eigenspaces of $\bPi_y$, which are the YES and NO sectors, as a basis.
    For the NO sector, the dissipative term $- \eta \lb \ketbra{L+1}{L+1} \otimes \bPi_y \rb$ does not activate due to the projector $\bPi_y$ evaluating to $0$. Therefore, the dynamics on the NO sector are unitary.

    Let us now consider the YES sector. We use the standard ansatz for the normal modes of the clock register: $\mathbf{u}(t) = e^{\lambda t} \balpha$, where $\balpha \in \mathbb{C}^{(L+1)}$  is a vector of constant coefficients and $\lambda$ is an eigenvalue. This ansatz yields the following quadratic eigenvalue equation:
    \begin{align}
        \lambda^2 \balpha = - \lambda \eta \ketbra{L+1}{L+1}  \balpha - \lb 4\one - \sum_{l=1}^L \lb \ketbra{l+1}{l} + \ketbra{l}{l+1} \rb \rb \balpha.
    \end{align}
    We can now attempt to solve the above quadratic eigenvalue problem entry-by-entry for the components $\alpha_k$:
    \begin{align}
        k=1:& \quad \lambda^2 \alpha_1 = - 4\alpha_1 + \alpha_2 \\
        2 \leq k \leq L:& \quad  \lambda^2 \alpha_k = -4 \alpha_k +\alpha_{k+1} + \alpha_{k-1} \\
        k=L+1:& \quad \lambda^2 \alpha_{L+1}  = - \lambda \eta  \alpha_{L+1} - 4 \alpha_{L+1} + \alpha_{L}. 
    \end{align}
    We use the following standard ansatz for eigenvalues and amplitudes of a 1d nearest neighbor chain: 
    \begin{align}
        \lambda^2 &= - (4 - 2 \cos(\omega)) ,\label{eigval_ansatz}\\
        \alpha_k &= \sin (k \omega) \label{eigvec_ansatz}.
    \end{align}
    This ansatz satisfies the boundary condition at $k=1$:
    \begin{equation}
        (\lambda^2+4) \alpha_1 = 2 \cos (\omega) \sin(\omega) = \sin (2\omega) = \alpha_2.
    \end{equation}
    Using the fact that 
    \begin{equation}
        \sin(a) + \sin(b) = 2 \sin \lb \frac{a+b}{2} \rb \cos \lb \frac{a-b}{2} \rb,
    \label{sine_add}
    \end{equation} 
    it is also straightforward to see that the ansatz solves the bulk equations ($2 \leq k \leq L$) since
    \begin{equation}
        \sin((k+1) \omega) + \sin((k-1)\omega) = 2 \sin(k\omega) \cos(\omega).
    \end{equation}
    So far, $\omega$ can be any complex number. The boundary condition at $k=L+1$ restricts the allowed values of $\omega$ to a finite set of modes $\{\omega_j\}_{j=1}^{2L+2}$.
    In particular, we require that
    \begin{equation}
        2\cos(\omega) \sin((L+1)\omega) = -\lambda \eta \sin((L+1)\omega) + \sin(L\omega).
    \end{equation}
    Using again Eq.~\eqref{sine_add}, we obtain the following transcendental equation for $\omega$:
    \begin{equation}
        \sin((L+2)\omega) = -\lambda \eta \sin((L+1)\omega).
    \end{equation}
    Throughout the following analysis, we choose $\eta = \eta_0 (L+1)^3$, where $\eta_0$ is a sufficiently large positive constant independent of $L$. As we will show, this choice guarantees that all constants implicit in the remainder estimates below are independent of $L$ and of the mode index $j$, thus allowing us to safely ignore the remainder terms.
    Defining $\delta := 1/\eta \ll 1$, we perform perturbation theory in terms of $\delta$ to find an expression for $\omega$. Specifically, we use the following power series ansatz:
    \begin{equation}
        \omega = \omega^{(0)} + \delta \omega^{(1)} + \delta^2 \omega^{(2)} + \dots
    \end{equation}
    The transcendental equation in terms of $\delta$ reads
    \begin{equation}
        \delta \sin((L+2)\omega) = -\lambda \sin((L+1)\omega).
    \label{transcendental}
    \end{equation}
    At $0$th order, the left-hand side is equal to $0$, so
    \begin{equation}
        0 = -\lambda \sin((L+1)\omega) = \pm \sqrt{-(4-2\cos(\omega))} \sin((L+1)\omega).
    \end{equation}
    The above can be satisfied by having either
    \begin{align}
        \sin((L+1)\omega) &= 0, \quad \text{or} \label{sin_sols}\\
        -\lambda^2 = 4-2\cos(\omega) &= 0. \label{cos_sol}
    \end{align}
    The first condition, Eq.~\eqref{sin_sols} yields $L$ solutions,

    \begin{equation}
        \omega_j^{(0)} = \frac{\pi j}{L+1}, \quad j \in \{1, 2, \dots, L \},
    \end{equation}
    which then gives the following $2L$ eigenvalues at $0$th order in $\delta$:
    \begin{equation}
        \lambda_{j, \pm}^{(0)} = \pm i \sqrt{4 - 2\cos \lb \omega_j^{(0)} \rb} = \pm i \sqrt{4 - 2\cos \lb \frac{\pi j}{L+1} \rb}.
    \end{equation}
    Note that these $0$th order eigenvalues are purely imaginary, meaning there is no dissipation at this order.
    At $1$st order, we have the following:
    \begin{align}
         \sin \lb (L+2)\omega_j^{(0)} \rb &= \sin \lb (L+2)\omega_j^{(0)} \rb \nn
         &= \sin  \lb \pi j + \frac{\pi j}{L+1} \rb \nn
         &= (-1)^j \sin  \lb \frac{\pi j}{L+1} \rb \nn
         &{=} -\lambda_{j, \pm}^{(0)} (-1)^j (L+1) \omega_{j}^{(1)} ,
    \end{align}
    where the last equality follows from noting that the quantity that we're computing is proportional to the first order shift in $\omega_j$.  Thus solving for the value of the shift implies  that
    \begin{equation}
        \omega_{j, \pm}^{(1)} = -\frac{\sin  \lb \frac{\pi j}{L+1} \rb}{(L+1) \lambda_{j, \pm}^{(0)}},
    \end{equation}
    which is purely imaginary due to $\lambda_{j, \pm}^{(0)}$ in the denominator.
    This leads to the following expression for the ($2L$ out of $2L+2$) eigenvalues up to $1$st order in $\delta$:
    \begin{equation}
    \begin{split}
        \lambda_{j, \pm} &= \pm i \sqrt{4 - 2 \cos \lb \omega_j^{(0)} + \delta  \omega_{j, \pm}^{(1)} \rb} + \cO \lb \delta^2 \rb \\
        &= \pm i \sqrt{4 - 2\cos \lb \omega_j^{(0)} \rb} \pm i \frac{\delta \omega_{j, \pm}^{(1)} \sin \lb \omega_j^{(0)} \rb}{\sqrt{4- 2 \cos \lb \omega_j^{(0)} \rb}} + \cO \lb \delta^2 \rb \\
        &= \pm i \sqrt{4 - 2\cos \lb \frac{\pi j}{L+1} \rb} - \frac{\delta \sin^2 \lb \frac{\pi j}{L+1} \rb}{(L+1) \lb 4 - 2 \cos \lb \frac{\pi j}{L+1} \rb \rb} + \cO \lb \delta^2 \rb.
    \end{split}
    \label{osc_modes_bound}
    \end{equation}
    In order to be able to safely ignore the remainder term, we need to ensure that the constants hidden in the expression $\cO \lb \delta^2 \rb$ do not scale badly with the other parameters $L$ and $j$. The analytic implicit function theorem provides this guarantee. To be more explicit, let us define the following rescaled correction to the $0$th order term of $\omega$:
    \begin{equation}
        \zeta := (L+1) \lb \omega - \omega^{(0)}_j \rb.
    \end{equation}
    Then we can rewrite the transcendental equation in Eq.~\eqref{transcendental} as follows:
    \begin{equation}
        \mathcal{F} \lb \delta, \zeta \rb := \delta \sin \lb \omega_j^{(0)} + \zeta + \frac{\zeta}{L+1} \rb + \lambda_{\pm} \lb \zeta \rb \sin \lb \zeta \rb = 0,
    \end{equation}
    where $\lambda_{\pm} = \pm i \sqrt{4 - 2 \cos(\omega)} =  \pm i \sqrt{4 - 2 \cos \lb \omega_j^{(0)} + \frac{\zeta}{L+1} \rb}$.
    Note that $\mathcal{F}$ is complex analytic in a fixed neighborhood around the point $(\delta, \zeta) = (0,0)$.
    Additionally, the derivative of $\mathcal{F}$ w.r.t.~$\zeta$ at the point $(\delta, \zeta) = (0,0)$ is equal to $\lambda_{\pm}(\omega_j^{(0)})$ and its absolute value is bounded away from $0$. Specifically,
    \begin{equation}
        \Bigg| \partial_\zeta \mathcal{F} \big|_{\delta = \zeta =  0} \Bigg| \geq \sqrt{2}.
    \end{equation}
    The analytic implicit function theorem therefore gives a common radius of convergence and uniformly bounded Taylor remainders. In particular,
    \begin{align}
        \zeta  &= - \delta \frac{\sin \lb \omega_j^{(0)} \rb}{\lambda_{j,\pm}^{(0)}} + \cO \lb \delta^2 \rb \\
        \implies \omega_{j,\pm} &= \omega_j^{(0)} - \delta \frac{\sin \lb \omega_j^{(0)} \rb}{(L+1) \lambda_{j,\pm}^{(0)}} + \cO \lb\frac{\delta^2}{L+1} \rb,
    \end{align}
    with the constants in the remainder term independent of $L$.
    This proves the uniform $\cO \lb \delta^2 \rb$ remainder in Eq.~\eqref{osc_modes_bound}. Since
    \begin{equation}
        \sin \lb \omega_j^{(0)} \rb \geq \frac{2}{L+1}, \quad 4 - 2 \cos \lb  \omega_j^{(0)} \rb \leq 6,
    \end{equation}
    the absolute value of the first-order real part of $\lambda_{j,\pm}$ obeys
    \begin{equation}
       \frac{\delta \sin^2 \lb \frac{\pi j}{L+1} \rb}{(L+1) \lb 4 - 2 \cos \lb \frac{\pi j}{L+1} \rb \rb} \geq \frac{2 \delta}{3 (L+1)^3}.
    \end{equation}
    Our choice of $\eta = \eta_0 (L+1)^3$ ensures that the $\cO \lb \delta^2 \rb$ remainder is at most $\frac{\delta}{3 (L+1)^3}$ for sufficiently large $\eta_0$. Hence,
    \begin{equation}
        -\mathrm{Re}\lb\lambda_{j,\pm}\rb\geq\frac{1}{3\eta (L+1)^3}\in\Omega\lb\frac{1}{\eta L^3}\rb.
    \end{equation}

    Let us now discuss the remaining two eigenvalues. Eq.~\eqref{cos_sol} yields one solution at $0$th order:
    \begin{align}
        \lambda_{\mathrm{mid}}^{(0)} &= 0 \\
        \omega_{\mathrm{mid}}^{(0)} &= \pm i \log \lb 2 + \sqrt{3} \rb.
    \end{align}
    Note that the sign of the expression for $\omega_{\mathrm{mid}}^{(0)}$ does not matter as either branch yields the same eigenvalue, so we will choose the ``+'' branch going forward.
    To find the $1$st-order contribution to $\lambda_{\mathrm{mid}}$, it will be useful to directly consider the power series expansion of $\lambda_{\mathrm{mid}}$ in terms of $\delta$ rather than going through a series expansion of $\omega$ as we did before. The reason this would be problematic is because we would need to expand the square root function near $0$ which is not well defined.
    Plugging the ansatz $\lambda_{\mathrm{mid}} = \tilde{\lambda}_{\mathrm{mid}}^{(0)} + \delta \tilde{\lambda}_{\mathrm{mid}}^{(1)} + \dots$ into the transcendental equation in Eq.~\eqref{transcendental} and eliminating all terms apart from the $1$st-order contributions, we see that
    \begin{equation}
        \sin  \lb (L+2) \omega_{\mathrm{mid}}^{(0)} \rb = - \tilde{\lambda}_{\mathrm{mid}}^{(1)} \sin \lb (L+1) \omega_{\mathrm{mid}}^{(0)} \rb.
    \end{equation}
    Hence,
    \begin{equation}
        \tilde{\lambda}_{\mathrm{mid}}^{(1)} = - \frac{ \sin  \lb (L+2) \omega_{\mathrm{mid}}^{(0)} \rb}{ \sin \lb (L+1) \omega_{\mathrm{mid}}^{(0)} \rb} = - \frac{ \sinh \lb (L+2) \log \lb 2 + \sqrt{3} \rb \rb}{ \sinh \lb (L+1) \log \lb 2 + \sqrt{3} \rb \rb}.
    \end{equation}
    In the limit of large $\eta$ and large $L$, we then have the following expression for $\lambda_{\mathrm{mid}}$ to first order:
    \begin{equation}
        {\lambda_{\mathrm{mid}}^{(1)}} = \tilde{\lambda}_{\mathrm{mid}}^{(0)} + \delta \tilde{\lambda}_{\mathrm{mid}}^{(1)}  \sim -\frac{2+\sqrt{3}}{\eta}. 
    \end{equation}
    The analytic implicit function theorem again allows us to  bound the remainder uniformly. In particular, we have that
    \begin{equation}
        \lambda_{\mathrm{mid}} = \delta \tilde{\lambda}_{\mathrm{mid}}^{(1)} + \cO \lb \delta^3 \rb,
    \end{equation}
    with the constants of the remainder term independent of $L$.
    This eigenvalue leads to significantly faster decay than the slowest decaying eigenvalues which have real part scaling like $- \eta^{-1} L^{-3}$ as discussed above.

    There should be a total of $2L+2$ eigenvalues but we have only discussed $2L+1$ so far. The last missing eigenvalue can be located via the Gershgorin circle theorem. To simplify the argument, we consider a system of first-order ODEs obtained from the system of second-order ODEs in Eq.~\eqref{second_order_ODE} by introducing $\tilde{\bv}(t) := \dot{\tilde{\bx}}(t)$ as a new variable:
    \begin{equation}
        \frac{d}{dt}
        \begin{pmatrix}
            \tilde{\bx} \\
            \tilde{\bv}
        \end{pmatrix}
        =
        \underbrace{
        \begin{pmatrix}
            \mathbf{0} & \one \\
            -\lb 4 \one - \bJ \otimes \one \rb & -\eta \lb \bPi_{L+1} \otimes \bPi_y \rb
        \end{pmatrix}
        }_{=: \bG_{\mathrm{ord:} 1}}
        \begin{pmatrix}
            \tilde{\bx} \\
            \tilde{\bv}
        \end{pmatrix}.
    \end{equation}
    Since the eigenvalues of the above matrix $\bG_{\mathrm{ord:} 1}$ are exactly the same eigenvalues $\{\lambda\}$ of the second-order system of ODEs in Eq.~\eqref{second_order_ODE}, we apply the Gershgorin circle theorem to one $(2L+2)$-dimensional YES clock block, where $\bPi_y$ acts as the identity. Specifically, for $\eta>12$, the disc of radius $6$ centered at $-\eta$ is disjoint from the other eigenvalue discs, all contained in the disc of radius $6$ centered at $0$. The isolated disc contains exactly one eigenvalue, counted with algebraic multiplicity. Hence there exists an eigenvalue $\lambda_{\mathrm{fast}}$ close to $-\eta$ such that
    \begin{equation}
        \left| \lambda_{\mathrm{fast}} - (-\eta) \right| \leq \max \left\{1, \norm{\lb 4 \one - \bJ \otimes \one \rb} \right\} \leq 6.
    \end{equation}
    We can again invoke the analytic implicit function theorem to uniformly bound the error term. Specifically, we have that
    \begin{equation}
        \lambda_{\mathrm{fast}} = -\eta + \cO \lb \frac{1}{\eta} \rb,
    \end{equation}
    with the constants in the remainder being independent of $L$.
    
    Together with the $2L$ distinct oscillatory eigenvalues and the middle eigenvalue, it gives $2L+2$ distinct eigenvalues of the reduced YES clock block. They exhaust its dimension, so that block is diagonalizable.
    This eigenvalue decays much faster than any of the other $2L+1$ eigenvalues.
    In summary, we have therefore shown that the real part of any of the $2L+2$ eigenvalues of our dissipative chain scales like
    \begin{equation}
        -\mathrm{Re} \lb \lambda \rb \in \Omega \lb \frac{1}{\eta L^3} \rb.
    \end{equation}

    In the following norm estimates, $\boldsymbol\psi(t)$ is the physical energy vector of \Cref{eq:sim_trans}, so that $\boldsymbol\psi(t)=e^{-\bC t}\boldsymbol\psi(0)$. The physical initial energy vector for the stated velocities $+1,-1$ has norm $\sqrt2$, so $\ket{\psi_0}=\boldsymbol\psi(0)/\sqrt2$ and $\norm{\boldsymbol\psi(t)}^2/\norm{\boldsymbol\psi(0)}^2=\norm{e^{-\bC t}\ket{\psi_0}}^2=E(t)/E(0)$. The conditional YES and NO states are normalized at $t=0$; in the YES-sector energy estimate below, $\bx,\bv$ and their rotated versions denote the corresponding normalized conditional vectors. These constant rescalings leave the linear evolution equations unchanged.

    Let us take a step back now. Ultimately, we want to show that the problem of deciding whether $\frac{\norm{\boldsymbol\psi(t)}^2}{\norm{\boldsymbol\psi(0)}^2}$, the squared norm of the state evolved under $-\bC$ relative to its initial value, is at least $2/3$ or at most $1/3$ is \BQP-complete.
    We will perform a series of decompositions of $\frac{\norm{\boldsymbol\psi(t)}^2}{\norm{\boldsymbol\psi(0)}^2}$ to simplify the analysis. First, recall that the dynamics can be neatly separated into the dynamics of two invariant subspaces which do not mix under $-\bC$: the YES sector and the NO sector. Thus,
    \begin{equation}
        \frac{\norm{\boldsymbol\psi(t)}^2}{\norm{\boldsymbol\psi(0)}^2} = p_{\mathrm{YES}} \norm{\boldsymbol\psi_{\mathrm{YES}}(t)}^2 + p_{\mathrm{NO}} \norm{\boldsymbol\psi_{\mathrm{NO}}(t)}^2, 
    \end{equation}
    where $p_{\mathrm{YES}}$ ($p_{\mathrm{NO}}$) denotes the probability mass of the initial state in the YES (NO) sector. Since the dynamics in the NO sector are unitary, we have that $\norm{\boldsymbol\psi_{\mathrm{NO}}(t)}^2 = 1$. To deal with the YES branch, we utilize the block form of $\norm{\boldsymbol\psi_{\mathrm{YES}}(t)}^2$ to decompose it as follows:
    \begin{equation}
    \begin{split}
        \norm{\boldsymbol\psi_{\mathrm{YES}}(t)}^2 &= \norm{\bv(t)}^2 + \norm{\bK^{1/2} \bT \bx(t)}^2 \\
        &= \norm{\frac{d\bx(t)}{dt}}^2 + \bx^\dagger(t) \bT^\dagger \bK \bT \bx(t) \\
        &= \norm{\frac{d\bx(t)}{dt}}^2 + \bx^\dagger(t) \bH_{\cl}' \bx(t) \\
        &\leq \norm{\frac{d\bx(t)}{dt}}^2 + \norm{\bx(t)}^2 \norm{ \bH_{\cl}'} \\
        &\leq \norm{\frac{d\bx(t)}{dt}}^2 + 6 \norm{\bx(t)}^2 \\
        &\leq 6 \lb \norm{\frac{d\bx(t)}{dt}}^2 + \norm{\bx(t)}^2 \rb \\
        &\leq 6 \lb \norm{\frac{d\tbx(t)}{dt}}^2 + \norm{\tbx(t)}^2 \rb.
    \end{split}
    \label{eq:yesNorm}
    \end{equation}
    The last line follows from the fact that $\bS$ is unitary, so $\norm{\bx} = \norm{\bS^\dagger \bx} = \norm{\tbx}$.
    Our previous calculations gave us the eigenmodes $\balpha^{(\lambda)}$ for the clock register in the frame rotated under $\bS$. All we need to do now is decompose $\tbx(t)$ in terms of the eigenmodes to be able to bound $\norm{\boldsymbol\psi_{\mathrm{YES}}(t)}^2$. The situation is complicated by the fact that $-\bC$ is not a normal matrix, so the eigenvectors do not form an orthonormal basis. Instead of working with $-\bC$, it will be easier to analyze $\bG_{\mathrm{ord:} 1}$ which corresponds to considering the evolution of the block vector $\lb \tbx(t), \tilde{\bv}(t) \rb^\top$. In fact, we can simplify the situation further by only focusing on the action of $\bG_{\mathrm{ord:} 1}$ on the $(2L+2)$-dimensional clock space, which is possible because $\bG_{\mathrm{ord:} 1}$ can be diagonalized individually on the clock space and the computational space (where the $W_l$'s are acting on). Define the following reduced version of $\bG_{\mathrm{ord:} 1}$ which describes its action on the clock register for the YES sector:
    \begin{equation}
        \tilde{\bG}_{\mathrm{ord:} 1}^{(\mathrm{YES})} :=
        \begin{pmatrix}
            \mathbf{0} & \one \\
            -\lb 4 \one - \bJ \rb & -\eta \bPi_{L+1}
        \end{pmatrix}
        \in \mathbb{C}^{(2L+2) \times (2L+2)}.
    \end{equation}
    We will now decompose $\tilde{\bG}_{\mathrm{ord:} 1}^{(\mathrm{YES})}$ in terms of left and right eigenvectors. 
    It is not difficult to see that the right eigenvectors of $\tilde{\bG}_{\mathrm{ord:} 1}^{(\mathrm{YES})}$ are of the form
    \begin{equation}
        \vec{\br}_{\lambda} = 
        \begin{pmatrix}
            \balpha^{(\lambda)} \\
            \lambda \balpha^{(\lambda)}
        \end{pmatrix},
    \end{equation}
    where the $\lambda$'s are the previously identified $2L+2$ eigenvalues. To find the left eigenvectors, one can start with the ansatz 
    \begin{equation}
        \vec{\boell}^\top_{\lambda, \mathrm{trial}} = 
        \begin{pmatrix}
            \lb {\balpha^{(\lambda)}} \rb^\top ,\quad
            \lambda \lb {\balpha^{(\lambda)}} \rb^\top
        \end{pmatrix},
    \end{equation}
    to realize that the first component actually needs to be of the form $- \balpha^\top \lb 4 \one - \bJ \rb$. From the quadratic eigenvalue equation, we also find that
    \begin{equation}
        - \balpha^\top \lb 4 \one - \bJ  \rb = \lambda^2 \balpha^\top + \lambda \eta \balpha^\top \bPi_{L+1} .
    \end{equation}
    Thus, the left eigenvectors of $\tilde{\bG}_{\mathrm{ord:} 1}^{(\mathrm{YES})}$ can be written as follows:
    \begin{equation}
        \vec{\boell}^\top_{\lambda} = 
        \begin{pmatrix}
            \lb {\balpha^{(\lambda)}} \rb^\top \lb \lambda \one + \eta \bPi_{L+1} \rb ,\quad 
            \lb {\balpha^{(\lambda)}} \rb^\top
        \end{pmatrix}.
    \end{equation}
    We therefore have the following spectral decomposition of $\tilde{\bG}_{\mathrm{ord:} 1}^{(\mathrm{YES})}$ in terms of left and right eigenvectors:
    \begin{equation}
        \tilde{\bG}_{\mathrm{ord:} 1}^{(\mathrm{YES})} = \sum_{j=1}^{2L+2} \lambda_j \frac{\vec{\br}_j \vec{\boell}_j^\top}{\vec{\boell}_j^\top \vec{\br}_j}.\label{eq:Gord}
    \end{equation}

    Next, we need to lower bound the inner-product-like quantities $\vec{\boell}_j^\top \vec{\br}_j$. Before doing that though, let us first recall what our input state for the whole computation is. To derive the second-order differential equation for the displacement vector $\bx$, we made the assumption that $\bx(0) = \bzero$. Additionally, we assume that only the oscillators associated with the clock register being in the state $\ket{1}$ have a nonzero velocity. In the $(2L+2)$-dimensional space of $\tilde{\bG}_{\mathrm{ord:} 1}$, our initial state can therefore be written as
    \begin{equation}
        \tilde{\phi}_0 =
        \begin{pmatrix}
            \mathbf{0} \\
            \ket{1}_{\mathrm{clock}}
        \end{pmatrix}.
    \end{equation}
    Note also that
    \begin{equation}
    \begin{split}
         \vec{\boell}_j^\top \vec{\br}_j &= \lb \balpha^{(j)} \rb^\top \lb \lambda_j \one + \eta \bPi_{L+1} \rb \balpha^{(j)} + \lambda_j \lb \balpha^{(j)} \rb^\top \balpha^{(j)} \\
         &= 2 \lambda_j \lb \balpha^{(j)} \rb^\top \balpha^{(j)} + \eta \lb \alpha^{(j)}_{ L+1} \rb^2.
    \end{split}
    \end{equation}

    Let us first look at $\vec{\boell}_j^\top \vec{\br}_j$ for the $2L$ oscillatory modes. Recall that 
    \begin{equation}
        \alpha_k^{(j, \pm)} = \sin (k \omega_{j, \pm}), \quad \omega_{j,\pm} = \frac{\pi j}{L+1} - \frac{1}{\eta} \frac{\sin  \lb \frac{\pi j}{L+1} \rb}{(L+1) \lambda_{j, \pm}^{(0)}} + \cO \lb \frac{\delta^2}{L+1} \rb,
    \end{equation}
    with the hidden constants independent of $L$.
    The above expansion of $\omega_{j,\pm}$ implies that the perturbation of $\omega_{j,\pm}$ around $\omega_{j,\pm}^{(0)}$ satisfies
    \begin{equation}
        k \lb \omega_{j,\pm} - \omega_{j,\pm}^{(0)} \rb \in \cO \lb \eta^{-1} \rb
    \end{equation}
    for all $k \in \{1,2, \dots, L+1 \}$. Thus, 
    \begin{equation}
        \alpha_k^{(j,\pm)}=\sin \lb k \omega_{j,\pm}^{(0)} \rb + \cO(\eta^{-1}),
    \end{equation}
    uniformly in both $k$ and $j$. Next, we use the fact that
    \begin{equation}
        \sum_{k=1}^{L+1}\sin^2 \lb k \omega_{j,\pm}^{(0)} \rb = \frac{L+1}{2},
    \end{equation}
    to obtain
    \begin{equation}
        \lb\balpha^{(j)}\rb^\top\balpha^{(j)} = \frac{L+1}{2} + \cO \lb \frac{L+1}{\eta} \rb.
    \end{equation}
    Also note that $\alpha_{L+1}^{(j)} \in \cO \lb \eta^{-1} \rb$.
    Putting it all together, we therefore have that for any of the $2L$ oscillatory modes,
    \begin{equation}
        \left|\vec{\boell}_j^\top\vec{\br}_j\right|\in\Theta(L),\qquad
        \norm{\balpha^{(j)}}\in\cO\lb\sqrt{L}\rb.
    \end{equation}
    Since $\vec{\boell}_j^\top\tilde{\phi}_0 = \alpha_1^{(j)} \in \cO(1)$ and $|\lambda_j| \in \cO(1)$, this implies that
    \begin{equation}
        \norm{\frac{\vec{\br}_j\vec{\boell}_j^\top}{\vec{\boell}_j^\top\vec{\br}_j}\tilde{\phi}_0}\in\cO(1).
    \end{equation}
    Next, let us consider the remaining two eigenvectors associated with $\lambda_{\mathrm{mid}}$ and $\lambda_{\mathrm{fast}}$. Recall that we did not obtain a power series expansion for $\omega_{\mathrm{mid}}$ (or $\omega_{\mathrm{fast}}$) as part of finding an expansion for $\lambda_{\mathrm{mid}}$ due to the square root singularity at $0$. In order to find the components of $\balpha^{(\mathrm{mid})}$ and $\balpha^{(\mathrm{fast})}$, we instead use the eigenvalue ansatz from Eq.~\eqref{eigval_ansatz} to find an expression for $z := e^{i \omega}$ in terms of $\lambda_{\mathrm{mid}}$ or $\lambda_{\mathrm{fast}}$ which we can then use to find an expression for the components of $\balpha^{(\mathrm{mid})}$ (or $\balpha^{(\mathrm{fast})}$). In the following, we will mostly drop the subscripts for notational simplicity but will write them out explicitly when needed for clarity. The first part of the analysis is the same for $\balpha^{(\mathrm{mid})}$ and $\balpha^{(\mathrm{fast})}$.
    From Eq.~\eqref{eigval_ansatz} we have that
    \begin{equation}
        \lambda^2 = - (4 - 2 \cos(\omega)) = - \lb 4 - \lb z + \frac{1}{z} \rb \rb.
    \end{equation}
    Multiplying both sides by $z$ and rearranging the terms, we find that 
    \begin{align}
        &z^2 - (4+\lambda^2) z + 1 = 0 \\
        &\implies z_{\pm} = \frac{4+\lambda^2}{2} \pm \sqrt{\lb \frac{4+\lambda^2}{2} \rb^2 - 1}.
    \end{align}
    It is straightforward to verify that the two roots are inverses of each other. Looking at the components of the ansatz for $\balpha$ in Eq.~\eqref{eigvec_ansatz}, we see that both roots result in the same vector (up to a global phase) since $\sin \lb k \omega \rb \propto \lb z^k - z^{-k}\rb$.
    In the following, we choose to work with the plus branch, $z_+ = \frac{4+\lambda^2}{2} + \sqrt{\lb \frac{4+\lambda^2}{2} \rb^2 - 1}$. 
    Our previous uniform eigenvalue estimates give
    \begin{equation}
        z_{+,\mathrm{fast}}=\eta^2 \lb 1 + \cO \lb \eta^{-2} \rb \rb,\qquad
        z_{+,\mathrm{mid}}= \lb 2+\sqrt{3} \rb \lb 1 + \cO \lb \eta^{-2} \rb \rb .
    \end{equation}
    We can safely drop the $\cO(\eta^{-2})$ remainder terms in the following discussion of the components of $\balpha^{(\mathrm{fast})}$ and $\balpha^{(\mathrm{mid})}$ since for our choice of $\eta = \eta_0 (L+1)^3$,
    \begin{equation}
        \lb 1 + \cO \lb \eta^{-2} \rb \rb^k = e^{\cO \lb \frac{L+1}{\eta^2} \rb}
    \end{equation}
    is uniformly bounded and approaches $1$ in the large $L$ limit.
    The components of $\balpha^{(\mathrm{fast})}$ are then given by
    \begin{equation}
    \begin{split}
         \alpha^{(\mathrm{fast})}_k &= \sin \lb k \, \omega_{\mathrm{fast}} \rb \\
         &= \frac{1}{2i} \lb e^{ik\omega_{\mathrm{fast}}} - e^{-ik\omega_{\mathrm{fast}}}\rb \\
         &= \frac{1}{2i} \lb z_{+, \mathrm{fast}}^k - z_{+, \mathrm{fast}}^{-k} \rb \\
         &\sim \frac{1}{2i} \eta^{2k}.
    \end{split}
    \end{equation}
    This means that 
    \begin{align}
        \left|\lb \balpha^{(\mathrm{fast})} \rb^\top \balpha^{(\mathrm{fast})}\right| = \left| \alpha^{(\mathrm{fast})}_{ L+1} \right|^2 + \cO \lb \eta^{4L+2} \rb \in \Theta \lb \eta^{4L+4} \rb .
    \end{align}
    Therefore,
    \begin{equation}
    \begin{split}
         \left| \vec{\boell}_{\mathrm{fast}}^\top \vec{\br}_{\mathrm{fast}} \right| &= \left| -2 \eta \lb \alpha^{(\mathrm{fast})}_{ L+1} \rb^2 + \eta \lb \alpha^{(\mathrm{fast})}_{ L+1} \rb^2 + \cO \lb \eta^{4L+3} \rb \right| \\
         &\in \Theta \lb \eta^{4L+5} \rb.
    \end{split}
    \end{equation}
    Additionally,
    \begin{equation}
        \norm{\balpha^{(\mathrm{fast})}} \in \cO \lb \sqrt{L} \, \eta^{2(L+1)} \rb,
    \end{equation}
    meaning that
    \begin{equation}
        \norm{\frac{\vec{\br}_{\mathrm{fast}} \vec{\boell}_{\mathrm{fast}}^\top}{\vec{\boell}_{\mathrm{fast}}^\top \vec{\br}_{\mathrm{fast}}} \tilde{\phi}_0} \in \cO \lb L \rb.
    \end{equation}
   
    A similar argument shows that 
    \begin{equation}
        \left|\lb \balpha^{(\rmid)} \rb^\top \balpha^{(\rmid)}\right| = \lb \alpha^{(\mathrm{mid})}_{ L+1} \rb^2 + \cO \lb \lb 2 + \sqrt{3} \rb^{2L+1} \rb \in \Theta \lb \lb 2 + \sqrt{3} \rb^{2L+2} \rb.
    \end{equation}
    Therefore,
    \begin{equation}
    \begin{split}
         \left| \vec{\boell}_{\mathrm{mid}}^\top \vec{\br}_{\mathrm{mid}} \right| &= \left| 2\lambda_{\mathrm{mid}} \lb \alpha^{(\mathrm{mid})}_{ L+1} \rb^2 + \eta \lb \alpha^{(\mathrm{mid})}_{ L+1} \rb^2 + \cO \lb \lb 2 + \sqrt{3} \rb^{2L+1} \rb \right| \\
         &\in \Theta \lb \eta \lb 2 + \sqrt{3} \rb^{2L+2} \rb.
    \end{split}
    \end{equation}
    Additionally,
    \begin{equation}
        \norm{\balpha^{(\mathrm{mid})}} \in \cO \lb \sqrt{L}  \lb 2 + \sqrt{3} \rb^{L+1} \rb,
    \end{equation}
    meaning that
    \begin{equation}
        \norm{\frac{\vec{\br}_{\mathrm{mid}} \vec{\boell}_{\mathrm{mid}}^\top}{\vec{\boell}_{\mathrm{mid}}^\top \vec{\br}_{\mathrm{mid}}} \tilde{\phi}_0} \in \cO \lb \frac{L}{\eta} \rb.
    \end{equation}

    Now we have all the ingredients needed to upper bound the norm of the part of the initial state that is in the YES sector. Specifically, from the use of~\eqref{eq:yesNorm} and the eigenvector form provided in~\eqref{eq:Gord} combined with the observation that the eigenvalues of the exponential of $\tilde{\mathbf{G}}_{\rm ord: 1}$ are of the form $e^{\lambda_j t}$
    \begin{align}
        \norm{\boldsymbol\psi_{\mathrm{YES}}(t)}^2 &\leq 6 \norm{\sum_{j=1}^{2L+2} e^{\lambda_j t} \frac{\vec{\br}_j \vec{\boell}_j^\top}{\vec{\boell}_j^\top \vec{\br}_j} \tilde{\phi}_0 }^2 \in \cO \lb L^2 e^{-\frac{c}{\eta L^3}t} \rb,
    \end{align}
    for some positive constant $c \in \mathbb{R}_{>0}$. Now recall that the overall squared norm of the evolved state relative to its initial value can be written as
    \begin{equation}
        \frac{\norm{\boldsymbol\psi(t)}^2}{\norm{\boldsymbol\psi(0)}^2} = p_{\mathrm{YES}} \norm{\boldsymbol\psi_{\mathrm{YES}}(t)}^2 + p_{\mathrm{NO}} \norm{\boldsymbol\psi_{\mathrm{NO}}(t)}^2. 
    \end{equation}
    For NO instances we have that
    \begin{equation}
        \frac{\norm{\boldsymbol\psi(t)}^2}{\norm{\boldsymbol\psi(0)}^2} \geq p_{\mathrm{NO}} \norm{\boldsymbol\psi_{\mathrm{NO}}(t)}^2 
        \geq 1 - \epsilon' \geq \frac{5}{6}.
    \end{equation}
    For YES instances on the other hand we have that
    \begin{align}
        \frac{\norm{\boldsymbol\psi(t)}^2}{\norm{\boldsymbol\psi(0)}^2} &\leq (1-\epsilon') \norm{\boldsymbol\psi_{\mathrm{YES}}(t)}^2 + \epsilon' \norm{\boldsymbol\psi_{\mathrm{NO}}(t)}^2 \\
        &\leq \frac{5}{6}  \norm{\boldsymbol\psi_{\mathrm{YES}}(t)}^2 + \frac{1}{6}.
    \end{align}
    It therefore suffices to simulate for time $T \in  \cO \lb \eta L^3 \log(L) \rb \subseteq \cO \lb \poly(n) \rb$, such that $\norm{\boldsymbol\psi_{\mathrm{YES}}(t)}^2 \leq 1/5$, in order to distinguish a YES from a NO instance with a constant precision estimate of $\frac{\norm{\boldsymbol\psi(T)}^2}{\norm{\boldsymbol\psi(0)}^2}$. Here we used that $\eta = \eta_0 (L+1)^3 \in \cO \lb \poly(n) \rb$ and the fact that $L \in \cO \lb \poly(n) \rb$ by construction.

    To show containment of this problem in \BQP, we can use LCHS as discussed in the main text to simulate the evolution operator $e^{-\bC T}$. 
    Recall that the initial conditions for the oscillators are such that all masses are initially in their rest positions, the first mass has velocity $+1$, the second mass has velocity $-1$ and all other masses do not have any initial velocity. This corresponds to the following easy-to-prepare initial state:
    \begin{equation}
        \ket{\psi_0} = 
        \begin{pmatrix}
             \ket{1}_{\mathrm{clock}}\ket{0}^{\otimes q}_{\mathrm{comp}} \ket{-}_{\mathrm{anc}} \\
        \mathbf{0}
        \end{pmatrix}.
    \end{equation}
    Since $T \in \cO \lb  \poly(n)\rb$ and $\bC$ can be block-encoded using $\cO \lb \poly(n) \rb$ gates with block-encoding constant scaling like $\eta \in \cO \lb \poly(n) \rb$, the overall number of gates scales like $\cO \lb \poly(n) \rb$ and thus the problem is in \BQP.

\end{proof}

\end{document}